\documentclass[11pt,a4paper]{article}
\usepackage{url}

\usepackage{fullpage}
\usepackage{amsthm,amssymb}
\usepackage{amsmath}
\usepackage{amstext}
\usepackage{enumerate}

\usepackage[all=normal,bibliography=tight]{savetrees}

\usepackage{graphicx}
\usepackage{xspace}
\usepackage{comment}
\usepackage{todonotes}
\usepackage{caption}
\usepackage{subcaption}

\graphicspath{{figures/}}

\newcommand{\executeiffilenewer}[3]{%
\ifnum\pdfstrcmp{\pdffilemoddate{#1}}%
{\pdffilemoddate{#2}}>0%
{\immediate\write18{#3}}\fi%
}
\newcommand{\includesvg}[2]{%
\executeiffilenewer{figures/#2.svg}{figures/#2.pdf}%
{inkscape -z -D --file=figures/#2.svg %
--export-pdf=figures/#2.pdf}%
{\includegraphics[width=#1]{figures/#2.pdf}}}%

\newcommand{\svg}[2]{\includesvg{#1}{#2}}

\newcommand{\etal}{et al.~}

\newcommand{\ie}{i.e.~}
\newcommand{\eg}{e.g.~}
\newcommand{\wloge}{w.l.o.g.~}
\newcommand{\mc}{\mathcal}
\newcommand{\poly}{\mathrm{poly}}
\newcommand{\dist}{\mathrm{dist}}

\DeclareSymbolFont{AMSb}{U}{msb}{m}{n}
\DeclareSymbolFontAlphabet{\mathbb}{AMSb}

\newtheorem{theorem}{Theorem}[section]

\newtheorem{claim}[theorem]{Claim}
\newtheorem{lemma}[theorem]{Lemma}
\newtheorem{corollary}[theorem]{Corollary}

\newtheorem{definition}[theorem]{Definition}
\newtheorem{reduction}{Reduction Rule}[section]
\newtheorem{example2}[theorem]{Example}

\newcommand{\knip}[1]{}

\newcommand{\pST}{\textsc{Planar Steiner Tree}\xspace}
\newcommand{\problemST}{\textsc{Steiner Tree}\xspace}
\newcommand{\pSF}{\textsc{Planar Steiner Forest}\xspace}
\newcommand{\problemSF}{\textsc{Steiner Forest}\xspace}

\newcommand{\problemMWC}{\textsc{Multiway Cut}\xspace}
\newcommand{\pSAT}{$3$-\textsc{SAT}\xspace}
\newcommand{\rSAT}{$R$-\textsc{SAT}\xspace}
\newcommand{\nG}{\ensuremath{\hat{G}}}

\newcommand{\intr}{\mathrm{int}}
\newcommand{\eps}{\varepsilon}
\newcommand{\prm}{\ensuremath{\partial}}
\newcommand{\tree}{T}
\newcommand{\terms}{S}
\newcommand{\termpairs}{\mathcal{S}}

\newcommand{\Tapx}{\ensuremath{\tree_{apx}}}
\newcommand{\Oh}{\ensuremath{\mathcal{O}}}
\newcommand{\Bb}{\ensuremath{\mathcal{B}}}
\newcommand{\cqed}{\renewcommand{\qedsymbol}{$\lrcorner$}}
\newcommand{\MR}{\mathit{MR}}

\newcommand{\nice}{\ensuremath{\tau}}

\newcommand{\weismb}{w}
\newcommand{\wei}[1]{\weismb(#1)}
\newcommand{\portalbound}{\theta}
\newcommand{\plane}{\Pi}
\newcommand{\Z}{\mathbb{Z}}
\newcommand{\mou}[2]{(#1 \wedge #2)}
\newcommand{\xP}{\mathfrak{P}}
\newcommand{\regs}{\mathcal{R}}
\newcommand{\regspm}[1]{\mathcal{R}^{#1}}
\newcommand{\proj}{\pi}
\newcommand{\extproj}{\tilde{\proj}}
\newcommand{\Proj}{\xi}
\newcommand{\mouth}{\kappa}
\newcommand{\porset}{\mathbf{P}}
\newcommand{\coreface}{f_{\mathtt{core}}}
\newcommand{\closeB}{B_{\textrm{close}}}

\newcommand{\pemwcname}{\textsc{Planar Edge Multiway Cut}}
\newcommand{\pemwc}{\textsc{PEMwC}}
\def\case #1{{\bf Case}\ {\it #1:}}

\newcommand{\defproblemnoparam}[3]{
  \vspace{1mm}
\noindent\fbox{
  \begin{minipage}{0.95\textwidth}
  #1 \\
  {\bf{Input:}} #2  \\
  {\bf{Task:}} #3
  \end{minipage}
  }
  \vspace{1mm}
}

\begin{document}
\pagenumbering{gobble}
\thispagestyle{empty}

\date{}
\title{Network Sparsification for Steiner Problems on Planar and Bounded-Genus Graphs%
\thanks{An extended abstract of this work has appeared at FOCS 2014.
The research leading to these results has received funding from the European Research Council under the European Union's Seventh Framework Programme (FP/2007-2013) / ERC Grant Agreement n.~267959
(Marcin Pilipczuk, Micha\l{} Pilipczuk), ERC Grant Agreement n.~259515 (Marcin Pilipczuk, Piotr Sankowski) and Foundation for Polish Science (Marcin Pilipczuk, Piotr Sankowski) and
Polish funds for years 2011-2014 for co-financed international projects.}}
\author{%
Marcin Pilipczuk\footnote{Institute of Informatics, University of Warsaw, Poland, \texttt{malcin@mimuw.edu.pl}. Part of the research was done while the author was at University of Bergen, Norway.} \and
Micha{\l} Pilipczuk\footnote{Institute of Informatics, University of Warsaw, Poland, \texttt{michal.pilipczuk@mimuw.edu.pl}. Research was done while the author was at University of Bergen, Norway.} \and
Piotr Sankowski\footnote{Institute of Informatics, University of Warsaw, Poland, \texttt{sank@mimuw.edu.pl}.} \and
Erik Jan van Leeuwen\footnote{Department of Information and Computing Sciences, Utrecht University, The Netherlands, \texttt{e.j.vanleeuwen@uu.nl}. Part of the research was done while the author was at Sapienza University of Rome, Italy and at Max-Planck Institut f\"{u}r Informatik, Saarland Informatics Campus, Saarbr\"{u}cken, Germany.}
}
\maketitle
%\vspace{-0.5cm}
\begin{abstract}
We propose polynomial-time algorithms that sparsify planar and bounded-genus graphs while preserving optimal or near-optimal solutions to Steiner problems.
%We propose polynomial-time algorithms that sparsify Steiner problems on planar and bounded-genus graphs.
Our main contribution is a polynomial-time algorithm that, given an unweighted graph $G$ embedded on a surface of genus $g$ and a designated face $f$ bounded by a simple cycle of length $k$, uncovers a set $F \subseteq E(G)$ of size polynomial in $g$ and $k$ that contains an optimal Steiner tree for \emph{any} set of terminals that is a subset of the vertices of $f$.

We apply this general theorem to prove that:
\begin{itemize}
\item given an unweighted graph $G$ embedded on a surface of genus $g$ and a terminal set $\terms \subseteq V(G)$, one can in polynomial time find a set $F \subseteq E(G)$ that contains an optimal Steiner tree $T$ for $\terms$ and that has size polynomial in $g$ and $|E(T)|$;
%\item an analogous result holds for the \textsc{Steiner Forest} problem;
\item an analogous result holds for an optimal Steiner forest for a set $\mc{\terms}$ of terminal pairs;
\item given an unweighted planar graph $G$ and a terminal set $\terms \subseteq V(G)$, one can in polynomial time
find a set $F \subseteq E(G)$ that contains an optimal (edge) multiway cut $C$ separating $\terms$ (i.e., a cutset that intersects any path with endpoints in different terminals from $\terms$) and that has size polynomial in $|C|$.
\end{itemize}
In the language of parameterized complexity, these results imply the first polynomial kernels for \textsc{Steiner Tree} and \textsc{Steiner Forest} on planar and bounded-genus graphs (parameterized by the size of the tree and forest, respectively) and for \textsc{(Edge) Multiway Cut} on planar graphs (parameterized by the size of the cutset).
%Polynomial kernels for \problemST and similar ``subset'' problems were called for in [Demaine, Hajiaghayi, Computer J., 2008] as part of the quest to widen the reach of the broad theory of bidimensionality ([Demaine \etaln, JACM 2005], [Fomin \etaln, SODA 2010]).
%\problemST and similar ``subset'' problems were identified in [Demaine, Hajiaghayi, Computer J., 2008] as important to the quest to widen the reach of the theory of bidimensionality ([Demaine~\etaln, JACM 2005], [Fomin~\etaln, SODA 2010]). Therefore, our results can be seen as a leap forward to achieve this broader goal.

Additionally, we obtain a weighted variant of our main contribution: a polynomial-time algorithm that, given an edge-weighted plane graph $G$, a designated face $f$ bounded by a simple cycle of weight $\wei{f}$, and an accuracy parameter $\eps > 0$, uncovers a set $F \subseteq E(G)$ of total weight at most $\poly(\eps^{-1}) \wei{f}$ that, for any set of terminal pairs that lie on $f$, contains a Steiner forest within additive error $\eps \wei{f}$ from the optimal Steiner forest.
%This result deepens the understanding of the recent framework of approximation schemes for network design problems on planar graphs
%([Klein, SICOMP 2008], [Borradaile, Klein, Mathieu, ACM TALG 2009], and later works)
%by explaining the structure of the solution space within a brick of the so-called \emph{mortar graph}, the central notion of this framework.
\end{abstract}

\clearpage

\pagenumbering{arabic}

%!TEX root = pst-kernel.tex

\section{Introduction}\label{sec:intro}
Preprocessing algorithms seek out and remove chunks of instances of hard problems that are irrelevant or easy to resolve. The strongest preprocessing algorithms reduce instances to the point that even an exponential-time brute-force algorithm can solve the remaining instance within limited time. The power of many preprocessing algorithms can be explained through the relatively recent framework of kernelization~\cite{downey-fellows:book,niedermeier:book}. In this framework, each problem instance $I$ has an associated parameter $k(I)$, often the desired or optimal size of a solution to the instance. %(the standard parameter). 
Then a \emph{kernel} is a polynomial-time algorithm that preprocesses the instance so that its size shrinks to at most $g(k(I))$, for some computable function $g$. If $g$ is a polynomial, then we call it a \emph{polynomial kernel}.

The ability to measure the strength of a kernel through the function $g$ has led to a concerted research effort to determine, for each problem, the function $g$ of smallest order that can be attained by a kernel for it.
Initial insight into this function, in particular a proof of its existence, is usually given by a \emph{parameterized algorithm}: an algorithm that solves an instance $I$ in time $g(k(I)) \cdot |I|^{O(1)}$. 
%The question whether a kernel even exists is usually answered by giving a \emph{parameterized algorithm}~\cite{downey-fellows:book,niedermeier:book}: an algorithm that solves an instance $I$ in time $g(k(I)) \cdot |I|^{O(1)}$. 
Such an algorithm implies a kernel with the same function $g$,
while, if the considered problem is decidable, then any kernel immediately gives a parameterized algorithm as well~\cite{downey-fellows:book,niedermeier:book}.
     %: if $|I| < g(k(I))$, then return the instance, and otherwise the parameterized algorithm actually takes polynomial time. 
%The question whether a kernel, and thus such a function $g$, even exists is usually answered by giving a \emph{parameterized algorithm}~\cite{downey-fellows:book,niedermeier:book}: an algorithm that solves an instance $I$ in time $g(k(I)) \cdot n^{O(1)}$. 
However, if the problem is NP-hard, then this approach can only yield a kernel of superpolynomial size, unless P$=$NP. Therefore, different insights are needed to find the function $g$ of smallest order, and in particular to find a polynomial kernel. This fact, combined with the discovery that for many problems the existence of a polynomial kernel would imply a collapse in the polynomial hierarchy~\cite{hans:lb,lance:lb,andrew:lb}, has recently led to a spike in research on polynomial kernels.

A focal point of research into polynomial kernels are problems on planar graphs. Many problems that on general graphs have no polynomial kernel or even no kernel at all, possess a polynomial kernel on planar graphs.
%With few exceptions~\cite{planar-capdom}, all problems considered thus far on planar graphs have a parameterized algorithm, and thus a kernel. Moreover, many problems actually possess a polynomial kernel on planar graphs (the most notable exceptions are \textsc{Longest Path} and similar `pattern-search' problems~\cite{hans:lb}).
The existence of almost all of these polynomial kernels can be explained from the theory of bidimensionality~\cite{metakernel,bidim:jacm,bidim:kernels}. The core assumption behind this theory is that the considered problem is {\em{bidimensional}}: informally speaking, the solution to an instance must be dense in the input graph. However, this assumption clearly fails for a lot of problems, which has led to gaps in our understanding of the power of preprocessing algorithms for planar graphs. In their survey, Demaine and Hajiaghayi \cite{dh:enc,dh:cj} pointed out `subset' problems, in particular \problemST, as an important research goal
in the quest to generalize the theory of bidimensionality.

In this paper, we pick up this line of research and positively resolve the question to the existence of a polynomial kernel on planar graphs for three well-known `subset' problems: \problemST, \problemSF, and \problemMWC. 
We remark that the theory of bidimensionality does not apply to any of these three problems, and that for the first two problems a polynomial kernel on general graphs is unlikely to exist~\cite{dom:ids} and for the third the existence of a polynomial kernel on general graphs is a major open problem~\cite{worker2013-opl,dagstuhl2012,stefan-magnus2}. %Additionally, the theory of bidimensionality does not apply to any of these three problems. 
All kernelization results in this paper are a consequence of a single, generic sparsification algorithm for Steiner trees in planar graphs, which is of independent interest. This sparsification algorithm extends to edge-weighted planar graphs, and we demonstrate its impact on approximation algorithms for problems on planar graphs, in particular on the EPTAS for \problemST on planar graphs~\cite{klein:planar-st-eptas}.

\subsection{Reading guide}
The paper presents three views on our results, with increasing level of detail. In the first view, Section~\ref{ss-intro:results} states our results, and briefly describes our techniques and how they (vastly) differ from previous papers on planar graph problems. We discuss possible limitations and extensions in Section~\ref{ss-intro:discussion}. The second view (in Section~\ref{sec:overview}) provides a rich overview of the proofs of our results. Finally, the third view (Sections~\ref{sec:preliminaries} through~\ref{sec:sf-lb}) gives full and detailed proofs.

\subsection{Results} \label{ss-intro:results}
We present an overview of the three major results that make up this paper. First, we describe the generic sparsification algorithm for Steiner trees in planar graphs. Second, we show how this sparsification algorithm powers the kernelization results in this paper. Third, we exhibit the extension of the sparsification algorithm to edge-weighted planar graphs, and its implications for approximation algorithms on planar graphs.

\medskip
\noindent\textbf{The Main Theorem.} In our main contribution, we characterize the behavior of Steiner trees in bricks. In our work, a \emph{brick} is simply a connected plane graph $B$ with one designated face formed by a simple cycle~$\prm B$, which \wloge is the outer (infinite) face of the plane drawing of $B$, and called the \emph{perimeter of $B$}. 
Recall that a \emph{Steiner tree} of a graph $G$ is a tree in $G$ that contains a given set $\terms \subseteq V(G)$ (called \emph{terminals}).
We also say that the Steiner tree \emph{connects} $\terms$.
In the unweighted setting, a Steiner tree $\tree$ that connects $\terms$ is \emph{optimal} if every Steiner tree that connects $\terms$ has at least as many edges as $\tree$.
We apply our characterization of Steiner trees in bricks to obtain the following sparsification algorithm:

\begin{theorem}[Main Theorem]\label{thm:main}
Let $B$ be a brick. Then one can find in $\Oh(|\prm B|^{142} \cdot |V(B)|)$ time a subgraph $H$ of $B$ such that
\begin{enumerate}
\renewcommand{\theenumi}{\roman{enumi}}
\renewcommand{\labelenumi}{(\theenumi)}
\item $\prm B \subseteq H$,
\item $|E(H)| = \Oh(|\prm B|^{142})$, and
\item for every set $\terms \subseteq V(\prm B)$, $H$ contains some optimal Steiner tree in $B$ that connects $\terms$.
\end{enumerate}
\end{theorem}
%We also prove an analogue of this theorem for graphs of bounded genus, with a polynomial dependence on the genus in the size bound --- see Theorem~\ref{thm:main-genus}.

The result of Theorem~\ref{thm:main} is stronger than just a polynomial kernel, because the graph $H$ contains an optimal Steiner tree for \emph{any} terminal set that is a subset of the brick's perimeter. The result fits in a line of sparsification algorithms that reduce an instance and enable fast queries or computations (unknown at the current time) on the original instance, such as sparsification algorithms that approximately preserve vertex distances (so-called graph spanners)~\cite{spanners-origin,spanners-greedy}, that preserve connectivity~\cite{ibaraki}, or that conserve flows and cuts~\cite{mimicking,benczur-karger,bss09,krauthgamer}. 
Such sparsification algorithms are a common tool in, among others, dynamic graph algorithms~\cite{Eppstein97}, especially for planar graphs~\cite{Eppstein96,EppsteinGIS98,Diks07,KleinS98,Subramanian93}.
%To the best of our knowledge, however, this is the first sparsification algorithm that preserves Steiner trees.

We also emphasise that the purely combinatorial (non-algorithmic) statement of Theorem~\ref{thm:main}, which asserts the existence of a subgraph $H$ that has property (iii) and polynomial size, is nontrivial and, in our opinion, interesting on its own.
A naive construction of a subgraph $H$ that has property (iii) would mark an optimal Steiner tree for each set $\terms \subseteq V(\prm B)$. Combined with the observation that any optimal Steiner tree of a set $\terms \subseteq V(\prm B)$ has size at most $|\prm B|$ (as $\prm B$ is a Steiner tree that connects $\terms$), we obtain a bound on the size of $H$ of 
%Observe that a naive construction of a subgraph $H$ with property (iii) would mark an optimal Steiner tree for each choice of the terminal set on the perimeter, leading to a bound on the size of $H$ of 
$|\prm B|\cdot 2^{|\prm B|}$.
%as opposed to the polynomial bound of Theorem~\ref{thm:main}.
%much worse than the polynomial bound of Theorem~\ref{thm:main}.
The polynomial bound of Theorem~\ref{thm:main} presents a significant improvement over this naive bound.
% than the ones used in previous works. %\cite{klein:planar-st-eptas,ours-stacs}.
%and, consequently, all our further arguments in the proof of Theorem~\ref{thm:main} are based on completely different ideas.
%In both previous works, the brick is cut into so-called strips, and then each strip is cut with a `perpendicular' column.
%In~\cite{ours-stacs} the algorithm quite closely follows the approach of the EPTAS~\cite{klein:planar-st-eptas}, cutting the brick using shortest paths into strips, and then cutting a strip with a `perpendicular' column.

The starting point of our work is the observation that an optimal Steiner tree for some choice of terminals on the perimeter decomposes the brick into smaller subbricks (see Figure~\ref{fig-intro:partition}), on which we can subsequently recurse.
However, the depth of the recursion may become too large if for \emph{any} optimal Steiner tree for \emph{any} choice of terminals on the perimeter, there are one or two subbricks that have perimeter almost equal to $|\prm B|$, as in Figure~\ref{fig-intro:double-comet}. Therefore, the main part of our proof aims to understand the structure of the brick when this happens.
%Therefore, the main part of our proof aims to understand the structure of the brick when \emph{no} such optimal Steiner tree can be used to decompose the brick.
%Intuitively, this happens when among the subbricks that are `cut out' by any optimal Steiner tree for any terminal set,
%there are always one or two subbricks that have perimeter almost equal to $|\prm B|$, as in Figure~\ref{fig-intro:double-comet}.
In this case, we show that any optimal tree for any terminal set avoids a well-defined inside of the brick (the core), and give an algorithm to find it.
Using the core and a deep topological analysis of the brick, we then find a cycle $C$ of length $\Oh(|\prm B|)$ that lies close to the perimeter of $B$ and that separates the core from all vertices of degree at least three of some optimal solution, for any set of terminals (see Figure~\ref{fig-intro:pineapple}).
Therefore, for any set of terminals, there is some optimal solution whose intersection with the area inside $C$ is a disjoint union of shortest paths, and thus we can sparsify this area by keeping a shortest path inside $C$ between any pair of vertices of $C$.
After this, we decompose the area between $C$ and the perimeter of the brick into several smaller pieces, which we recursively sparsify.
% and we also sparsify the area inside the cycle.
Using an inductive argument, we show that this yields the polynomial bound on the size of the returned graph $H$.

\begin{figure}[tb]
\centering
\begin{subfigure}{.33\textwidth}
\centering
\includegraphics[width=.9\linewidth]{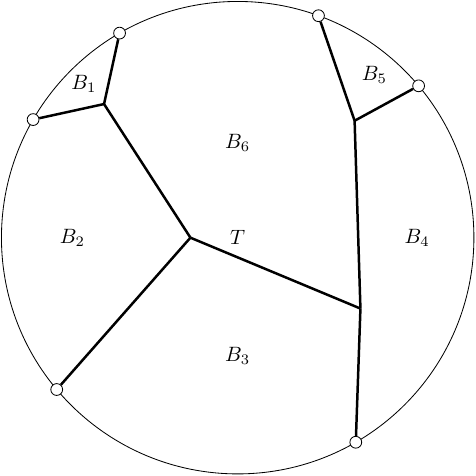}
\caption{}
\label{fig-intro:partition}
\end{subfigure}%
\begin{subfigure}{.33\textwidth}
 \centering
 \includegraphics[width=.9\linewidth]{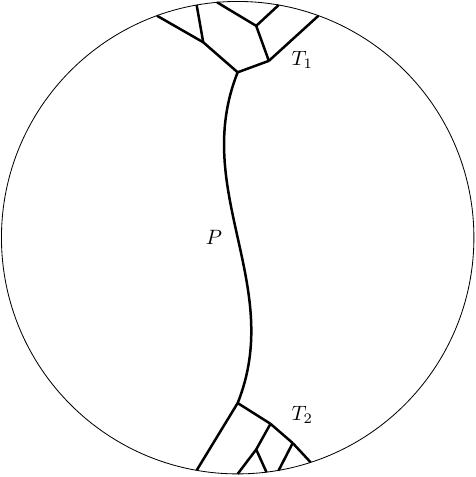}
 \caption{}
 \label{fig-intro:double-comet}
\end{subfigure}%
\begin{subfigure}{.33\textwidth}
 \centering
 \includegraphics[width=.9\linewidth]{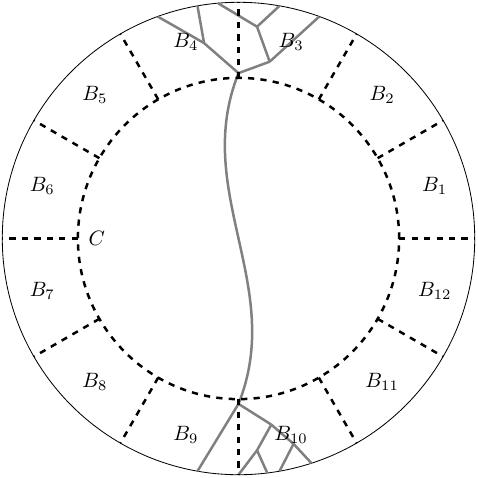}
 \caption{}
 \label{fig-intro:pineapple}
\end{subfigure}
\caption{(a) shows an optimal Steiner tree $T$ and how it partitions
the brick $B$ into smaller bricks $B_1,\ldots,B_r$.
  (b) shows an optimal Steiner tree that connects a set of vertices on the perimeter of $B$ and that consists of two small trees $T_{1},T_{2}$ that are connected by a long path~$P$;
note that both bricks neighbouring $P$ may have perimeter very close to $|\prm B|$.
  (c) shows a cycle $C$ that (in particular) hides the small trees $T_1,T_2$ in the ring between $C$ and $\prm B$, and a subsequent decomposition of $B$ into smaller bricks.}
\label{fig-intro:intro-comet}
\end{figure}

We give a more detailed overview of the proof of Theorem~\ref{thm:main} in Section~\ref{sec:overview}.
The full proof is contained in Sections~\ref{sec:preliminaries}--\ref{sec:dp}. 
We also prove an analogue of Theorem~\ref{thm:main} for graphs of bounded genus, with a polynomial dependence on the genus in the size bound. This analogue is sketched in Section~\ref{ss-over:genus} and presented in full in Section~\ref{sec:genus}.

The approach that we take in this paper is very different from previous approaches to tackle problems on planar graphs or on bricks. In particular, our ideas are disjoint from those developed in both an EPTAS~\cite{klein:planar-st-eptas} and a subexponential-time parameterized algorithm~\cite{ours-stacs} for \pST{}. In those works, a brick was cut into so-called strips and then each strip was cut with a `perpendicular column'. Therefore, already our starting observation (to use an optimal Steiner tree to decompose the brick) seems novel. Moreover, to the best of our knowledge, there is no work that aims to understand the behavior of a Steiner tree in a brick when all optimal Steiner trees leave one or two large subbricks (as in Figure~\ref{fig-intro:double-comet}). Most of our paper is devoted to developing the tools and techniques to understand this case. We also stress that we do not employ any techniques used in the theory of bidimensionality. In particular, we do not use any tools from Graph Minors theory, such as the Excluded Grid Theorem~\cite{dh:linear-grid,rst:grid} --- the engine of the theory of bidimensionality.

%Observe that the proof of Theorem~\ref{thm:main} is largely disjoint from previous approaches to tackle problems on planar graphs. In particular, we do not use any tools from Graph Minors theory, such as the Excluded Grid Theorem~\cite{dh:linear-grid,rst:grid} --- the engine of the theory of bidimensionality. %Instead, we use a divide-and-conquer principle: we decompose the brick $B$ into smaller pieces, which we recursively sparsify.
%A divide-and-conquer approach for \pST{} on bricks was explored in both an EPTAS~\cite{klein:planar-st-eptas} and a subexponential-time parameterized algorithm~\cite{ours-stacs}. In both works, a brick was cut into so-called strips and then each strip was cut with a `perpendicular column'. Although at some level we still follow a divide-and-conquer approach, we employ completely different decomposition methods.

%Note that the above approach is also disparate from the theory of bidimensionality. In particular, we do not use any tools from Graph Minors theory, such as the Excluded Grid Theorem~\cite{dh:linear-grid,rst:grid}, which are the engine of bidimensionality theory.

\medskip
\noindent\textbf{Applications of Theorem~\ref{thm:main}.}
We give three applications of Theorem~\ref{thm:main}. For each application, we state the result and its significance, and give an intuition of the proof. More detailed sketches of the proofs are provided in Section~\ref{ss-over:applications}, and for details we refer to Section~\ref{sec:applications}.

For the first application of Theorem~\ref{thm:main}, we consider \problemST{}. For this problem, a polynomial kernel on general graphs would imply a collapse of the polynomial hierarchy~\cite{dom:ids}. At the same time, the core assumption of bidimensionality theory fails, and whether a polynomial kernel exists for \textsc{Steiner Tree} on planar graphs was hitherto unknown. Using Theorem~\ref{thm:main}, we can resolve the existence of a polynomial kernel for \problemST{} on planar graphs.

\begin{theorem}\label{thm:pst-intro}
Given a $\pST{}$ instance $(G,\terms)$,
one can in
$\Oh(k_{OPT}^{142} |G|)$ time find a set $F \subseteq E(G)$
of $\Oh(k_{OPT}^{142})$ edges that contains an optimal Steiner tree
connecting $\terms$ in $G$, where $k_{OPT}$ is the size of an optimal
Steiner tree.
\end{theorem}

We emphasise two aspects of Theorem~\ref{thm:pst-intro}. First, the proposed algorithm \emph{does not} need to be given an optimal solution nor its size, even though the running time and output size of the algorithm are polynomial in the size of an optimal solution. Second, the running time of the algorithm can be bounded by $\Oh(|G|^2)$: if $|G|$ is smaller than the promised kernel bound, then the algorithm may simply return the input graph without any modification. Similar remarks hold also for the second and third applications of Theorem~\ref{thm:main} that we present later.

Intuitively, Theorem~\ref{thm:pst-intro} is almost a direct consequence of Theorem~\ref{thm:main}: we compute a $2$-approximation to the optimal Steiner tree, cut the plane open along it, and then make the resulting cycle the outer face (see Figure~\ref{fig-intro:cutopen}), as done in the EPTAS for this problem~\cite{klein:planar-st-eptas}. Since all terminals lie on the outer face of the cut-open graph, we apply Theorem~\ref{thm:main} to it, and project the resulting graph $H$ back to the original graph.

\begin{figure}[bt]
\centering
\includegraphics[width=.6\linewidth]{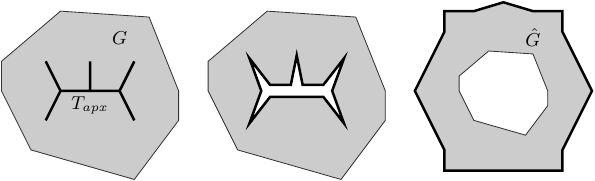}
\caption{The process of cutting open the graph $G$ along the approximate Steiner tree.}
\label{fig-intro:cutopen}
\end{figure}

For the second application of Theorem~\ref{thm:main}, we modify the approach of Theorem~\ref{thm:pst-intro} for the closely related \problemSF{} problem on planar graphs. Recall that a \emph{Steiner forest} that \emph{connects} a family $\mathcal{\terms} \subseteq V(G) \times V(G)$ of terminal pairs in a graph $G$ is a forest in $G$ such that both vertices of each pair in $\mc{\terms}$  are contained in the same connected component of the forest.

\begin{theorem}\label{thm:psf-intro}
Given a $\pSF{}$ instance $(G,\termpairs)$,
one can in
$\Oh(k_{OPT}^{710} |G|)$ time find a set $F \subseteq E(G)$
of $\Oh(k_{OPT}^{710})$ edges that contains an optimal Steiner forest
connecting $\termpairs$ in $G$, where $k_{OPT}$ is the size of an optimal
Steiner forest.
\end{theorem}
Using the analogue of Theorem~\ref{thm:main} for bounded-genus graphs, we extend Theorems~\ref{thm:pst-intro} and~\ref{thm:psf-intro} to obtain a polynomial kernel for \problemST{} and even \problemSF{} on such graphs (see Section~\ref{sec:genus}).
Here, we assume that we are given an embedding of the input graph into a surface of genus $g$ such that the interior of each face is homeomorphic to an open disc.

For the third application of Theorem~\ref{thm:main}, we consider \textsc{Edge Multiway Cut} on planar graphs.
Recall that an \emph{edge multiway cut}\footnote{In the approximation algorithms literature, the term \emph{multiway cut} usually refers to an edge cut, i.e., a subset of edges of the graph, and the node-deletion variants of the problem are often much harder. However, from the point of view of parameterized complexity, there is usually little or no difference between edge- and node-deletion variants of cut problems,
  and hence one often considers the (more general) node-deletion variant as the `default one'.
  To avoid confusion, in this work we always explicitly state that we consider the edge-deletion variant.}
 in a graph $G$ is a set $X \subseteq E(G)$ such that no two vertices of a given set $\terms \subseteq V(G)$ are in the same component of $G \setminus X$.
A recent breakthrough in the application of matroid theory to kernelization problems~\cite{stefan-magnus1,stefan-magnus2} led to the discovery of a polynomial kernel for \textsc{Multiway Cut} on general graphs with a constant number of terminals. It is a major open question whether this problem has a polynomial kernel for an arbitrary number of terminals~\cite{worker2013-opl,dagstuhl2012,stefan-magnus2}.
Here, we show that such a polynomial kernel does exist for \textsc{Edge Multiway Cut} on planar graphs.

\begin{theorem}\label{thm:emwc-intro}
Given a \pemwcname{} instance $(G,\terms)$,
one can in polynomial time find a set $F \subseteq E(G)$
of $\Oh(k_{OPT}^{568})$ edges that contains an optimal solution
to $(G,\terms)$, where $k_{OPT}$ is the size of this optimal solution.
\end{theorem}
The proof of this theorem is based on a well-known relation between a multiway cut in a planar graph $G$ and a Steiner tree in the dual of $G$. Hence, we apply Theorem~\ref{thm:main} to a cut-open dual of the input graph. However, to bound the diameter of the initial brick, we need to bound the diameter
of the dual of $G$. 
To this end, we show that edges are irrelevant to the problem if they are `far' from some carefully chosen, laminar family of minimal cuts. Such edges may then be contracted safely, leading to the needed bound.

We note that in contrast to the work on polynomial kernels for \textsc{Multiway Cut} mentioned before~\cite{stefan-magnus1,stefan-magnus2}, we do not rely on matroid theory.
%To this end, we employ the framework of important separators~\cite{Marx06} and show that an edge that is `far' from some carefully chosen laminar family of important separators is irrelevant for the problem, and may be safely contracted.

%We also note that the running times of the algorithms of Theorems~\ref{thm:main},~\ref{thm:pst-intro} and~\ref{thm:psf-intro} can be bounded by $\Oh(n^2)$: if $n$ is smaller than the promised kernel bound, then the algorithm may simply return the input graph without any modification. In Theorem~\ref{thm:emwc-intro}, a similar trick may be applied, but the running time is superquadratic due to multiple computations of minimum cuts in the graph, and we refrain from computing the running time explicitly for sake of clarity.

As an immediate consequence of Theorem~\ref{thm:pst-intro} and Theorem~\ref{thm:emwc-intro}, we observe that by plugging the kernels promised by these theorems into the algorithms of Tazari~\cite{tazari:mfcs10}
for \pST{} or its modification for \pemwcname{} (provided for completeness in Section~\ref{sec:emwc-subexp}), or the algorithm of Klein and Marx~\cite{marx:mwc-upper} for \pemwcname{}, respectively, we obtain faster parameterized algorithms for both problems.

\begin{corollary} \label{cor:subexp}
Given a planar graph $G$, a terminal set $\terms \subseteq V(G)$, and an integer $k$,
one can
\begin{enumerate}
\item
in $2^{\Oh(\sqrt{k \log k})} + \Oh(k^{142} |V(G)|)$ time decide
whether the \pST{} instance $(G,\terms)$ has a solution with at most $k$ edges;
\item
in $2^{\Oh(\sqrt{k} \log k)} + \mathrm{poly}(|V(G)|)$ time decide
whether the \pemwcname{} instance $(G,\terms)$ has a solution with
at most $k$ edges;
\item
in $2^{\Oh(|\terms|+\sqrt{|\terms|} \log k)} + \mathrm{poly}(|V(G)|)$ time decide
whether the \pemwcname{} instance $(G,\terms)$ has a solution with
at most $k$ edges.
\end{enumerate}
\end{corollary}
This corollary improves on the subexponential-time algorithm for \pST{} previously proposed by the authors~\cite{ours-stacs}, and on the algorithm for \pemwcname{} by Klein and Marx~\cite{marx:mwc-upper} if $k = o(\log |V(G)|)$.
As Tazari's algorithm extends to graphs of bounded genus, combining it with our kernelization algorithm, we obtain the first subexponential-time algorithm for \problemST{} on graphs of bounded genus. The running time is a computable function of the genus times the running time of the planar case --- see Corollary~\ref{cor:pst-genus-subexp}.

We also remark that a similar corollary is unlikely to exist for the case of \pSF{}.
In Section~\ref{sec:sf-lb} we observe that the lower bound for \textsc{Steiner Forest}
on graphs of bounded treewidth of Bateni \etal\cite{marx-bateni}, with minor modifications,
shows also that \pSF{} does not admit a subexponential-time algorithm unless the
Exponential Time Hypothesis of Impagliazzo, Paturi, and Zane~\cite{eth} fails.

\begin{theorem} \label{thm:psf-eth}
Unless the Exponential Time Hypothesis fails,
no algorithm can decide in $2^{o(k)} \poly(|G|)$ time whether \pSF{} instances $(G,\mc{\terms})$ have a solution with at most $k$ edges.
%no algorithm can decide in time $2^{o(k)} \poly(|G|)$ whether a \pSF{} instance $(G,\mc{\terms})$ has a solution with at most $k$ edges.
\end{theorem}

\medskip
\noindent\textbf{Edge-Weighted Planar Graphs.}
Although the decomposition methods in the proof of Theorem~\ref{thm:main}
were developed with applications in unweighted graphs in mind,
they can be modified for graphs with positive edge weights
(henceforth called \emph{edge-weighted graphs}).
That is, we show the following weighted and approximate variant of Theorem~\ref{thm:main}:

\begin{theorem}\label{thm:weighted}
Let $\eps > 0$ be a fixed accuracy parameter,
and let $B$ be an edge-weighted brick with weight function $\weismb$.
Then one can find in $\poly(\eps^{-1})\, |B| \log |B|$ time a subgraph
$H$ of $B$ such that\footnote{In this paper, we denote by $\wei{H}$ the total weight of all the edges of a graph $H$.}
\begin{enumerate}
\renewcommand{\theenumi}{\roman{enumi}}
\renewcommand{\labelenumi}{(\theenumi)}
\item $\prm B \subseteq H$,
\item $\wei{H} \leq \poly(\eps^{-1}) \wei{\prm B}$, and
\item for every set $\mathcal{\terms} \subseteq V(\prm B) \times V(\prm B)$
there exists a Steiner forest $F_H$ that connects $\mathcal{\terms}$ in $H$
such that $\wei{F_H} \leq \wei{F_B} + \eps \wei{\prm B}$ for any Steiner forest $F_B$ that connects $\mathcal{\terms}$ in $B$.
\end{enumerate}
\end{theorem}

Notice that, contrary to Theorem~\ref{thm:main},
we state Theorem~\ref{thm:weighted} in the language of Steiner \emph{forest}, not Steiner tree.
The reason is that the allowed error in Theorem~\ref{thm:weighted} is additive,
and therefore the forest statement seems significantly stronger than the tree one.
Observe that for Theorem~\ref{thm:main}, it would be of no consequence to state it in the language of Steiner forest instead of in the language of Steiner tree.

The proof of Theorem~\ref{thm:weighted} extends the techniques developed for Theorem~\ref{thm:main} to edge-weighted planar graphs, and then wraps this extension into the mortar graph framework developed by Borradaille, Klein, and Mathieu~\cite{klein:planar-st-eptas}. Therefore, the main leap to prove Theorem~\ref{thm:weighted} turns out to be a slight variant of Theorem~\ref{thm:weighted}, where $\mc{\terms}$ is allowed to contain at most $\portalbound$ terminal pairs and the obtained bound for $\wei{H}$ depends polynomially both on $\eps^{-1}$ and~$\portalbound$. We call this the $\portalbound$-variant of Theorem~\ref{thm:weighted}.

The proof of this $\portalbound$-variant considers first a base case where $\mathcal{\terms}$ consists of a single terminal pair and only a $(1+\eps)$ multiplicative error in the weight of the forest $F_{H}$ is allowed. 
%EJ: we are talking abound the variant here, not the theorem itself
%The statement of Theorem~\ref{thm:weighted} for the base case has been shown by 
This base case follows immediately from work by Klein~\cite{klein:stoc06} in the context of an approximation scheme for \textsc{Subset TSP}. The base case, together with all the structural results and decomposition methods developed in the proof of Theorem~\ref{thm:main} (extended to edge-weighted planar graphs), then powers the proof of the $\portalbound$-variant of Theorem~\ref{thm:weighted}. We provide a more detailed sketch of the proof of Theorem~\ref{thm:weighted} in Section~\ref{ss-over:weighted}, and a full proof in Section~\ref{sec:weights}.

%EJ: tree or forest
Theorem~\ref{thm:weighted} influences the known polynomial-time approximation schemes for network design as follows.
The mortar graph framework of Borradaile, Klein, and Mathieu~\cite{klein:planar-st-eptas} may be understood as a method to decompose a brick into cells, such that each cell is equipped with $\portalbound$ evenly-spaced portal vertices, and there is an approximate Steiner tree that for each cell uses a subset of the portal vertices to enter and leave the cell. Then it suffices to preserve an approximate or optimal Steiner tree for any subset of portal vertices. Previously, only a bound that is exponential in $\portalbound$ on the preserved subgraph of each cell was known~\cite{klein:planar-st-eptas}. The impact of our work, and particularly of the $\portalbound$-variant of Theorem~\ref{thm:weighted}, is that the dependency on $\portalbound$ can be reduced to a polynomial. This observation is not only used to prove Theorem~\ref{thm:weighted}, but also leads to deeper understanding of the mortar graph framework.

Observe that one can directly derive an EPTAS for \pST{} from Theorem~\ref{thm:weighted}:
cut the input graph $G$ open along a $2$-approximate Steiner tree (as in the kernel, see Figure~\ref{fig-intro:cutopen}),
apply Theorem~\ref{thm:weighted} to the resulting brick $B$,
and project the obtained graph $H$ back onto the original graph.
An optimal Steiner tree in $G$ becomes an optimal Steiner forest in $B$, and thus the projection
of $H$ preserves an approximate Steiner tree for the input instance.
Since the total weight of $H$ is within a multiplicative factor $\poly(\eps^{-1})$
of the weight of the optimal solution for the input instance, an application of Baker's shifting technique~\cite{baker} can find an approximate solution in $H$
in $2^{\poly(\eps^{-1})} |H| \log |H|$ time.
However, we note that the polynomial dependency on $\eps$ in the exponent is worse
than the one obtained by the currently known \mbox{EPTAS}~\cite{klein:planar-st-eptas}, despite our substantially improved reduction of the cells. This is because that EPTAS utilizes Baker's technique in a more clever way that is aware of the properties of the mortar graph, and is indifferent to the actual replacement within each cell.

%We provide a more detailed sketch of the proof of Theorem~\ref{thm:weighted} in Section~\ref{ss-over:weighted}, and a full proof in Section~\ref{sec:weights}.

%Proving a stronger version of Theorem~\ref{thm:weighted} with a multiplicative error of $(1+\eps)$ (i.e., $\wei{F_H} \leq (1+\eps)\wei{F_B}$) remains open. Observe that $\prm B$ is always a feasible Steiner forest for any choice of $\mathcal{\terms} \subseteq V(\prm B) \times V(\prm B)$, hence such a version would be indeed stronger than Theorem~\ref{thm:weighted}.

\subsection{Discussion} \label{ss-intro:discussion}
%!TEX root = pst-kernel.tex

%\section{Summary and Open Problems}\label{sec:conc}

%We presented a method that sparsifies a given planar graph, while preserving optimal Steiner trees that connect vertices on a designated face. Our result implies polynomial kernels for \pST{}, \pSF{}, and \pemwcname{}, resolving several open problems in parameterized complexity and providing methods that reach beyond the bidimensionality theory.

A drawback of our methods is that the exponents in the kernel bounds and the polynomial dependency on $\eps^{-1}$ in the weighted variant are currently large, making the results theoretical.
However, we see the strength of our results in that we prove that a polynomial kernel actually exists  --- thus proving that \pST, \pSF, and \pemwcname{} belong to the class of problems that have a polynomial kernel --- rather than in the actual size bound.
Encouraged by the recent progress in understanding distance sparsifiers in planar graphs~\cite{GoranciHP17,KrauthgamerR17},
we conjecture that the correct dependency in Theorem~\ref{thm:main} is quadratic, with a grid being the worst-case scenario.

Another limitation of our methods is that we need to parameterize by the \emph{number of edges} of the Steiner tree. 
A subsequent work of Suchy~\cite{Suchy15} extended the result to the parameter 
{\em{number of non-terminal vertices of the tree}}.
Very recently, Marx and the first two authors~\cite{dirtsp-pst} proved 
that \pST{} does not admit a subexponential algorithm in planar graphs for the parameter
\emph{number of terminals} (under the standard assumption of the Exponential Time Hypothesis)
  and showed that this refutes an existence of a polynomial
kernel (with this parameter) that does not increase the value of the parameter in the reduction.
This shows that probably the number of edges or non-terminal vertices are the most general parameterizations for which we can obtain a polynomial kernel in planar graphs.

Similarly, one may consider graph-separation problems with vertex-based parameters, such as \textsc{Odd Cycle Transversal} or the node-deletion variant of \textsc{Multiway Cut}. On planar graphs, both of these problems are some sort of Steiner problem on the dual graph. It would be interesting to show polynomial kernels for these problems (without using the matroid framework~\cite{stefan-magnus2}).

To generalize our methods, it would be interesting to lift our results to more general graph classes, such as graphs with a fixed excluded minor. For \textsc{Edge Multiway Cut}, even the bounded-genus case remains open.
Further work is also needed to improve the allowed error in Theorem~\ref{thm:weighted}. Currently, this error is an additive error of $\eps \wei{\prm B}$. In other words, a near-optimal Steiner forest is preserved only for ``large'' optimal forests, that is, for ones of size comparable to the perimeter of $B$.
Is it possible to improve Theorem~\ref{thm:weighted} to ensure a $(1+\eps)$ multiplicative error? That is, to obtain a variant of Theorem~\ref{thm:weighted} where the graph $H$ satisfies $\wei{F_H} \leq (1+\eps)\wei{F_B}$, and thus to preserve near-optimal Steiner forests at \emph{all scales}?
Finally, since our methods handle problems that are beyond the reach of the theory of bidimensionality, our contribution might open the door to a more general framework that is capable of addressing a broader range of problems.

\subsection{Related work}
The three problems considered in this paper (\problemST{}, \problemSF{}, and \textsc{Edge Multiway Cut}) are all NP-hard~\cite{st:np-hard,dahlhaus}
and unlikely to have a PTAS~\cite{st:max-snp-hard,dahlhaus} on general graphs. However, they do admit constant-factor approximation algorithms~\cite{steiner-approx,sf-approx,mwc-approx}.

\problemST{} has a $2^{|\terms|} \cdot \textrm{poly}(|G|)$-time, polynomial-space algorithm on general graphs~\cite{jesper:icalp09}; the exponential factor is believed to be optimal~\cite{ccc}, but an improvement has not yet been ruled out under the Strong Exponential Time Hypothesis. The algorithm for \problemST{} implies a $(2|\mathcal{\terms}|)^{|\mathcal{\terms}|} \cdot \textrm{poly}(|G|)$-time, polynomial-space algorithm for \problemSF{}. On the other hand, \textsc{Edge Multiway Cut} remains NP-hard on general graphs even when $|\terms|=3$~\cite{dahlhaus}, while for the parameterization by the size of the cut $k$, a $1.84^{k} \cdot \textrm{poly}(|G|)$-time algorithm is known~\cite{mwc-yixin}.

Neither \problemST{} nor \problemSF{} admits a polynomial kernel on general graphs~\cite{dom:ids}, unless the polynomial hierarchy collapses. Recently, a polynomial kernel was given for \textsc{Edge} and \textsc{Node Multiway cut} for a constant number of terminals or deletable terminals~\cite{stefan-magnus2}; nevertheless, the question for a polynomial kernel in the general case remains open.

\problemST{}, \problemSF{}, and \textsc{Edge Multiway Cut} all remain NP-hard on planar graphs~\cite{planar-st:np-hard,dahlhaus}, even in restricted cases.
All three problems do admit an EPTAS on planar
graphs~\cite{klein:planar-st-eptas,klein-efficient-psf,bateni-mwc}, and \problemST{}  admits an EPTAS on bounded-genus graphs~\cite{cora:genus}.
As mentioned before, for many graph problems on planar graphs, both polynomial kernels and subexponential-time algorithms follow from the theory of bidimensionality~\cite{bidim:jacm,bidim:kernels}. However, the theory neither applies to \problemST{}, \problemSF{}, nor \textsc{Edge Multiway Cut}.

We are not aware of any previous kernelization results for \problemST{}, \problemSF{}, or \textsc{Edge Multiway Cut} on planar graphs.
The question of the existence of a subexponential-time algorithm for \pST{} was first explicitly
pursued by Tazari~\cite{tazari:mfcs10}. He showed that such a result would be implied by a subexponential or polynomial kernel. The current authors adapted the main ideas of the EPTAS for
\pST{}~\cite{klein:planar-st-eptas} to show a subexponential-time algorithm~\cite{ours-stacs}, without actually giving a kernel beforehand. The algorithm of~\cite{ours-stacs} in fact finds subexponentially many subgraphs of subexponential size, one of which is a subexponential kernel if the instance is a YES-instance.
Finally, for \textsc{Edge Multiway Cut} on planar graphs, a $2^{\Oh(|\terms|)}\cdot |G|^{\Oh(\sqrt{|\terms|})}$-time
algorithm is known~\cite{marx:mwc-upper} and believed to be optimal~\cite{marx:mwc-lower}.

%!TEX root = pst-kernel.tex

\section{Overviews of the proofs} \label{sec:overview}

In this section, we give a more detailed overview of the proof of Theorem~\ref{thm:main},
its weighted counterpart Theorem~\ref{thm:weighted}, and
discuss its corollaries (Theorems~\ref{thm:pst-intro}, \ref{thm:psf-intro} and~\ref{thm:emwc-intro}).

Before we start, we set up some notation.
For a subgraph $H$ of $G$, we silently identify $H$ with the edge set of $H$; that is,
all our subgraphs are edge-induced. For a brick $B$, $\prm B[a,b]$ denotes the subpath of $\prm B$ obtained by traversing $\prm B$ in counter-clockwise direction from $a$ to $b$.
By $\plane$ we denote the standard Euclidean plane.
For a closed curve $\gamma$ on $\plane$, we say that $\gamma$ {\em{strictly encloses}}
$c \in \plane$ if $c \notin \gamma$ and $\gamma$ is not continuously
retractable in $\plane \setminus \{c\}$ to a single point,
and $\gamma$ {\em{encloses}} $c$ if it strictly encloses $c$ or $c \in \gamma$.
This notion naturally translates to cycles and walks in a plane graph $G$
(strictly) enclosing vertices, edges, and faces of $G$.

\subsection{Overview of the proof of Theorem~\ref{thm:main}}
The idea behind the proof of Theorem~\ref{thm:main} is to apply it recursively on subbricks (subgraphs enclosed by a simple cycle) of the given brick $B$. The main challenge is to devise an appropriate way to decompose $B$ into subbricks, so that their ``measure'' decreases. Here we use
the \emph{perimeter} of a brick as a potential that measures the progress of the
algorithm.

Intuitively, we would want to do the following. Let $T$ be a tree in $B$ that connects a subset of the vertices on the perimeter of $B$. Then $T$ splits $B$ into a number of smaller bricks $B_1,\ldots,B_r$, formed by the finite faces of $\prm B \cup T$ (see Figure~\ref{fig-over:partition}). We recurse on bricks $B_i$, obtaining graphs $H_i \subseteq B_i$, and return $H := \bigcup_{i=1}^r H_i$. We can prove that this decomposition yields a polynomial bound on $|H|$ if (i) all bricks $B_i$ have multiplicatively smaller perimeter than $B$, and (ii) the sum of the perimeters of the subbricks is linear in the perimeter of $B$.

\begin{figure}[tbh]
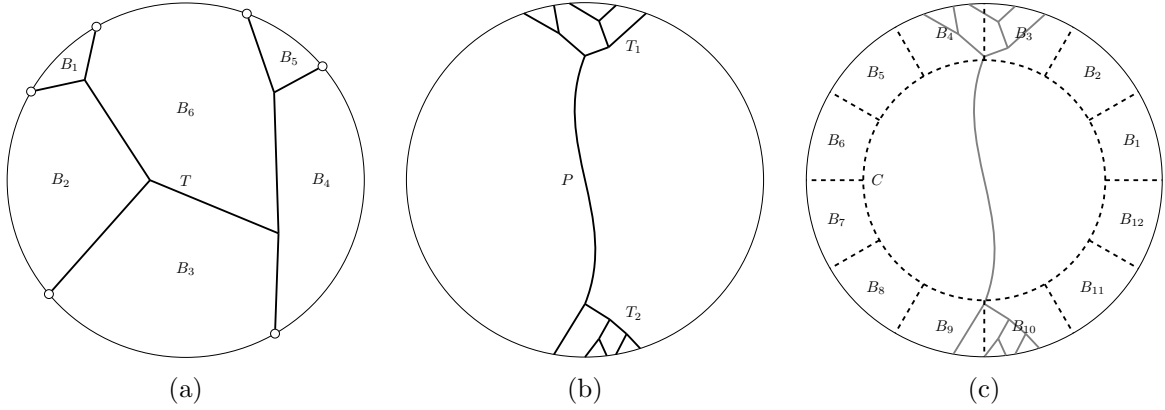

\centering
\begin{subfigure}{.33\textwidth}
\centering
\includegraphics[width=.9\linewidth]{figures/fig-over-partition}
\caption{}
\label{fig-over:partition}
\end{subfigure}%
\begin{subfigure}{.33\textwidth}
 \centering
 \includegraphics[width=.9\linewidth]{figures/fig-over-double-comet}
 \caption{}
 \label{fig-over:double-comet}
\end{subfigure}%
\begin{subfigure}{.33\textwidth}
 \centering
 \includegraphics[width=.9\linewidth]{figures/fig-over-pineapple}
 \caption{}
 \label{fig-over:pineapple}
\end{subfigure}
\caption{(Figure~\ref{fig-intro:intro-comet} repeated); (a) shows an optimal Steiner tree $T$ and how it partitions
the brick $B$ into smaller bricks $B_1,\ldots,B_r$.
  (b) shows an optimal Steiner tree that connects a set of vertices on the perimeter of $B$ and that consists of two small trees $T_{1},T_{2}$ that are connected by a long path~$P$.
  (c) shows a cycle $C$ that (in particular) hides the small trees $T_1,T_2$ in the ring between $C$ and $\prm B$, and a subsequent decomposition of $B$ into smaller bricks.}
\label{fig-over:intro-comet}
\end{figure}

In this approach, 
there are two clear issues that need to be solved. 
The first issue is that we need an algorithm to decide whether there is a tree $T$ for which the induced set of subbricks satisfies conditions (i) and (ii). We design a dynamic programming algorithm that either correctly decides that no such tree exists, or finds a set of subbricks of $B$ that satisfies condition (i) and (ii). In the latter case, we can recurse on each of those subbricks.%, as described before.

The second issue is that there might be no trees $T$ for which the induced set of subbricks satisfies conditions (i) and (ii). In this case, optimal Steiner trees, which are a natural candidate for such partitioning trees $T$, behave in a specific way. For example, consider the tree of Figure~\ref{fig-over:double-comet}, which consists of two small trees $T_1$, $T_2$ that lie on opposite sides of the brick $B$ and that are connected through a shortest path $P$ (of length slightly less than $|\prm B|/2$). Then both faces of $\prm B \cup T$ that neighbour $P$ may have perimeter almost equal to $|\prm B|$, thus blocking our default decomposition approach.

To address this second issue, we propose a completely different decomposition. Intuitively, we find a cycle $C$ of length linear in $|\prm B|$ that lies close to $\prm B$, such that all vertices of degree three or more of any optimal Steiner tree are hidden in the ring between $C$ and $\prm B$ (see Figure~\ref{fig-over:pineapple}). We then decompose the ring between $\prm B$ and $C$ into a number of smaller bricks. We recursively apply Theorem~\ref{thm:main} to these bricks, and return the result of these recursive calls together with a set of shortest paths inside $C$ between any pair of vertices on $C$.

In Section~\ref{ss-over:bricks} below, we formalise the above notions and give the algorithm that addresses the first issue. Then, Section~\ref{ss-over:decomp:default} describes the default decomposition, whereas Section~\ref{ss-over:decomp:other} describes the alternative decomposition that addresses the second issue. The full proof appears in Sections~\ref{sec:preliminaries} through~\ref{sec:dp}.

\subsubsection{Deciding on the Decomposition}\label{ss-over:bricks}
In this section, we present some of the basic notions of our paper and describe the algorithm that decides which of the two possible decompositions is used.

%A brick $B'$ is {\em{a subbrick}} of a brick $B$ if $B'$ is a subgraph of $B$ that consists of all edges enclosed by $\prm B'$.

\begin{definition}
For a brick $B$, a {\em brick covering} of $B$ is a family $\Bb=\{B_1,\ldots,B_p\}$ of bricks, such that {\em (i)} each $B_i$, $1 \leq i \leq p$,
is a subbrick of $B$, and {\em (ii)} each face of $B$ is contained in at least one brick $B_i$, $1 \leq i \leq p$.
A brick covering is called a {\em{brick partition}} if each face of $B$ is contained in {\em{exactly one}} brick $B_i$.
\end{definition}
\noindent We note that if $\Bb = \{B_1,\ldots,B_p\}$ is a brick partition of $B$,
then every edge of $\prm B$ belongs to the perimeter of exactly one brick $B_i$,
while every edge strictly enclosed by $\prm B$ either is in the interior of exactly one brick $B_i$,
or lies on the perimeters of exactly two bricks $B_i$, $B_j$ for $i\neq j$.

Any connected set $F\subseteq B$ will be called a {\em connector}.
Let $\terms$ be the set of vertices of $\prm B$ adjacent to at least one edge of $F$;
the elements of $\terms$ are the {\em anchors} of the connector $F$.
We then say that $F$ {\em connects} $\terms$.
For a connector $F$, we say that $F$ is {\em optimal} if there is no connector $F'$ with $|F'|<|F|$
that connects a superset of the anchors of $F$.
Clearly, each optimal connector $F$ induces a tree, whose every leaf is an anchor of $F$.
We say that a connector $F\subseteq B$ is {\em brickable} if the boundary of every inner face of $\prm B \cup F$ is a simple cycle, i.e., these boundaries form subbricks of $B$. Let $\Bb$ be the corresponding brick partition of $B$. Observe that $\sum_{B'\in \Bb} |\prm B'|\leq |\prm B|+2|F|$. %Note that every optimal connector is brickable.

Next, we define the crucial notions for partitions and coverings that are used for the default decomposition.

\begin{definition}
The {\em{total perimeter}} of a brick covering $\Bb = \{B_1,\ldots,B_p\}$ is defined as $\sum_{i=1}^p |\prm B_i|$.
For a constant $c > 0$, $\Bb$ is {\em{$c$-short}} if the total perimeter of $\Bb$ is at most $c \cdot |\prm B|$.
For a constant $\nice > 0$, $\Bb$ is {\em{$\nice$-nice}} if $|\prm B_i| \leq (1-\nice) \cdot |\prm B|$ for each $1 \leq i \leq p$.

Similarly, a brickable connector $F\subseteq B$, with $\Bb = \{B_1,\ldots,B_p\}$ being the corresponding brick partition,
is {\em{$c$-short}} if $\Bb$ is $c$-short,
is simply {\em{short}} if it is $3$-short,
and is \emph{$\nice$-nice} if $\Bb$ is $\nice$-nice.
%if {\emph{(i)}} $|F| < |\prm B|$, and {\emph{(ii)}} $|\prm B_{i}| \leq (1-\nice) \cdot |\prm B|$ for all $i=1,\ldots,p$.
\end{definition}

Observe that if $F\subseteq B$ is a brickable connector, then $F$ is $c$-short if $|F| \leq |\prm B| \cdot (c-1)/2$, and $F$ is short if $|F|\leq |\prm B|$. Moreover, if $F$ is an optimal connector, then $F$ is a short brickable connector, as $F$ must be a tree of length at most $|\prm B|$. %Such a tree is called a \emph{$3$-short tree} (or just {\emph{short}} instead of $3$-short, for simplicity). 
Now we are ready to give the algorithm that decides what decomposition to use.

\begin{theorem}\label{thm-over:nice-testing}
Let $\nice > 0$ be a fixed constant.
Given a brick $B$, in $\Oh(|\prm B|^8 |B|)$ time
one can either correctly conclude that no short $\nice$-nice tree exists in $B$
or find a $3$-short $\nice$-nice brick covering of $B$.
\end{theorem}

The proof of Theorem~\ref{thm:nice-testing}, omitted in this overview and provided in full detail in Section~\ref{sec:dp}, is a technical modification of the classical algorithm of Erickson et al.~\cite{erickson}. That algorithm computes an optimal Steiner tree in a planar graph assuming that all the terminals lie on the boundary of the infinite face. It uses the Dreyfus-Wagner dynamic-programming approach, where a state consists of a subset of already connected terminals, and the current ``interface'' vertex; the main observation is that only states with consecutive terminals on the boundary are relevant, yielding a polynomial bound on the number
of them.
In our case, we can proceed similarly: our state consists of the leftmost and rightmost chosen terminal, the ``interface'' vertex inside the brick, the total length of the tree, and the length of the leftmost and rightmost path in the constructed tree. Consequently, the terminals are chosen on-the-fly.

In case some short $\nice$-nice tree exists, for technical reasons we cannot ensure that the output of the algorithm of Theorem~\ref{thm-over:nice-testing} will actually be a brick partition corresponding to some short $\nice$-nice tree. Instead, the algorithm may output a brick covering, but one that is guaranteed to be $3$-short and $\nice$-nice. This is sufficient for our purposes.

We can now formally describe the main line of reasoning of our sparsification algorithm. Let $\nice > 0$ be some constant chosen later. If $|\prm B| \leq 2/\nice$, then for each $\terms \subseteq V(\prm B)$ we compute an optimal Steiner tree that connects $\terms$ using the algorithm of Erickson et al.~\cite{erickson}, and take the union of all such trees. If $|\prm B| > 2/\nice$, then we run the algorithm of Theorem~\ref{thm-over:nice-testing} for $B$ and $\nice$. If the algorithm returns a $3$-short $\nice$-nice brick covering, then we proceed to the default decomposition, formalized in Section~\ref{ss-over:decomp:default} below. Otherwise, if the algorithm of Theorem~\ref{thm-over:nice-testing} concluded that $B$ does not contain any short $\nice$-nice tree, then we proceed to the arguments in Section~\ref{ss-over:decomp:other}. We show that in all cases we obtain a subgraph of $B$ that satisfies conditions (i)-(iii) of Theorem~\ref{thm:main}.

\subsubsection{The Default Decomposition} \label{ss-over:decomp:default}
Suppose that the algorithm of Theorem~\ref{thm-over:nice-testing} returns a $3$-short $\nice$-nice brick covering $\Bb = \{B_{1},\ldots,B_{p}\}$ of $B$. We can then use this brick covering as a decomposition and recurse on each brick individually. This is formalized in the following lemma.

\begin{lemma}\label{lem-over:recursion}
Let $c, \nice > 0$ be constants.
Let $B$ be a brick and let $\Bb = \{B_1,\ldots,B_p\}$ be a $c$-short $\nice$-nice brick covering of $B$.
Assume that the algorithm of Theorem~\ref{thm:main} was applied recursively to bricks $B_1,\ldots, B_p$,
and let $H_1,\ldots,H_p$ be the subgraphs output by this algorithm for $B_1,\ldots,B_p$, respectively,
where $|H_i|\leq C \cdot |\prm B_i|^\alpha$ for some constants $C>0$ and $\alpha\geq 1$
such that $(1-\nice)^{\alpha-1}\leq \frac{1}{c}$.
Let $H=\bigcup_{i=1}^p H_i$.
Then $H$ satisfies conditions (i)-(iii) of Theorem~\ref{thm:main}, with $|H|\leq C\cdot |\prm B|^\alpha$.
\end{lemma}
\begin{proof}
To see that $H$ satisfies condition (i), note that every edge of $\prm B$ is in the perimeter of some brick $B_i$, and that $\prm B_i\subseteq H_i$ for every $i=1,2,\ldots,p$. Therefore, $\prm B\subseteq H$.

To see that $H$ satisfies condition (ii), recall that $\Bb$ is $c$-short and that $|\prm B_i|\leq (1-\nice) \cdot |\prm B|$ for each $i=1,2,\ldots,p$. Therefore, $|\prm B_i|^\alpha \leq |\prm B_i|\cdot (1-\nice)^{\alpha-1}|\prm B|^{\alpha-1}$, and
$$
|H| \leq \sum_{i=1}^p |H_i| \leq C\cdot \sum_{i=1}^p|\prm B_i|^\alpha
\leq C\cdot (1-\nice)^{\alpha-1}|\prm B|^{\alpha-1}\cdot \sum_{i=1}^p |\prm B_i|
\leq c\cdot (1-\nice)^{\alpha-1} C \cdot |\prm B|^{\alpha} \leq C \cdot |\prm B|^\alpha.$$

Finally, to see that $H$ satisfies condition (iii), let $\terms\subseteq V(\prm B)$ be a set of terminals lying on the perimeter of $B$, and let $T$ be an optimal Steiner tree connecting $\terms$ in $B$ that contains a minimum number of edges that are not in $H$.
We claim that $T \subseteq H$. Assume the contrary, and let $e \in T \setminus H$.
Since each face of $B$ is contained in some brick of $\Bb$, there exists a brick $B_i$ such that $\prm B_i$ encloses $e$.
As $\prm B_i \subseteq H_i \subseteq H$, we infer $e \notin \prm B_i$.
Consider the subgraph of $T$ strictly enclosed by $\prm B_i$,
and let $X$  be the connected component of this subgraph that contains $e$.
Clearly, $X$ is a connector inside $B_i$.
Since $H_i$ is obtained by a recursive application of Theorem~\ref{thm:main},
there exists a connected subgraph $D \subseteq H_i$ that connects the anchors of $X$
and that satisfies $|D| \leq |X|$. Let $T' = (T \setminus X) \cup D$.
Observe that $|T'| \leq |T|$ and that $T'$ contains strictly less edges that are not in $H$ than $T$ does.
Since $D$ connects the anchors of $X$ in $H_i$, $T'$ still connects the anchors of $T$ in $B$, that is, $T'$ connects $\terms$.
However, $T$ is an optimal Steiner tree that connects $\terms$, and thus $T'$ is also an optimal Steiner tree that connects $\terms$. Since $T'$ contains strictly less edges that are not in $H$ than $T$,
this contradicts the choice of $T$. Hence, $T \subseteq H$.
\end{proof}

\subsubsection{The Alternative Decomposition --- Mountain Ranges and the Core} \label{ss-over:decomp:other}
Suppose that the algorithm of Theorem~\ref{thm-over:nice-testing} decides that no short $\nice$-nice tree exists in $B$. As mentioned before, we want to find a cycle $C$ of length linear in $|\prm B|$ that is close to $\prm B$, such that all vertices of degree three or more of any optimal Steiner tree are hidden in the ring between $C$ and $\prm B$ (see Figure~\ref{fig-over:pineapple}). In the following, we use a constant $\delta\in (0,\frac{1}{2})$, which depends on $\nice$ and is chosen later.

\begin{definition}
A \emph{$\delta$-carve} $L$ from a brick $B$ is a pair $(P,I)$, where $P$, called the \emph{carvemark}, is a path in $B$ between two distinct vertices $a,b \in V(\prm B)$ of length at most $(\frac{1}{2}-\delta) \cdot |\prm B|$, and $I$, called the \emph{carvebase}, is a shortest of the two paths $\prm B[a,b], \prm B[b,a]$.
The subgraph enclosed by the closed walk $P\cup I$ is called the {\emph{interior}} of a $\delta$-carve.
\end{definition}

Of particular interest will be the following special type of $\delta$-carves.

\begin{definition}
For %constant $\delta \in (0, 1/2)$ and
 fixed $l,r \in V(\prm B)$, a \emph{$\delta$-mountain} of $B$ for $l,r$
is a $\delta$-carve $M$ in $B$ such that
\begin{enumerate}
\item $l$ and $r$ are the endpoints of the carvemark and carvebase of $M$;
\item the edges of the carvemark can be partitioned into two paths $P_{L}, P_{R}$, where $P_{L}$ is a shortest $l$--$P_{R}$ path in the interior of $M$
and $P_{R}$ is a shortest $r$--$P_{L}$ path in the interior of $M$.
\end{enumerate}
\end{definition}

We write $M= \mou{P_L}{P_R}$ to exhibit the partition of the carvemark into paths $P_L$ and $P_R$.
We use $v_M$ to denote the unique vertex of $V(P_L) \cap V(P_R)$. %, and we call this the \emph{summit} of the $\delta$-mountain.
We also say that a $\delta$-mountain $M$ {\em{connects}} the vertices $l$ and $r$.

The following lemma motivates why we are interested in $\delta$-mountains.
For a tree $T$, $T[a,b]$ denotes the unique path in $T$ between vertices $a$ and $b$.

\begin{lemma}\label{lem-over:broom-in-mountain}
Let $B$ be a brick and let $T$ be an optimal Steiner tree
connecting $V(T) \cap V(\prm B)$ in $B$.
Let $uv \in T$ be an edge of $T$, where $v$ is of degree at least $3$ in $T$,
and let $T_v$ be the connected component of $T \setminus \{uv\}$ containing $v$, rooted at $v$.
Let $l$ and $r$ be the leftmost and rightmost elements of $V(T_v) \cap V(\prm B)$,
that is, $V(T_v) \cap V(\prm B) \subseteq V(\prm B[l,r])$ and $T[l,r] \cup \prm B[r,l]$
encloses $uv$.
Assume furthermore that $|\prm B[l,r]| < |\prm B|/2$.
Then $M := \mou{T[l,v]}{T[r,v]}$ is a $\delta$-mountain connecting $l$ and $r$, for any
$\delta < 1/ 2 - |T[l,r]|/|\prm B|$.
\end{lemma}
\begin{proof}
As $v$ is of degree at least $3$ in $T$, $v$ has degree at least $2$ in $T_v$, and $T[l,v] \cap T[r,v] = \{v\}$. Therefore,
$T[l,v] \cup T[v,r] = T[l,r]$, and $T[l,r]$ induces a $\delta$-carve $M$ with carvebase $\prm B[l,r]$.%, as $|\prm B[l,r]| < |\prm B|/2$.

Suppose that $M$ is not a $\delta$-mountain if we take $P_L = T[l,v]$ and $P_R = T[r,v]$. Without loss of generality, there exists a path $P$ enclosed by $M$ that connects $l$ with $w \in V(P_R)$,
$V(P_R) \cap V(P) = \{w\}$, and $|P| < |T[l,v]|$.
Let $D$ be the subgraph of $M$ enclosed by the closed walk
$T[l,v] \cup P \cup T[v,w]$.
Define $T' := (T \setminus D) \cup T[v,w] \cup P$. As $T[v,w] \cup T[l,v] \subseteq D$, $|T'| < |T|$. By the definition of $l$ and $r$,
$T[l,v] \setminus P$ does not contain any vertex of $\prm B$. Therefore, $T'$ is a connected subgraph of $B$ connecting $V(T) \cap V(\prm B)$,
a contradiction to the optimality of $T$.
\end{proof}
The above lemma shows that small subtrees of optimal Steiner trees in $B$
are hidden in $\delta$-mountains. Here, `small' means that the leftmost and rightmost
path in the subtree have total length at most $(1/2 - \delta)\cdot |\prm B|$. Note that
an optimal Steiner tree in $B$ has total size smaller than $|\prm B|$, as $\prm B$
without an arbitrary edge connects any subset of $V(\prm B)$. Therefore,
if we choose $\delta$ appropriately, then we can `hide' almost an entire optimal non-$\nice$-nice Steiner tree
in at most two $\delta$-mountains. To hide most of \emph{all} optimal Steiner trees, we consider unions of $\delta$-mountains. For fixed vertices $l,r \in V(\prm B)$, the \emph{$\delta$-mountain range} is the closed walk $W_{l,r}$ in $B$ such that a face $f$ of $B$ is enclosed by $W_{l,r}$ if and only if $f$ belongs to some $\delta$-mountain that connects $l$ and $r$. %We show that the $\delta$-mountain range can be found in polynomial time.

\begin{theorem}[Mountain Range Theorem]\label{thm-over:mountain-range}
Fix $\nice \in [0,1/4)$ and $\delta \in [2\nice,1/2)$, and assume that $B$ does not admit any short $\nice$-nice tree. Then for any fixed $l, r \in V(\prm B)$ with $|\prm B[l,r]| < |\prm B|/2$, $W_{l,r}$ has length at most $3\cdot |\prm B[l,r]|$. Moreover, the set of the faces enclosed by $W_{l,r}$ can be computed in $\Oh(|B|)$ time.
\end{theorem}

\begin{proof}[Proof sketch]
By case analysis, omitted in this overview and provided in full detail in Section~\ref{sec:mountains}, we deduce
that the set of all inclusion-wise maximal $\delta$-mountains essentially looks as in Figure~\ref{fig-over:rangefull}, i.e., 
for any two maximal mountains there exists exactly one region of the plane that is in one of them but not in the other one. 

\begin{figure}[tbh]
\centering
\begin{subfigure}{.5\textwidth}
\centering
\svg{.9\linewidth}{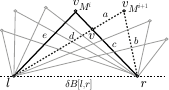}
\caption{}
\label{fig-over:rangefull}
\end{subfigure}%
\begin{subfigure}{.5\textwidth}
 \centering
 \includegraphics[width=.6\linewidth]{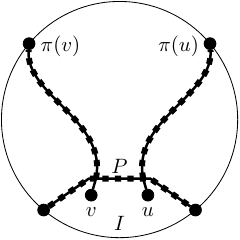}
 \caption{}
 \label{fig-over:corestep}
\end{subfigure}%
\caption{(a) shows a mountain range.
(b) shows a short $\nice$-nice tree occurring if $\pi(v)$ and $\pi(u)$ are far from $I$.}\label{fig-over:mountain-core}
\end{figure}

Let $\{ M^i=(P^i_L, P^i_R)\}_{i=1}^s$ be the set of all these maximal $\delta$-mountains,
ordered from left to right. By induction, we show that the perimeter
of the union of the first $i$ $\delta$-mountains, denoted $p^i$,
is at most $|\prm B[l,r]| + |P^1_R| + |P^i_L|$.
This statement clearly holds for $i=1$, and for $i=s$ it
proves the bound on the perimeter
of the $\delta$-mountain range promised by Theorem~\ref{thm-over:mountain-range}.

For the inductive step, define $b=|P^{i+1}_R|$ and $e=|P^i_L|$.
Let $v$ be the first point on $P^{i+1}_L$ that lies on $P^{i}_R$. We denote the distance (along $P^{i+1}_L$) from $l$ to $v$ as $d$
and the distance from $v$ to $v_{M^{i+1}}$ as $a$. Finally, we denote by $c$ the distance (along $P^{i}_R$) from $r$ to $v$.
These definitions are illustrated in Figure~\ref{fig-over:rangefull}.
Observe that $d\ge e$, because $M^i$ is a $\delta$-mountain. Similarly, observe that $c \ge b$, because $M^{i+1}$ is a $\delta$-mountain.
Hence, 
$p^{i+1} - p^i = a + b-c \le a \le a+d-e = |P^{i+1}_L| - |P^{i}_L|.$
This concludes the inductive step.
In this overview, we omit the description of the algorithm that finds the mountain range.
\end{proof}

We now designate $\Oh(\nice^{-1})$ vertices on $\prm B$, and construct the union $\mathcal{M}$ of all $\delta$-mountain ranges for each pair of designated vertices. Using the following deep theorem, we can show that $\mathcal{M}$ is not the entire brick.

\begin{theorem}[Core Theorem]\label{thm-over:core}
For any $\nice \in (0,\frac{1}{4})$ and any $\delta \in [2\nice, \frac{1}{2})$, if $B$ has no short $\nice$-nice tree, then there exists a face of $B$ that is not enclosed by any $\delta$-carve. Moreover, such a face can be found in $O(|B|)$ time.
\end{theorem}
\begin{proof}[Proof sketch]
Suppose, for sake of contradiction, that all faces of $B$ are enclosed by some $\delta$-carve.
We first observe that, for any brickable short tree $T$ with diameter not more than $(\frac{1}{2}-\delta)\cdot |\prm B|$, 
there exists an interval $I_T$ of $\prm B$ of length at most $(\frac{1}{2} -\frac{\delta}{2}) \cdot |\prm B|$ 
such that all anchors of $T$ are in $I_T$. If no such interval exists, then every brick induced
by $T$ has perimeter less than $(\frac{1}{2}-\delta)\cdot |\prm B| + (\frac{1}{2} +\frac{\delta}{2}) \cdot |\prm B| \le (1-\nice) \cdot |\prm B|$. 
Hence, $T$ would be $\nice$-nice, a contradiction.

Define a map $v\to \proj(v)$ for $v\in V(B)$ such that $\proj(v)$ is a vertex of $\prm B$ closest to $v$.
The main observation is that if $v$ and $u$ belong to the interior
of some $\delta$-carve $(P,I)$, then the distance between $\proj(v)$ and $\proj(u)$ 
along $\prm B$ is at most $(\frac{1}{2} -\frac{\delta}{2}) \cdot |\prm B|$. 
To see this, consider the shortest paths $P_v$ from $v$ to $\pi(v)$. 
These paths can be used to form a tree $T$, consisting of $P$, the subpath of $P_v$ to $\proj(v)$ from the last point of $P_v$ on $P$, 
and the subpath of $P_u$ to $\proj(u)$ from the last point of $P_u$ on $P$ (see Figure~\ref{fig-over:corestep}).
We observe that the diameter of $T$ is bounded by $|P| \le (\frac{1}{2}-\delta)\cdot |\prm B|$, 
because the paths that make up $T$ always have length at most the corresponding part of $P$. 
Moreover, as $T$ has only four leaves, $|T|$ is bounded by twice the diameter of $T$, so $T$ is short. 
Hence, $\pi(v)$, $\pi(u)$, and $V(P)\cap V(\prm B)$ lie on the interval $I_T$, as observed above. 
We extend $\proj$ to the edges of $B$ by mapping $uv$
onto the shorter subpath between $\proj(u)$ and $\proj(v)$ on $\prm B$. Now consider a face $f$ that is enclosed by $(P,I)$. We note that no point of any edge of $f$ is mapped to a point lying exactly opposite on $\prm B$ to any point in $V(P)\cap V(\prm B)$, as such points cannot belong to $I_T$. Hence, all edges of $f$ are mapped to an interval of $\prm B$. 
Since an interval is a simply connected metric space, we can extend $\proj$ from the boundary of face 
$f$ to its interior in a continuous manner such that the whole face $f$ is mapped into it.
Consequently, since every face of $B$ can be enclosed by a $\delta$-carve,
we have constructed a retraction of a closed disc onto its boundary.
This contradicts Borsuk's non-retraction theorem~\cite{badger}.%, bitches!
\end{proof}

As each $\delta$-mountain is also a $\delta$-carve, $\mathcal{M}$ does not contain
an arbitrarily chosen core face $\coreface$ promised by Theorem~\ref{thm-over:core}.
Hence, the union of the perimeters of the $\delta$-mountain ranges that make up $\mathcal{M}$
contains a cycle $C_0$ that separates $\coreface$ from the mountain ranges.
Moreover, as we construct only $\Oh(\nice^{-2})$ mountain ranges,
each of perimeter $\Oh(|\prm B|)$ by Theorem~\ref{thm-over:mountain-range}, we have that $|C_0| = \Oh(|\prm B|)$; see Figure~\ref{fig-over:Czero}.

\begin{figure}
\centering
\begin{subfigure}{.5\textwidth}
\centering
  \includegraphics[width=.6\linewidth]{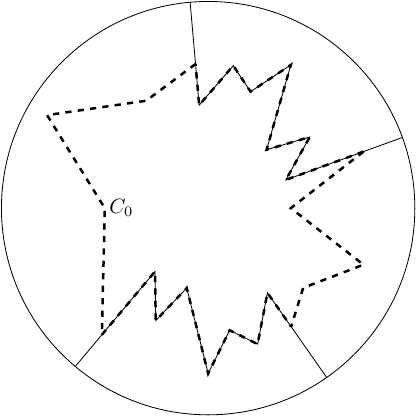}
\caption{}
\label{fig-over:Czero}
\end{subfigure}%
\begin{subfigure}{.5\textwidth}
\centering
  \includegraphics[width=.45\linewidth]{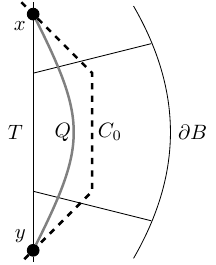}
\caption{}
\label{fig-over:slide}
\end{subfigure}
\caption{(a) shows the cycle $C_0$ formed by the union of the perimeters of the mountain ranges; example mountain ranges are drawn solid.
(b) shows how to shortcut the tree $T$ (solid) with a shortest $xy$-path $Q$ (gray).}\label{fig-over:slidefull}
\end{figure}

We observe that certain optimal Steiner trees in $B$ may behave nontrivially in the subgraph enclosed by $C_0$, and in particular, may still have a vertex of degree three or more that is enclosed by $C_{0}$. However, this behavior is easily dealt with as follows. Consider the situation in Figure~\ref{fig-over:slide}. If $Q$ is a shortest path between $x$ and $y$, then we may replace the part of the tree to the left of $Q$ by $Q$. Hence, we shortcut $C_0$ whenever possible while keeping $\coreface$ enclosed by $C_0$. By choosing $\delta = 4 \nice$, we then obtain the following result.

\begin{theorem}\label{thm-over:ananas}
% Pineapple is "ananas" in any language except English, I insist on using `ananas' at least in the latex references :)
Let $\nice \in (0,1/36]$.
Assume that $B$ does not admit any short $\nice$-nice tree and that $|\prm B| > 2/\nice$.
Then one can in $\Oh(|B|)$ time compute a simple
cycle $C$ in $B$ with the following properties:
\begin{enumerate}
\item the length of $C$ is at most $\frac{16}{\nice^2} |\prm B|$;
\item for each vertex $x \in V(C)$, there exists a path from $x$ to $V(\prm B)$
of length at most $(\frac{1}{4} - 2\nice) \cdot |\prm B|$ such that no edge of the path
is strictly enclosed by $C$;
\item $C$ encloses $\coreface$, where $\coreface$ is a face of $B$, promised by Theorem~\ref{thm-over:core}, that is not enclosed by any $2\nice$-carve;
\item for any $S \subseteq V(\prm B)$,
there exists an optimal Steiner tree $T_S$ that connects $S$
in $B$ such that no vertex of degree at least $3$ in $T_S$
is strictly enclosed by $C$.
\end{enumerate}
\end{theorem}

Finally, we are ready to describe the decomposition. Apply the algorithm of Theorem~\ref{thm-over:ananas} to $B$, and let $C$ denote the resulting cycle. We can then decompose the brick as in Figure~\ref{fig-over:pineapple}, meaning that the area between $C$ and $\prm B$ is partitioned into a number of small subbricks of total perimeter $\Oh(|\prm B|)$. Here we use the second property of $C$ that is promised by Theorem~\ref{thm-over:ananas} to build the sides of the subbricks. We recursively apply Theorem~\ref{thm:main} to these subbricks, and let $H$ denote the union of the resulting subgraphs. Then we add to $H$ for each pair of vertices of $C$ a shortest path in the area enclosed by $C$ between the two vertices if that shortest path has length at most $|\prm B|$.
The linear bound on the total perimeter of the subbricks enables a similar analysis as in the proof of Lemma~\ref{lem-over:recursion}. We then choose $\nice = 1/36$. This concludes the proof of Theorem~\ref{thm:main}. 

\subsection{Extending to graphs of bounded genus} \label{ss-over:genus}
In this section, we informally argue how to extend Theorem~\ref{thm:main} to graphs of bounded genus. A detailed statement of the result and a full proof can be found in Section~\ref{sec:genus}.

We use the framework of Borradaile et al.~\cite{cora:genus}:
the idea is to reduce the bounded-genus case to the planar case by cutting the graph embedded on a surface of bounded genus into a planar graph using a cutset of small size.
That is, as in~\cite{cora:genus},
     given a brick embedded on a surface of genus $g$ (i.e., a graph with a designated face),
we may cut along a number of ``short'' cutpaths to make the brick planar,
at the cost of extending the perimeter and the diameter of the brick by an additive factor
of $\Oh(gd)$, where $d$ is the diameter of $G$.
However, in our case, $d$ can be bounded by the perimeter of the brick, as vertices further from the perimeter may be safely discarded.

\subsection{Applications of Theorem~\ref{thm:main}} \label{ss-over:applications}
In this section, we briefly sketch how to prove Theorems~\ref{thm:pst-intro},~\ref{thm:psf-intro}, and~\ref{thm:emwc-intro}. Full proofs appear in Section~\ref{sec:applications}.

\begin{proof}[Proof sketch of Theorem~\ref{thm:pst-intro}]
We manipulate the graph such that all terminals lie on the outer face. We first find a $2$-approximate Steiner tree $\Tapx$ for $\terms$ in $G$. We then cut the plane open along $\Tapx$, cf.~\cite{klein:planar-st-eptas}. That is, we create an Euler tour of $\Tapx$ that traverses each edge twice in different directions and respects the plane embedding of $\Tapx$. Then we duplicate every edge of $\Tapx$, replace each vertex $v$ of $\Tapx$ with $d-1$ copies of $v$, where $d$ is the degree of $v$ in $\Tapx$, and distribute the copies in the plane embedding so that we obtain a new face $F$ with boundary corresponding to the aforementioned Euler tour. Fix the embedding of the resulting graph $\nG$ such that $F$ is its outer face. Note that the terminals $\terms$ lie only on the outer face of $\nG$, and that $|\prm \nG| \le 4k_{OPT}$. Apply Theorem~\ref{thm:main} to $\nG$ to obtain $\hat{H}$, which is of size $\Oh(|\prm \nG|^{142}) = \Oh(k_{OPT}^{142})$. As an optimal Steiner tree $T$ in $G$ splits into a family of trees in $\nG$ that each connect subsets of $V(T) \cap V(\prm \nG)$, the projection of $\hat{H}$ onto $G$ yields the desired set $F \subseteq E(G)$.
\end{proof}
To prove Theorem~\ref{thm:psf-intro}, we compute a simple approximate solution and remove all edges that are farther from a terminal than the size of this approximate solution. We then apply the same idea as in Theorem~\ref{thm:pst-intro} to each of the resulting connected components.

The idea behind the proof of Theorem~\ref{thm:emwc-intro} is that the \textsc{Edge Multiway Cut} problem becomes a \textsc{Steiner Forest}-like problem in the dual graph.
Hence, we cut open the dual of $G$ similarly as we cut open $G$ in Theorem~\ref{thm:pst-intro}:
for each terminal $t$ of $G$, we take the cycle $C_t$ in the dual of $G$ that consists of
all edges incident to $t$, and cut the dual along a short connected subgraph containing
all cycles $C_t$ for all terminals of $G$.
We show that to preserve an optimal solution for \textsc{Edge Multiway Cut} in $G$
it suffices to preserve an optimal Steiner tree for any choice of the terminals on the perimeter
of the obtained brick. Hence, to apply Theorem~\ref{thm:main}, we need to bound
the length of the perimeter, that is, the length of the subgraph of the dual of $G$
that we cut along. By standard reductions, the total length of the cycles
$C_t$ (i.e., the total number of edges incident to terminals) is bounded
by $2k_{OPT}$, where $k_{OPT}$ is the optimal solution size.
Hence, it suffices to bound the diameter of the dual of $G$.

To this end, we fix a terminal $t$ and choose an inclusion-wise maximal laminar family of 
minimal separators that separate $t$ from the remaining terminals and that are maximally ``pushed away''
from $t$ (that is, they are important separators in the sense of~\cite{Marx06}).
By the ``pushed away'' property of the chosen family, each chosen separator is of different size, 
and as there are at most $2k_{OPT}$ edges incident
to the terminals, the largest chosen separator is of size at most $2k_{OPT}$. Hence,
there are $\Oh(k_{OPT}^2)$ edges in this chosen laminar family of minimal separators.

The essence of the proof is to show that an edge that is ``far'' from the chosen family of
separators
is irrelevant for the problem, and may be safely contracted.
Here, ``far'' means $ck_{OPT}$ for some universal constant $c$. Intuitively,
if such an edge $e$ is chosen in an optimal solution $X$, then the connected component of $X$
of the dual of $G$ that contains $e$ lives between two separators from the chosen family,
and we can show that it can be replaced by (a part of) one of these two separators.

Hence, after this reduction is performed exhaustively, the diameter of the dual of $G$
is bounded by $\Oh(k_{OPT}^3)$. Consequently,
  cutting the graph open and applying Theorem~\ref{thm:main}
leads to a polynomial kernel.

Using the extension of Theorem~\ref{thm:main} to graphs of bounded genus, we can extend Theorem~\ref{thm:pst-intro} and~\ref{thm:psf-intro}, and part 1 of Corollary~\ref{cor:subexp} to such graphs (see Section~\ref{sec:genus} for details).

\subsection{The weighted variant: overview of the proof of Theorem~\ref{thm:weighted}}\label{ss-over:weighted}

We now focus on the weighted variant, and sketch the proof of Theorem~\ref{thm:weighted}. A full proof appears in Section~\ref{sec:weights}.

We start with a base case, where $\mathcal{\terms}$ consists of a single terminal pair and $H$ must contain a Steiner forest $F_{H}$ that connects $\mc{\terms}$ such that $\wei{F_{H}} \leq (1+\eps) \wei{F_{B}}$ for any Steiner forest $F_{B}$ in $B$ that connects $\mc{\terms}$.
This base case has been already resolved by Klein~\cite{klein:stoc06}
in the context of an approximation scheme for \textsc{Subset TSP}.

With the base case of a single terminal pair in mind, we move to the $\portalbound$-variant of Theorem~\ref{thm:weighted}, where $\mathcal{\terms}$ is allowed to contain only
$\portalbound$ terminal pairs and the obtained bound for $\wei{H}$
depends polynomially both on $\eps^{-1}$ and $\portalbound$.
In this proof, we use the entire power of the structural results and decomposition
methods developed for the proof of Theorem~\ref{thm:main}, adjusted to the edge-weighted case.
In short, we show that if we decompose each brick recursively into smaller bricks,
stopping when the perimeter of the brick drops below some
threshold $\poly(\eps/\portalbound) \wei{\prm B}$, then we can take the single-pair
graph $H$ developed previously in each such small brick, and the union
of all such graphs has the desired properties. 
The crux of the analysis is that the bound $\portalbound$ ensures that we can ``buy'' the entire perimeter of each small brick in which some vertex of degree at least three of an optimal Steiner forest of $B$ is present.

Finally, we use the partitioning methods from the EPTAS~\cite{klein:planar-st-eptas}, the so-called mortar graph framework,
to derive Theorem~\ref{thm:weighted}
from the $\portalbound$-variant. The mortar graph constructed by~\cite{klein:planar-st-eptas} is essentially a brickable connector. We call the bricks induced by this connector \emph{cells}. The mortar graph has the property that there exists a near-optimal
Steiner forest in $B$ that crosses each cell at most
$\alpha(\eps) = o(\eps^{-5.5})$ times. Therefore, we construct the mortar graph of the input brick and then apply $\portalbound$-variant to each cell 
independently, for an appropriate choice of $\portalbound = \poly(\eps^{-1})$.
This then yields the desired graph $H$.

%!TEX root = pst-kernel.tex

\section{Preliminaries}\label{sec:preliminaries}

We use standard graph notation, see e.g.~\cite{diestel}.
All our graphs are undirected and, unless otherwise stated, simple.
For a graph $G$, by $V(G)$ and $E(G)$ we denote its vertex- and edge-set, respectively.
For $v \in V(G)$, the neighbourhood of $v$ is defined as $N_G(v) = \{u: uv \in E(G)\}$
and the closed neighbourhood of $v$ as $N_G[v] = N_G(v) \cup \{v\}$. We extend these notions
to sets $X \subseteq V(G)$ as $N_G[X] = \bigcup_{v \in X} N_G[v]$ and $N_G(X) = N_G[X] \setminus X$.
We omit the subscript if the graph is clear from the context.

For a subgraph $H$ of $G$, we silently identify $H$ with the edge set of $H$; that is,
all our subgraphs are edge-induced. In particular, this applies to all paths, walks, and cycles; we
treat them as sequences of edges.

In this paper, we work with both unweighted and edge-weighted graphs.
An \emph{edge-weighted graph} is a graph $G$ equipped with a weight function $\weismb: E(G) \to (0,+\infty)$.
We explicitly disallow zero-cost edges in the input graph.
For any edge $e \in E(G)$, the value $\wei{e}$ is the \emph{length} or \emph{weight} of the edge $e$.
For any subgraph $H$ of $G$ (in particular, for any cycle or path in $G$), the \emph{length} or \emph{weight} of $H$
is defined as $\wei{H} = \sum_{e \in H} \wei{e}$.
An \emph{unweighted graph} is an edge-weighted graph with weight function $\wei{e} = 1$ for each edge $e$, i.e., $\wei{H} = |H|$ for any subgraph $H$.

The distance between two vertices is the length of a shortest path
between them. The distance between two vertex sets is the minimum distance
between pairs of vertices in the sets.
The distance between two (sets of) edges is the minimum distance between
the endpoints of the edges.
By $\dist_G(X,Y)$ we denote the distance between objects (vertices, vertex sets, edge sets) $X$ and $Y$ in the graph $G$.

By $\plane$ we denote the standard euclidean plane.
Let $G$ be a plane graph, that is, a graph embedded on plane $\plane$.
Let $\gamma$ be a closed curve on the plane, that is, a continuous image of a circle.
We say that $\gamma$ {\em{strictly encloses}} a point $c$ on the plane if $c$ does not lie on $\gamma$ and $\gamma$ is not the neutral element of the fundamental group
of $\plane \setminus \{c\}$ (note that this fundamental group is isomorphic to $\Z$) or, equivalently, $c$ does not lie on $\gamma$ and $\gamma$
is not continuously retractable to a single point in $\plane \setminus \{c\}$.
We say that $\gamma$ {\em{encloses}} $c$ if $\gamma$ strictly encloses $c$ or $c$ lies on $\gamma$.
We often identify closed walks in the graph $G$ with the closed curves that they induce in the planar embedding; thus, we can say that a closed walk in the graph (strictly) encloses $c$.
We extend these notions to vertices, edges, and faces of a graph $G$: a vertex is (strictly) enclosed if its drawing on the plane is (strictly) enclosed,
and edge is (strictly) enclosed if all interior points of its drawing are (strictly) enclosed, and a face is (strictly) enclosed if all points of its interior
are (strictly) enclosed. We also say that a closed walk in $G$ (strictly) encloses some object if the drawing of this closed walk
(strictly) encloses the object.
Note that if $C$ is a simple cycle in $G$, then its drawing is a closed curve without self-intersections, and the notion of (strict) enclosure
coincides with the intuitive meaning of these terms.

%We say that a vertex, edge, or face of $G$ is \emph{enclosed by a simple cycle $C$} of $G$, if it is contained in the interior of the bounded Jordan region defined by the Jordan curve that the cycle $C$ induces in the plane embedding. We now introduce the key definitions that will be used in all our reasonings.

\begin{definition}
A connected plane graph $B$ is called a {\em brick} if the boundary of the infinite face of $B$ is a simple cycle. This cycle is then called the {\em perimeter} of the brick, and denoted by $\prm B$. The {\em interior} of the brick, denoted $\intr B$, is the graph induced by all the edges not lying on the perimeter, that is, $\intr B := B \setminus \prm B$.
\end{definition}

Note that for a brick $B$, all the edges of $\intr B$ as well as all the vertices of $\intr B$ not lying on $\prm B$, are strictly
enclosed by $\prm B$. For a brick $B$, every face of $B$ enclosed by $\prm B$ is called an {\em inner face}. 

For a path $P$,
we denote by $P[a,b]$ the subpath of $P$ starting in vertex $a$ and ending in vertex $b$. This definition is extended to the perimeter
$\prm B$ of a brick $B$ in the following way. We denote by $\prm B[a,b]$ the subpath of $\prm B$ obtained by traversing $\prm B$ in 
counter-clockwise direction from $a$ to $b$. On the other hand, for a tree $T$ we denote by $T[a,b]$ the unique path in $T$ between 
vertices $a$ and $b$. 

We also need the following notation. Let $T$ be a tree embedded in the plane, and let $uv$ be an edge of $T$.
The {\em{subtree of $T$, rooted at $v$, with parent edge $uv$}} is the connected component of $T \setminus \{uv\}$ that contains $v$,
rooted in $v$, equipped with the following order on the children of each node $w$: order the children of $w$ in counter-clockwise
order starting from the parent of $w$ if $w \neq v$ and with the edge $uv$ if $w = v$.
%Generally, by the {\em{leftmost}} (or {\em{rightmost}}) object
%of the subtree of $T$, rooted at $v$ with parent edge $uv$, is the first (last) such object in the pre-order traversal of the subtree (visiting children in the aforementioned order).
We say that $a$ and $b$ are the {\em{leftmost and rightmost elements of $V(T_v) \cap V(\prm B)$}}, respectively,
if $a, b \in V(T_v) \cap V(\prm B)$, $V(T_v) \cap V(\prm B) \subseteq V(\prm B[a,b])$, and the face of $\prm B \cup V(T_v)$
that contains the edge $uv$ is incident to the edges of $\prm B[b,a]$.

\subsection{Problem definitions}

For completeness, we formally state the problems considered in this paper.\\

\defproblemnoparam{\pST{}}{An edge-weighted planar graph $G$, a set of terminals $\terms \subseteq V(G)$.}{Find a connected subgraph $T$ of $G$ of minimum possible length such that
  $\terms \subseteq V(T)$ (i.e., $T$ connects $\terms$).}

\defproblemnoparam{\pSF{}}{An edge-weighted planar graph $G$, a family of pairs of terminals
  $\termpairs \subseteq V(G) \times V(G)$.}{Find a subgraph $H$ of $G$ of minimum possible length such that
 for each $(s,t) \in \termpairs$, the terminals $s$ and $t$ lie in the same connected
   component of $H$.}
\ \\

Observe that \pST{} reduces to \pSF{} by taking the family $\termpairs$ to be $\terms \times \terms$.

As we study \pemwcname{} only in unweighted graphs, we state this problem in the unweighted
setting only.\\

\defproblemnoparam{\pemwcname{} (\pemwc{})}{A planar graph $G$, a set of terminals $\terms \subseteq V(G)$.}{Find a minimum set of edges $X$ such that no two terminals lie in the same connected component of $G \setminus X$.}
\ \\

In the bounded-genus case, we assume that the input graph
is given together with an embedding into a surface of genus $g$
such that the interior of each face is homeomorphic to an open disc.

%!TEX root = pst-kernel.tex

\section{The case of a nicely decomposable brick}\label{sec:bricks}

Sections~\ref{sec:bricks}--\ref{sec:dp} are devoted
to the proof of Theorem~\ref{thm:main}.
However, in most places we take a more general view and argue about edge-weighted graphs, as we would like to re-use the
obtained structural results in the weighted variant, discussed in Section~\ref{sec:weights}.
Hence, unless otherwise stated, all graphs are equipped with a weight function $\weismb$.

We first give formal definitions of the brick decomposition and related notions, and proceed to define what it means for a brick to be nicely decomposable. Then we explain how Theorem~\ref{thm:main} can be applied recursively.

\begin{definition}
We say that a brick $B'$ is {\em{a subbrick}} of a brick $B$ if $B'$ is a subgraph of $B$
consisting of all edges enclosed by $\prm B'$.
\end{definition}

\begin{definition}
For a brick $B$, a {\em brick covering} of $B$ is a family $\Bb=\{B_1,B_2,\ldots,B_p\}$ of bricks, such that {\em (i)} each $B_i$, $1 \leq i \leq p$,
is a subbrick of $B$, and {\em (ii)} each face of $B$ is contained in at least one brick $B_i$, $1 \leq i \leq p$.
A brick covering is called a {\em{brick partition}} if each face of $B$ is contained in {\em{exactly one}} brick $B_i$.
\end{definition}

Let us now discuss the notion of brick partition.
If $\Bb = \{B_1,B_2,\ldots,B_p\}$ is a brick partition of $B$,
then it follows that every edge of $\prm B$ belongs to the perimeter of exactly one brick $B_i$,
while every edge of $\intr B$ either is in the interior of exactly one brick $B_i$,
or lies on perimeters of exactly two bricks $B_i$, $B_j$ for $i\neq j$.
%The set of edges of the interior of $B$ that lie on the perimeter of some $B_i$
%is called the {\em skeleton} of $\Bb$, i.e., the skeleton is the set $(\bigcup_{B'\in \Bb} \prm B')\setminus \prm B$.

Any connected set $F\subseteq B$ will be called a {\em connector}.
Let $\terms$ be the set of vertices of $\prm B$ adjacent to at least one edge of $F$;
the elements of the set $\terms$ will be called the {\em anchors} of the connector $F$.
We then say that $F$ {\em connects} $\terms$.
For a connector $F$, we say that $F$ is {\em optimal} if there is no connector $F'$ with $\wei{F'}<\wei{F}$
that connects a superset of the anchors of $F$.
Clearly, each optimal connector $F$ induces a tree, whose every leaf is an anchor of $F$.
For a connector $F$, every part of $\prm B$ between two consecutive anchors of $F$ will be called an {\emph{interval of $F$}}.

We say that a connector $F\subseteq B$ is {\em brickable} if the boundary of every inner face of $\prm B \cup F$ is a simple cycle, i.e., these boundaries form subbricks of $B$.
Let $\Bb$ be the corresponding brick partition of $B$; observe that then $\sum_{B'\in \Bb} \wei{\prm B'}\leq \wei{\prm B}+2\wei{F}$.
Note that a tree is brickable if and only if all its leaves lie on $\prm B$
and, consequently, every optimal connector is brickable.
We now move to the definition of one of the crucial notions that explains which partitions and coverings
can be used for the recursive step.

\begin{definition}
The {\em{total perimeter}} of a brick covering $\Bb = \{B_1,\ldots,B_p\}$ is defined as $\sum_{i=1}^p \wei{\prm B_i}$.
For a constant $c > 0$, $\Bb$ is {\em{$c$-short}} if the total perimeter of $\Bb$ is at most $c \cdot \wei{\prm B}$.
For a constant $\nice > 0$, $\Bb$ is {\em{$\nice$-nice}} if $\wei{\prm B_i} \leq (1-\nice) \cdot \wei{\prm B}$ for each $1 \leq i \leq p$.

Similarly, a brickable connector $F\subseteq B$, with $\Bb = \{B_1,\ldots,B_p\}$ the corresponding brick partition,
is {\em{$c$-short}} if $\Bb$ is $c$-short,
simply {\em{short}} if it is $3$-short,
and $F$ is \emph{$\nice$-nice} if $\Bb$ is $\nice$-nice.
%if {\emph{(i)}} $|F| < |\prm B|$, and {\emph{(ii)}} $|\prm B_{i}| \leq (1-\nice) \cdot |\prm B|$ for all $i=1,\ldots,p$.
\end{definition}

Observe that for a brickable connector $F\subseteq B$, if $\wei{F} \leq \wei{\prm B} \cdot (c-1)/2$, then $F$ is $c$-short, and in particular if $\wei{F}\leq \wei{\prm B}$, then $F$ is $3$-short.
Moreover, if $F$ is a tree with leaves on $\prm B$ and of length at most $\wei{\prm B}$, then $F$ is a short brickable connector.
Such a tree is called a \emph{$3$-short tree} (or just {\emph{short}} instead of $3$-short, for simplicity).
The following theorem is needed to make our proof algorithmic.

\begin{theorem}\label{thm:nice-testing}
Let $\nice > 0$ be a fixed constant.
Given an unweighted brick $B$, in $\Oh(|\prm B|^8 |B|)$ time
one can either correctly conclude that no short $\nice$-nice tree exists in $B$
or find a short $\nice$-nice brick covering of $B$.
\end{theorem}

A slightly more technical variant of Theorem~\ref{thm:nice-testing}, in the edge-weighted setting,
  is stated in Section~\ref{sec:dp}.
The proofs of Theorem~\ref{thm:nice-testing} and its edge-weighted counterpart, given in Section~\ref{sec:dp}, are a technical modification of the classical algorithm of Erickson et al.~\cite{erickson} that computes an optimal Steiner tree in a planar graph assuming that all the terminals lie on the boundary of the infinite face. For technical reasons, we cannot ensure that if some short $\nice$-nice tree exists, then the output of the algorithm of Theorem~\ref{thm:nice-testing} will actually be a brick partition corresponding to some short $\nice$-nice tree. Instead, the algorithm may output a brick covering, but one that is guaranteed to be short and nice for some choice of constants.
Fortunately, this property is sufficient for our needs.

Armed with Theorem~\ref{thm:nice-testing} and the notion of brick partition and covering,
we may now describe the recursive step in the algorithm of Theorem~\ref{thm:main}.
Thus, in the rest of this section we work with unweighted bricks only, and $\wei{H} = |H|$
for any subgraph $H$.
The following lemma is the main technical contribution of this section.

\begin{lemma}\label{lem:recursion}
Let $c, \nice > 0$ be constants.
Let $B$ be an unweighted brick and let $\Bb = \{B_1,\ldots,B_p\}$ be a $c$-short $\nice$-nice brick covering of $B$.
Assume that the algorithm of Theorem~\ref{thm:main} was applied recursively to bricks $B_1,\ldots, B_p$,
and let $H_1,\ldots,H_p$ be the subgraphs output by this algorithm for $B_1,\ldots,B_p$, respectively,
where $|H_i|\leq C \cdot |\prm B_i|^\alpha$ for some constants $C>0$ and $\alpha\geq 1$
such that $(1-\nice)^{\alpha-1}\leq \frac{1}{c}$.
Let $H=\bigcup_{i=1}^p H_i$.
Then $H$ satisfies conditions {\em (i)-(iii)} of Theorem~\ref{thm:main}, with $|H|\leq C\cdot |\prm B|^\alpha$.
\end{lemma}
\begin{proof}
To see that $H$ satisfies condition (i), note that every edge of $\prm B$ is in the perimeter of some brick $B_i$, and that $\prm B_i\subseteq H_i$ for every $i=1,2,\ldots,p$. Therefore, $\prm B\subseteq H$.

To see that $H$ satisfies condition (ii), recall that $\Bb$ is $c$-short and that $|\prm B_i|\leq (1-\nice) \cdot |\prm B|$ for each $i=1,2,\ldots,p$. Therefore, $|\prm B_i|^\alpha \leq |\prm B_i|\cdot (1-\nice)^{\alpha-1}|\prm B|^{\alpha-1}$, and
\begin{eqnarray*}
|H| &\leq & \sum_{i=1}^p |H_i| \\
& \leq & C\cdot \sum_{i=1}^p|\prm B_i|^\alpha\\
& \leq & C\cdot (1-\nice)^{\alpha-1}|\prm B|^{\alpha-1}\cdot \sum_{i=1}^p |\prm B_i|\\
& \leq & c\cdot (1-\nice)^{\alpha-1} \cdot C \cdot |\prm B|^{\alpha} \\
& \leq & C \cdot |\prm B|^\alpha.
\end{eqnarray*}

Finally, to see that $H$ satisfies condition (iii), let $\terms\subseteq V(\prm B)$ be a set of terminals lying on the perimeter of $B$, and let $T$ be an optimal Steiner tree connecting $\terms$ in $B$ that contains a minimum number of edges that are not in $H$.
We claim that $T \subseteq H$. Assume the contrary, and let $e \in T \setminus H$.
Since each face of $B$ is contained in some brick of $\Bb$, there exists a brick $B_i$ such that $\prm B_i$ encloses $e$.
As $\prm B_i \subseteq H_i \subseteq H$, we infer $e \notin \prm B_i$.
Consider the subgraph $T \cap \intr B_i$ (i.e., the part of $T$ strictly enclosed by $\prm B_i$)
and let $X$  be the connected component of this subgraph that contains $e$.
Clearly, $X$ is a connector inside $B_i$.
Since $H_i$ is obtained by a recursive application of Theorem~\ref{thm:main},
there exists a connected subgraph $D \subseteq H_i$ that connects the anchors of $X$
and that satisfies $|D| \leq |X|$. Let $T' = (T \setminus X) \cup D$.
Observe that $|T'| \leq |T|$ and that $T'$ contains strictly less edges that are not in $H$ than $T$ does.
Since $D$ connects the anchors of $X$ in $H_i$, $T'$ still connects the anchors of $T$ in $B$, that is, $T'$ connects $\terms$.
However, $T$ is an optimal Steiner tree that connects $\terms$, and thus $T'$ is also an optimal Steiner tree that connects $\terms$. Since $T'$ contains strictly less edges that are not in $H$ than $T$,
this contradicts the choice of $T$. Hence, $T \subseteq H$.
\end{proof}

We may now sketch the first step of our kernelization algorithm of Theorem~\ref{thm:main};
a formal argument is provided in Section \ref{sec:finish}.
We run the algorithm of Theorem~\ref{thm:nice-testing} for the brick $B$
and some fixed small constant $\nice > 0$ (to be chosen later).
If the algorithm returns a short $\nice$-nice brick covering
$\Bb = \{B_1,B_2,\ldots,B_p\}$, then we recurse on each brick $B_i$, obtaining
a graph $H_i$ of size bounded polynomially in $|\prm B_i|$.
By Lemma \ref{lem:recursion},
the assumptions of shortness and $\nice$-niceness
yield a polynomial bound on $|\bigcup_{i=1}^p H_i|$ in terms of
$|\prm B|$, where the exponent $\alpha$ is chosen large enough so that $(1-\nice)^{\alpha-1}<\frac{1}{3}$.
If the algorithm of Theorem~\ref{thm:nice-testing} concluded that brick $B$ does not contain any short $\nice$-nice tree, then we proceed to the arguments in the
next sections with this assumption.

%!TEX root = pst-kernel.tex

\section{Carves and the core}\label{sec:core}

Let $B$ be a possibly edge-weighted brick.
We are now working with the assumption that $B$ does not contain any short $\nice$-nice tree for some $\nice > 0$.
In this section, we define the notion of carving a small portion of the brick, which will be a crucial technical ingredient in our further reasonings.
In particular, we formalize the intuition that if no short $\nice$-nice tree can be found, then $B$ contains a well-defined middle region, and each attempt of carving out some part of $B$ using a limited budget cannot affect this middle region. In the following, we use a constant $\delta\in (0,\frac{1}{2})$ to be determined later.

We start by formalizing what we mean by `carving'.

\begin{definition}
A \emph{$\delta$-carve} $L$ from a brick $B$ is a pair $(P,I)$, where $P$ (called the \emph{carvemark}) is a path in $B$ between two distinct vertices $a,b \in V(\prm B)$ of length at most $(\frac{1}{2}-\delta) \cdot \wei{\prm B}$, and $I$ (called the \emph{carvebase}) is a shortest of the two paths $\prm B[a,b], \prm B[b,a]$. If $P$ has only two common vertices with $\prm B$, i.e., $V(P)\cap V(\prm B)=\{a,b\}$, then the $\delta$-carve is {\emph{strict}}. The subgraph enclosed by the closed walk $P\cup I$ is called the {\emph{interior}} of a $\delta$-carve.
\end{definition}

Observe that if a $\delta$-carve $(P,I)$ is strict, then $P\cup I$ is a simple cycle and thus the interior of $(P,I)$ is a brick. We often identify a strict $\delta$-carve with this brick.

In the following lemma, we observe that if a brick does not admit any short $\nice$-nice trees, then the carvebases cannot be much longer than the carvemarks.

\begin{lemma}\label{lem:pluseps}
For any $\nice,\delta \in (0,\frac{1}{2})$, $\delta>\nice$, if $B$ admits no short $\nice$-nice tree, then the base $I$ of any $\delta$-carve $(P,I)$ in $B$ has length at most $\wei{P}+\nice \wei{\prm B}$.
\end{lemma}
\begin{proof}
Consider a $\delta$-carve $L=(P,I)$ with the carvemark $P$ between vertices $a,b$, such that $I=\prm B[a,b]$. Let $I'=\prm B[b,a]$. Assume on the contrary that $\wei{I}>\wei{P}+\nice \wei{\prm B}$. Then $\wei{I}\in (\wei{P}+\nice \wei{\prm B},\frac{1}{2}\wei{\prm B}]$. Hence, $\wei{I'}\in [\frac{1}{2}\wei{\prm B},(1-\nice)\wei{\prm B}-\wei{P})$. Clearly, $P$ is a brickable connector in $B$. Let $\Bb$ be the corresponding brick partition of $B$. Note that each brick $B'\in \Bb$ has its perimeter contained entirely in either $P\cup I$ or $P\cup I'$. Since
$$\wei{P\cup I}, \wei{P\cup I'} \leq \wei{P} + (1-\nice)\wei{\prm B}-\wei{P}\leq (1-\nice)\wei{\prm B},$$
$B'$ has perimeter at most $(1-\nice)\wei{\prm B}$. As $\wei{P} \leq \wei{\prm B}$, $P$ is a short $\nice$-nice tree, a contradiction.
\end{proof}

By applying Lemma~\ref{lem:pluseps} to the maximum length of a carvemark, that is, $(\frac{1}{2}-\delta) \cdot \wei{\prm B}$, we obtain the following corollary.

\begin{corollary}\label{cor:small-base}
For any $\nice \in (0,\frac{1}{4})$ and any $\delta \in [2\nice, \frac{1}{2})$, if $B$ admits no short $\nice$-nice tree, then the base of any $\delta$-carve $L=(P,I)$ in $B$ has length at most $(\frac{1}{2} - \frac{\delta}{2}) \cdot \wei{\prm B}$. In particular, $\wei{P}+\wei{I}\leq (1-\frac{3}{2}\delta)\wei{\prm B}< (1-\nice)\wei{\prm B}$.
\end{corollary}

Note that Corollary~\ref{cor:small-base} implies that, under its assumptions, the base of a carve is unique. Moreover, we can make
the following observation. Recall that a tree $T$ in $B$ is brickable if and only if all its leaves lie on $\prm B$.

\begin{lemma}\label{lem:short-connector}
For any $\nice \in (0,\frac{1}{4})$ and any $\delta \in [2\nice, \frac{1}{2})$, if $B$ admits no short $\nice$-nice tree, then
for any brickable short tree $T$ with diameter not bigger than $(\frac{1}{2}-\delta)\cdot \wei{\prm B}$ there exists an interval $I_T$ of $\prm B$
of length at most $(\frac{1}{2} -\frac{\delta}{2})\wei{\prm B}$ such that all anchors of $T$ are in $I_T$.
\end{lemma}
\begin{proof}
Observe that $T$ is short, but not $\nice$-nice. Hence, there exists a brick $B'$ induced by $T$ of perimeter bigger than $(1-\nice)\wei{\prm B}$.
The intersection of $\prm B'$ with $T$ cannot be longer than the diameter of $T$, so $\prm B' \setminus T$, which is
an interval $I$ on $\prm B$, has length at least
$$\left(1-\nice\right)\wei{\prm B} - \left(\frac{1}{2}-\delta\right)\wei{\prm B} = \left(\frac{1}{2}-\nice + \delta\right)
\wei{\prm B} \ge \left(\frac{1}{2}+\frac{\delta}{2}\right)\wei{\prm B}.$$
All other anchors of $T$ need to be contained in the interval $I_T = \prm B \setminus I$,
which is of length at most $(\frac{1}{2} -\frac{\delta}{2})\wei{\prm B}$.
\end{proof}

\begin{figure}[h]
\centering
\svg{0.8\textwidth}{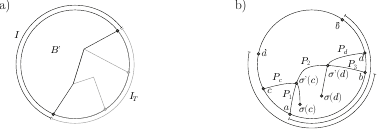}
\caption{Panel (a) illustrates the proof of Lemma~\ref{lem:short-connector}, whereas
panel (b) refers to Claim~\ref{cl:multi-projection}.}
\label{fig:connectors}
\end{figure}

We now proceed to defining the region that can be carved out by some $\delta$-carve.

\begin{definition}
A subgraph $F$ of $B$ can be \emph{$\delta$-carved} if there is a $\delta$-carve $L$ of $B$ such that $F$ is also a subgraph of the interior of $L$.
\end{definition}

In particular, a vertex, edge, or face of $B$ can be $\delta$-carved if there is a $\delta$-carve $L$ of $B$ that encloses this vertex, edge, or face. One can also define a similar notion for strict $\delta$-carves. The following lemma shows that in the case that is of our interest, the two notions coincide.

\begin{lemma}\label{lem:strict-equiv}
For any $\nice \in (0,\frac{1}{4})$ and any $\delta \in [2\nice, \frac{1}{2})$, if a brick $B$ admits no short $\nice$-nice tree, then a face $f$ of $B$ can be $\delta$-carved if and only if it can be strictly $\delta$-carved.
\end{lemma}
\begin{proof}
By definition, a face that can be strictly $\delta$-carved can also be $\delta$-carved. Therefore, we proceed to prove the converse. Let $f$ be a face enclosed by a $\delta$-carve $L=(P,I)$, where $I=\prm B[a,b]$ for two vertices $a,b \in V(\prm B)$. We may assume that $|P \cap V(\prm B)|$ is minimum among all $\delta$-carves that $\delta$-carve $f$.

We prove that $|P\cap V(\prm B)|=2$, and thus $L$ is strict. Suppose for sake of contradiction that $|P \cap V(\prm B)| > 2$, and let $c$ be any internal vertex of $P$ that lies on $\prm B$. We consider two cases.

In the first case, suppose that $c\in V(I)$. Observe that $L_1=(P[a,c],\prm B[a,c])$ and $L_2=(P[c,b],\prm B[c,b])$ are both $\delta$-carves, because $\wei{\prm B[a,c]},\wei{\prm B[c,b]}<\wei{\prm B[a,b]}\leq \frac{1}{2}\wei{\prm B}$ and also $\wei{P[a,c]},\wei{P[c,b]}\leq \wei{P}\leq (\frac{1}{2}-\delta) \cdot \wei{\prm B}$. Moreover, at least one of these $\delta$-carves encloses $f$. Since the carvemarks of $L_1$ and $L_2$ contain less vertices of $\prm B$ than $L$ does, this contradicts the choice of $L$.

In the second case, suppose that $c\notin V(I)$. Observe that $P$ is a brickable tree of diameter and size at most $(\frac{1}{2}-\delta) \cdot \wei{\prm B}$
that connects three anchors $a$, $b$, and $c$. Consequently, let $I_P$ be the interval whose existence is asserted by Lemma~\ref{lem:short-connector} for $P$.
As $a,b \in V(I_P)$ and $\wei{\prm B[b,a]} > \wei{\prm B}/2$ by Corollary~\ref{cor:small-base}, we have that $\prm B[a,b] \subseteq I_P$.
Therefore, either $\prm B[a,c]$ or $\prm B[c,b]$ is contained in $I_P$, and thus has length at most $(\frac{1}{2} -\frac{\delta}{2})\wei{\prm B}$.
Without loss of generality, assume that it is $\prm B[a,c]$.
In this case, $L'=(P[a,c],\prm B[a,c])$ is a $\delta$-carve that encloses $f$, because it encloses a superset of the faces enclosed by $(P,I)$.
Since the carvemark of $L'$ contains less vertices of $\prm B$ than $L$ does, we contradict the choice of $L$.
%
%
%Let again $L_1=(P[a,c],\prm B[a,c])$ and $L_2=(P[c,b],\prm B[c,b])$. We claim that at least one of these pairs is a $\delta$-carve. Observe that this will be a contradiction, as this $\delta$-carve would enclose a supergraph of what $L$ enclosed, while containing strictly less vertices of $\prm B$ on the carvemark. Observe that it still holds that $\wei{P[a,c]},\wei{P[c,b]}\leq \wei{P}\leq (\frac{1}{2}-\delta) \cdot \wei{\prm B}$. Assume then that neither $L_1$ nor $L_2$ is a $\delta$-carve; hence, it must necessarily hold that $\wei{\prm B[a,c]}, \wei{\prm B[c,b]}>\frac{1}{2}\wei{\prm B}$. Consequently, $\wei{\prm B[c,a]},\wei{\prm B[b,c]}\leq \frac{1}{2}\wei{\prm B}$. Now recall that $P$ is a brickable connector in $B$, and let $\Bb$ be the corresponding brick partition of $B$. Again, since $a,b\in V(\prm B)$ and $\wei{P}\leq \wei{\prm B}$, $P$ would be a short $\nice$-nice tree in $B$ unless at least one of the bricks of $\Bb$ had perimeter larger than $(1-\nice)\wei{B}$. However, each of these bricks have perimeter entirely contained in $P\cup I$, $P\cup \prm B[b,c]$, or $P\cup \prm B[c,a]$, and
%$$\wei{P\cup I}, \wei{P\cup \prm B[b,c]}, \wei{P\cup \prm B[c,a]}\leq \wei{P}+\frac{1}{2}\wei{\prm B}\leq (1-\delta)\wei{\prm B}<(1-\nice)\wei{\prm B},$$
%which is a contradiction.
\end{proof}

%Since $\wei{P} \leq (\frac{1}{2}-\delta) \cdot \wei{\prm B} < \wei{\prm B}$, $P$ constitutes an $\nice$-nice tree in $B$, a contradiction.
%Suppose that $\wei{\prm B[a,b]}$ forms the base of $B'$ and that it has length greater than $(\frac{1}{2} - \frac{\delta}{2}) \cdot \wei{\prm B}$.
%Observe that $P$ is a brickable connector that decomposes $B$ into two subbricks: $B'$ and another brick with perimeter $P\cup \prm B[b,a]$, which we denote by $B''$.
%Since $\wei{\prm B[a,b]} \leq \wei{\prm B[b,a]}$ by the definition of a carve, $\wei{\prm B[b,a]} > (\frac{1}{2} - \frac{\delta}{2}) \cdot \wei{\prm B}$.
%This implies that $\wei{\prm B[a,b]}, \wei{\prm B[b,a]} < (\frac{1}{2} + \frac{\delta}{2}) \cdot \wei{\prm B}$.
%Then

%The intuition is that, in a brick that does not admit any short $\nice$-nice tree for some $\nice\in (0,1/4)$, parts of the brick that can be $\delta$-carved can be seen as areas that are located near the perimeter and not in the middle. Hence, the idea is to prove that there exists some part of $B$ that cannot be $\delta$-carved. This is precisely the statement of the next theorem, which is the main result of this section.
%EJ: we already gave the above intuition at the beginning of the section. Therefore, a shorter paragraph should suffice here.
We now present the main result of this section: there is a middle region of $B$ that cannot be carved out of $B$ using a limited budget, \ie by a $\delta$-carve for some appropriate choice of $\delta$.

\begin{theorem}[Core Theorem]\label{thm:core}
For any $\nice \in (0,\frac{1}{4})$ and any $\delta \in [2\nice, \frac{1}{2})$, if $B$ has no short $\nice$-nice tree, then there exists a face of $B$ that cannot be $\delta$-carved. Moreover, such a face can be found in $O(|B|)$ time.
\end{theorem}
\begin{proof}
We first prove the existential statement, and then show how the proof can be made algorithmic.

%Before we start the proof, let us observe that without loss of generality we may assume that all the inner faces of $B$ are triangles. Indeed, let $\bar{B}$ be any supergraph of $B$ where every inner face of $B$ has been triangulated; note that $\prm B=\prm \bar{B}$. Clearly, if a face $M$ of $B$ could be $\delta$-carved in $B$, then the same $\delta$-carve also $\delta$-carves every face of $\bar{B}$ contained in $M$. Hence, if we prove that there is a face $\bar{M}$ in $\bar{B}$ that cannot be $\delta$-carved, then the face $M$ of $B$ that contains $\bar{M}$ also cannot be $\delta$-carved, and we will be done.

%We proceed with the assumption that all the faces of $B$ are triangular, so in particular the areas enclosed by them are homeomorphic to closed discs. For the sake of contradiction, assume that every face of $B$ can be $\delta$-carved.

Define maps $v\to \proj(v)$ and $v\to \Proj(v)$ for $v\in V(B)$, such that $\proj(v)$ is a vertex of $\prm B$ closest to $v$, and $\Proj(v)$ is a shortest path between $v$ and $\proj(v)$. We can assume for any two vertices $v_1,v_2\in V(B)$ that $\Proj(v_1)$ and $\Proj(v_2)$, when traversed from $v_1$ and $v_2$ respectively, are either disjoint, or when they meet they continue together towards the same vertex of $\prm B$ (implying that $\proj(v_1)=\proj(v_2)$). Such a property can be ensured by constructing maps $\proj,\Proj$ in the following manner: attach a super-terminal $s_0$ adjacent to every vertex of $\prm B$ with unit-weight edges, and apply a linear-time shortest-path algorithm~\cite{planar-sp} from $s_0$. In the obtained shortest-path tree, vertices of $\prm B$ are children of the root $s_0$. For each subtree $T_{v'}$ rooted in a child $v'$ of $s_0$, we set $\proj(v)=v'$ for every vertex $v\in T_{v'}$, and we set $\Proj(v)$ as the path from $v$ to $v'$ in $T_{v'}$.
Note that by the definition of maps $\proj,\Proj$, for any $v\in V(\prm B)$, $\proj(v)$ is equal to $v$ and $\Proj(v)$ is a path of length zero that consists of the single vertex $v$.

Now fix some strict $\delta$-carve $L=(P,\prm B[a,b])$, where $a,b \in V(\prm B)$ are the endpoints of the carvemark $P$ of $L$. Let $B'$ be the subbrick enclosed by $L$.

\begin{claim}\label{cl:multi-projection}
There is an interval $I_{L}$ on $\prm B$ of length at most $(\frac{1}{2} - \frac{\delta}{2}) \cdot \wei{\prm B}$ that {\em{(i)}} contains the carvebase of $L$, and {\em{(ii)}} contains $\proj(v)$ for any $v \in V(B')$.
\end{claim}
\begin{proof}
Let $D := \pi(V(B')) \setminus V(\prm B[a,b])$. If $D = \emptyset$, then $I_L := \prm B[a,b]$ satisfies the desired conditions by Corollary~\ref{cor:small-base}, so assume otherwise.
Let $\sigma: D \rightarrow V(B')$ be any mapping such that $\pi(\sigma(c))=c$ for any $c \in D$.
Note that $\Proj(\sigma(c))$ intersects $P$; let $\sigma'(c)$ be the vertex of $V(\Proj(\sigma(c))) \cap V(P)$ that is closest to $c$ on $\Proj(\sigma(c))$
and let $P_c := \Proj(\sigma(c))[c,\sigma'(c)]$.
Observe that, by the construction of the paths $\Proj(\cdot)$,
for distinct $c,d \in D$, the paths $\Proj(\sigma(c))$ and $\Proj(\sigma(d))$ are vertex-disjoint.

We now show that, for any $c,d \in D$ (where possibly $c=d$), there exists an interval $I_{c,d} \subseteq \prm B$ such that $\wei{I_{c,d}} \leq (\frac{1}{2}-\frac{\delta}{2}) \wei{\prm B}$, $c,d \in V(I_{c,d})$ and $\prm B[a,b] \subseteq I_{c,d}$.
Consider the subgraph $T_{c,d} := P \cup P_c \cup P_d$ (see Figure~\ref{fig:connectors}b).
Observe that $T_{c,d}$ is a brickable tree in $B$ with anchors $a$, $b$, $c$, and $d$.
Without loss of generality, assume that $a$, $\sigma'(c)$, $\sigma'(d)$, and $b$ lie on $P$ in this order (possibly $\sigma'(c) = \sigma'(d)$ if $c=d$).
Denote $P_1 = P[a,\sigma'(c)]$, $P_2 = P[\sigma'(c),\sigma'(d)]$, and $P_3 = P[\sigma'(d), b]$.
As $\Proj(\sigma(c))$ is a shortest path between $\sigma(c)$ and $V(\prm B)$, $\wei{P_c} \leq \wei{P_1}$ and, symmetrically, $\wei{P_d} \leq \wei{P_3}$.
Consequently, the diameter of $T_{c,d}$ is bounded by $\wei{P}$, which is at most $(\frac{1}{2} - \delta)\wei{\prm B}$ by definition, and thus
  $$\wei{T_{c,d}} \leq \wei{P} + \wei{P_c} + \wei{P_d} \leq \wei{P} + \wei{P_1} + \wei{P_3} \leq 2\wei{P} < \wei{\prm B}.$$
Hence, Lemma~\ref{lem:short-connector} applies to $T_{c,d}$, and we obtain an interval of length at most $(\frac{1}{2}-\frac{\delta}{2})\wei{\prm B}$
that contains $a$, $b$, $c$, and $d$. For any $c,d \in D$, let us denote the interval obtained 
this way by $I_{c,d}$. As $\wei{\prm B[b,a]} > \wei{\prm B}/2$, we have $\prm B[a,b] \subseteq I_{c,d}$. Hence, $I_{c,d}$ has the claimed properties.

We now find the interval $I_{L}$.
Traverse $\prm B$ in counter-clockwise direction from $a$ and let $b'$ be the last vertex for which $\wei{\prm B[a,b']} \leq (\frac{1}{2}-\frac{\delta}{2})\wei{\prm B}$.
Symmetrically, traverse $\prm B$ in clockwise direction from $b$ and let $a'$ be the last vertex for which $\wei{\prm B[a',b]} \leq (\frac{1}{2}-\frac{\delta}{2})\wei{\prm B}$.
Observe that $a',a,b,b'$ lie on $\prm B$ in this counter-clockwise order and $a' \neq b'$.
Moreover, note that for any $c,d \in D$, it follows from the properties of $I_{c,d}$ that $I_{c,d} \subseteq \prm B[a',b']$, and thus $D \subseteq \prm B[a',b']$.
%Note that $I^a \cup I^b = \prm B[a',b']$ is an interval on $\prm B$ of length at most $(1-\delta)\wei{\prm B}$.
Let $c_0$ and $d_0$ be the vertices of $D$ that are closest to $a'$ and $b'$ on $\prm B[a',b']$, respectively
(possibly $c_0=d_0$ if $|D| = 1$).
We claim that $I_L := I_{c_0,d_0}$ satisfies the conditions of the claim. By the properties of $I_{c_0,d_0}$ proven above, the length of $I_{L}$ is at most $(\frac{1}{2} - \frac{\delta}{2})\wei{\prm B}$ and $\prm B[a,b] \subseteq I_{c_0,d_0}$. Hence, property (i) is satisfied. If $c_0 = d_0$, then $|D| = 1$, and property (ii) is satisfied by the construction of $I_{c_0,d_0}$. If $c_0 \neq d_0$, then $\prm B[d_0,c_0] \not\subseteq \prm B[a',b']$ and, consequently, $\prm B[c_0,d_0] \subseteq I_{c_0,d_0}$.
Since $\prm B[a,b] \subseteq I_{c_0,d_0}$ and $D \subseteq \prm B[c_0,d_0]$, we infer that $\pi(V(B')) \subseteq V(I_{c_0,d_0})$. Hence, property (ii) is satisfied. This finishes the proof of the claim.
\cqed\end{proof}

Armed with Claim~\ref{cl:multi-projection}, we can proceed to the proof of the existential statement of Theorem~\ref{thm:core}. The proof strategy is as follows: given the map $\proj: V(B)\to V(\prm B)$, we extend $\proj$ to a map $\extproj$ such that:
\begin{itemize}
\item[{\em (i)}] $\extproj$ is a continuous map from the closed disk enclosed by $\prm B$ to its boundary;
\item[{\em (ii)}] $\extproj$ is the identity when restricted to the boundary of this disk, i.e., to $\prm B$.
\end{itemize}
We will define the extension $\extproj$  using Claim~\ref{cl:multi-projection} and the assumption that every face of $B$ can be $\delta$-carved. Such a mapping $\extproj$, however, would be a retraction of a closed disc onto its boundary. This contradicts Borsuk's non-retraction theorem~\cite{badger}%, bitches!
, which states that such a retraction cannot exist.

We proceed with the construction of $\extproj$. We first extend the map $\proj$ to the edges of $B$. Consider any edge $vw$ of $B$. Since $vw$ lies on the perimeter of some face of $B$, there exists a $\delta$-carve $L$ that encloses $vw$. By Lemma~\ref{lem:strict-equiv}, we can assume that $L$ is strict. By Claim~\ref{cl:multi-projection}, $\proj(v)$ and $\proj(w)$ both lie on $I_{L}$, which is of length at most $(\frac{1}{2} - \frac{\delta}{2}) \cdot \wei{\prm B}$. Hence, among the two intervals $\prm B[\proj(v),\proj(w)]$ and $\prm B[\proj(w),\proj(v)]$, one is of length at most $(\frac{1}{2} - \frac{\delta}{2}) \cdot \wei{\prm B}$ and one is of length at least $(\frac{1}{2} + \frac{\delta}{2}) \cdot \wei{\prm B}$. Therefore, we map the edge $vw$ in a continuous manner onto the shorter of these two intervals in such a way that the distance between any two points on the embedding of $vw$ is proportional to the distance of their images on this shorter interval. Note that the image of $vw$ is a subinterval of $I_{L}$ for every $L$ that strictly $\delta$-carves $vw$. By Claim~\ref{cl:multi-projection}, $I_L$ and $I_{L'}$ for strict $\delta$-carves $L$ and $L'$ can share only a subinterval.
Moreover, observe that if $vw\in \prm B$, then $\proj(v)=v$, $\proj(w)=w$ and $\extproj$ is the identity on $vw$. Hence, property {\em (ii)} of $\extproj$ is already satisfied.

It remains to define $\extproj$ on faces of $B$. Let $f$ be any face of $B$. Since we assumed that every face of $B$ can be $\delta$-carved, there exists some $\delta$-carve $L$ that encloses $f$. Again, by Lemma~\ref{lem:strict-equiv}, we can assume that $L$ is strict. As we have observed, $\proj(u)\in I_L$ for every $u$ on the boundary of $f$ and $\extproj(e)\subseteq I_{L}$ for every edge $e$ on the boundary of $f$. Since $I_{L}$ is an interval, which is a simply connected metric space, %pol: jednospójna; definition of a simply connected space X from wikipedia: every continuous map from circle S^1 into X can be continuously extended to a continuous map from the whole disk into X
we can extend $\extproj$ from the boundary of face $f$ to its interior in a continuous manner such that the whole face $f$ is mapped into $I_L$.
%I am not 100 percent  with this explanation, but I fear that super-formal explanation of this (straightening to a disk via homeomorphism and contracting disc to a point) is not necessary.

By construction, $\extproj$ is continuous and maps the closed disc enclosed by $\prm B$ onto its boundary such that $\prm B$ is fixed in this mapping. Hence, $\extproj$ is a retraction of a disc onto its boundary, contradicting Borsuk's non-retraction theorem. Hence, there must be an inner face of $B$ that cannot be $\delta$-carved and the existential statement is proved.

Finally, we present how to find such a face in time $O(|B|)$.
As discussed earlier, we construct the mapping $\proj$ by first placing a super-terminal $s_{0}$ on the outer face of $B$, attaching it to each vertex of $V(\prm B)$ with a unit-weight edge,
and then constructing a shortest-path tree from $s_0$ in the obtained plane graph in linear time~\cite{planar-sp}.
Observe now that we have in fact proven not only that some face $f_0$ cannot be $\delta$-carved, but also that for some face $f_0$, the images of the vertices of $f_0$ are not contained in an interval of length at most $(\frac{1}{2} - \frac{\delta}{2}) \cdot \wei{\prm B}$ on $\prm B$ --- otherwise, the extended mapping $\extproj$ could be constructed, leading to a contradiction. Clearly, given the mapping $\proj$ we can identify such a face $f_0$ in $O(|B|)$ time by performing a linear-time check on the boundary of each face of $B$. By Claim~\ref{cl:multi-projection}, any face for which this check fails cannot be $\delta$-carved.
\end{proof}

%!TEX root = pst-kernel.tex

\section{Mountains}\label{sec:mountains}
In this section, we start to develop the tools that we need to find a cycle $C$ of length $\Oh(\wei{\prm B})$ that lies close to the perimeter of $B$ and that separates the core from all vertices of degree at least three of some optimal solution for any set of terminals on $\prm B$. To this end, we need a deep and rigorous understanding of the brick. Then, in Section~\ref{sec:sliding}, we exploit this understanding to actually find the cycle $C$.

Before we start, we need the following notion.
For a path $P$ in a brick $B$ connecting $a$ and $b$, and a real $0 \leq \mouth \leq \wei{P}$, we define the \emph{vertex at distance $\mouth$ from $a$ on $P$}, denoted $v(P,a,\mouth)$ as follows.
If there exists $v \in V(P)$ such that $\wei{P[a,v]} = \mouth$, then $v(P,a,\mouth) = v$.
Otherwise, we find the unique edge $xy \in P$ such that $\wei{P[a,x]} < \mouth < \wei{P[a,y]}$, subdivide it by inserting a new vertex $v$ such that $\wei{xy} = \wei{xv} + \wei{vy}$
and $\wei{P[a,x]} + \wei{xv} = \mouth$, and set $v(P,a,\mouth) = v$.
If we speak about a vertex at distance $\mouth$ from $a$ on $P$ in $B$, and an edge $xy$ needs to be subdivided to obtain
$v(P,a,\mouth)$, then we abuse notation and identify the original brick $B$ and path $P$ with the brick $B$ and the path $P$ with the edge $xy$ subdivided.
Observe that this subdivision does not change any metric properties of the brick $B$.

The main notion in this section are $\delta$-carves of a special form which are defined as follows.

\begin{definition}
For a constant $\delta \in (0, 1/2)$ and fixed $l,r \in V(\prm B)$, a \emph{$\delta$-mountain} of $B$ for $l,r$
is a $\delta$-carve $M$ in $B$ such that
\begin{enumerate}
\item $l$ and $r$ are the endpoints of the carvemark and carvebase of $M$;
\item there exists a real $\mouth_M$, $0 \leq \mouth_M \leq \wei{M}$, such that if
we define $v_M = v(M,l,\mouth_M)$, $P_L = M[l,v_M]$ and $P_R = M[v_M,r]$, then
$P_{L}$ is a shortest $l$--$P_{R}$ path in the subgraph enclosed by $M$
and $P_{R}$ is a shortest $r$--$P_{L}$ path in the subgraph enclosed by $M$.
\end{enumerate}
\end{definition}
We denote a mountain either by $M$ to refer to the subgraph of $B$ enclosed by the carve, or,
if we want to exhibit the choice of $\mouth_M$ and the partition of the carvemark into paths $P_L$ and $P_R$, we write $\mou{P_L}{P_R}$.
By abusing notation, we may write $M = \mou{P_L}{P_R}$.
We call the vertex $v_M$ the \emph{summit} of the mountain.
We also say that a mountain $M$ {\em{connects}} the vertices $l$ and $r$.

We want to stress that mountains are discrete objects. Observe that, formally, a mountain is a carve $M$ only, and the definition speaks about the existence
of a real $\mouth_M$ and a vertex $v_M$ (that may not exist in $B$, if we need to subdivide some edge to obtain it).
Hence, a mountain is a discrete object in $B$, and there are only a finite number of mountains in a fixed brick $B$.

Throughout this section, when we discuss a (finite) family of mountains in $B$ and prove some structural properties of them,
we will assume that the summits $v_M$ exist in $B$. In particular, if we use notation $M = \mou{P_L}{P_R}$, then we implicitly assume that the summit $v_M$ is (already) present in $B$. 
In the unweighted setting, one may observe that $\mouth_M$ can always be taken to be integral, and then $v_M$ always exists in the brick $B$.
In the edge-weighted setting, we can ensure that $v_{M}$ exists by subdividing some edges. % in the analysis. 
Observe that subdividing some edges of $B$ does not change the family of mountains with fixed endpoints $l$ and $r$.
However, when we move to the algorithmic part --- where we discuss how to find some specific mountains in a brick $B$ ---
we will need to be careful not to assume that $v_M$ is present in the brick $B$.

%Observe also that, if $B$ is unweighted, $\mouth_M$ can always be taken to be integral, and then $v_M$ always exists in the brick $B$. Hence, in the unweighted case, we do not need to subdivide any edges of $B$ in the analysis.

Before we move on to the properties of $\delta$-mountains, we give an intuition why we study this notion. Assume that among the terminals $\terms$ lying on the boundary of the brick, one can distinguish a small set $Y\subseteq \terms$ that are ``close enough'' to each other and considerably ``far away'' from $\terms\setminus Y$. Intuitively, an optimal Steiner tree connecting $\terms$ should gather all of $Y$ in one subtree $T_v$ such that the leftmost and rightmost elements of $Y$ on the interval of $\prm B$ containing $Y$, denote them by $l$ and $r$, correspond to the leftmost and the rightmost anchors of $T_v$. Consider the $\delta$-carve induced by the path in $T_v$ joining $l$ and $r$, with carvebase $\prm B[l,r]$. Observe that if this $\delta$-carve was not a $\delta$-mountain with summit $v$, then there would exist a shorter path inside this $\delta$-carve that could be used as a shortcut to decrease the cost of $T$.
This is formalized in the following lemma.

\begin{lemma}\label{lem:broom-in-mountain}
Let $B$ be a brick and $T$ be an optimal Steiner tree
that connects $S := V(T) \cap V(\prm B)$ in $B$. Let $uv \in T$ be an edge of $T$, where $v$ is of degree at least $3$ in $T$,
and let $T_v$ be the subtree of $T$ rooted at $v$ with parent edge $uv$.
Let $a$ and $b$ be the leftmost and rightmost elements of $V(T_v) \cap V(\prm B)$ and let $l, r \in \prm B[b, a]$ be two vertices such that $l \neq r$ and $a,b \in \prm B[l, r]$. 
Let $P_L = T[v,a] \cup \prm B[l, a]$ and $P_R = T[v,b] \cup \prm B[b, r]$. If $\wei{\prm B[l,r]} < \wei{\prm B}/2$, then $M := \mou{P_L}{P_R}$ is a $\delta$-mountain, connecting $l$ and $r$, for any
$\delta < 1/ 2 - (\wei{P_L} + \wei{P_R})/\wei{\prm B}$.
\end{lemma}
\begin{proof}
Recall that, by the definition of the leftmost and rightmost elements of $V(T_v) \cap V(\prm B)$, we have that $V(T_v) \cap V(\prm B) \subseteq V(\prm B[a,b])$.
As $v$ is of degree at least $3$ in $T$, it is of degree at least $2$ in $T_v$ and $P_L \cap P_R = \{v\}$. Therefore,
$P_L \cup P_R$ is a path and it induces a $\delta$-carve $M$ with carvebase $\prm B[l,r]$, as $\wei{\prm B[l,r]} < \wei{\prm B}/2$.

Suppose that $M$ is not a mountain. Without loss of generality, there exists a path $P$ enclosed by $M$ that connects $l$ with $w \in V(P_R)$ such that
$V(P_R) \cap V(P) = \{w\}$ and $\wei{P} < \wei{P_L}$. By construction, $P$ passes through $a$ and $P[l,a] = \prm B[l,a]$.
Let $D$ be the subgraph of $M$ enclosed by the closed walk
$P_L[a,v] \cup P[a,w] \cup P_R[v,w]$.
Define $T' := (T \setminus D) \cup P_R[v,w] \cup P[a,w]$. As $P_R[v,w] \cup P_L[a,v] \subseteq D$, $\wei{T'} < \wei{T}$. By the definition of $a$ and $b$,
$P_L \setminus P$ does not contain any vertex of $\prm B$. Therefore, $T'$ is a connected subgraph of $B$ connecting $V(T) \cap V(\prm B)$,
a contradiction to the minimality of $T$.
\end{proof}

The goal of this section is essentially to prove that if we take the union of all maximal $\delta$-mountains with fixed $l$ and $r$, then the perimeter of the resulting subgraph has length bounded linearly in the length of the carvebase. This intuition is captured by the following theorem.

\begin{theorem}[Mountain Range Theorem]\label{thm:mountain-range}
Fix $\nice \in [0,1/4)$ and $\delta \in [2\nice,1/2)$ and assume $B$ does not admit any $\nice$-nice $3$-short tree.
Then for any fixed $l, r \in V(\prm B)$ with $\wei{\prm B[l,r]} < \wei{\prm B}/2$,
there exists a closed walk $W_{l,r}$ in $B$ of length at most $3\wei{\prm B[l,r]}$ such that, for each face $f$ of $B$,
$f$ is enclosed by $W_{l,r}$ if and only if $f$ belongs to some $\delta$-mountain connecting $l$ and $r$.
Moreover, the set of the faces enclosed by $W_{l,r}$ can be computed 
in $O(|B|)$ time.
%in time $\Oh(|B| \log |B|)$ in the edge-weighted setting and in time $\Oh(|\prm B| \cdot |B|)$ in the unweighted setting.
\end{theorem}
The set of the faces enclosed by $W_{l,r}$ is called the \emph{$\delta$-mountain range} of $l,r$.

Observe that in Lemma~\ref{lem:broom-in-mountain}, the discussed mountain has summit $v$ that belongs to $B$ (i.e., we do not need to subdivide any edge).
However, in the edge-weighted setting we need to allow the mountains to have summits in the middle of some edges
to obtain the statement of Theorem~\ref{thm:mountain-range}.

The rest of this section is devoted to the proof of Theorem~\ref{thm:mountain-range}.
Henceforth, we assume that $\nice \in [0,1/4)$, $\delta \in [2\nice,1/2)$, $l,r \in V(\prm B)$ are fixed.
Whenever we speak about a mountain, we mean a $\delta$-mountain connecting $l$ and $r$.

\subsection{Preliminary simplification steps}

We start the proof of Theorem~\ref{thm:mountain-range} with the following simplification step.
We attach to $B$ two paths $\xP$, $\xP'$ connecting $l$ and $r$, being copies of $\prm B[l,r]$ and $\prm B[r,l]$, respectively,
drawn in the outer face of $B$ in such a manner that $\xP \cup \xP'$ is the infinite
face and $\xP \cup \prm B[l,r]$ and $\xP' \cup \prm B[r,l]$ are two finite faces of the constructed graph $B'$.
Note that $B'$ is also a brick of perimeter $\wei{\prm B}$, and that all $\delta$-mountains connecting $l$ and $r$
in $B$ are also $\delta$-mountains in $B'$ (with carvebase $\prm B[l,r]$ replaced by $\xP$)
with the additional property that the $\delta$-carves of these mountains
are strict. Moreover, as $\wei{\xP'} > \wei{\prm B}/2$,
any mountain that is present in $B'$ but not in $B$ is induced by the $\delta$-carve
$(\xP, \xP)$ and any choice of the summit; note that this $\delta$-carve is enclosed by any other
$\delta$-mountain in $B'$, and does not influence the output graph of Theorem~\ref{thm:mountain-range}.
Hence, by somewhat abusing the notation and denoting the modified brick $B'$ by $B$ again, we may
assume that all $\delta$-mountains connecting $l$ and $r$ are induced by strict $\delta$-carves,
possibly with the exception of the trivial $\delta$-carve $(\prm B[l,r], \prm B[l,r])$.
We silently ignore the existence of the latter in the upcoming arguments, and assume that whenever we pick
a mountain, it is induced by a strict $\delta$-carve.

Hence, for any $\delta$-mountain $M = \mou{P_L}{P_R}$, the closed walk
$P_L \cup P_R \cup \prm B[l,r]$ is actually a simple cycle in $B$, denoted $\prm M$.

\subsection{Maximal mountains}
In this subsection, we describe two properties of mountains that will be crucial in the remainder of the proof of Theorem~\ref{thm:mountain-range}. The first property is the following easy consequence of the definition of a mountain.
\begin{lemma}\label{lem:shortcut}
Let $M = \mou{P_L}{P_R}$ be a mountain and let $a,b \in V(\prm M)$ be such
that $\prm M[a,b]$ is contained entirely in $P_L$ or entirely in $P_R$.
If there exists a path $Q$ with endpoints in $a$ and $b$ that is enclosed by $\prm M$, then $\wei{Q} \geq \wei{\prm M[a,b]}$.
\end{lemma}
\begin{proof}
By symmetry, without loss of generality assume $\prm M[a,b]$ is a subpath of $P_L$.
Note that $Q' := \prm M[v_M, a] \cup Q \cup \prm M[b,l]$ is a path connecting $l$ and $P_R$,
enclosed by $\prm M$. Hence, $\wei{Q'} \geq \wei{P_L}$ and the lemma follows.
\end{proof}

We now define what it means for a mountain to be maximal. Observe that since $\prm M$ is a simple cycle for each mountain $M$ in $B$, the subgraph enclosed by $\prm M$ is defined by the set of faces of $B$ enclosed by $\prm M$. A mountain $M$ is called {\em{maximal}} if this set of faces is inclusion-wise maximal, among the set of all $\delta$-mountains connecting $l$ and $r$.
Note that in the proof of Theorem~\ref{thm:mountain-range} we may actually look for the union of all faces enclosed by maximal mountains.

The second property is actually a condition under which a mountain cannot be maximal.
\begin{lemma}
\label{lem:largermountain}
Let $M = \mou{P_L}{P_R}$ be a mountain.
Let $u,w \in V(\prm M)$ and let $P$ be a path between $u$ and $w$ such that:
\begin{enumerate}
\item $P$ does not contain any edge strictly enclosed by $\prm M$ and, moreover, the closed walk $\prm M[u,w] \cup P$
encloses $M$;
\item $P \neq \prm M[w, u]$;
\item $\wei{P} \leq \wei{\prm M[w,u]}$.
\end{enumerate}
Then $M$ is not a maximal mountain.
\end{lemma}
\begin{proof}
First note that if $P$ is a path satisfying the assumptions of the lemma, then there exists a subpath of $P$ also satisfying the assumptions
for which no internal vertex lies on $\prm M$ (recall that all edge weights are positive).
%that does not contain any internal vertex on $\prm M$. 
Hence, %we may in the rest of the proof assume that all internal vertices of $P$ do not lie in $\prm M$. In particular, 
$u,w$ lie on the carvemark of $M$.
Let $M^\ast$ denote the carve obtained by replacing $\prm M[w,u]$ with $P$ in the carve $M$.
We assume that $P$ and $u,w \in V(\prm M)$ have been chosen such that the number of faces contained in $M^\ast$ is minimum (satisfying
the previous assumption that $P$ does not contain any internal vertices on $\prm M$).
As $\prm B[l,r] \subseteq \prm M$ and $\prm M[u,w] \cup P$ encloses $M$, we have that $u$ is closer to $l$ on $P_L \cup P_R$ than $w$ is. Since $\wei{P} \leq \wei{\prm M[w,u]}$, $M^\ast$ is also a $\delta$-carve.

We now consider two cases. First, suppose that $u$ and $w$ both lie on $P_L$ or both lie on $P_R$; by symmetry, assume that they both lie on $P_L$.
Partition the carvemark of $M^\ast$ into $P_L^\ast$ and $P_R^\ast$ by taking $P_L^\ast$ equal to $P_L$ with $\prm M[w,u]$ substituted by $P$, and taking $P_R^\ast$ equal to $P_R$. Note that thus $\wei{P_L^\ast}\leq \wei{P_L}$.
We claim that $M^\ast$ treated as $\mou{P_L^\ast}{P_R^\ast}$ is also a $\delta$-mountain. Together with the observation that $M^\ast$ encloses a proper superset of the faces enclosed by $M$ (since no edge of $P$ is enclosed by $M$), this contradicts that $M$ is maximal.

For sake of contradiction, assume that $M^\ast$ is not a $\delta$-mountain. Suppose that there exists a shortest path $Q$ in $M^\ast$ between $r$ and some $x\in P_L^\ast$ that is shorter than $P_R^\ast=P_R$ --- see Figure~\ref{fig:largermountain}~(a). Observe that then $Q$ must meet $P_L$, and let $x'$ be the first point of intersection of $Q$ and $P_L$, counting from $r$. We infer that $Q[r,x']$ is entirely contained in $M$ and that $\wei{Q[r,x']}\leq \wei{Q}<\wei{P_R}$. Since $x'\in V(P_L)$, this contradicts that $M$ is a mountain.
Therefore, there exists a shortest path $Q$ in $M^\ast$ between $l$ and some $x\in P_R^\ast=P_R$ that is shorter than $P_L^\ast$ --- see Figure~\ref{fig:largermountain}~(b). Since $\wei{Q}<\wei{P_L^\ast}\leq \wei{P_L}$, $Q$ must contain an edge that is not enclosed by $M$, since otherwise existence of $Q$ would contradict the fact that $M$ is a mountain. Then $Q$ contains some subpath $Q[a,b]$ where $a,b\in V(\prm M)$ but no internal vertex
of $Q[a,b]$ lies on $\prm M$.
By choosing $a$ and $b$ so that the number of faces enclosed by $Q[a,b] \cup \prm M[b,a]$ is minimized, we can moreover assume that no edge of $Q$ is strictly enclosed by $Q[a,b]\cup \prm M[b,a]$.
As $Q$ is a shortest path, we have that $\wei{Q[a,b]} \leq \wei{\prm M[b,a]}$. By the choice of $P$ as the path that minimizes the number of faces enclosed by $M^\ast$, we infer that $Q$ would be a better candidate for $P$ unless $Q[a,b]=P$ (and thus, $(a,b)=(u,w)$). 
%By the choice of $Q$, all the other edges of $Q$ are enclosed by $M$. 
By the definition of $Q[a,b]$, all edges of $Q$ not on $Q[a,b]$ are enclosed by $M$.
We infer that $\wei{P_L[l,u]}=\wei{Q[l,u]}$, since $P_L[l,u]$ is a shortest path in $M$ between $l$ and $u$, and $Q[l,u]$ is enclosed by $M$. Similarly, $\wei{P_L[w,v_M]}=\wei{Q[w,x]}$, since $P_L[w,v_M]$ is a shortest path between $w$ and $P_R$ and $Q[w,x]$ is enclosed by $M$. Thus we have that $\wei{Q}=\wei{Q[l,u]}+\wei{P}+\wei{Q[w,x]}=\wei{P_L[l,u]}+\wei{P}+\wei{P_L[w,v_M]}=\wei{P_L^\ast}$, a contradiction with the choice of $Q$.

Now we consider the case when $u$ lies on $P_L$ and $w$ lies on $P_R$. As $\wei{\prm M^\ast} \leq \wei{\prm M}$, observe that it is possible to find a vertex $v_{M^\ast}$ on $P$ (possibly by  subdividing some edge of $P$) such that $\wei{\prm M^{\ast}[v_{M^\ast},l]}\leq \wei{P_L}$ and $\wei{\prm M^{\ast}[r,v_{M^\ast}]}\leq \wei{P_R}$.\footnote{We remark here that this is the sole point in the argumentation that forces us to allow mountains with summits in the middle of some edge.} Let $P_L^\ast=\prm M^{\ast}[v_{M^\ast},l]$ and $P_{R}^\ast=\prm M^{\ast}[r,v_{M^\ast}]$. We again claim that $M^\ast$ treated as $\mou{P_L^\ast}{P_R^\ast}$ is a $\delta$-mountain, which in the same manner brings a contradiction.

\begin{figure}[t]
\centering
\svg{0.8\textwidth}{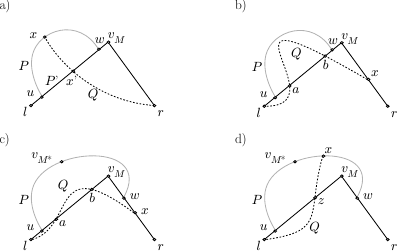}
\caption{The cases considered in the proof of Lemma~\ref{lem:largermountain}.}
\label{fig:largermountain}
\end{figure}

Assume that this is not the case, and without loss of generality suppose that there is a shortest path $Q$ in $M^\ast$ between $l$ and $x \in P_{R}^\ast$ that is shorter than $P_L^\ast$. The case that there is a path between $r$ and $P_{L}^\ast$ shorter than $P_{R}^\ast$ is symmetric.
If $Q$ does not contain any edge not enclosed by $\prm M$, then $x$ must in fact lie on $P_R$ and $Q$ is also a shorter path than $P_L$ in $M$ between $l$ and $P_R$, a contradiction.
Assume now that $Q$ contains a subpath $Q[a,b]$ where $a,b \in V(\prm M)$ but every internal vertex of $Q[a,b]$ is not enclosed by $\prm M$ --- see Figure~\ref{fig:largermountain}~(c). We now employ a very similar reasoning as in the previous case. 
Again, by choosing $a$ and $b$ that minimize the number of faces enclosed by $Q[a,b] \cup \prm M[b,a]$, we may assume that no edge of $Q$ is strictly enclosed by $Q[a,b]\cup \prm M[b,a]$. Since $Q$ is a shortest path, we have that $\wei{Q[a,b]} \leq \wei{\prm M[b,a]}$. By the choice of $P$ as the path that minimizes the number of faces enclosed by $M^\ast$, we infer that $Q$ would be a better candidate for $P$ unless $Q[a,b]=P$ (and hence $(a,b)=(u,w)$). 
%By the choice of $Q$ we have also that all the other edges of $Q$ are enclosed by $M$. 
By the definition of $Q[a,b]$, all edges of $Q$ not on $Q[a,b]$ are enclosed by $M$.
Since $v_{M^\ast}$ lies on $P=Q[a,b]$, we infer that $x=v_{M^\ast}=w$. Moreover, again we have that $\wei{P_L[l,u]}=\wei{Q[l,u]}$, since $P_L[l,u]$ is a shortest path in $M$ between $l$ and $u$, and $Q[l,u]$ is enclosed by $M$. Therefore, $\wei{Q}=\wei{Q[l,u]}+\wei{P}=\wei{P_L[l,u]}+\wei{P}=\wei{P_L^\ast}$, a contradiction with the choice of $Q$.

We are left with the case when $x$ is not enclosed by $\prm M$ and $Q$ can be partitioned into $Q[l,z]$ and $Q[z,x]$, where $z \in V(\prm M)$, $z \neq x$, $Q[l,z]$ is enclosed by $\prm M$,
and no edge of $Q[z,x]$ is enclosed by $\prm M$ --- see Figure~\ref{fig:largermountain}~(d). Since $\wei{Q} < \wei{P_L^\ast} \leq \wei{P_L}$ and $M$ is a mountain, $z \in V(P_L) \setminus \{v_M\}$.
We note that $\wei{Q[l,z]} = \wei{P_L[l,z]}$, since $M$ is a mountain, and both $Q[l,z]$ and $P_L[l,z]$ are shortest paths in $M$.
As $\wei{Q} < \wei{P_L^\ast}$, we have $\wei{Q[z,x]}+\wei{P_L[u,z]} < |P_L^\ast[u,v_{M^\ast}]$. Define $\overline{P} := Q[z,x] \cup P_R^\ast[x,w]$, and observe that
\begin{align*}
\wei{\overline{P}}-\wei{\prm M[w, z]} & = \wei{Q[z,x]} + \wei{P_R^\ast[x,w]} - \wei{\prm M[w, z]} \\
& < \wei{P_R^\ast[x,w]}+\wei{P_L^\ast[u,v_{M^\ast}]} - \wei{\prm M[w, z]}-\wei{P_L[u,z]}\\
  &\leq \wei{P} - \wei{\prm M[w,u]}.
\end{align*}
Since $\wei{P} \leq \wei{\prm M[w,u]}$ by assumption, $\wei{\overline{P}} < \wei{\prm M[w,z]}$. 
We infer that $\overline{P}$, instead of $P$, would define a carve with a strictly smaller
number of faces than $M^\ast$, a contradiction; note here that $\overline{P}\neq P$, since then the left-hand side and the right-hand side of the inequality above would need to be equal. This contradicts the choice of $Q$.
%
%Due to symmetric argument, there is no shorter path in $M^\ast$ between $r$ and $P_{L}^\ast$ than $P_{R}^\ast$. Therefore, $M^\ast$ is a $\delta$-mountain and $\prm M^\ast$ enclosest a strict superset of faces enclosed by $\prm M$.
%
\end{proof}

\begin{corollary}\label{cor:maxmountain}
Let $M = \mou{P_L}{P_R}$ be a maximal mountain with summit $v_M$.
Then $\dist_B(v_{M},l) = \dist_B(P_R,l)$ and $\dist_B(v_{M},r) = \dist_B(P_L,r)$.
\end{corollary}
\begin{proof}
We prove $\dist_B(v_{M},l) = \dist_B(P_R,l)$; the other case is symmetric. Clearly, $\dist_B(v_{M},l) \geq \dist_B(P_R,l)$, so it remains to prove an
inequality in the other direction.
Let $P$ be a shortest path between $P_R$ and $l$.
We claim that $P$ is actually enclosed by $M$; if this is the case then, by the definition of mountain, $\wei{P} \geq \wei{P_L} \geq \dist_B(v_{M},l)$ and the lemma is proven.

Assume the contrary, and let $Q$ be a subpath of $P$ with endpoints $u,w \in V(\prm M)$, such that all edges of $Q$ are not enclosed by $M$
and, moreover, the closed walk $\prm M[u,w] \cup Q$ encloses $M$.
By Lemma~\ref{lem:largermountain}, $\wei{Q} > \wei{\prm M[w,u]}$, a contradiction to the fact that $P$ is a shortest path in~$B$.
\end{proof}

\subsection{Untangling maximal mountains}
We now show a result that implies that the boundaries of two distinct maximal mountains $M^1 = \mou{P_L^1}{P_R^1}$ and $M^2 = \mou{P_L^2}{P_R^2}$ cannot cross each other (in a topological sense) more than twice, because then we can find a shortcut either inside one of the mountains (which contradicts Lemma~\ref{lem:shortcut}) or outside one of the mountains (which contradicts Lemma~\ref{lem:largermountain}).
%Our goal is now to analyse how two $\delta$-mountains can interfere. That is, we consider two $\delta$-mountains $M^1 = \mou{P_L^1}{P_R^1}$ and $M^2 = \mou{P_L^2}{P_R^2}$ and we would like to understand how they can intersect in the brick $B$.
We assume that both summits of $M^1$ and $M^2$ are present in $B$, that is, the corresponding edges have already been subdivided if needed.

\subsubsection{From mountains to curves}
%Recall that we consider two $\delta$-mountains $M^1 = \mou{P_L^1}{P_R^1}$ and $M^2 = \mou{P_L^2}{P_R^2}$ and consider the following construction.
To build a topological understanding of how the two mountains interact, we build a representation of them as Jordan curves. %These curves divide the plane into several regions, which we can use to build understanding of intersection between the two mountains.

First, we duplicate each edge of $B$ to obtain a brick $B_2$; the copies of the edges are drawn in parallel in the plane, without
any other part of $B_2$ in between.
Second, we project $\prm M^1$ and $\prm M^2$ onto $B_2$ in the following manner. For each $e \in \prm M^1$ ($e \in \prm M^2$)
we choose one copy of $e$ to belong to $\prm M^1$ ($\prm M^2$) in $B_2$. If $e \in \prm M^1 \cap \prm M^2$, then
one copy of $e$ belongs to $\prm M^1$ and the second one to $\prm M^2$ in $B_2$, so that $\prm M^1$
and $\prm M^2$ are edge-disjoint in $B_2$. By abuse of notation, we often consider $\prm M^1$ and $\prm M^2$ both as walks in $B$ and in $B_2$.

A vertex $v$ is a {\em{traversal vertex}} if both $\prm M^1$ and $\prm M^2$ pass though $v$ and they cross in $v$ in the graph $B_2$;
that is, among the four edges of $\prm M^1 \cup \prm M^2$ incident to $v$ considered in counter-clockwise order around $v$, the odd-numbered edges belong to one mountain,
and the even-numbered to the second mountain.
In the process of choosing copies of an edge $e \in \prm M^1 \cup \prm M^2$,
we minimize the number of traversal vertices of $B_2$ and, minimizing this number,
we secondly minimize the number of traversal vertices of $B_2$ that are not equal to $l$ or $r$.
Clearly, if $e \in \prm M^1 \triangle \prm M^2$, the choice of the copy of $e$ does not influence the set of transversal vertices,
but the aforementioned minimization criterium regularizes the choice whenever $e \in \prm M^1 \cap \prm M^2$.
In particular, we note the following.

\begin{lemma}\label{lem:border-traversal}
No internal vertex of $\prm B[l,r]$ is a traversal vertex.
\end{lemma}
\begin{proof}
Assume otherwise, let $x \in V(\prm B[l,r])$, $x \neq l,r$ be a traversal vertex.
Consider the following change: for each edge $e \in \prm B[x,r]$, swap the copies of $e$ that belong to $\prm M^1$ and $\prm M^2$.
In this manner, $x$ stops to be a traversal vertex, all internal vertices of $\prm B[x,r]$ are traversal vertices
if and only if they were traversal vertices before the change, and $r$ may become a traversal vertex.
Thus, we either decrease the number of traversal vertices, or do not change it while decreasing the number of traversal vertices
not equal to $l$ and $r$. This contradicts the minimization criterium for the choice of $\prm M^1$ and $\prm M^2$ in $B_2$.
\end{proof}

Now, for each $v \in V(B_2) = V(B)$ we pick a small closed disc $D_v$ in the plane, with the drawing of $v$ at its centre,
and with radius small enough so that $D_v$ contains $v$ and small starting segments of a drawing of each edge of $B_2$ incident to $v$.
For $\alpha=1,2$, we associate the following closed Jordan curve $\gamma^\alpha$ with the cycle $\prm M^\alpha$ in $B_2$: we take the drawing of $\prm M^\alpha$
and for each $v \in V(\prm M^\alpha)$ we replace $D_v \cap \prm M^\alpha$ with the straight line segment $S_v^\alpha$ connecting the two points of $\prm D_v \cap \prm M^\alpha$.
We note that $\prm D_v \cap \prm M^\alpha$ consists of exactly two points since $\prm M^\alpha$ is a simple cycle.
Moreover, $S_v^\alpha \subseteq D_v$.
Consequently, $\gamma^\alpha$ is a closed Jordan curve without self-intersections.
The important properties of this construction are summarized in the following lemmata.

\begin{lemma}\label{lem:gammas}
$\gamma^1 \cap \gamma^2$ consists of exactly one point in each disc $D_v$ where $v$ is a traversal vertex, and nothing more.
Moreover, for each $p \in \gamma^1 \cap \gamma^2$, the curves $\gamma^1$ and $\gamma^2$ traverse each other in the following sense:
there exists an open neighbourhood $O_p$ of $p$ in the plane such that $\gamma^\alpha \cap O_p$ splits $\gamma^{3-\alpha} \cap O_p$
into two connected sets for $\alpha=1,2$.
In particular, $|\gamma^1 \cap \gamma^2|$ is finite and even.
\end{lemma}
\begin{proof}
The first claim follows from the fact that $\prm M^1$ and $\prm M^2$ are edge-disjoint in $B_2$, so the points of $\prm D_v \cap \prm M^1$
and $\prm D_v \cap \prm M^2$ are pairwise distinct, and the segments $S_v^1$ and $S_v^2$ intersect if and only if $v$ is a traversal vertex.
For any traversal vertex $v$, if we take a small open disc $O_p$ centred in $S_v^1 \cap S_v^2$ and contained in $D_v$, then $O_p \cap \gamma^1$ and
$O_p \cap \gamma^2$ are two straight segments intersecting in the centre of $O_p$, which proves the second claim.
\end{proof}

\begin{lemma}\label{lem:curves-disjoint}
$\gamma^1 \cap \gamma^2 \not= \emptyset$.
\end{lemma}
\begin{proof}
Note that all finite faces incident to $\prm B[l,r]$ are enclosed by $\prm M^1$ and $\prm M^2$.
Consequently, if $\gamma^1 \cap \gamma^2 = \emptyset$, then $\gamma^1$ encloses $\gamma^2$ or vice versa. Therefore, $M^{1} \subseteq M^{2}$ or vice versa, which contradicts that $M^{1}$ and $M^{2}$ are two distinct maximal mountains.
\end{proof}

\subsubsection{Regions, elementary regions, and their properties} \label{sss:basics}
Observe that since $\gamma^1 \cap \gamma^2 \neq \emptyset$ by Lemma~\ref{lem:curves-disjoint}, the curves $\gamma^1$ and $\gamma^2$ induce a set of Jordan regions in the plane; denote this set by $\regs$. The goal of this section is to analyse $\regs$.
%Since we want to understand maximal mountains, from this point to the end of Section~\ref{sss:basics} we assume that $\gamma^1$ and $\gamma^2$ intersect in at least one point.
%Note, that if the mountains $M^1$ and $M^2$ are maximal then $\gamma^1$ and $\gamma^2$ need to cross, as otherwise
%$\prm M^1$ encloses $\prm M^2$ or vice versa (note that they both enclose finite faces of $B$ incident to $\prm B[l,r]$).

Lemma~\ref{lem:gammas} immediately implies the following.
\begin{lemma}\label{lem:region-border}
For each region $R \in \regs$,
the border of $R$ can be partitioned into an even number of subcurves $\gamma_1,\gamma_2,\ldots,\gamma_{2s}$ of positive length,
    appearing on the border in counter-clockwise order,
where $\gamma_1,\gamma_3,\ldots,\gamma_{2s-1} \subseteq \gamma^1$ and $\gamma_2,\gamma_4,\ldots,\gamma_{2s} \subseteq \gamma^2$.
The number $s$ and the choice of the curves is unique up to a cyclic shift of the indices.
\end{lemma}
Moreover, note that, since $\prm M^1$ and $\prm M^2$ are simple cycles, a face incident to $\prm M^1$ is enclosed by $\prm M^1$
($\prm M^2$) if and only if it lies to the left, if we walk along $\prm M^1$ ($\prm M^2$) in  counter-clockwise direction.
By this observation, and by the construction of the curves $\gamma^1$ and $\gamma^2$, the following is immediate.

\begin{lemma}\label{lem:plus-minus}
For each $\alpha=1,2$ and for each region $R \in \regs$, the set $R \setminus \bigcup_{v \in V(B_2)} D_v$ is either completely
enclosed by $\prm M^\alpha$ or no point of this set is strictly enclosed by $\prm M^\alpha$.
\end{lemma}

Lemmata~\ref{lem:region-border} and \ref{lem:plus-minus} motivate the following definitions.

\begin{definition}[elementary region]
We say that a region $R \in \regs$ is {\em{elementary}} if its border can be partitioned
into two curves $\gamma_1,\gamma_2$ with $\gamma_1 \subseteq \gamma^1$ and $\gamma_2 \subseteq \gamma^2$.
That is, $s=1$ in the statement of Lemma~\ref{lem:region-border} for the region $R$.
\end{definition}

\begin{definition}
We partition $\regs = \regspm{++} \cup \regspm{+-} \cup \regspm{-+} \cup \regspm{--}$
as follows: $R \in \regs$ belongs to $\regspm{++} \cup \regspm{+-}$ if and only if $R \setminus \bigcup_{v \in V(B_2)} D_v$
is enclosed by $\prm M^1$, and to $\regspm{-+} \cup \regspm{--}$ otherwise.
Similarly, $R$ belongs to $\regspm{++} \cup \regspm{-+}$ if and only if $R \setminus \bigcup_{v \in V(B_2)} D_v$
is enclosed by $\prm M^2$, and to $\regspm{+-} \cup \regspm{--}$ otherwise.
\end{definition}

We also define a {\em{curve-arc}}, which is a subcurve of $\gamma^1$ or $\gamma^2$ that connects two points of $\gamma^1 \cap \gamma^2$, but does not contain any point of this intersection as an interior point. The following property of curve-arcs is immediate from Lemma~\ref{lem:gammas}.

\begin{lemma}\label{lem:nei-regions}
If $\gamma$ is a curve-arc, then exactly two regions are incident to $\gamma$: one of these regions
belongs to $\regspm{++} \cup \regspm{--}$, and the other to $\regspm{+-} \cup \regspm{-+}$.
\end{lemma}

We now show that there, in fact, exist elementary regions.

\begin{lemma}
\label{lemma-jordan}
There exist at least two elementary regions in $\regspm{-+} \cup \regspm{--}$.
\end{lemma}
\begin{proof}
Consider the infinite region $R_\infty$ in $\regs$. The border of this region
cannot be fully contained in one of the curves $\gamma^1$ and $\gamma^2$, because they intersect.
Take any curve-arc $\gamma_1 \subseteq \gamma^1$ incident to $R_\infty$ and cut open $\gamma^1$ by removing $\gamma_1$ to obtain a Jordan arc $\gamma^{1}_{\times}$.
Order the intersection points of $\gamma^1$ with $\gamma^2$ along Jordan arc $\gamma^1_\times$. Consider now the
set $\mathcal{C}^2$ of curve-arcs of $\gamma^2$ that are not enclosed by $\gamma^1$.
For each Jordan arc
$\gamma \in \mathcal{C}^2$ tie a pair parenthesis to its endpoints. We associate the opening parenthesis with the first endpoint of $\gamma$ along
$\gamma^1_{\times}$ , whereas we associate
the closing parenthesis with the second one. Observe that Jordan arcs in $\mathcal{C}^2$ cannot intersect, hence, when
we list the parenthesis along $\gamma^1_{\times}$ we obtain a valid parenthesis expression $E$.
We have to consider two cases.

First, suppose that the first and the last parenthesis
in $E$ belong to the same pair given by arc $\gamma$. We observe that
the infinite region in $\regs$ is elementary, as its boundary is formed by $\gamma_1$ and $\gamma$.
To obtain the second elementary region, observe that there has to be a pair of innermost
matching parenthesis in $E$ corresponding
to some arc $\gamma'$. The Jordan region enclosed by $\gamma'$ and the part of $\gamma^1$ between the endpoints of $\gamma'$ is
the second elementary region not enclosed by $\gamma^1$.

Second, suppose that the first and the last parenthesis in $E$ do not form a matching pair. Then $E$ can be
decomposed into the concatenation of two valid parenthesis expressions $E_1$ and $E_2$. Both of them
need to contain a pair of innermost matching parenthesis, which induce two elementary regions.
\end{proof}

Note that the arguments of Lemma~\ref{lemma-jordan} can be modified to exhibit two elementary regions in $\regspm{+-} \cup \regspm{--}$.

%\subsubsection{Cushions and properties of elementary regions}
We introduce some more notation with respect to regions.
For a region $R \in \regs$, we associate a closed walk $W_2(R)$ in $B_2$
that corresponds to the border of $R$ in the obvious manner. Note that
the walk $W_2(R)$ contains each edge of $B_2$ at most once (since $\prm M^1$ and $\prm M^2$ are edge-disjoint).
It may visit a vertex $v \in V(B_2)$ more than once, but it never {\em{traverses}} itself in such a vertex:
if we walk along $W_2(R)$ in counter-clockwise direction (defined by the border of $R$)
and we enter a vertex $v$ along an edge $e \in E(B_2)$, then we leave the vertex $v$ with the edge
of $W_2(R)$ incident to $v$ being the first such edge in counter-clockwise order after $e$.
We also define a walk $W(R)$ in $B$ as the projection of the walk $W_2(R)$ onto $B$.

We say that a vertex $v$ belongs to $\prm R$ for some region $R \in \regs$ (written $v \in \prm R$)
if and only if the border of $R$ intersects $D_v$; equivalently, if $W(R)$ visits $v$.
Similarly, we say that a region $R$ is incident to an edge $e \in E(B)$ or $e \in E(B_2)$
if and only if $W(R)$ or $W_2(R)$ contains $e$.

Consider now an elementary region $R \in \regs$. According to the definition,
its border splits into curves $\gamma_1$ and $\gamma_2$, where $\gamma_\alpha \subseteq \gamma^\alpha$ for $\alpha=1,2$.
Consequently, since $\prm M^1$ and $\prm M^2$ are simple cycles in $B$ and $B_2$,
the walk $W_2(R)$ splits into paths $P_2^1(R)$ and $P_2^2(R)$ in $B_2$
and the walk $W(R)$ splits into paths $P^1(R)$ and $P^2(R)$ in $B$,
where $P_2^\alpha(R)$ and $P^\alpha(R)$ are a subpath of $\prm M^\alpha$ in $B_2$ and $B$, respectively.

Using this notation, we can present the following implication of the minimization criterium assumed
in the projection of $\prm M^1$ and $\prm M^2$ onto $B_2$.
\begin{lemma}\label{lem:true-region}
For any elementary region $R \in \regs$, there exists
a face $f_2$ of $B_2$ enclosed by $W_2(R)$ that is not a face between two copies of an edge of $B$,
and thus, there exists a face $f$ of $B$ enclosed by $W(R)$.
\end{lemma}
\begin{proof}
If such faces $f_2$ and $f$ do not exist, then $P^1(R) = P^2(R)$, as $P^1(R)$ and $P^2(R)$ are simple paths.
Let $v_0, e_1, v_1, e_2, \ldots, e_s, v_s$ be the vertices
and edges of $P^1(R)= P^2(R)$, and let $P^\alpha_2(R) = v_0, e_{1,\alpha}, v_1, e_{2,\alpha}, \ldots, e_{s,\alpha}, v_s$ for $\alpha=1,2$;
the edges $e_{i,1}$ and $e_{i,2}$ are the two copies of $e_i$ in $B_2$.
As $R$ is a region, $\gamma^1$ and $\gamma^2$ intersect in $D_{v_0}$ and $D_{v_s}$, hence $v_0$ and $v_s$ are traversal vertices.
Consider the following modification to $\prm M^1$ and $\prm M^2$ in $B_2$: for each $1 \leq i \leq s$, we swap $e_{i,1}$ with $e_{i,2}$,
so that $e_{i,1}$ now belongs to $\prm M^2$ and $e_{i,2}$ belongs to $\prm M^1$. After this operation, for any $1 \leq i < s$,
the vertex $v_i$ is a traversal vertex if and only if it was a traversal vertex before the operation, while $v_0$ and $v_s$
discontinue to be traversal vertices. This contradicts the minimization criterium for the choice of $\prm M^1$ and $\prm M^2$
in $B_2$.% and finishes the proof of the lemma.
\end{proof}

%\subsubsection{Finding elementary regions}
\subsubsection{Two maximal mountains form a range}
Intuitively, elementary regions that are finite and do not belong to $\regspm{++}$
often give grounds to applying Lemma~\ref{lem:largermountain} and to the conclusion that
$M^1$ or $M^2$ is not maximal.
In this argumentation, we need to watch out for the following
special case. Informally speaking, the cushion is the artificial face created between the two copies of $\prm B[l,r]$ when we duplicated the edges of $B$;
however, it can contain some other faces if $l$ or $r$ is not a traversal vertex of the mountains $M^1$ and $M^2$.
\begin{definition}[cushion]
An elementary region $R \in \regs$ is called a {\em{cushion}}
if $W(R)$ contains two copies of $\prm B[l,r]$ (one from $\prm M^1$ and one from $\prm M^2$),
and $W_2(R)$ contains all  edges of $\prm M^1 \cup \prm M^2$ incident to $l$ or $r$.
\end{definition}

We are now ready for a crucial definition that is necessary to prove Theorem~\ref{thm:mountain-range}.

\begin{definition}[range]
\label{definition-forms-a-range}
We say that $M^1$ and $M^2$ {\em form a range} when the following condition hold
\begin{itemize}
%\item infinite region $R_{\infty}$ in $\mathcal{R}(\prm M,\prm M^2)$ is elementary,
%\item $P_L \subseteq \prm R_{\infty}$ and $P_R' \subseteq \prm R_{\infty}$,
\item there is exactly one region $R^{+-}$ in $\regspm{+-}$ and one region $R^{-+}$ in $\regspm{-+}$;
\item $R^{+-}$ and $R^{-+}$ are elementary and neither of them is a cushion;
\item $v_{M^2} \in V(W(R^{-+})) \setminus V(P^1(R^{-+}))$ and $v_{M^1} \in V(W(R^{+-})) \setminus V(P^2(R^{+-}))$.
\end{itemize}
\end{definition}
%If $M$ precedes $M^2$ or vice versa we say that $M$ and $M^2$ {\em form a mountain range}.

The main step of the proof of Theorem~\ref{thm:mountain-range}, which we take in this section, is to show that every pair of maximal $\delta$-mountains forms a range. Observe that a necessary condition for $M^1$ and $M^2$ to form a range is that $\gamma^1$ and $\gamma^2$ cross only in two points. The following lemma is used to establish this condition.

\begin{lemma}
\label{lemma-pseudocircles}
If there exist two elementary regions $R_1,R_2 \in \regs$ that have a common incident curve-arc, then $|\gamma^1 \cap \gamma^2| = 2$.
\end{lemma}
\begin{proof}
Let $\gamma$ be the common incident curve-arc between $R_1$ and $R_2$, and let $a$ and $b$ be its endpoints.
Without loss of generality,
we assume that $\gamma \subseteq \gamma^1$. We have that $\prm R_1 \setminus \gamma \subseteq \gamma^2$ and $\prm R_2 \setminus \gamma \subseteq \gamma^2$,
so $\gamma^2 = (\prm R_1 \cup \prm R_2) \setminus \gamma$. Hence, $\gamma^1$ and $\gamma^2$ cross only in $a$ and $b$.
\end{proof}

%Using this lemma, we can split the possible configurations of $M^1$ and $M^2$ into the following cases.
We can now split the possible configurations of $M^1$ and $M^2$ into the following cases.

\begin{lemma}
\label{lemma-jordan2}
One of the following holds:
\begin{enumerate}[(i)]
\item there exists a finite elementary region $R \in \regspm{--}$;\label{case:mm}\label{case:jordan:first}
\item there exists an elementary region $R \in \regspm{+-} \cup \regspm{-+}$, such
that $v_{M^1} \notin V(W(R)) \setminus V(P^2(R))$ and $v_{M^2} \notin V(W(R)) \setminus V(P^1(R))$;\label{case:pm}
\item the infinite region $R_{\infty}\in \regs$ is
elementary and is not incident to $v_{M^1}$ nor $v_{M^2}$;\label{case:inf}
\item there exists a finite elementary region that is a cushion;\label{case:cushion}\label{case:jordan:last}
\item $M^1$ and $M^2$ form a range.\label{case:range}
\end{enumerate}
\end{lemma}
\begin{proof}
From Lemma~\ref{lemma-jordan} we know that there exist two elementary regions $R_1^1,R_2^1 \in \regspm{--} \cup \regspm{-+}$
and two elementary regions $R_1^2,R_2^2 \in \regspm{--} \cup \regspm{+-}$.
Regions $R_i^1$ and $R_j^2$ may be sometimes equal.
Up to symmetry, we have the following cases.

\case{$R_1^1=R_1^2$ and $R_2^1=R_2^2$}
In this case, both $R_1^1=R_1^2$ and $R_2^1=R_2^2$ belong to $\regspm{--}$.
Hence, one of these two elementary regions is not infinite and Case~(\ref{case:mm}) holds.

\case{$R_1^1=R_1^2$ and $R_2^1\neq R_2^2$} If $R_1^1=R_1^2$ is not infinite, then Case~(\ref{case:mm}) holds.
Hence, assume the contrary, which implies that the infinite region is elementary. Now, neither $R_2^1$ nor $R_2^2$ can be infinite.
On the other hand, if one of them belongs to $\regspm{--}$ then Case~(\ref{case:mm}) holds.
We are left with the case $R_2^1 \in \regspm{-+}$ and $R_2^2 \in \regspm{+-}$.
By Lemma~\ref{lem:nei-regions}, $R_2^1$ and $R_2^2$ are not incident to a common curve-arc.
We note that, as $\prm M^1$ and $\prm M^2$ are simple cycles, only one region $R\in \{R_2^1,R_2^2\}$ may satisfy $v_{M^1} \in V(W(R)) \setminus V(P^2(R))$
and only one region $R\in \{R_2^1,R_2^2\}$ may satisfy $v_{M^2} \in V(W(R)) \setminus V(P^1(R))$.
Hence, if Case~(\ref{case:pm}) does not hold for both $R_2^1$ and $R_2^2$,
we need to have that $v_{M^1} \in V(W(R_2^1)) \setminus V(P^2(R_2^1))$ and $v_{M^2} \in V(W(R_2^2)) \setminus V(P^1(R_2^2))$ or vice versa (i.e., with the roles of $R_2^1$ and $R_2^2$ swapped).
In particular, neither $v_{M^1}$ nor $v_{M^2}$ is a traversal vertex.
If Case~(\ref{case:inf}) does not hold, then since the infinite region $R_1^1=R_1^2$ is elementary, either $v_{M^1}$ or $v_{M^2}$ has to be on the border of the infinite region $R_1^1=R_1^2$.
This implies that either $R_2^1$ or $R_2^2$ shares a curve-arc with $R_1^1=R_1^2$. By applying Lemma~\ref{lemma-pseudocircles}
to these two incident elementary regions we know that $\gamma^1$ and $\gamma^2$ cross exactly twice.
Consequently, each set $\regspm{++}$, $\regspm{+-}$, $\regspm{-+}$ and $\regspm{--}$ has size exactly one and
all regions in $\regs$ are elementary.
Moreover, as $v_{M^1}$ or $v_{M^2}$ is on the border of the infinite region $R_1^1=R_1^2$,
we infer that in fact
$v_{M^1} \in V(W(R_2^2)) \setminus V(P^2(R_2^2))$
and
$v_{M^2} \in V(W(R_2^1)) \setminus V(P^1(R_2^1))$.
If $R_2^1$ or $R_2^2$ is a cushion, we have Case~(\ref{case:cushion}).
Otherwise, $R^{+-}=R_2^2$ and $R^{-+}=R_2^1$ fulfills Definition~\ref{definition-forms-a-range}.

\case{All four $R_1^1,R_1^2,R_2^1,R_2^2$ are different.}
If at least two of these regions belong to $\regspm{--}$, then one is finite and we have Case~(\ref{case:mm}). Therefore, at least three of the regions belong to $\regspm{+-} \cup \regspm{-+}$. Lemma~\ref{lem:nei-regions} implies that at most one of them has $v_{M^1} \in V(W(R)) \setminus V(P^2(R))$ and at most one has $v_{M^2} \in V(W(R)) \setminus V(P^1(R))$. Therefore, at least one of the regions satisfies Case~(\ref{case:pm}).
%Moreover, due to Lemma~\ref{lem:nei-regions}, only one region $R$ among those four regions may belong to $\regspm{+-} \cup \regspm{-+}$ and have $v_{M^1} \in V(W(R)) \setminus V(P^2(R))$ and only one region $R$ may have $v_{M^2} \in V(W(R)) \setminus V(P^1(R))$. Hence, if at least three of this regions belong to $\regspm{+-} \cup \regspm{-+}$ then at least one of them satisfies Case~(\ref{case:pm}).
\end{proof}

%\subsubsection{Refuting non-range cases}

In the next lemmata we show that when one of the
Cases~\ref{case:jordan:first}--\ref{case:jordan:last} of Lemma~\ref{lemma-jordan2} holds,
then either $M^1$ or $M^2$ is not maximal.
Our main tools in the upcoming arguments are Lemmata~\ref{lem:shortcut} and~\ref{lem:largermountain}.

\begin{lemma}\label{lemma:mountain-range-1}
If Case~(\ref{case:mm}) in Lemma~\ref{lemma-jordan2} holds, then $M^1$ or $M^2$ is not maximal.
\end{lemma}
\begin{proof}
Let $R$ be the elementary region $R$ promised by Case~(\ref{case:mm}).
By Lemma~\ref{lem:true-region}, $W(R)$ encloses at least one finite face of $B$ and $P^1(R) \neq P^2(R)$.
If $\wei{P^1(R)} \leq \wei{P^2(R)}$, then we can apply Lemma~\ref{lem:largermountain} to $P^1(R)$ and $M^2$, implying that there
exists a mountain that strictly contains $M^2$. Otherwise, \ie if $\wei{P^1(R)} >\wei{P^2(R)}$, then we can apply
Lemma~\ref{lem:largermountain} to $P^2(R)$ and $M^1$, implying that there exists a mountain that strictly contains $M^1$.
\end{proof}

\begin{lemma}\label{lemma:mountain-range-2}
If Case~(\ref{case:pm}) in Lemma~\ref{lemma-jordan2} holds, then $M^1$ or $M^2$ is not maximal.
\end{lemma}
\begin{proof}
Let $R$ be the elementary region $R$ promised by Case~(\ref{case:pm}).
By Lemma~\ref{lem:true-region}, $W(R)$ encloses at least one finite face of $B$ and $P^1(R) \neq P^2(R)$.
Without loss of generality, assume that $R \in \regspm{-+}$.
If $\wei{P^1(R)} \geq \wei{P^2(R)}$, then we can apply Lemma~\ref{lem:largermountain} to $P^2(R)$ and $M^1$, implying that there exists
a mountain that strictly contains $M^1$. Hence, we are left with the case $\wei{P^1(R)} < \wei{P^2(R)}$.

Note that $P^1(R)$ is enclosed by $\prm M^2$.
Let $v_1,v_2,\ldots,v_s$ be the vertices of $V(P^1(R)) \cap V(P^2(R))$, in the order
of their appearance on $P^2(R)$. Note that $s \geq 2$, as $v_1, v_s$ are the endpoints of $P^1(R)$ and $P^2(R)$.
Moreover, $v_1,v_2,\ldots,v_s$ is also the order of the appearance of vertices of $V(P^1(R)) \cap V(P^2(R))$ on $P^1(R)$,
as $P^1(R)$ is a simple path and is enclosed by $\prm M^2$.
As $\wei{P^1(R)} < \wei{P^2(R)}$, there exists an index $1 < i \leq s$ such that $\wei{P^1(R)[v_{i-1},v_i]} < \wei{P^2(R)[v_{i-1},v_i]}$.

By the properties of Case~(\ref{case:pm}), $v_{M^2}$ is not in $V(P^2(R)) \setminus V(P^1(R))$; in particular,
$v_{M^2}$ is not an internal vertex of $P^2(R)[v_{i-1},v_i]$.
As $R \in \regspm{-+}$, that is, $R \setminus \bigcup_{v \in V(B^2)} D_v$ is not enclosed by $\prm M^1$, and since
$W(R)$ encloses $C := P^1(R)[v_{i-1},v_i] \cup P^2(R)[v_{i-1},v_i]$, $C$ cannot enclose any face incident
to an edge of $\prm B[l,r]$. Consequently, $P^2(R)[v_{i-1},v_i]$ is a subpath of $P_L^2$ or $P_R^2$.
However, as $\wei{P^1(R)[v_{i-1},v_i]} < \wei{P^2(R)[v_{i-1},v_i]}$ and $P^1(R)[v_{i-1},v_i]$ is enclosed
by $\prm M^2$, this contradicts Lemma~\ref{lem:shortcut} and finishes the proof of the lemma.
\end{proof}

\begin{lemma}\label{lemma:mountain-range-3}
If Case~(\ref{case:inf}) in Lemma~\ref{lemma-jordan2} holds, then $M^1$ or $M^2$ is not maximal.
\end{lemma}
\begin{proof}
Let $R=R_\infty$ be the elementary infinite region promised by Case~(\ref{case:inf}).
Let $a$ and $b$ be the endpoints of $P^1(R)$ and $P^2(R)$, and let $Q^\alpha = \prm M^\alpha \setminus P^\alpha(R)$ for $\alpha=1,2$.
Note that $Q^1 \neq P^2(R)$ (and, symmetrically, $Q^2 \neq P^1(R)$), as otherwise $\prm M^1$ encloses $M^2$; however,
in this case $\gamma^1$ and $\gamma^2$ would be disjoint, due to the minimization criterium used in the construction of $\prm M^1$ and $\prm M^2$ in $B_2$.
Consequently, if $\wei{Q^1} \geq \wei{P^2(R)}$ or $\wei{Q^2} \geq \wei{P^1(R)}$, then we may apply Lemma~\ref{lem:largermountain}
either to the pair $(P^2(R), M^1)$ or to the pair $(P^1(R), M^2)$, finishing the proof of the lemma. Hence, we are left with the case $\wei{Q^1} < \wei{P^2(R)}$ and $\wei{Q^2} < \wei{P^1(R)}$.

By Lemma~\ref{lem:border-traversal}, for exactly one $\alpha \in \{1,2\}$
all edges of the path $\prm M^\alpha[l,r]$ are incident to the infinite face in $B_2$, and
all edges of $\prm M^{3-\alpha}[l,r]$ are not incident to the infinite face in $B_2$.
Moreover, neither $a$ nor $b$ is an internal vertex of $\prm B[l,r]$.
Consequently, at least one of the paths $P^1(R)$ and $P^2(R)$ does not contain any edge of $\prm B[l,r]$.
Without loss of generality, assume it is $P^1(R)$.
Moreover, by the properties of Case~(\ref{case:inf}), $P^1(R)$ does not contain $v_{M^1}$.
Hence, $P^1(R)$ is a subpath of $P_L^1$ or $P_R^2$. However, $\wei{Q^2} < \wei{P^1(R)}$ and $Q^2$ is enclosed by $\prm M^1$.
This contradicts Lemma~\ref{lem:shortcut}.
\end{proof}

\begin{lemma}\label{lemma:mountain-range-4}\label{lem:mountain-range-last}
If Case~(\ref{case:cushion}) in Lemma~\ref{lemma-jordan2} holds, then $M^1$ or $M^2$ is not maximal.
\end{lemma}
\begin{proof}
Let $R$ be the cushion promised by Case~(\ref{case:cushion}).
By the definition of a cushion, $\prm B[l,r]$ is a subpath of both $P^1(R)$ and $P^2(R)$.
As $W_2(R)$ encloses all faces of $B_2$ between the copies of the edges of $\prm B[l,r]$, $R \in \regspm{+-} \cup \regspm{-+}$.
Without loss of generality assume that $P^1_2(R)[l,r]$ is incident to the infinite face of $B_2$ and thus $R \in \regspm{+-}$.
Let $a$ be the endpoint of $P^1(R)$ that lies closer to $l$ than to $r$, and let $b$ be the other endpoint;
note that also on $P^2(R)$ the endpoint $a$ is closer to $l$ than to $r$.

Assume that $P^1(R)[a,l] = P^2(R)[a,l]$. Consider the following operation: for each edge $e$
of $P^1(R)[a,l]$, we swap which copy of $e$ in $B_2$ belongs to $\prm M^1$ and which to $\prm M^2$.
In this manner, an internal vertex of $P^1(R)[a,l]$ is a traversal vertex if and only if it was traversal vertex prior to the operation,
whereas $a$ discontinues to be a traversal vertex and $l$ becomes a traversal vertex.
Consequently, the operation does not change the total number of traversal vertices while strictly
decreasing the number of traversal vertices that are not equal to $l$ or $r$, a contradiction to the choice of $\prm M^1$ and $\prm M^2$ in $B_2$.

We infer that the closed walks $P^1(R)[a,l] \cup P^2(R)[a,l]$ and $P^1(R)[r,b] \cup P^2(R)[r,b]$
enclose each at least one face of $B$. The vertex $v_{M^1}$ cannot lie both on $P^1(R)[a,l]$ and $P^1(R)[r,b]$;
without loss of generality assume it does not lie on $P^1(R)[a,l]$, and $P^1(R)[a,l]$ is a subpath of $P_L^1$.
If $\wei{P^1(R)[a,l]} \leq \wei{P^2(R)[a,l]}$ then we may apply Lemma~\ref{lem:largermountain}
to the pair $(P^1(R)[a,l], M^2)$.
Otherwise, $\wei{P^2(R)[a,l]} < \wei{P^1(R)[a,l]}$. However, $P^2(R)[a,l]$ is enclosed by $\prm M^1$
and $P^1(R)[a,l]$ is a subpath of $P_L^1$. This contradicts Lemma~\ref{lem:shortcut}.
\end{proof}

As a consequence of the above lemmata, we infer the following.
\begin{corollary}\label{cor:range}
Any two distinct maximal mountains form a range.
\end{corollary}

\subsection{The range of all maximal mountains}

We now analyse the structure of all maximal mountains, using the crucial property established in Corollary~\ref{cor:range} that any two distinct maximal mountains form a range.
Recall that, formally, a mountain is only a carve in $B$, and therefore, there is only a finite number of mountains.
Hence, we may assume that some edges of $B$ have been subdivided, so that each mountain with endpoints $l$ and $r$
can choose its summit among the vertices of $B$.

We start with the following observation that the mountain range relation implies an order on the set of maximal mountains.

\begin{lemma}\label{lem:mountain-order}
Let two mountains $M^1 = \mou{P_{L}^1}{P_{R}^1}$ and $M^2 = \mou{P_{L}^2}{P_{R}^2}$ form a range.
Then $\wei{P_L^1} < \wei{P_L^2}$ or $\wei{P_L^2} < \wei{P_L^1}$.
Moreover, if $\wei{P_L^1}<\wei{P_L^2}$, then $P^1(R^{-+})$ is a subpath of $\prm B[l,r] \cup P_R^1$, where $R^{-+}$ is the elementary region in $\regspm{-+}$.
\end{lemma}
\begin{proof}
Consider the unique regions $R^{-+} \in \regspm{-+}$ and $R^{+-} \in \regspm{+-}$.
Note that, by Lemma~\ref{lem:nei-regions}, they do not share any curve-arc that makes up their borders
and, consequently, for $\alpha=1,2$, $P^\alpha(R^{+-})$ and $P^\alpha(R^{-+})$ are edge-disjoint.
Moreover, by definition of forming a range,
$v_{M^2}$ does not lie on $P^2(R^{+-})$ and $v_{M^1}$ does not lie on $P^1(R^{-+})$.

We also infer from Lemma~\ref{lem:nei-regions} that, since $|\regspm{-+}| = |\regspm{+-}|=1$,
there are only four curve-arcs, all incident to $R^{-+}$ or $R^{+-}$ and, consequently,
$\regspm{--}$ consist only of the infinite region $R_\infty$ and $\regspm{++}$ consists
only of one region $R^{++}$.

Consider the faces of $B_2$ between the copies of the edges of $\prm B[l,r]$.
By Lemma~\ref{lem:border-traversal}, they are all enclosed by $W_2(R)$ for a single region $R$.
Moreover, as they are enclosed by only one of $\prm M^1$ and $\prm M^2$ in $B_2$, $R = R^{+-}$ or $R=R^{-+}$.
As neither $R^{+-}$ nor $R^{-+}$ is a cushion, exactly one of the vertices $l$ and $r$ is a traversal vertex,
   and an endpoint of all four paths $P^1(R^{+-})$, $P^2(R^{+-})$, $P^1(R^{-+})$ and $P^2(R^{-+})$.
We note that we may assume $r$ to be the traversal vertex, as the other case
can be reduced to this one by swapping the copies of the edges of $\prm B[l,r]$ in $B_2$ between $\prm M^1$ and $\prm M^2$.
Moreover, by symmetry between $M^1$ and $M^2$, without loss of generality we may assume that
the faces of $B_2$ between the copies of the edges of $\prm B[l,r]$ are enclosed by $W_2(R^{+-})$;
this implies that $\prm M^1[l,r]$ is incident to the infinite face of $B_2$.
Hence, $l$ lies on $P^2(R^{+-})$, that is, $P^1(R^{+-})[l,r] = P^2(R^{+-})[l,r] = \prm B[l,r]$.
Let $a$ be the intersection of $\gamma^1$ and $\gamma^2$ different than $r$,
and, at the same time, the endpoint of the paths
$P^1(R^{+-})$, $P^2(R^{+-})$, $P^1(R^{-+})$ and $P^2(R^{-+})$.
As $v_{M^1}$ lies on $P^1(R^{+-})$, the path $P^2(R^{+-})[l,a]$ is a path connecting
$l$ with $P_R^1$ that is enclosed by $\prm M^1$. Consequently,
$\wei{P_L^1} \leq \wei{P^2(R^{+-})[l,a]} < \wei{P_L^2}$, where the last inequality follows from the fact
that $v_{M^2}$ lies on $P^2(R^{-+})$ and $v_{M^2} \neq a$, thus $P^2(R^{+-})[l,a]$ is a proper subpath of $P_L^2$.
The second part of the lemma is immediate from the above discussion.

Note that if we would assume that the faces of $B_2$ between the copies of the edges
of $\prm B[l,r]$ are enclosed by $W_2(R^{-+})$, the roles of $M^1$ and $M^2$
would change in the above reasoning and we would obtain $\wei{P_L^2} < \wei{P_L^1}$.
This concludes the proof of the lemma.
\end{proof}

Now we proceed to analyse the union of all maximal mountains.

\begin{lemma}
\label{lemma:range-length}
There exists a closed walk $W$ of length at most $3\wei{\prm B[l,r]}$ such that
$W$ encloses a face $f$ if and only if $f$ is contained in some maximal mountain.
\end{lemma}
\begin{proof}
Let $\{ M^i=(P^i_L, P^i_R)\}_{i=1}^s$ be the set of all maximal mountains such that $\wei{P^i_L} < \wei{P^j_L}$ for $1\le i <j\le s$.
By induction, we show closed walks $W^1,W^2,\ldots,W^s$ such that for each $i=1,2,\ldots,s$, the following holds:
\begin{enumerate}
\item $W^i$ contains $P^i_R \cup \prm B[l,r]$ as a subpath.
\item If we define $\widehat{\gamma}^i$ to be the closed curve in the plane $\plane$ obtained by traversing $W^i$ in the direction so that the $P^i_R \cup \prm B[l,r]$
is traversed from $l$ to $v_{M^i}$, then, for any face $f$ of $B$ and any point $c$ in the interior of $f$:
\begin{itemize}
\item if $f$ belongs to one of the mountains $M^1,M^2,\ldots,M^i$ then $\widehat{\gamma}^i$ is a positive element of the fundamental group $\Gamma_c \cong \mathbb{Z}$
of $\plane \setminus \{c\}$;
\item otherwise, $\widehat{\gamma}^i$ is the neutral element of this group.
\end{itemize}
In particular, $W^i$ encloses $f$ if and only if $f$ is contained in one of
the mountains $M^1,M^2,\ldots,M^i$.
\item $\wei{W^i} \leq \wei{P^1_R} + \wei{P^i_L}+ \wei{\prm B[l,r]}$.
\end{enumerate}
Here, property 2 formalizes the intuition that maximal mountains look as they do in Figure~\ref{fig:range}. In reality, the boundaries of the mountains may actually intersect often (but not cross more than twice), which is why we need this formal property.

For $i=1$, the induction hypothesis holds by taking $W^1 := \prm M^1$. Now assume that the induction hypothesis holds for $W^i$. Consider mountains $M^i$ and $M^{i+1}$ and apply Lemma~\ref{lem:mountain-order} to them;
by abuse of notation, we denote the appropriate paths as $P^i$ and $P^{i+1}$ instead of $P^1$ and $P^2$.

From Corollary~\ref{cor:range} and Definition~\ref{definition-forms-a-range}, we know that there exists a unique region $R \in \regspm{-+}$.
Recall that $P^i_R \cup \prm B[l,r]$ is a subpath of $W^i$, and that $\wei{P^{i}_L} < \wei{P^{i+1}_L}$ by the chosen order.
Hence, by Lemma~\ref{lem:mountain-order}, $P^i(R)$ is a subpath of $W^i$.
We define $W^{i+1}$ as $W^i$ with
$P^i(R)$ replaced with $P^{i+1}(R)$.
Moreover, as $v_{M^{i+1}}$ lies on $P^{i+1}(R)$, it follows that
$P^{i+1}_R \cup \prm B[l,r]$ is a subpath of $W^{i+1}$.

Let $\gamma_W^{i+1}$ be the closed curve obtained by traversing $W(R)$
in counter-clockwise direction, that is
in $\gamma_W^{i+1}$ the path $P^i(R)$
is traversed from the endpoint closer to $v_{M^i}$ to the endpoint
closer to or on the carvebase $\prm B[l,r]$.
Consider any face $f$ of $B$ and any point $c$ in its interior.
Note that, in the fundamental group $\Gamma_c \cong \mathbb{Z}$ of $\plane \setminus \{c\}$,
we have $\widehat{\gamma}^i + \gamma_W^{i+1} = \widehat{\gamma}^{i+1}$.
If $f$ is enclosed by $\gamma_W^{i+1}$, that is, by $W(R)$, then
$\gamma_W^{i+1}$ is a positive element of $\Gamma_c$, and otherwise
it is the neutral element.
We infer that the second condition is satisfied
for the curve $\widehat{\gamma}^{i+1}$,
due to the induction hypothesis, and since
$W(R)$ encloses a face $f$ if and only if $f$ is contained in $M^{i+1}$, but not in $M^i$.

Thus, to finish the proof of the induction step we need to show the bound on
the length of $W^{i+1}$.

Define $b=\wei{P^{i+1}_R}$ and $e=\wei{P^i_L}$.
Let $v$ be the first point on $P^{i+1}_L$ that lies on $P^{i}_R$. We denote the distance (along $P^{i+1}_L$) from $l$ to $v$ as $d$
and the distance from $v$ to $v_{M^{i+1}}$ as $a$. Finally, we denote by $c$ the distance (along $P^{i}_R$) from $r$ to $v$.
These definitions are illustrated in Figure~\ref{fig:range}.

\begin{figure}[h]
\centering
  \svg{0.4\textwidth}{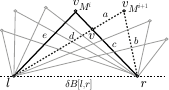}
\caption{(Figure~\ref{fig-over:rangefull} repeated) Illustration of the inductive proof in Lemma~\ref{lemma:range-length}.}
\label{fig:range}
\end{figure}
Observe that $d\ge e$ because $M^i$ is a mountain. Similarly, observe that $c \ge b$ because $M^{i+1}$ is a mountain.
Hence, we have
\[
\wei{W^{i+1}} -\wei{W^i} = a + b-c \le a \le a+d-e = \wei{P^{i+1}_L} - \wei{P^{i}_L}.
\]
Using the induction hypothesis with the above inequality we obtain
\begin{align*}
\wei{W^{i+1}} &= \wei{W^{i+1}} - \wei{W^{i}}  + \wei{W^{i}} \\
    &\le \wei{P^1_R} + \wei{P^i_L}+ \wei{\prm B[l,r]} + \wei{P^{i+1}_L} - \wei{P^{i}_L} \\
    &=\wei{P^1_R} + \wei{P^{i+1}_L}+ \wei{\prm B[l,r]}.
\end{align*}
This proves the induction. Hence, $W := W^s$ satisfies the conditions of the lemma as
\begin{align*}
\wei{W^{s}} &\le \wei{P^1_R} + \wei{P^{s}_L}+ \wei{\prm B[l,r]} \\
    &\le \wei{\prm B[l,r]} + \wei{\prm B[l,r]} + \wei{\prm B[l,r]} \\
    &= 3\wei{\prm B[l,r]}.
\end{align*}
\end{proof}

We remark here that, although formally the reasoning of Lemma~\ref{lemma:range-length} has been done in the presence of all summits of mountains, the obtained walk $W$ projects back to the original brick $B$, where edges have not been subdivided. 

\subsection{Finding the mountain range}

Finally, we show the algorithm to compute the $\delta$-mountain range. We need the following technical observation.
Consider a plane drawing of $B$ in which the segment $\prm B[l,r]$ is drawn as a horizontal
segment and $l$ is the left end of it. The {\it leftmost shortest} path from $l$ to $v$ is the shortest path that lies as
much as possible to the left in the drawing of $B$. Symmetrically, we define the {\it rightmost shortest} path from $r$ to $v$. Note that these notions are well defined, as they correspond to taking furthest counter-clockwise and clockwise objects around $l$ and $r$, respectively, in the semi-plane above segment $\prm B[l,r]$ that contains brick $B$.

Observe the following connection between left- and rightmost shortest paths and maximal mountains.
\begin{lemma}
\label{lemma-unique-mountain}
For fixed $xy \in B$, there
exists at most one maximal $\delta$-mountain $M^{x,y}$ of $B$ that can choose the summit on the edge $xy$ (possibly in $x$ or $y$) and has $x$ closer to $l$ on the carvemark $M^{x,y}$ than $y$.
Moreover, the carvemark of $M^{x,y}$ consists of the leftmost shortest path from $l$ to $x$ in $B$, the edge $xy$
and the rightmost shortest path from $r$ to $y$ in $B$.
\end{lemma}
\begin{proof}
Let $M$ be a maximal mountain that contains $xy$ on the carvemark $M$, such that there exists a witness
$\mouth_M$ with $\wei{M[l,x]} \leq \mouth_M \leq \wei{M[l,y]}$.
Let $P$ be any shortest path between $l$ and $x$ in $B$. We claim that $P$ is enclosed by $M$, and a symmetrical claim holds for any shortest path between $r$ and $y$ in $B$.
Note that this statement would conclude the proof of the lemma.

Assume the contrary, and let $Q = P[a,b]$ be any subpath of $P$ whose all edges and internal vertices are not enclosed by $M$, but both endpoints $a,b$ of $Q$ lie on the carvemark $M$.
By Lemma~\ref{lem:largermountain}, $\wei{Q} > \wei{M[a,b]}$ or $\wei{Q} > \wei{M[b,a]}$, a contradiction to the assumption that $P$ is a shortest path.
The arguments for paths connecting $y$ and $r$ are symmetric.
\end{proof}

Observe that for a fixed vertex $u \in V(\prm B)$, the union of all leftmost shortest paths from $u$ to every $v \in V(B)$ is a shortest-path tree rooted at $u$;
we call it the \emph{leftmost shortest-path tree} rooted at $u$.
An analogous claim holds for rightmost shortest paths by symmetry.
We show that this shortest-path tree can be found efficiently for fixed $u$.

\begin{lemma}\label{lem:leftmost-shortest}
For any fixed $u \in V(\prm B)$, the leftmost shortest-path tree rooted at $u$ (and, symmetrically, the rightmost one) can be found
in $O(|B|)$ time.
%in $\Oh(|B| \log |B|)$ time in the edge-weighted setting, and $\Oh(|\prm B| \cdot |B|)$ time in the unweighted setting.
\end{lemma}
\begin{proof}
The approach is the same as one proposed by Klein~\cite{klein:mssp}. First, we find a shortest path tree from $u$ in linear time~\cite{planar-sp}. Let $d(v)$ denote the distance from $u$ to $v$ for any $v \in V(B)$. Let $H$ be the following directed graph. The vertex set of $H$ is $V(B)$. Then $H$ contains the arc $(v,w)$ if and only if $vw$ is an edge of $B$ and $d(w)=d(v)+\wei{wv}$. Observe that $H$ is acyclic, as all edge weights are positive. Now it suffices to find a leftmost search tree (see \eg\cite{leftmost}) in $H$. This can be done in linear time using a simple depth-first search, which visits the neighbours of a vertex in left-to-right order. By the construction of $H$, this immediately translates into a leftmost shortest path tree in $B$. A rightmost shortest-path tree can be found symmetrically.
\end{proof}

In the next lemma, we make use of the left- and rightmost shortest-path trees
and conclude the proof of Theorem~\ref{thm:mountain-range}.

\begin{lemma}
\label{lemma:range-algorithm}
The union of all finite faces of $\delta$-mountains for fixed $l,r$ can be computed 
in $O(|B|)$ time.
%in time $\Oh(|B| \log |B|)$ in the weighted setting and in time $\Oh(|\prm B| \cdot |B|)$ in the unweighted setting.
\end{lemma}
\begin{proof}
Using Lemma~\ref{lem:leftmost-shortest} we compute the leftmost shortest-path tree rooted at $l$
and the rightmost shortest-path tree rooted at $r$. Denote these trees $T_l$ and $T_r$, respectively.

By traversing the tree $T_l$ from the root to its leaves, we compute for each $v \in V(B)$ the value
$$d_l(v) = \min \{\dist_B(v',r): v'\textrm{ is an ancestor of }v\textrm{ in the tree }T_l\}.$$
Symmetrically, we compute values $d_r(v)$ in the tree $T_r$ (taking into account distances to $l$). This takes $O(|B|)$ time.

Let $Z$ be the set of pairs $(x,y) \in V(B) \times V(B)$ such that $xy \in B$,
$\dist_B(x,l) \leq d_r(y)$, $\dist_B(y,r) \leq d_l(x)$ and
$d_r(y) + d_l(x) \geq \dist_B(x,l) + \wei{xy} + \dist_B(y,r)$.
For every  $(x,y) \in Z$, consider a walk $M^{x,y}$ in $B$ that consists of
the leftmost shortest path from $l$ to $x$ (i.e., the path from $x$ to the root $l$
 in $T_l$), the edge $xy$ and the rightmost shortest path from $r$ to $y$ (i.e., the path from $y$ to the root $r$ in $T_r$).
We observe the following equivalence, captured in the next two claims.
\begin{claim}\label{cl:xy-is-mountain}
For every $(x,y) \in Z$, $M^{x,y}$ is a mountain.
\end{claim}
\begin{proof}
First observe that $M^{x,y}[l,x]$ and $M^{x,y}[y,r]$ cannot share a vertex, as otherwise
$d_r(y) + d_l(x) \leq \wei{M^{x,y}[l,x]} + \wei{M^{x,y}[y,r]} = \dist_B(x,l) + \dist_B(y,r)$, a contradiction
to the properties of the pairs in $Z$ and the fact that $\wei{xy} > 0$. Hence, $M^{x,y}$ is a path.

We claim that $M^{x,y}$ is a mountain for $\mouth = d_r(y)$. By the properties of pairs in $Z$ we have $\wei{M^{x,y}[l,x]} \leq \mouth \leq \wei{M^{x,y}[l,y]}$ and, consequently,
the candidate summit $v := v(M^{x,y},l,\mouth)$ is located on the edge $xy$ (possibly at one of the endpoints).
If needed, subdivide the edge $xy$ with the vertex $v$.
As $\dist_B(x,l) \leq d_r(y)$, by the definition of $d_r(y)$, we have $\dist_B(l,M^{x,y}[v,r]) \geq d_r(y)$.
Regarding the distances from $r$, first observe that $\wei{vy} + \dist_B(y,r) = \wei{M^{x,y}[v,r]}$ and, hence,
any path in $B$ connecting $r$ and $v$ that passes through $y$ is of length at least $\wei{M^{x,y}[v,r]}$.
Second, note that
\begin{align*}
\dist_B(V(M^{x,y}[l,x]),r) &= d_l(x) \geq \dist_B(x,l) + \wei{xy} + \dist_B(y,r) - d_r(y) \\&= \wei{M^{x,y}} - \mouth = \wei{M^{x,y}[v,r]}.
\end{align*}
\cqed\end{proof}
\begin{claim}\label{cl:max-mountain-in-Z}
Let $M$ be a maximal mountain. Then $M = M^{x,y}$ for some $(x,y) \in Z$.
\end{claim}
\begin{proof}
Let $\mouth_M$ be a real that witnesses that $M$ is a mountain and let $xy \in M$
be such that $\wei{M[l,x]} \leq \mouth_M \leq \wei{M[y,r]}$ (i.e., the summit of $M$ is on the edge $xy$, possibly in one of the endpoints).
Let $v = v(M,l,\mouth)$.
If needed, subdivide the edge $xy$ with the vertex $v$.
By Lemma~\ref{lemma-unique-mountain} we have that $M[l,x]$ is the leftmost shortest-path between $l$ and $x$ and $M[y,r]$ is the rightmost shortest-path between $y$ and $r$.
By Corollary~\ref{cor:maxmountain}, $d_l(x) = \dist_B(V(M[l,x]),r) \geq \wei{M} - \mouth_{M} = \wei{M[y,r]} + \wei{vy} = \dist_B(y,r) + \wei{vy}$
and symmetrically $d_r(y) \geq \dist_B(x,l) + \wei{xv}$. By adding up these two inequalities we obtain
$d_r(y) + d_l(x) \geq \dist_B(x,l) + \wei{xy} + \dist_B(y,r)$. Consequently, $(x,y) \in Z$ and $M^{x,y} = M$ by the construction of $M^{x,y}$.
\cqed\end{proof}
By Claims~\ref{cl:xy-is-mountain} and~\ref{cl:max-mountain-in-Z},
our goal is to compute the set of all finite faces that are enclosed by some mountain $M^{x,y}$ for $(x,y) \in Z$.

To achieve this goal, we first construct the directed dual $B^\ast_\rightarrow$ of $B$, that is, we take the undirected dual $B^\ast$
and replace each edge with two arcs in both directions.
Then, we would like to assign integer weights to the arcs of $B^\ast_\rightarrow$ in the following manner. First, set all weights to zero.
Second, for each $(x,y) \in Z$, add $+1$ to the weight of each arc that corresponds to an edge of $\prm M^{x,y}$ and ends in
the face enclosed by $M^{x,y}$, and add $-1$ to the weight of the arc in the opposite direction.
It is easy to observe that the weighted graph $B^\ast_\rightarrow$ defined in this manner has no non-null cycles, and for any face $f$,
the sum of weights on any path from the outer face to $f$ in $B^\ast_\rightarrow$ equals the number of mountains $M^{x,y}$, $(x,y) \in Z$
that enclose $f$. Consequently, given $B^\ast_\rightarrow$ it is straightforward to compute the union of all finite faces of $\delta$-mountains
for fixed $l$ and $r$.

However, inspecting the perimeters of all mountains $M^{x,y}$ for $(x,y) \in Z$ may take quadratic time. Luckily, one can compute the weights of $B^\ast_\rightarrow$ in $O(|B|)$ time as follows.
Start with all weights of $B^\ast_\rightarrow$ set to zero.
Then, traverse $T_l$ from the leaves to its root and for each edge $e \in E(T_l)$, compute
$\zeta_l(e)$: the number of pairs $(x,y) \in Z$ such that $x$ lies in the tree of $T_l \setminus \{e\}$ that does not contain $l$.
Similarly, compute the values $\zeta_r(e)$ for each $e \in E(T_r)$ that count the number of pairs $(x,y) \in Z$ such that $y$ lies
in the tree of $T_r \setminus \{e\}$ that does not contain $r$.
Observe that for each $e \in E(T_l)$, there are exactly $\zeta_l(e)$ mountains $M^{x,y}$ for which $e$ lies on the left slope of $M^{x,y}$.
Moreover, in all of these mountains, if we orient $e$ towards the root $l$ of $T_l$, the face that lies on the left-hand side of $e$ is not enclosed by $M^{x,y}$, and the one that lies on the right-hand side is enclosed by $M^{x,y}$.
Hence, we may proceed as follows: for each $e \in E(T_l)$, add weight $\zeta_l(e)$ to the arc of $B^\ast_\rightarrow$ that traverses
the edge $e$, keeping the closer-to-root endpoint of $e$ to the right hand side, and add weight $-\zeta_l(e)$
to the other arc of $B^\ast_\rightarrow$ corresponding to the edge $e$.
Similarly, for each $e \in E(T_r)$, add weight $\zeta_r(e)$ to the arc of $B^\ast_\rightarrow$ that traverses
the edge $e$ keeping the closer-to-root endpoint of $e$ to the left hand side, and add weight $-\zeta_r(e)$
to the other arc of $B^\ast_\rightarrow$ corresponding to the edge $e$.
Finally, observe that each mountain $M^{x,y}$ contains the baseline $\prm B[l,r]$ and there are exactly $|Z|$ such mountains.
To support this, for each $e \in \prm B[l,r]$, add weight $|Z|$ to all arcs that traverse an edge of $\prm B[l,r]$ and start
in the outer face, and add weight $-|Z|$ to such arcs that end in the outer face.
In this manner we have constructed the graph $B^\ast_\rightarrow$ in $O(|B|)$ time, and concluded the proof of Lemma~\ref{lemma:range-algorithm}.
\end{proof}

%!TEX root = pst-kernel.tex

\section{Taming sliding trees}\label{sec:sliding}
In the previous section, we took a major step towards finding a cycle $C$ of length $\Oh(\wei{\prm B})$ that lies close to the perimeter of $B$ and that separates the core from all vertices of degree at least three of some optimal solution for any set of terminals on $\prm B$. In fact, Lemma \ref{lem:broom-in-mountain} shows that short subtrees of optimal Steiner trees in $B$
are hidden in $\delta$-mountains. Here, `short' means that the leftmost and rightmost
path in the subtree have total length at most $(1/2 - \delta)\wei{\prm B}$. Note that
an optimal Steiner tree in $B$ has total size smaller than $\wei{\prm B}$, as $\prm B$
without an arbitrary edge connects any subset of $V(\prm B)$. Therefore,
for small $\delta$, we can `hide' almost an entire optimal Steiner tree $T$
in at most two $\delta$-mountains. 
In this section we study what is left outside these mountains.

Before we describe the main result of this section, we need an additional notion.
Let $B$ be an edge-weighted brick. For an edge $uv \in E(B)$ we say that \emph{each point of $uv$ is at distance at most $d$ from $V(\prm B)$}
if $uv \in \prm B$ or
$\dist_B(u,V(\prm B)) \leq d$, $\dist_B(v,V(\prm B)) \leq d$ and, additionally, $\dist_B(u,V(\prm B)) + \dist_B(v,V(\prm B)) + \wei{uv} \leq 2d$.
Equivalently, we may require that $uv \in \prm B$ or whenever we subdivide the edge $uv$, replacing it with a new vertex $x$ and edges $ux$, $vx$ with positive lengths satisfying
$\wei{ux} + \wei{vx} = \wei{uv}$, we have $\dist_B(x,V(\prm B)) \leq d$.
For a subgraph $H$ of $B$, we say that \emph{each point of $H$ is at distance at most $d$ from $V(\prm B)$}
if each vertex and each point of each edge of $H$ is at distance at most $d$ from $V(\prm B)$.

With this definition, we are ready to state the main theorem of this section.
\begin{theorem}\label{thm:ananas}
% Pineapple is "ananas" in any language except English, I insist on using `ananas' at least in the latex references :)
Let $\nice \in (0,1/36]$ be a fixed constant.
Assume that $B$ does not admit a short $\nice$-nice tree.
Then one can compute a simple
cycle $C$ in $B$ with the following properties:
\begin{enumerate}[(i)]
\item the length of $C$ is at most $\frac{16}{\nice^2} \wei{\prm B}$; \label{p:ananas:length}
\item each point of $C$ is within distance at most $(\frac{1}{4} - 2\nice)\wei{\prm B}$ from $V(\prm B)$;\label{p:ananas:close}
\item for each vertex $x \in V(C)$ there exists a shortest path from $x$ to $V(\prm B)$
such that no edge of the path
is strictly enclosed by $C$;\label{p:ananas:cuts}
\item $C$ encloses $\coreface$, where $\coreface$ is any
arbitrarily chosen face of $B$ promised by Theorem \ref{thm:core}
that is not carved by any $2\nice$-carve;\label{p:ananas:coreface}
\item for any $S \subseteq V(\prm B)$
there exists an optimal Steiner tree $T_S$ connecting $S$
in $B$ such that no vertex of degree at least $3$ in $T_S$
is strictly enclosed by $C$.\label{p:ananas:correct}
\end{enumerate}
The computation takes $\Oh(|B| \log\log |B|)$ time in the edge-weighted setting
and $\Oh(|B|)$ time in the unweighted setting.
\end{theorem}

We begin the proof of Theorem \ref{thm:ananas}
with a construction. %that, at the first glance, may look obscure.
Then we show how it interacts with optimal Steiner trees in $B$.

Let $\porset \subseteq V(\prm B)$ be a set of {\em{pegs}} on $\prm B$, such that
for any $v \in V(\prm B)$, there exist pegs $p_\leftarrow(v)$ and $p_\rightarrow(v)$ 
with $v \in V(\prm B[p_\leftarrow(v),p_\rightarrow(v)])$ and  $\wei{\prm B[p_\leftarrow(v),v]},\wei{\prm B[v,p_\rightarrow(v)]} \leq \nice \wei{\prm B}/2$. Here, possibly $p_\leftarrow(v)=v$ or $p_\rightarrow(v) = v$.
We choose the set of pegs $\porset$ in the following greedy manner.
We take an arbitrary vertex $v_0 \in V(\prm B)$ as a first peg and then we traverse $\prm B$ starting from $v_0$ twice, once clockwise
and once counter-clockwise. In each pass, we take as a next peg the first vertex that is of distance larger than $\nice \wei{\prm B}/2$ from the previously
placed peg.
As each pass chooses at most $2/\nice$ pegs, $|\porset| \leq 4/\nice$.

Let $\delta = 4\nice$.
For any $l, r \in \porset$, $l \neq r$,
apply Theorem \ref{thm:mountain-range} to find
the mountain range $\MR_{l,r}$ for $\delta$-mountains with endpoints
$l$ and $r$. Recall that
$\MR_{l,r}$ is a set of faces of $B$. 
Let $\MR = \bigcup_{l,r \in \porset,l \neq r} \MR_{l,r}$.
As $|\porset|$ is a constant, by Theorem~\ref{thm:mountain-range}
$\MR$ is computable within the desired time bound.

Since each $\delta$-mountain is a $\delta$-carve, $\coreface \notin \MR$.
Let $\widehat{\coreface}$ be the connected component of
$B^\ast \setminus \MR$ containing $\coreface$, where $B^\ast$
is the dual of $B$ without the outer face.
Let $C(\coreface)$ be the simple cycle in $B$ around $\widehat{\coreface}$.
Clearly, each edge of $C(\coreface)$ belongs either to some $\MR_{l,r} \setminus \prm B$ or to $\prm B$. Therefore, by Theorem \ref{thm:mountain-range},
$$\wei{C(\coreface)} \leq |\porset|(|\porset|-1) (1-2\delta)\wei{\prm B} + \wei{\prm B} \leq \frac{16}{\nice^2} \wei{\prm B}.$$

Now let $\closeB$ be the set of edges of $B$ of which each point is at distance
at most $(\frac{1}{4}-\frac{\delta}{2})\wei{\prm B} = (\frac{1}{4}-2\nice)\wei{\prm B}$
from $V(\prm B)$; note that $\closeB$ can be computed
in $\Oh(|B|)$ time by creating a super-terminal vertex $t$ in the outer face of $B$,
connecting it by unit-length edges to all vertices of $V(\prm B)$, and running a shortest-path algorithm from $t$ in the obtained plane graph in linear time~\cite{planar-sp}.
Observe that each edge of $C(\coreface)$ belongs to $\closeB$, since in the definition of $\MR_{l,r}$ we consider $4\nice$-mountains and $\nice \leq \frac{1}{36}$.

Consider now the subgraph $H$ of $B$ that contains all edges of $\closeB$ that are enclosed by $C(\coreface)$.
Let $\coreface^H$ be the face of $H$ that contains $\coreface$. As $C(\coreface)$ is a subgraph of $H$, $\coreface^H$ is a finite face of $H$.
Define $C$ to be some shortest cycle in $H$ separating the outer face of $H$ from $\coreface^H$; such a cycle exists as $\coreface^H$ is finite.
Observe that $C$ corresponds to a minimum cut between $\coreface^H$ and the outer face of $H$ in the dual of $H$. Hence, $C$ can be found in $\Oh(|B| \log \log |B|)$ time in the edge-weighted setting~\cite{cut-loglog}
and in $\Oh(|B|)$ time in the unweighted setting~\cite{cut-unit-linear}.
%EJ: Reif is too slow now
%We remark here that, with the given time bounds of Theorem~\ref{thm:ananas}, we may actually use Reif's classical algorithm~\cite{reif}: in the edge-weighted setting it works in $\Oh(|B| \log |B|)$ time if we use the linear shortest-path algorithm in planar graphs, and in the undirected setting observe that the distance between $C(\coreface)$ and $\coreface^H$ is at most $|\prm B|$ in the graph $H$, and Reif's algorithm finds the cycle $C$ in $\Oh(|B| \log |\prm B|)$ time.

We claim that the cycle $C$ satisfies all the requirements of Theorem~\ref{thm:ananas}.
Since $C(\coreface)$ is a candidate for $C$, $\wei{C} \leq \wei{C(\coreface)} \leq \frac{16}{\nice^2} \wei{\prm B}$ and property~\eqref{p:ananas:length} is satisfied.
Properties~\eqref{p:ananas:close} and~\eqref{p:ananas:coreface} follows directly from the construction of $C$.

Regarding property~\eqref{p:ananas:cuts}, consider any $x \in V(C)$ and let $P_x$ be a shortest
path between $x$ and $V(\prm B)$ that uses the minimum number of edges
strictly enclosed by $C$. Since $x\in V(\closeB)$, in particular $\dist_B(x,V(\prm B)) \leq (\frac{1}{4}-2\nice)\wei{\prm B}$, it is clear that also all edges of $P_x$ are in $\closeB$.
Assume now that $P_x$ contains some edge strictly enclosed by $C$. Then $P_x$ contains a subpath $P_x'$ between two vertices $y,z \in V(C)$ that is strictly enclosed by $C$.
By the choice of $P_x$ we infer that $\wei{C[y,z]},\wei{C[z,y]} > \wei{P_x'}$.
Since every edge of $P_x'$ is in $\closeB$, we infer that either $C[y,z] \cup P_x'$ or $C[z,y] \cup P_x'$ is a cycle that separates $\coreface^H$ from the outer face
in $H$ of length strictly shorter than $\wei{C}$, a contradiction to the choice of $C$.
Hence, no edge of $P_x$ is strictly enclosed by $C$, and property~\eqref{p:ananas:cuts} follows.

The following lemma proves that $C$ satisfies the remaining condition, property~\eqref{p:ananas:correct}, and thus finishes the proof of Theorem~\ref{thm:ananas}.
\begin{lemma}\label{lem:slide-correctness}
For any set $S \subseteq V(\prm B)$ there exists
an optimal Steiner tree $T_S$ connecting $S$ in $B$ such that
no vertex of degree at least $3$ in $T_S$ is strictly
enclosed by $C$.
\end{lemma}
\begin{proof}
Let $T$ be an optimal Steiner tree in $B$ for some set of terminals;
clearly, it is also optimal for the set of terminals $S := V(T) \cap V(\prm B)$.
Note that $T$ is a brickable connector and let $\Bb=\{B_1,B_2,\ldots,B_s\}$ the corresponding brick partition, i.e., $B_1,B_2,\ldots,B_s$ are the bricks induced by the 
faces of $T \cup \prm B$. Recall that $\sum_{i=1}^s \wei{\prm B_i} \leq \wei{\prm B} + 2\wei{T}$.

For each brick $B_i$, let $a_i,b_i \in V(\prm B)$ be such that
$\prm B[a_i,b_i] = \prm B_i \setminus T$.
Since $T$ is an optimal Steiner tree for some choice of terminals on $\prm B$, we have that $T$ is short.
By assumption we have that $T$ is not $\nice$-nice, so there exists a brick
$B_i$ with $\wei{\prm B_i} > (1-\nice)\wei{\prm B}$. 
Let $B_i$ be such a large brick.
Note that $\prm B[b_i,a_i]$ connects $S$, so
$\wei{T} \leq \wei{\prm B[b_i,a_i]}$.
We infer that
\begin{equation}\label{eq:BiTdiff}
\nice \wei{\prm B} > \wei{\prm B} - \wei{\prm B_i} = \wei{\prm B[b_i,a_i]} - \wei{\prm B_i \cap T} \geq
\wei{T} - \wei{\prm B_i \cap T} = \wei{T \setminus \prm B_i}.
\end{equation}

Note that $\prm B_i = \prm B[a_i,b_i] \cup T[a_i,b_i]$, where by $T[x,y]$ we define the unique path in $T$ between $x$ and $y$.
Let $v_a$ and $v_b$ be vertices on $T[a_i,b_i]$ such that
$$\wei{T[a_i,v_a]},\wei{T[b_i,v_b]} \leq \min \left(\wei{T[a_i,b_i]}/2, \left(\frac{1}{2}-6\nice\right)\wei{\prm B}\right)$$
and, moreover, both $T[a_i,v_a]$ and $T[b_i,v_b]$ are as long as possible.
Note that possibly $v_a=v_b$, but vertices $a_i,v_a,v_b,b_i$ appear on $T[a_i,b_i]$ in this order.
In particular, $v_a \neq b_i$ and $v_b \neq a_i$.

Let $Z$ be the union of $\{a_i,b_i\}$ with the set of vertices of $T[a_i,b_i]$ of degree at least $3$
in $T$.
Let $w_a$ be the vertex of $Z$ and 
$T[a_i,v_a]$ that is closest to $v_a$
and let $e_a$ be the edge that precedes $T[w_a,a_i]$ on $T[b_i,a_i]$.
Let $T_a$ be the subtree of $T$
rooted at $w_a$ with the parent edge $e_a$.
Note that the rightmost element of $V(T_a) \cap V(\prm B)$
is $a_i$; let $c$ be the leftmost element of $V(T_a) \cap V(\prm B)$.
By \eqref{eq:BiTdiff}, $\wei{T[w_a,c]} \leq \nice \wei{\prm B}$.
Therefore $\wei{T[w_a,c]} + \wei{T[w_a,a_i]} \leq (\frac{1}{2}-5\nice)\wei{\prm B}$.

Assume that $c \neq a_i$ and $\wei{\prm B[a_i,c]} \leq \wei{\prm B} / 2$.
As $\wei{T[a_i,c]} \leq (\frac{1}{2}-5\nice)\wei{\prm B}$, we infer that $(T[a_i,c],\prm B[a_i,c])$ is a $\delta$-carve and, by Lemma~\ref{lem:pluseps},
$\wei{\prm B[a_i,c]} \leq (\frac{1}{2}-4\nice)\wei{\prm B}$.
Let $C := \prm B[a_i, c] \cup T[a_i ,c ]$, which is a closed walk.
Consider the subgraph $T'$ created from $T$ by first deleting any edge enclosed by $C$, and then adding the closed walk $C$ instead.
Note that $\prm B_i$ is enclosed by $T'$ and $\wei{\prm B_i} > (1-\nice)\wei{\prm B}$, thus 
$$\wei{T'} \leq \wei{T}-(1-\nice)\wei{\prm B}+(1-9\nice)\wei{\prm B}\leq \wei{T}- 8\nice\wei{\prm B}.$$
However, as $T'$ includes $C$, $T'$ also connects $S$, a contradiction to the choice of $T$.

Therefore $c = a_i = w_a$ or $\wei{\prm B[c,a_i]} < \wei{\prm B}/2$.
Consider the second case. Again, we observe that $(T[a_i,c],\prm B[c,a_i])$ is a $\delta$-carve and, by Lemma~\ref{lem:pluseps},
$\wei{\prm B[c,a_i]} < (\frac{1}{2} - 4\nice)\wei{\prm B}$.
We now use the pegs $p_\rightarrow(a_i),p_\leftarrow(c) \in \porset$.
By the choice of $\porset$, $\wei{\prm B[a_i,p_\rightarrow(a_i)]} + \wei{\prm B[p_\leftarrow(c),c]} \leq \nice \wei{\prm B}$.
By Lemma \ref{lem:broom-in-mountain},
$\mou{(\prm B[p_\leftarrow(c),c] \cup T[w_a,c])}{(T[w_a,a_i] \cup \prm B[a_i,p_\rightarrow(a_i)])}$
is a $\delta$-mountain
and, by Theorem \ref{thm:mountain-range} and the construction of $C(\coreface)$,
no edge of the subtree of $T$
rooted at $w_a$ with parent edge $e_a$ is strictly enclosed by $C(\coreface)$, and, hence, by $C$ as well.
Clearly, this last claim is also true in the case $c = a_i = w_a$.

Symmetrically, the same argumentation can be made for $w_b$ being the first
vertex of $Z$ on $T[v_b,b_i]$, with its preceding edge $e_b$.

Now, if $T[w_a,w_b]$ does not contain any internal vertex from $Z$, then every vertex of degree at least $3$ in $T$ is contained either in $T_a$ or in $T_b$, and hence
the lemma is proven for $T_S = T$. Therefore, assume otherwise. In particular,
by the choice of $w_a$ and $w_b$, $v_a \neq v_b$, $v_av_b \notin T$ and
$\wei{T[a_i,b_i]} > (1-12\nice)\wei{\prm B}$. As $\prm B[b_i,a_i]$ connects
$S$, $\wei{\prm B[b_i,a_i]} \geq \wei{T[a_i,b_i]} > (1-12\nice)\wei{\prm B}$
and $\wei{\prm B[a_i,b_i]} < 12 \nice \wei{\prm B}$. 

Consider two consecutive vertices $w_1,w_2$ from $Z$ on $T[a_i,b_i]$.
Note that $(T \setminus T[w_1,w_2]) \cup \prm B[a_i,b_i]$ connects $S$.
Therefore, by the minimality of $T$, $\wei{T[w_1,w_2]} < 12 \nice \wei{\prm B}$.
Recall that $\wei{T \setminus \prm B_i} \leq \nice\wei{\prm B}$,
and, in particular, any vertex of $Z$ is connected
with $\prm B$ with a path in $T$ of length at most $\nice \wei{\prm B}$.
We infer that any edge of $T[a_i,b_i]$ lies on some path of length at most $14 \nice \wei{\prm B}$
with endpoints in $V(\prm B)$ and thus, belongs to $\closeB$ since
$\nice \leq 1/36$. % THIS IS THE PLACE WHERE WE ARE TIGHT WITH EPS=1/36

Let us now take any brick $B_j\neq B_i$. Observe that $\wei{\prm B_j\cap T}\leq 13\nice\wei{\prm B}$, since $\prm B_j\cap \prm B_i$ is either empty or an interval of length at most $12\nice \wei{\prm B}$, and $\wei{T\setminus \prm B_i}\leq \nice\wei{\prm B}$.
Recall that $\prm B[a_j,b_j]=\prm B_j \setminus T$. Assume first that $\wei{\prm B[a_j,b_j]}>\frac{1}{2}\wei{\prm B}$. Observe that then $\wei{\prm B[b_j,a_j]}\leq \frac{1}{2}\wei{\prm B}$ and, since $\prm B[b_j,a_j]$ connects $S$, we would obtain that $\wei{T}\leq \frac{1}{2}\wei{\prm B}$ by the optimality of $T$. On the other hand, $\wei{T}\geq \wei{B_i\setminus \prm B[a_i,b_i]}\geq (1-13\nice)\wei{\prm B}$. Since $\nice\leq \frac{1}{36}$, we obtain a contradiction.

Therefore, $\wei{\prm B[a_j,b_j]}\leq \frac{1}{2}\wei{\prm B}$. Since $\wei{T[a_j,b_j]}=\wei{\prm B_j\cap T}\leq 13\nice\wei{\prm B}$, $\delta=4\nice$ and $\nice\leq \frac{1}{36}$,
we obtain that $(T[a_j,b_j],\prm B[a_j,b_j])$ is a $\delta$-carve. As a result, we infer that $\coreface$ is not inside $B_j$. Since $B_j$ was chosen arbitrarily, $\coreface$ belongs to $B_i$.

Assume that some edge of $T$ is strictly enclosed by $C$.
As $\coreface$ belongs to both $B_i$ and $C$, this implies that
a subpath $T[x,y]$ of $T[a_i,b_i]$ ($x,y \in V(C)$) is strictly
enclosed by $C$. Without loss of generality assume that $T[x,y] \cup C[x,y]$
encloses $\coreface$, that is, $B_i$ lies on the same side of $T[x,y]$ as $C[x,y]$.
Consequently, any edge of $T$ incident to an internal
vertex of $T[x,y]$ is enclosed by $T[x,y] \cup C[y,x]$.
As each edge of $T[a_i,b_i]$ belongs to $\closeB$,
by the construction of $C$ we obtain $\wei{T[x,y]} \geq \wei{C[y,x]}$.
Construct $T'$ from $T$ by removing any edge enclosed by $C[y,x] \cup T[x,y]$
and adding $C[y,x]$ instead. Clearly, $\wei{T'} \leq \wei{T}$, $T'$ connects $S$
and $T'$ contains strictly less edges strictly enclosed by $C$.
By repeating this argument for all subpaths $T[x,y]$, we obtain a 
subgraph $T_S$ connecting $S$ and without any edge strictly enclosed by $C$.
This finishes the proof of the lemma.
\end{proof}
This concludes the proof of Theorem~\ref{thm:ananas}.

%!TEX root = pst-kernel.tex

\section{A polynomial kernel: concluding the proof of Theorem \ref{thm:main}}\label{sec:finish}

In this section, we conclude the proof of Theorem \ref{thm:main}.
That is, we assume that the brick~$B$ is unweighted.

Fix $\nice = 1/36$ and choose $\alpha$
such that 
$$(1-\nice)^{\alpha - 1} < \textstyle \frac{1}{3}$$
and
$$(1-3\nice)^{\alpha-1} < 1/202177.$$
(In particular, $\alpha > 141$.)
We show an algorithm that 
runs in $\Oh(|\prm B|^\alpha |B|)$ time
and returns a subgraph $H$
of size bounded by $\beta |\prm B|^\alpha$ for sufficiently large $\beta$
such that
$$202177 (1-3\nice)^{\alpha-1} + 108838883520 / \beta \leq 1.$$
For example, $\alpha = 142$ and $\beta = 2\,159\,872\,407\,596$ suffices.

First, consider the base case $|\prm B| \leq 2/\nice = 72$.
For each subset $S \subseteq V(\prm B)$,
we compute in $\Oh(|B|)$ time an optimal Steiner tree
using the algorithm of Erickson~\etal\cite{erickson}
for the set $S$ and add it to graph $H$. Note that the size
of the computed tree is at most $71$, as $\prm B$ without an arbitrary
edge connects $V(\prm B)$. Therefore, in $\Oh(|B|)$ time
we obtain a graph $H$ of size at most $71 \cdot 2^{72}$,
which is at most $\beta |\prm B|^\alpha$ for any $\beta\geq 1$, as $\alpha > 141$
and $|\prm B| \geq 3$.

Now, consider the recursive case.
Using the algorithm of Theorem~\ref{thm:nice-testing}, we test in $c_1 |\prm B|^8 \cdot|B|$ time whether
$B$ admits a short $\nice$-nice tree, for some constant $c_1$.
If the algorithm returns a short $\nice$-nice brick covering $\Bb = \{B_1,B_2,\ldots,B_p\}$, then we recurse on each brick $B_i$ separately, obtaining a subgraph~$H_i$.
By Lemma \ref{lem:recursion} and the choice of $\alpha$,
we may return the subgraph $H := \bigcup_{i=1}^p H_i$.
As for the time complexity, 
assume that the $i$-th recursive call took at most $c|\prm B_i|^\alpha |B_i|$ time.
Then, as the brick covering $\Bb$ is short and $\nice$-nice, we obtain that the total 
time spent is bounded by 
$$\left(c_1 |\prm B|^8  + c \sum_{i=1}^p |\prm B_i|^\alpha\right) |B| \leq |\prm B|^\alpha |B| \left(c_1 + 3c(1-\nice)^{\alpha-1}\right),$$
which is at most $c |\prm B|^\alpha |B|$ for sufficiently large $c$, by the choice of
$\alpha$.

Assume then that the algorithm of Theorem~\ref{thm:nice-testing} decided that no short $\nice$-nice tree exists in $B$.
First, we find some core face $\coreface$, using Theorem \ref{thm:core},
that cannot be $2\nice$-carved.
Then we employ Theorem \ref{thm:ananas}
to find a cycle $C$ of length at most
$\frac{16}{\nice^2} |\prm B| = 20736 |\prm B|$ that encloses $\coreface$.
Mark a set $X \subseteq V(C)$ such that the distance between any two consecutive
vertices of $X$ on $C$ is at most $2\nice |\prm B| = |\prm B|/18$. 
As $|\prm B| > 72$, we may greedily mark such set $X$ of size at
most $\frac{5}{4} \frac{|C|}{2\nice |\prm B|} \leq 466560$.
For each $x \in X$, we compute a shortest path $P_x$ from $x$ to $V(\prm B)$
that does not contain any edge strictly enclosed by $C$.
Note that this computation can be done by a simple breadth-first search from
$V(\prm B)$ in the graph obtained from $B$ by removing all edges strictly enclosed by $C$.
Moreover, in this manner, for any $x,y \in X$, the intersection of $P_x$
and $P_y$ is a common (possibly empty) suffix. By condition (ii) of Theorem~\ref{thm:ananas}, each path $P_x$ is of length at most $(\frac{1}{4}-2\nice)|\prm B|=\frac{7}{36}|\prm B|$.
For $x\in X$, let $\pi(x)$ be the second endpoint of $P_x$.

Let $x,y \in X$ be two vertices that are consecutive (in counter-clockwise direction) on $C$ and consider the walk $P := P_x \cup C[x,y] \cup P_y$.
Note that $|P| \leq \frac{4}{9}|\prm B|$, as $|P_{x}|,|P_{y}| \leq \frac{7}{36}|\prm B|$ and $C[x,y] \leq \frac{1}{18} |\prm B|$. We claim that:
\begin{equation} \label{eq:finish:P}
|\prm B[\pi(x),\pi(y)] \cup P| \leq (1-3\nice)|\prm B|.
\end{equation}
If $\pi(x) = \pi(y)$, then $|\prm B[\pi(x),\pi(y)] \cup P| \leq |P| \leq \frac{4}{9}|\prm B|$, and~\eqref{eq:finish:P} follows from the choice of $\nice$. Therefore, suppose that $\pi(x)\neq \pi(y)$. Then $P_x$ and $P_y$ do not intersect.
Let $x'$ be the vertex of $V(P_x) \cap V(C[x,y])$ that lies closest to $\pi(x)$ on $P_y$, and define $y'$ similarly with respect to $P_y$.
Observe that $x'$ lies closer to $x$ on $C[x,y]$ than $y'$, as otherwise $P_x[x,x']$ and $P_y[y,y']$ would
intersect (recall that neither $P_x$ nor $P_y$ contains an edge strictly enclosed by $C$).
Hence, $C[x',y']$ is a subpath of $C[x,y]$.
Define $P' = P_x[\pi(x),x'] \cup C[x',y'] \cup P_y[y',\pi(y)]$.
Observe that $P'$ is simple path of length at most $|P| \leq \frac{4}{9}|\prm B|$.
Then, either $(P',\prm B[\pi(x),\pi(y)])$ or $(P',\prm B[\pi(y),\pi(x)])$ is a $(2\nice)$-carve.
%Moreover, since $P'$ contains the subpath $C[x',y']$, and this subpath neighbours the interior of $C$ to the left, if we traverse $C$ counter-clockwise, 
Note that 
$P' \cup \prm B[\pi(y),\pi(x)]$ encloses $C$, and thus in particular $\coreface$.
Hence, it must be $(P,\prm B[\pi(x),\pi(y)])$ that is a $(2\nice)$-carve.
By Lemma~\ref{lem:pluseps} we infer that $|\prm B[\pi(x),\pi(y)]|\leq \frac{17}{36}|\prm B|$, and 
thus $|\prm B[\pi(x),\pi(y)] \cup P| \leq \frac{33}{36}|\prm B|$. Then~\eqref{eq:finish:P} follows from the choice of $\nice$.

Consider now the closed walk $W_x = \prm B[\pi(x),\pi(y)] \cup P$.
Let $H_x$ be the graph consisting of all edges of $W_x$ that neighbour the outer face
of $W_x$ treated as a planar graph; note that $W_x$ and $H_x$ are computable in linear time
for fixed $x$. By definition, each doubly-connected component of $H_x$ is a cycle or a bridge.
For each doubly-connected component that is a cycle, we create a brick consisting of all edges
of $B$ that are enclosed by this cycle. Let $\Bb_x$ be the family of obtained bricks.
Observe that $\Bb_x$ is computable in linear time and a face of $B$ is enclosed by some brick
of $\Bb_x$ if and only if it is enclosed by $W_x$.
Moreover, by~\eqref{eq:finish:P},
$$\sum_{B' \in \Bb_x} |\prm B'| \leq |W_x| \leq (1-3\nice)|\prm B|.$$
Therefore,
\begin{equation}\label{eq:Bxbound}
\sum_{x \in X}\sum_{B' \in \Bb_x}|\prm B'| \leq |C| + |\prm B| + 2|X| \frac{7}{36}|\prm B|
\leq 202177|\prm B|.
\end{equation}
We recurse on each brick $B' \in \Bb_x$, obtaining a graph $H(B')$.
Furthermore, for each $x,y \in V(C)$, we mark one shortest path $Q_{x,y}$ between $x$ and $y$
in $B$, if its length is at most $|\prm B|$. We define
$$H := \left(\bigcup_{x \in X} \bigcup_{B' \in \Bb_x}H(B')\right) \cup \left(\bigcup_{x,y \in X} Q_{x,y}\right).$$
By Theorem \ref{thm:ananas}, for any choice of terminals on $V(\prm B)$, there exists
an optimal Steiner tree contained in $H$. Note here that by Theorem \ref{thm:ananas} we may assume that every connection strictly enclosed by $C$ is realized by some marked shortest path $Q_{x,y}$.

We now bound the size of $H$.
For each $x \in X$ and $B' \in \Bb_x$ we have $|H(B')| \leq \beta |\prm B'|^\alpha$.
Moreover, each $Q_{x,y}$ is of length at most $|\prm B|$. Hence,
\begin{align*}
|H| &\leq \beta\sum_{x \in X}\sum_{B' \in \Bb_x} |\prm B'|^\alpha + \binom{|X|}{2} |\prm B| \\
    &\leq \beta 202177|\prm B|^\alpha (1-3\nice)^{\alpha-1} + 108838883520 |\prm B| \\
    & \leq \beta |\prm B|^\alpha.
\end{align*}
(The last inequality follows from the choice of $\alpha$ and $\beta$.)

Regarding time bound, note that all computations, except for the recursive
calls, can be done in $c_2 |\prm B|^3 |B|$ time, for some constant $c_2$.
Therefore the total time spent is
$$\left(c_2 |\prm B|^3  + c \sum_{x \in X} |\prm B_x|^\alpha\right) |B| \leq |\prm B|^\alpha |B| \left(c_2 + 202177 c (1-3\nice)^{\alpha-1}\right)$$
which is at most $c |\prm B|^\alpha |B|$ for sufficiently large $c$, by the choice of
$\alpha$.

%!TEX root = pst-kernel.tex

\section{Dynamic programming to find nice subgraphs}\label{sec:dp}

\newcommand{\rzeska}[1]{\widehat{#1}}
\newcommand{\exB}{\rzeska{B}}
\newcommand{\Tree}{\mathbb{T}}
\newcommand{\Emb}{\pi}
\newcommand{\prmup}{\prm^\uparrow}
\newcommand{\krok}{\lambda}
\newcommand{\rnd}{\mathtt{rnd}}
\newcommand{\impV}[1]{\mathbf{I}(#1)}

Our goal in this section is to prove the two algorithmic statements mentioned Section~\ref{sec:bricks}.
\begin{theorem}[Theorem \ref{thm:nice-testing} recalled]\label{thm:nice-testing:copy}
Let $\nice > 0$ be a fixed constant.
Given an unweighted brick $B$, in $\Oh(|\prm B|^8 |B|)$ time
one can either correctly conclude that
no short $\nice$-nice tree exists in $B$
or find a short $\nice$-nice brick covering of $B$.
\end{theorem}
\begin{theorem}\label{thm:nice-testing-wei}
Let $0 < \nice \leq \frac{1}{4}$ be a fixed constant.
Given an edge-weighted brick $B$, in $\Oh(\nice^{-14} |B| \log |B|)$ time
one can either correctly conclude that no $3$-short $\nice$-nice tree exists in $B$
or find a $(3+2\nice)$-short $(\nice/2)$-nice brick covering $\Bb$ of $B$ with the following
additional properties:
\begin{enumerate}
\item each finite face of $B$ is enclosed by at most $7$ bricks $B' \in \Bb$;\label{p:nice-testing-wei:seven}
\item $\bigcup_{B' \in \Bb} \prm B'$ is connected.
\end{enumerate}
\end{theorem}

The idea of the proofs of Theorems~\ref{thm:nice-testing} and~\ref{thm:nice-testing-wei} is to perform a dynamic-programming
algorithm similar to the algorithm of Erickson et al.~\cite{erickson}
for finding an optimal Steiner tree for a given set of terminals on the outer face.
However, as we impose some restrictions on the faces that the tree cuts out of the brick $B$,
the outcome of the algorithm may no longer be a tree. We start by formalizing
what we can actually find.

Construct the {\em{extended brick}} $\exB$ as follows: take $B$ and
for every $a \in V(\prm B)$ add a degree-$1$ vertex $\rzeska{a}$ attached
to $a$ with an edge of zero weight, drawn outside the cycle $\prm B$. (We remark here that the weight of the edge $a\rzeska{a}$ does not have any real significance in the sequel.)
We denote $\rzeska{\prm} B = \{\rzeska{a}: a \in V(\prm B)\}$.

We define an {\em{ordered tree}} $\Tree$ as a rooted tree where every vertex has imposed some linear order on its children. This naturally induces a linear order on the set of leaves of $
\Tree$. The following definition captures the objects found by our dynamic-programming algorithms.
\begin{definition}[embedded tree]
An {\em{embedded tree}} is a pair $(\Tree, \Emb)$
where $\Tree$ is an edge-weighted ordered tree with at least one edge, rooted at vertex $r(\Tree)$,
and $\Emb$ is a homomorphism from $\Tree$ into $\exB$
such that $\Emb(v) \in V(B)$ for any non-leaf vertex of $\Tree$
and $\Emb$ assigns the leaves of $\Tree$ to vertices
of $\rzeska{\prm}B$.
We require that the order of the leaves of $\Tree$ coincides with the counter-clockwise
order of their images on $\rzeska{\prm}B$ under the homomorphism $\Emb$.

We say that an embedded tree is {\em{leaf-injective}} if $\Emb$ is injective on the set of leaves of $\Tree$.
\end{definition}
Here, by a homomorphism $\Emb$ from a graph $G$ to a graph $H$ we mean
a function $\Emb: E(G) \cup V(G) \to E(H) \cup V(H)$ that matches edges
to edges and vertices to vertices and, if $\Emb(uv) = u'v'$, then
$\{\Emb(u), \Emb(v)\} = \{u',v'\}$ and $\wei{u'v'} = \wei{uv}$.
As all edges $a\rzeska{a}$ are of weight zero in $\exB$, and all edges of $B$ have positive weight, we may restrict ourselves 
to embedded trees where an edge has weight zero if and only if it is adjacent to a leaf.

We measure the length of an embedded tree as in all weighted graphs.
Note that the edges incident to leaves of an embedded tree do not contribute to the length of the tree. In the unweighted case, we will mostly be working with leaf-injective embedded trees, while in the weighted case it will be more convenient to drop this assumption.

Recall that for two vertices $a,b \in \prm B$, by $\prm B[a,b]$ we denote the subpath of $\prm B$
between $a$ and $b$, obtained by traversing $\prm B$ in counter-clockwise direction. If $a=b$, then $\prm B[a,b] = \emptyset$.
We define $\prmup B[a,b]$ to be equal $\prm B[a,b]$ unless $a=b$; in this case $\prmup B[a,b] = \prm B$.

An embedded tree $(\Tree, \Emb)$ is {\em{$\nice$-nice}} if for any two
consecutive leaves $\rzeska{l_a},\rzeska{l_b}$ in $\Tree$ the following holds.
Let $\Emb(\rzeska{l_a}) = \rzeska{a}$ and $\Emb(\rzeska{l_b}) = \rzeska{b}$
and let $l_a,l_b$ be the parents of $\rzeska{l_a}, \rzeska{l_b}$ in $\Tree$,
respectively; note that $\Emb(l_a) = a$ and $\Emb(l_b)=b$, and possibly $a = b$.
Let $u$ be the
lowest common ancestor of $\rzeska{l_a}$ and $\rzeska{l_b}$ in $\Tree$. Then for $(\Tree, \Emb)$ to be {\em{$\nice$-nice}} we require that
\begin{equation}
\wei{\prm B[a,b]} + \wei{\Tree[u,l_a]} + \wei{\Tree[u,l_b]} \leq (1-\nice)\wei{\prm B}.\label{eq:dp:eps}
\end{equation}
An embedded tree is {\em{fully $\nice$-nice}} if additionally~\eqref{eq:dp:eps} holds
for $\rzeska{l_a}$ being the last leaf of $\Tree$,
$\rzeska{l_b}$ being the first leaf of $\Tree$, $u = r(\Tree)$ and
$\prm B[a,b]$ replaced by $\prmup B[a,b]$.

The intuition behind this notion is that the image of $\Tree[w,l_a] \cup \Tree[w,l_b]$
under $\Emb$, together with $\prm B[a,b]$ (or $\prmup B[a,b]$ in the case of the last and the first leaf of $\Tree$), is likely to
yield a perimeter of an output brick $B_i$ in our algorithm.

We now formalize how to find a set of bricks promised 
by Theorem \ref{thm:nice-testing} and Theorem~\ref{thm:nice-testing-wei}, given a fully $\nice$-nice embedded tree.
\begin{lemma}\label{lem:dp:tree-to-bricks}
Given a fully $\nice$-nice embedded tree $(\Tree, \Emb)$ with $r$ leaves,
one can in $\Oh(r(|\Tree| + |B|))$ time compute
a $\nice$-nice brick covering $\Bb$ of $B$ of total perimeter at most
$\wei{\prm B} + 2\wei{\Tree}$ with the following additional properties:
\begin{enumerate}
\item each finite face of $B$ is enclosed by at most $r$ bricks of $\Bb$;
\item $\bigcup_{B' \in \Bb} \prm B'$ is connected.
\end{enumerate}
\end{lemma}
\begin{proof}
Let $\mathcal{F}$ be a family of pairs of two consecutive leaves of $\Tree$
and the pair $\mathbf{p}^\circ$ consisting of the last and the first leaf of $\Tree$.
For any $\mathbf{p} = (\rzeska{l_a},\rzeska{l_b}) \in \mathcal{F}$,
define $l_a,l_b,a,b,w$ as in the definition of a fully $\nice$-nice tree.
Define $C(\mathbf{p}) := \Emb(\Tree[l_a,w] \cup \Tree[w,l_b]) \cup \prm B[a,b]$
if $\mathbf{p} \neq \mathbf{p}^\circ$
and
$C(\mathbf{p}) := \Emb(\Tree[l_a,w] \cup \Tree[w,l_b]) \cup \prmup B[a,b]$ if $\mathbf{p} = \mathbf{p}^\circ$.
Observe that $C(\mathbf{p})$ is a closed walk in $B$.
Note that the fact that $\Tree$ is fully $\nice$-nice tree implies that the length of $C(\mathbf{p})$ is bounded
by $(1-\nice)\wei{\prm B}$.
Moreover, as edges of $\Tree$ not incident to a leaf contribute to exactly two
cycles $C(\mathbf{p})$, and each edge of $\prm B$ contributes to exactly one
such cycle, we have
\begin{equation}\label{eq:dp:sumC}
\sum_{\mathbf{p} \in \mathcal{F}} \wei{C(\mathbf{p})} = \wei{\prm B} + 2\wei{\Tree}.
\end{equation}
Let $H_0(\mathbf{p})$ be the subgraph of $B$ consisting of all edges that lie on $C(\mathbf{p})$. Clearly,
$H_0(\mathbf{p})$ is connected. Let $H(\mathbf{p})$ be the subgraph of $H_0(\mathbf{p})$ consisting of all edges of $H_0(\mathbf{p})$
that are adjacent to the outer face of $H_0(\mathbf{p})$. Note that $H(\mathbf{p})$ is connected,
$\prm B[a,b] \subseteq H(\mathbf{p})$ ($\prmup B[a,b] \subseteq H(\mathbf{p})$ if $\mathbf{p} \neq \mathbf{p}^\circ$) and the outer faces of $H(\mathbf{p})$ and $H_0(\mathbf{p})$ are equal.
Moreover, by the definition of $H(\mathbf{p})$, any doubly-connected component of $H(\mathbf{p})$
is either a simple cycle or a bridge. 

We construct a preliminary brick covering $\Bb_0$ as follows: for each $\mathbf{p} \in \mc{F}$ and
for each doubly-connected component $D$
of $H(\mathbf{p})$ that is a cycle, we insert into $\Bb_0$ a brick $B_i$ consisting of all edges
of $B$ that are enclosed by $D$; clearly $\prm B_i = D$ and $B_i$ is a subbrick of $B$.
Note that $\Bb_{0}$ can be computed within the desired running time. Indeed, $H(\mathbf{p})$ can be computed in $\Oh(|\Tree| + |B|)$ time, and the corresponding bricks can be computed in $\Oh(|B|)$ time. It remains to observe that $|\mc{F}| = r$, where $r$ is the number of leaves of $\Tree$.

We can now make several observations about $\Bb_{0}$.
First, as $\wei{C(\mathbf{p})} \leq (1-\nice)\wei{\prm B}$,
each brick in $\Bb_{0}$ has perimeter at most $(1-\nice)\wei{\prm B}$.
Second, for a fixed $\mathbf{p}$, the total perimeter of the bricks inserted into $\Bb_0$ is at most $\wei{H(\mathbf{p})} \leq \wei{C(\mathbf{p})}$.
Therefore, by \eqref{eq:dp:sumC}, the sum of the perimeters
of all bricks in $\Bb_0$ is bounded by $\wei{\prm B} + 2\wei{\Tree}$,
as desired.
Third, for a fixed cycle $C(\mathbf{p})$, the constructed bricks do not share an
enclosed finite face of $B$. Hence, each finite face of $B$ is enclosed by at most $r$
bricks of $\Bb_0$.

We now show that $\Bb_0$ is a brick covering of $B$, that is, 
   we prove that each face of $B$
is contained in some brick of $\Bb_0$. Let $f$ be any face of $B$ and let $c$ be an arbitrary
point of the plane in the interior of $f$. Let $\Gamma\cong \mathbb{Z}$ be the fundamental group of $\plane \setminus \{c\}$, and let $\iota$ be the mapping that assigns to each closed curve in $\plane \setminus \{c\}$ the corresponding element of $\Gamma$. 
For each $\mathbf{p} \in \mathcal{F}$, orient the walk $C(\mathbf{p})$
in the direction such that the part $\prm B[a,b]$ or $\prmup B[a,b]$ is traversed from $a$ to $b$
%(in the case $\mathbf{p} = \mathbf{p}^\circ$ and $a=b$, we orient $C(\mathbf{p})$ such that $\prmup B[a,b] = \prm B$ is oriented counter-clockwise).
(note that if $\mathbf{p} = \mathbf{p}^\circ$, then $a \not=b$ and $\prmup B[a,b] = \prm B[a,b]$, as $(\Tree,\Emb)$ is fully $\nice$-nice).
If $c$ belongs to the outer face of the graph $H(\mathbf{p})$,
then $C(\mathbf{p})$ is continuously retractable to a single point in $\plane \setminus \{c\}$, and thus $\iota(C(\mathbf{p}))$ is the neutral element of $\Gamma$.
On the other hand, $\iota(\prm B)$ is {\em{not}} the neutral element of this fundamental group, since it winds around $c$ exactly one time.
Observe that in this fundamental group we have equation
$$\sum_{\mathbf{p} \in \mathcal{F}} \iota(C(\mathbf{p})) = \iota(\prm B),$$
since for each $e\in E(\Tree)$ we have that $\Emb(e)$ is traversed by two different walks $C(\mathbf{p}_1)$, $C(\mathbf{p}_2)$, in different directions.
Therefore, for at least one $\mathbf{p}_0\in \mathcal{F}$ it must hold that $\iota(C(\mathbf{p}_0))$ is not the neutral element of $\Gamma$. Consequently, $c$ belongs to some bounded face of one of the constructed graphs $H(\mathbf{p})$,
and one of the bricks of $\Bb_0$ contains $f$.

Observe that $\Bb_0$ has all the required properties, except possibly the property
that $\bigcup_{B' \in  \Bb_0} \prm B'$ is connected. To ensure this property as well, we 
select a subfamily of $\Bb_0$ as follows. For each connected component $D$
of $\bigcup_{B' \in \Bb_0} \prm B'$, let $\Bb_D$ be the family of all bricks $B' \in \Bb_0$
with $\prm B' \subseteq D$. Let $D_0$ be the component of $\bigcup_{B' \in \Bb_0} \prm B'$
that contains $\prm B$.

We claim that if $D \neq D_0$ is a component of $\bigcup_{B' \in \Bb_0} \prm B'$, then $\Bb_0 \setminus \Bb_D$ is a brick covering of $B$ as well. Let $f$ be a face of $B$
that is incident to one of the edges of $D$, but is contained in the outer face of $D$.
As $D$ does not contain any edge of $\prm B$, $f$ is finite. Let $B' \in \Bb_0$ be a brick
such that $\prm B'$ encloses $f$. Clearly, $B' \notin \Bb_D$ and hence $\prm B'$ does not share
any vertex with $D$. As $f$ is incident with an edge of $D$, we infer that $\prm B'$ strictly encloses
all edges of $D$; in particular, $\prm B'$ encloses all faces that are enclosed by the bricks
of $\Bb_D$. Consequently, $\Bb_0 \setminus \Bb_D$ is a brick covering of $B$.

We now remove all bricks $\Bb_D$ from $\Bb_{0}$ for any component $D \neq D_0$ of $\bigcup_{B' \in \Bb_0} \prm B'$. By the above claim, we infer that the remainder, $\Bb_{D_0}$, is a brick covering of $B$.
As $\Bb_{D_0} \subseteq \Bb_0$, $\Bb_{D_0}$ inherited all other required properties:
in particular, it is $\nice$-nice and of total perimeter
at most $\wei{\prm B} + 2\wei{\Tree}$. Hence, the algorithm may output $\Bb_{D_0}$.
Observe that it can be computed from $\Bb_0$ in time linear in $|B|$ and the total size of $\Bb_0$.
\end{proof}

In the other direction, it is easy to see that a short $\nice$-nice tree in $B$
yields a fully $\nice$-nice embedded tree of small length.
\begin{lemma}\label{lem:dp:nice-to-tree}
If $B$ admits a $\nice$-nice tree $T$, then $B$ admits a
fully $\nice$-nice, leaf-injective embedded tree $(\Tree, \Emb)$ of length $\wei{T}$.
\end{lemma}
\begin{proof}
We construct $\Tree$ as follows: root $T$ at an arbitrary vertex $r \in V(T)$, for each $a \in V(T) \cap V(\prm B)$, add the edge $a\rzeska{a}$,
and for each internal vertex $p$ of $\Tree$, order its children in the counter-clockwise
order in which they appear on the plane (starting from the parent of $p$, or at arbitrary point for $r=p$).
As each leaf of $T$ lies on $V(\prm B)$, in this manner each leaf of $\Tree$ lies
in $\rzeska{\prm}B$. Therefore, if we take $\Emb$ to be the identity mapping,
$(\Tree, \Emb)$ is an embedded tree. By construction, $\wei{\Tree} = \wei{T}$ and $(\Tree,\Emb)$ is leaf-injective.
Moreover, for any two consecutive leaves $\rzeska{a}$ and $\rzeska{b}$ of $\Tree$,
if $w$ is the lowest common ancestor of $\rzeska{a}$ and $\rzeska{b}$,
then the value $\Tree[w,a] \cup \Tree[w,b] \cup \prm B[a,b]$ is the perimeter
of the face of $B[T \cup \prm B]$ that neighbours $\prm B[a,b]$. As $T$ is $\nice$-nice,
   we infer that $(\Tree,\Emb)$ is $\nice$-nice as well.
Finally, if $\rzeska{a}$ is the last leaf of $\Tree$ and $\rzeska{b}$ is the first leaf
of $\Tree$, then since $r$ has degree at least two in $\Tree$, $r$ is the lowest
common ancestor of $\rzeska{a}$ and $\rzeska{b}$ in $\Tree$,
and $\Tree[r,a] \cup \Tree[r,b] \cup \prmup B[a,b]$ is again the perimeter of the face 
of $B[T \cup \prm B]$ that neighbours $\prmup B[a,b]$. We infer that $(\Tree, \Emb)$
is fully $\nice$-nice and the lemma is proven.
\end{proof}

By Lemmata~\ref{lem:dp:tree-to-bricks} and~\ref{lem:dp:nice-to-tree}, it remains to find a fully $\nice$-nice embedded tree of small length.
Here the argumentation for the unweighted and the edge-weighted cases diverge. In both cases, we use a dynamic-programming algorithm.
However, in the unweighted case we are able to obtain the exact statement of Theorem~\ref{thm:nice-testing}; in the edge-weighted case, we need to perform some rounding
to fit into the $\Oh(|B| \log |B|)$ time frame, and therefore we may lose some `niceness' of the constructed tree.

\subsection{Finding a nice embedded tree in the unweighted setting}

For brevity we denote $n = |B|$ and $k = |\prm B|$.
\begin{lemma}\label{lem:dp:find-tree}
Assume $B$ is unweighted.
Given an integer $\ell$, in $\Oh(nk^4\ell^4)$ time one can 
find a fully $\nice$-nice leaf-injective embedded tree of length at most $\ell$
or correctly conclude that no such tree exists.
\end{lemma}
\begin{proof}
For each $v \in V(B)$, $a, b \in V(\prm B)$, $0 \leq k_a,k_b \leq \ell$,
we define $\mathcal{F}[v,a,b,k_a,k_b]$ to be the set of all leaf-injective embedded trees
$(\Tree,\Emb)$ that:
\begin{enumerate}
\item have length at most $\ell$;
\item are $\nice$-nice;
\item satisfy $\Emb(r(\Tree)) = v$;
\item map the first leaf of $\Tree$, $\rzeska{l_a}$, to $\rzeska{a}$ under $\Emb$,
 and the last leaf of $\Tree$, $\rzeska{l_b}$, to $\rzeska{b}$ under $\Emb$;
\item satisfy $|\Tree[r(\Tree), l_a]| \leq k_a$ and $|\Tree[r(\Tree), l_b]| \leq k_b$, 
  where $l_a$, $l_b$ are parents of $\rzeska{l_a}$, $\rzeska{l_b}$ in $\Tree$, respectively.
\end{enumerate}
Let $M[v,a,b,k_a,k_b] = \min \{ \wei{\Tree} : (\Tree, \Emb) \in \mathcal{F}[v,a,b,k_a,k_b]\}$.

Assume that $B$ admits a fully $\nice$-nice leaf-injective embedded tree $(\Tree,\Emb)$ of length at most $\ell$.
Let $\rzeska{l_a}$, $\rzeska{l_b}$ be the first and the last leaf of $\Tree$,
let $l_a,l_b$ be the parents of $\rzeska{l_a}, \rzeska{l_b}$ in $\Tree$, respectively,
and let $a = \Emb(l_a)$, $b = \Emb(l_b)$.
Note that $(\Tree, \Emb) \in \mathcal{F}[r(\Tree), a, b, |\Tree[r(\Tree), l_a]|, |\Tree[r(\Tree),l_b]|]$
and $|\Tree[r(\Tree), l_a]| + |\Tree[r(\Tree), l_b]| + |\prmup B[b,a]| \leq (1-\nice)k$, as
$(\Tree, \Emb)$ is fully $\nice$-nice.
In the other direction, if $(\Tree, \Emb) \in \mathcal{F}[v,a,b,k_a,k_b]$
and $k_a+k_b+|\prmup B[b,a]| \leq (1-\nice)k$, then $(\Tree, \Emb)$ is fully $\nice$-nice.
Therefore, it suffices to compute, for each choice of the parameters $v,a,b,k_a,k_b$,
the value $M[v,a,b,k_a,k_b]$ and one representative element $T[v,a,b,k_a,k_b] \in \mathcal{F}[v,a,b,k_a,k_b]$
of length $M[v,a,b,k_a,k_b]$, if $\mathcal{F}[v,a,b,k_a,k_b] \neq \emptyset$.

Clearly, for $v=a=b$, $M[v,a,b,k_a,k_b] = 0$ and $T[v,a,b,k_a,k_b]$ can be defined as
a two-vertex tree with root $r$, mapped to $v=a=b$, and a single leaf mapped to $\rzeska{a} = \rzeska{b}$.
These are the only embedded trees of zero length.

Consider a $\nice$-nice leaf-injective embedded tree $(\Tree, \Emb)$ with $0 < \wei{\Tree} \leq \ell$.
Let $\rzeska{l_a}, \rzeska{l_b}$ be the first and the last leaf of $\Tree$,
let $l_a, l_b$ be the parents of $\rzeska{l_a}, \rzeska{l_b}$ in $\Tree$
  and let $a = \Emb(l_a)$, $b = \Emb(l_b)$, $v = \Emb(r(\Tree))$. 
Let $k_a,k_b$ be such that $|\Tree[r(\Tree), l_a]| \leq k_a$ and $|\Tree[r(\Tree), l_b]| \leq k_b$.
Consider two cases: either $r(\Tree)$ is of degree one in $\Tree$ or larger.

In the first case, let $p$ be the only child of $r(\Tree)$; note that $p$ is not a leaf
as $\wei{\Tree} > 0$. Let $\Tree_1 = \Tree \setminus r(\Tree)$, rooted at $p$,
and let $\Emb_1$ be the mapping $\Emb$ restricted to $\Tree_1$. Clearly, $(\Tree_1,\Emb_1)$
is a $\nice$-nice embedded tree that belongs to $\mathcal{F}[\Emb(p), a, b, k_a-1, k_b-1]$.

In the other direction, consider the cell $\mathcal{F}[v,a,b,k_a,k_b]$.
We note that for any $w \in N_B(v)$ and $(\Tree_1,\Emb_1) \in \mathcal{F}[w, a, b, k_a-1, k_b-1]$,
if we extend $\Tree_1$ with a new root vertex $r$ mapped to $v$, with one child $r(\Tree_1)$,
then the extended tree belongs to $\mathcal{F}[v,a,b,k_a,k_b]$.

In the second case, split $\Tree$ into two trees $\Tree_1$ and $\Tree_2$, rooted at $r(\Tree)$:
$\Tree_1$ contains the subtree of $\Tree$ rooted in the first child of $r(\Tree)$, together
with the edge connecting it to $r(\Tree)$, and $\Tree_2$ contains the remaining edges of $\Tree$
(i.e., all but the first children of $r(\Tree)$, together with the edges connecting them to $r(\Tree)$).
Define $\Emb_1$ and $\Emb_2$ as restrictions of $\Emb$ to $\Tree_1$ and $\Tree_2$, respectively.
Let $\rzeska{l_c}$ be the last leaf of $\Tree_1$ and $\rzeska{l_d}$ be the first
leaf of $\Tree_2$. Define $l_c,l_d,c,d$ analogously to $l_a,l_b,a,b$. Observe that $l_c\neq l_d$ since $(\Tree,\Emb)$ is leaf-injective, but it may be that $l_a=l_c$ or $l_d=l_b$ in case $a=c$ or $b=d$.
Note that $(\Tree_1,\Emb_1) \in \mathcal{F}[v,a,c,k_a,|\Tree[r(\Tree), l_c]|]$
and $(\Tree_2,\Emb_2) \in \mathcal{F}[v,d,b,|\Tree[r(\Tree), l_d]|, k_b]$.
Moreover, $|\Tree[r(\Tree), l_c]| + |\Tree[r(\Tree), l_d]| + |\prm B[c,d]| \leq (1-\nice)k$,
as $\Tree$ is $\nice$-nice and $r(\Tree)$ is the lowest common ancestor of $l_c$ and $l_d$ in $\Tree$.

In the other direction, assume that for some $c,d \in \prm B[a,b]$ such that $c$ lies
strictly closer to $a$ than $d$ on $\prm B[a,b]$ (i.e., $\prm B[a,c]\subsetneq \prm B[a,d]\subseteq \prm B[a,b]$), and for some $k_c, k_d \leq \ell$ such that
$k_c + k_d + |\prm B[c,d]| \leq (1-\nice)k$, we have embedded trees
$(\Tree_1, \Emb_1) \in \mathcal{F}[v,a,c,k_a,k_c]$ and $(\Tree_2,\Emb_2) \in \mathcal{F}[v,d,b,k_d,k_b]$
such that $\wei{\Tree_1} + \wei{\Tree_2} \leq \ell$.
Define $\Tree$ as $\Tree_1 \cup \Tree_2$ with identified roots, rooted at $r(\Tree) = r(\Tree_1) = r(\Tree_2)$,
and order of the children of $r(\Tree)$ by first placing the children in $\Tree_1$ and then the children in $\Tree_2$, in the corresponding orders.
Moreover, define $\Emb = \Emb_1 \cup \Emb_2$.
Then in the embedded tree $(\Tree,\Emb)$ the first leaf is $\rzeska{l_a}$ with $\Emb(\rzeska{l_a}) = \rzeska{a}$
and the last leaf is $\rzeska{l_b}$ with $\Emb(\rzeska{l_b}) = \rzeska{b}$. The assumption that $c$ is strictly closer to $a$ than $d$ implies that $(\Tree,\Emb)$ is leaf-injective.
Furthermore, $\wei{\Tree} = \wei{\Tree_1} + \wei{\Tree_2} \leq \ell$.
Finally, the requirement $k_c + k_d + |\prm B[c,d]| \leq (1-\nice)k$ implies that
$(\Tree, \Emb)$ is $\nice$-nice. Hence, $(\Tree,\Emb) \in \mathcal{F}[v,a,b,k_a,k_b]$.

From the previous discussion, we infer that $M[v,a,b,k_a,k_b] = 0$ if $v=a=b$ and otherwise
$M[v,a,b,k_a,k_b]$ equals the minimum over the following candidates:
\begin{itemize}
\item if $k_a,k_b > 0$, for each $w \in N_B(v)$, we take $1 + M[w,a,b,k_a-1,k_b-1]$ as a candidate value;
\item for each $c,d \in \prm B[a,b]$ such that $c$ lies strictly closer to $a$ than $d$ on $\prm B[a,b]$,
  and for each integers $0 \leq k_c,k_d \leq \ell$ such that $k_c+k_d+|\prm B[c,d]| \leq (1-\nice)k$,
  we take $M[v,a,c,k_a,k_c] + M[v,d,b,k_d,k_b]$ as a candidate value, provided that this value does not exceed $\ell$.
\end{itemize}
We note that, in the aforementioned recursive formula,
to compute $M[v,a,b,k_a,k_b]$ we take into account at most $|N_B(v)| + k^2(1+\ell)^2$ other candidates,
in each computation taking into account values $M[v',a',b',k_a',k_b']$ with
$|\prm B[a',b']|$ strictly smaller than $|\prm B[a,b]|$.
We infer that the values $M[v,a,b,k_a,k_b]$ for all valid choice of the parameters $v,a,b,k_a,k_b$
can be computed in $\Oh(nk^4\ell^4)$ time. If we additionally store for each cell $M[v,a,b,k_a,k_b]$
which candidate attained the minimum value, we can read an optimal embedded tree $T[v,a,b,k_a,k_b]$
in linear time with respect to its size. This concludes the proof of the lemma.
\end{proof}

We may now conclude the proof of Theorem \ref{thm:nice-testing}.
Using Lemma \ref{lem:dp:find-tree} we look for a fully $\nice$-nice
embedded tree of length at most $k$. If one is found, we apply Lemma \ref{lem:dp:tree-to-bricks}
to obtain the desired family of bricks.
If the algorithm of Lemma \ref{lem:dp:find-tree} does not find any
embedded tree, Lemma \ref{lem:dp:nice-to-tree} allows us to conclude
that no short $\nice$-nice tree exists in $B$.

\subsection{Finding a nice embedded tree in the edge-weighted setting}

We start with the following observation that extends Lemma~\ref{lem:dp:nice-to-tree}.
\begin{lemma}\label{lem:dp:seven-leaves}
Let $B$ be an edge-weighted brick and let $0 < \nice \leq \frac{1}{4}$ be a constant.
If there exists a short $\nice$-nice tree $T$ in $B$, then there exists
an embedded fully $\nice$-nice tree $(\Tree,\Emb)$ in $B$ of length at most $\wei{T}$ and with at most $7$ leaves.
\end{lemma}
\begin{proof}
Let $T$ be as in the lemma statement, and construct $(\Tree,\Emb)$ as in the proof
of Lemma~\ref{lem:dp:nice-to-tree}.
That is, we construct $\Tree$ as follows: we root $T$ at an arbitrary vertex $r \in V(T)$,
for each $a \in V(T) \cap V(\prm B)$, add the edge $a\rzeska{a}$,
and for each internal vertex $p$ of $\Tree$, order its children in the counter-clockwise
order in which they appear on the plane (starting from the parent of $p$, or at arbitrary point for $r=p$).
The mapping $\Emb$ is the identity mapping.
Clearly, $\wei{\Tree} = \wei{T}$. Our goal is to trim $\Tree$ so that it is still fully $\nice$-nice,
but has at most $7$ leaves.

Assume $\Tree$ has at least $8$ leaves, as otherwise we are done.
Pick any four pairwise distinct leaves $\rzeska{l_1},\rzeska{l_2},\rzeska{l_3},\rzeska{l_4}$ of $\Tree$ with the following properties:
they lie in $\Tree$ in this order, no two of them are two consecutive leaves of $\Tree$, and $\rzeska{l_4}$ is not the last leaf of $\Tree$.
As $\Tree$ has at least $8$ leaves, this is always possible (e.g., we may take the first, third, fifth and seventh leaf of $\Tree$).
Let $l_i$ be the unique neighbour of $\rzeska{l_i}$ in $\Tree$.
Moreover, let $\rzeska{a_i} = \Emb(\rzeska{l_i})$ and $a_i = \Emb(l_i)$; note that $a_i$ is the unique neighbour of $\rzeska{a_i}$ 
in the extended brick $\exB$. We use a cyclic ordering for the index $i$, that is $l_5=l_1$, $a_5=a_1$ etc.
Observe that all $a_i$ are pairwise distinct, as we have started from a short $\nice$-nice tree $T$ (in other words, $(\Tree, \Emb)$ is leaf-injective).

For $i=1,2,3,4$, by $L_i$ we denote the set of leaves of $\Tree$ that lie between $\rzeska{l_i}$ and $\rzeska{l_{i+1}}$ (exclusive), in the circular order
of the leaves of $\Tree$. By the assumption on the leaves $\rzeska{l_i}$, all sets $L_i$ are nonempty.
For $i=1,2,3,4$, let $\Tree_i$ be a subtree of $\Tree$ defined as follows: for each $\rzeska{l} \in L_i$, we remove from $\Tree$ the path from
$\rzeska{l}$ to the closest vertex of $\Tree[\rzeska{l_i},\rzeska{l_{i+1}}]$ (recall that $\rzeska{l_5}= \rzeska{l_1}$).
Define $\Emb_i = \Emb|_{\Tree_i}$.
As we preserve the path $\Tree[\rzeska{l_i},\rzeska{l_{i+1}}]$ in $\Tree_i$, no new leaf has been introduced into $\Tree_i$
and $(\Tree_i,\Emb_i)$ is an embedded tree in $B$.

We claim that for at least one index $i$, the embedded tree $(\Tree_i,\Emb_i)$ is fully $\nice$-nice.
Assume the contrary. Since $(\Tree,\Emb)$ is fully $\nice$-nice, we infer that for each $i=1,2,3,4$:
$$\wei{\Tree[l_i,l_{i+1}]} + \wei{\prm B[a_i,a_{i+1}]} > (1-\nice) \wei{\prm B}.$$
Summing up, we infer that:
$$\wei{\prm B} + \sum_{i=1}^4 \wei{\Tree[l_i,l_{i+1}]} > 4(1-\nice) \wei{\prm B} \geq 3\wei{\prm B},$$
where the last inequality follows from the assumption $\nice \leq \frac{1}{4}$. However, note that
$$\sum_{i=1}^4 \wei{\Tree[l_i,l_{i+1}]} \leq 2\wei{\Tree} \leq 2\wei{T} \leq 2\wei{\prm B},$$
since $T$ is short. We have reached a contradiction.

Consequently, we may replace $(\Tree,\Emb)$ with $(\Tree_i,\Emb_i)$ for some $i \in \{1,2,3,4\}$, keeping
the fully $\nice$-niceness and decreasing the number of leaves. If we proceed with this procedure exhaustively,
we finally arrive at an embedded tree that is fully $\nice$-nice and has at most $7$ leaves.
\end{proof}

A \emph{branching vertex} is a vertex of an embedded tree $(\Tree,\Emb)$ with at least two children.
By Lemma~\ref{lem:dp:seven-leaves}, in the case $\nice \leq \frac{1}{4}$ we may look for embedded trees with at most $7$ leaves and, consequently,
   at most $6$ branching vertices.
If we are satisfied with any polynomial running time of the algorithm that finds a fully $\nice$-nice embedded tree,
observe that it suffices to guess the images of all leaves and branching vertices of the tree in question, and compute a shortest path between any pair of them.
However, if we aim for a $\Oh(\nice^{-14} |B| \log |B|)$ running time, then we need to proceed more carefully.
We will essentially follow the dynamic-programming algorithm of the unweighted case (i.e., Lemma~\ref{lem:dp:find-tree}) but due
to the existence of arbitrary real weights, we cannot directly use $k_a$ and $k_b$, the lengths of the leftmost and rightmost paths in the constructed tree, as dimensions
in the dynamic programming table.
Instead, we need to round them. The idea is to round independently the length of each maximal path consisting of vertices of degree two of the embedded tree in question;
as there are at most $13$ such paths, we control the error introduced by the rounding.

\begin{lemma}\label{lem:dp:find-tree-wei}
In $\Oh(\nice^{-14} |B| \log |B|)$ time
one can either correctly conclude that no fully $\nice$-nice embedded tree with at most $7$ leaves and of length at most $\wei{\prm B}$ exists in $B$,
or find a fully $(\nice/2)$-nice embedded tree in $B$ of length at most $(1+\nice)\wei{\prm B}$.
\end{lemma}
\begin{proof}
Greedily, we find a set $\porset \subseteq V(B)$ of at most $16/\nice$ \emph{pegs}, such that for any $v \in V(\prm B)$, if we traverse $\prm B$ from $v$ in clockwise direction, then
we encounter a peg at distance at most $\nice \wei{\prm B}/8$ (possibly, the peg is on $v$). Observe the following:
\begin{claim}\label{cl:dp:wei:1}
If there exists in $B$ a fully $\nice$-nice embedded tree with at most $7$ leaves and of length at most $\wei{\prm B}$,
then there exists a fully $(3\nice/4)$-nice embedded tree with at most $7$ leaves and of length at most $(1+7\nice/8)\wei{\prm B}$,
whose leaves are mapped to vertices of $\rzeska{\prm} B$ adjacent to pegs.
\end{claim}
\begin{proof}
Let $(\Tree,\Emb)$ be an embedded tree as in the statement. For each leaf $\rzeska{l_a}$ of $\Tree$, proceed as follows.
Let $l_a$ be the unique neighbour of $\rzeska{l_a}$ in $\Tree$ and let $\Emb(\rzeska{l_a}) = \rzeska{a}$ and $\Emb(l_a) = a$.
Traverse $\prm B$ from $a$ in clockwise direction and let $p(a)$ be the first peg encountered (possibly, $p(a)=a$).
Replace the edge $\rzeska{l_a}l_a$ in $\Tree$ with a copy of the path $\prm B[p(a),a]$ and the edge $\rzeska{p(a)}p(a)$, embedded
by $\Emb$ into $\prm B[p(a),a] \cup \{\rzeska{p(a)}p(a)\}$.
Note that the constructed tree is an embedded tree.
As $\wei{\prm B[p(a),a]} \leq \nice \wei{\prm B}/8$, the constructed tree is fully $(\nice-\nice/4)$-nice and we have
enlarged the length of $\Tree$ by at most $7\nice\wei{\prm B}/8$.
\cqed\end{proof}

Hence, we restrict ourselves to embedded trees whose leaves are mapped to the neighbours of pegs.
We branch into $\Oh(|\porset|^7) = \Oh(\nice^{-7})$ cases, guessing the number of leaves and their images in the tree in question.
That is, we are now given an integer $r \leq 7$ and a sequence $a_1,a_2,\ldots,a_r$ of pegs that appear on $\prm B$ in this counter-clockwise order (possibly $a_i = a_{i+1}$ for some $i$),
and we look for a fully $(3\nice/4)$-nice embedded tree of length at most $(1+7\nice/8)\wei{\prm B}$
with $r$ leaves that maps consecutive leaves to vertices $\rzeska{a_1},\rzeska{a_2},\ldots,\rzeska{a_r}$.

Denote $\krok = \frac{\nice}{104} \wei{\prm B}$.
As discussed earlier, to achieve the promised running time, we need to round the distances in the dynamic programming algorithm.
We will use $\krok$ as one unit of distance for rounding. For a real $x$, by $\rnd(x)$ we denote the smallest integer $k$ for which $k\krok \geq x$, that is,
$\rnd(x) = \lceil x/\krok \rceil$.
For an embedded tree $(\Tree,\Emb)$ with $\rho$ leaves, by $\impV{\Tree}$ we denote the set consisting of the root, all branching vertices, and all neighbours of leaves in the tree $\Tree$.
Observe that $|\impV{\Tree}| \leq 2\rho$.
Let $\Tree' \subseteq  \Tree$ be any subtree of $\Tree$. The set $\impV{\Tree}$ partitions the edge set of $\Tree'$ into a family of paths, with at most $|\impV{\Tree}|-1 \leq 2\rho-1$ paths of positive length; let $\mathcal{P}(\Tree')$ be the family of all these paths.
The \emph{rounded length of $\Tree'$}, denoted $\rnd(\Tree')$, equals $\sum_{P \in \mathcal{P}(\Tree')} \rnd(\wei{P})$.
Observe that
\begin{equation}\label{eq:dp:wei:rnd1}
\wei{\Tree'} \leq \krok \rnd(\Tree') \leq \wei{\Tree'} + (2\rho-1)\krok.
\end{equation}
We remark that this bound on $\rnd(\Tree')$ applies in particular to a $\Tree'$ that is a path between some branching vertex of $\Tree$ and a leaf of $\Tree$.

We now adjust the definition of niceness to the rounded distances.
An embedded tree $(\Tree, \Emb)$ is {\em{$\rnd$-$\nice'$-nice}} if, for any two
consecutive leaves $\rzeska{l_a},\rzeska{l_b}$ in $\Tree$ the following holds.
Let $\Emb(\rzeska{l_a}) = \rzeska{a}$ and $\Emb(\rzeska{l_b}) = \rzeska{b}$
and let $l_a,l_b$ be the parents of $\rzeska{l_a}, \rzeska{l_b}$ in $\Tree$,
respectively; note that $\Emb(l_a) = a$ and $\Emb(l_b)=b$, possibly $a = b$.
Let $w$ be the
lowest common ancestor of $\rzeska{l_a}$ and $\rzeska{l_b}$ in $\Tree$. Then the requirement for $\rnd$-$\nice'$-niceness is that
\begin{equation}
\wei{\prm B[a,b]} + \krok \rnd(\Tree[w,l_a]) + \krok \rnd(\Tree[w,l_b]) \leq (1-\nice')\wei{\prm B}.\label{eq:dp:eps-rnd}
\end{equation}
An embedded tree is {\em{fully $\rnd$-$\nice'$-nice}} if additionally~\eqref{eq:dp:eps-rnd} holds
for $\rzeska{l_a}$ being the last leaf of $\Tree$,
$\rzeska{l_b}$ being the first leaf of $\Tree$, $w = r(\Tree)$ and
$\prm B[a,b]$ replaced by $\prmup B[a,b]$.
Observe the following.

\begin{claim}\label{cl:dp:wei:rnd}
If an embedded tree is (fully) $\rnd$-$\nice'$-nice, then it is also (fully) $\nice'$-nice.
If an embedded tree with at most $7$ leaves is (fully) $\nice'$-nice and $\nice' > \nice/4$, then it is also (fully) $\rnd$-$(\nice' - \nice/4)$-nice.
\end{claim}
\begin{proof}
The claim follows by applying inequality~\eqref{eq:dp:wei:rnd1} to $\rnd(\Tree[w,l_a])$ and $\rnd(\Tree[w,l_b])$ in condition~\eqref{eq:dp:eps-rnd}.
\cqed\end{proof}
By Claims~\ref{cl:dp:wei:1} and~\ref{cl:dp:wei:rnd}, we may restrict ourselves to searching for a fully $\rnd$-$(\nice/2)$-nice embedded tree: in each of these claims we lose only $\nice/4$ on the niceness of the tree.

We are now ready to describe the main table for the dynamic-programming algorithm.
Define $L$ to be the largest integer such that $\krok L\leq (1+\nice)\wei{\prm B}$; observe that $L=\Oh(\nice^{-1})$. 
For each $v \in V(B)$, indices $1 \leq i_a \leq i_b \leq r$,
and integers $0 \leq k_a,k_b,\ell \leq L$,
we define the value $F[v,i_a,i_b,\ell,k_a,k_b]$ to be any embedded tree
$(\Tree,\Emb)$ that satisfies the following:
\begin{enumerate}
\item $(\Tree,\Emb)$ is $\rnd$-$(\nice/2)$-nice;
\item $\Emb(r(\Tree)) = v$;
\item $\Tree$ has $i_b-i_a+1$ leaves, mapped by $\Emb$ onto $\rzeska{a_{i_a}},\rzeska{a_{i_a+1}},\ldots,\rzeska{a_{i_b}}$ in this order;
\item $\Tree$ has rounded length at most $\ell$;
\item if $\rzeska{l_a}$ is the first leaf of $\Tree$ and $\rzeska{l_b}$ is the last leaf, then $\rnd(\Tree[\rzeska{l_a},r(\Tree)]) \leq k_a$
and $\rnd(\Tree[\rzeska{l_b},r(\Tree)]) \leq k_b$.
\end{enumerate}
We require that $F[v,i_a,i_b,\ell,k_a,k_b] = \bot$ if no such embedded tree exists.

The next two claims verify that computing all values $F[v,i_a,i_b,\ell,k_a,k_b]$ is sufficient for our needs.
\begin{claim}\label{cl:dp:wei:extract}
Assume that $F[v,1,r,L,k_a,k_b] = (\Tree,\Emb) \neq \bot$ for some $v, k_a,k_b$.
Moreover, assume that
\begin{equation}\label{eq:dp:wei:extract}
\wei{\prmup B[a_r,a_1]} + \krok k_a + \krok k_b \leq (1-\nice/2)\wei{\prm B}.
\end{equation}
Then $(\Tree,\Emb)$ is fully $(\nice/2)$-nice and has length at most $(1+\nice)\wei{\prm B}$.
\end{claim}
\begin{proof}
First, observe that $(\Tree,\Emb)$ is $\rnd$-$(\nice/2)$-nice by the properties of the cell $F[v,1,r,L,k_a,k_b]$.
Moreover, we have that the first leaf of $\Tree$ is mapped onto $\rzeska{a_1}$ and the last leaf is mapped onto $\rzeska{a_r}$.
Hence, inequality~\eqref{eq:dp:wei:extract} implies that $(\Tree,\Emb)$ is fully $\rnd$-$(\nice/2)$-nice.
By Claim~\ref{cl:dp:wei:rnd}, $(\Tree,\Emb)$ is fully $(\nice/2)$-nice.
Finally, note that since $\Tree$ has rounded length at most $L$, by~\eqref{eq:dp:wei:rnd1} the length of $\Tree$ is bounded by $L\krok\leq (1+\nice)\wei{\prm B}$.
\cqed\end{proof}
\begin{claim}\label{cl:dp:wei:extract2}
Assume that there exists in $B$ a fully $(3\nice/4)$-nice embedded tree $(\Tree,\Emb)$ of length
at most $(1+7\nice/8)\wei{\prm B}$, such that the leaves of $\Tree$ are mapped onto $a_1,a_2,\ldots,a_r$
in this order. Then $F[v,1,r,L,k_a,k_b] \neq \bot$ for some $v$ and $k_a,k_b$ satisfying~\eqref{eq:dp:wei:extract}.
\end{claim}
\begin{proof}
First, observe that by~\eqref{eq:dp:wei:rnd1} we have
$$\krok\rnd(\Tree) \leq \wei{\Tree}+13\krok\leq (1+7\nice/8)\wei{\prm B}+13\krok=(1+7\nice/8)\wei{\prm B}+\frac{13\nice}{104}\wei{\prm B}=(1+\nice)\wei{\prm B}.$$
By the definition of $L$, this means that $\rnd(\Tree)\leq L$.
Let $\rzeska{l_a}$ and $\rzeska{l_b}$ be the first and the last leaf of $\Tree$, respectively.
Let $k_a = \rnd(\Tree[\rzeska{l_a},r(\Tree)])$ and $k_b = \rnd(\Tree[\rzeska{l_b},r(\Tree)])$.
Note that $k_a,k_b \leq \rnd(\Tree) \leq L$. Hence, 
$(\Tree,\Emb)$ is a valid candidate for $F[v,1,r,L,k_a,k_b]$ where $v = \Emb(r(\Tree))$.
Moreover, by Claim~\ref{cl:dp:wei:rnd}, $(\Tree,\Emb)$ is fully $\rnd$-$(\nice/2)$-nice and hence
inequality~\eqref{eq:dp:wei:extract} is satisfied for $k_a$ and $k_b$.
\cqed\end{proof}

We now describe how to compute the values $F[v,i_a,i_b,\ell,k_a,k_b]$.
Initially, we set $F[a_i,i,i,0,k_a,k_b]$ to be a tree consisting of the edge $\rzeska{a_i}a_i$ with the identity mapping,
for each $1 \leq i \leq r$ and $0 \leq k_a,k_b \leq L$. Moreover, we set $F[v,i,i,0,k_a,k_b] = \bot$ for any $v \neq a_i$ and $0 \leq k_a,k_b \leq L$.
It is straightforward to verify that these are correct values of the entries $F[v,i_a,i_b,\ell,k_a,k_b]$ for $i_a=i_b$ and $\ell=0$.

Then, we compute the values $F[v,i_a,i_b,\ell,k_a,k_b]$ in order of increasing values $(i_b-i_a)$ and $\ell$.
That is, for fixed $i_a$, $i_b$, $\ell$, $k_a$, $k_b$, we want to compute the entries $F[v,i_a,i_b,\ell,k_a,k_b]$ for all $v \in V(B)$
in $\Oh(\nice^{-4} |B| \log |B|)$ time, assuming that all entries $F[v',i_a',i_b',\ell',k_a',k_b']$ were already computed
whenever $i_b'-i_a' \leq i_b -i_a$, $\ell' \leq \ell$, and at least one of this inequality is strict.

Consider now a cell $F[v,i_a,i_b,\ell,k_a,k_b]$ for $(i_b-i_a) + \ell > 0$.
If $F[v,i_a,i_b,\ell',k_a',k_b'] \neq \bot$ for some $\ell' \leq \ell$, $k_a' \leq k_a$, $k_b' \leq k_b$
and $(\ell,k_a,k_b) \neq (\ell',k_a',k_b')$, then we may copy the value of $F[v,i_a,i_b,\ell',k_a',k_b']$ and conclude.
Hence, assume otherwise.

Consider an embedded tree $(\Tree,\Emb)$ that satisfies all requirements for the cell $F[v,i_a,i_b,\ell,k_a,k_b]$.
There are two cases, depending on the degree of $r(\Tree)$.

If $r(\Tree)$ has at least two children in $\Tree$, let $\Tree_1$ be the subtree of $\Tree$ rooted at the first child of $\Tree$ (together
with the edge towards the root $r(\Tree)$) and let $\Tree_2 = \Tree \setminus \Tree_1$. Denote $\Emb_j = \Emb|_{\Tree_j}$ for $j=1,2$.
Let $i$ be such that $\Tree_1$ has $i-i_a$ leaves, that is, the last leaf of $\Tree_1$, denoted $\rzeska{l_c}$, is mapped onto $\rzeska{a_{i-1}}$ and the first
leaf of $\Tree_2$, denoted $\rzeska{l_d}$, is mapped onto $\rzeska{a_i}$. Observe that $i_a < i \leq i_c$.
Denote $\ell_1 = \rnd(\Tree_1)$, $\ell_2 = \rnd(\Tree_2)$, $k_c = \rnd(\Tree_1[r(\Tree),\rzeska{l_c}])$, and $k_d = \rnd(\Tree_2[r(\Tree),\rzeska{l_d})$.
Observe that $(\Tree_1,\Emb_1)$ is a feasible entry for $F[v,i_a,i-1,\ell_1,k_a,k_c]$ and $(\Tree_2,\Emb_2)$ is a feasible entry for $F[v,i,i_b,\ell_2,k_d,k_b]$.
Moreover, $\ell_1,\ell_2 \leq \ell$, $\ell_1 + \ell_2 = \ell$ and, since $(\Tree,\Emb)$ is $\rnd$-$(\nice/2)$-nice we have that
\begin{equation}\label{eq:dp:wei:join}
\wei{\prm B[a_{i-1},a_i]} + \krok k_c + \krok k_d \leq (1-\nice/2) \wei{\prm B}.
\end{equation}

In the other direction, assume that for some choice of $0 \leq \ell_1,\ell_2 \leq \ell$ with $\ell_1 + \ell_2 = \ell$, $i_a < i \leq i_b$
and $0 \leq k_c,k_d \leq L$ satisfying~\eqref{eq:dp:wei:join} we have $F[v,i_a,i-1,\ell_1,k_a,k_c] = (\Tree_1,\Emb_1) \neq \bot$
and $F[v,i,i_b,\ell_2,k_d,k_b] = (\Tree_2,\Emb_2)\neq \bot$. Define
$\Tree$ to be $\Tree_1 \cup \Tree_2$ with identified roots of $\Tree_1$ and $\Tree_2$, and $\Emb = \Emb_1 \cup \Emb_2$.
It is straightforward to verify that $(\Tree,\Emb)$ is a feasible entry for $F[v,i_a,i_b,\ell,k_a,k_b]$. Here observe that~\eqref{eq:dp:wei:join} ensures
that condition~\eqref{eq:dp:eps-rnd} is satisfied for leaves mapped to $\rzeska{a_{i-1}}$ and $\rzeska{a_i}$.
Moreover, in the dynamic programming the values $F[v,i_a,i-1,\ell_1,k_a,k_c]$ and $F[v,i,i_b,\ell_2,k_d,k_b]$ are already computed
when we consider the cell $F[v,i_a,i_b,\ell,k_a,k_b]$, since $i-1-i_a < i_b - i_a$, $i_b-i < i_b-i_a$ and $\ell_1,\ell_2 \leq \ell$.
Hence, we look for a feasible candidate for $F[v,i_a,i_b,\ell,k_a,k_b]$ among all values $\ell_1,\ell_2,i,k_c,k_d$ as above
and merge $(\Tree_1,\Emb_1)$ with $(\Tree_2,\Emb_2)$ whenever possible.
By the argumentation so far, whenever there exists a feasible candidate for $F[v,i_a,i_b,\ell,k_a,k_b]$ with root of degree at least two,
we find at least one such candidate.

In the remaining case, $r(\Tree)$ has exactly one child. Observe that $\Tree$
has more than one edge, as otherwise $i_a=i_b$, $v = a_{i_a}$ and $(\Tree,\Emb)$ is a feasible candidate for $F[v,i_a,i_b,0,0,0]$,
and we would have found $(\Tree,\Emb)$ in the first step.
Hence, $\impV{\Tree}$ contains at least two vertices. Let $x$ be the vertex of $\impV{\Tree} \setminus r(\Tree)$ that is closest to $r(\Tree)$.
Denote $u = \Emb(x)$ and $k = \rnd(\wei{\Tree[r(\Tree),x]})$; note that $k > 0$.
Define $\Tree'$ to be the tree $\Tree \setminus \Tree[r(\Tree),x]$, rooted at $x$, and $\Emb' = \Emb|_{\Tree'}$.
Observe that $(\Tree',\Emb')$ is an embedded tree and it is a feasible candidate for $F[u,i_a,i_b,\ell-k,k_a-k,k_b-k]$.

In the other direction, assume that $F[u,i_a,i_b,\ell-k,k_a-k,k_b-k] = (\Tree',\Emb')$ for some $u \in V(B)$, where $k \geq \rnd(\dist_B(u,v))$.
To obtain an embedded tree $(\Tree,\Emb)$, extend $(\Tree',\Emb')$ with a copy of a shortest path between $u$ and $v$ in $B$, mapped by $\Emb$ to its original,
connecting $r(\Tree')$ with a new root $r(\Tree)$ (mapped by $\Emb$ to $v$). It is straightforward to verify that $(\Tree,\Emb)$
is a feasible candidate for $F[v,i_a,i_b,\ell,k_a,k_b]$. We remark here that the rounded length of $\Tree$ may be strictly smaller
than $\rnd(\Tree') + \rnd(\dist_B(u,v))$ in the case when $r(\Tree')$ has degree one.

Hence, to verify whether there exists a feasible candidate for $F[v,i_a,i_b,\ell,k_a,k_b]$, we need to inspect all entries 
$F[u,i_a,i_b,\ell-k,k_a-k,k_b-k]$ where $u \in V(B) \setminus \{v\}$ and $k \geq \rnd(\dist_B(u,v))$. However, a naive implementation would take time quadratic in $|B|$.
We now show how to check all pairs $(v,u)$ using at most $L$ runs of Dijkstra's shortest-path algorithm in $B$,
which yields a $\Oh(L |B| \log |B|)$-time algorithm.
Iterate through all integers $k$ such that $1 \leq k \leq \min(\ell,k_a,k_b) \leq L$. Define $U$ to be the set of these vertices $u$ for which $F[u,i_a,i_b,\ell-k,k_a-k,k_b-k] \neq \bot$.
By a single run of Dijkstra's algorithm in $B$ starting from $U$, we may compute $\dist_B(v,U)$ for every $v \in V(B)$. Moreover, for each $v \in V(B)$ we can compute the closest vertex $u(v) \in U$
and a shortest path between $v$ and $u(v)$.
Then we inspect all $v \in V(B)$ and whenever $\rnd(\dist_B(v,U)) \leq k$, we may use the entry $F[u(v),i_a,i_b,\ell-k,k_a-k,k_b-k]$ to find a feasible candidate for $F[v,i_a,i_b,\ell,k_a,k_b]$.

We remark here that we do not need to explicitly keep the embedded trees as values of $F[v,i_a,i_b,\ell,k_a,k_b]$. It suffices to keep only a boolean that signals whether a feasible
candidate has been found and, if this is the case, how it was obtained. Then, the actual tree for a fixed cell $F[v,i_a,i_b,\ell,k_a,k_b]$ can be computed in $\Oh( |B| )$ time:
we need to reproduce at most $13$ shortest paths in the tree, each of which can be computed in linear time~\cite{planar-sp}.

We now analyze the running time. There is an $\Oh(\nice^{-7})$ overhead from guessing $r$ and the sequence $a_1,a_2,\ldots,a_r$.
In the dynamic-programming algorithm, in each step we need to keep track of at most $7$ integer variables ranging from $0$ to $L$ (namely, $\ell,k_a,k_b,\ell_1,\ell_2,k_c,k_d$).
Recall that $r \leq 7$. Hence, we obtain a running time of $\Oh(\nice^{-14} |B| \log |B|)$.
\end{proof}

We may now conclude the proof of Theorem \ref{thm:nice-testing-wei}.
By Lemma~\ref{lem:dp:seven-leaves}, if a short $\nice$-nice tree exists in $B$, then there exists a fully $\nice$-nice embedded tree with at most $7$ leaves and not larger length.
Using Lemma \ref{lem:dp:find-tree-wei} we look for such a tree;
if it indeed exists in $B$, we obtain a fully $(\nice/2)$-nice embedded tree of length at most $(1+\nice)\wei{\prm B}$.
In this case, we apply Lemma \ref{lem:dp:tree-to-bricks} to obtain the desired family of bricks.
If the algorithm of Lemma \ref{lem:dp:find-tree-wei} does not find any
embedded tree, Lemma \ref{lem:dp:seven-leaves} allows us to conclude
that no short $\nice$-nice tree exists in $B$.

%!TEX root = pst-kernel.tex

\newcommand{\north}{\mathbf{N}}
\newcommand{\south}{\mathbf{S}}
\newcommand{\west}{\mathbf{W}}
\newcommand{\east}{\mathbf{E}}
\newcommand{\Bleft}{\mathfrak{l}}
\newcommand{\Bright}{\mathfrak{r}}

\section{Weighted variant}\label{sec:weights}

We now focus on the weighted variant, and prove Theorem~\ref{thm:weighted}.
%Recall that for a graph $G$ and a family of terminal pairs $\mathcal{\terms} \subseteq V(G) \times V(G)$, a \emph{Steiner forest} connecting $\mathcal{\terms}$ in $G$ is a subgraph $F \subseteq G$ such that for each pair $(s,t) \in \mathcal{\terms}$ both $s$ and $t$ lie in the same connected component of $F$.
As described in the outline of the proof of Theorem~\ref{thm:weighted} (see Section~\ref{ss-over:weighted}),
we start with a base case, when $\mc{\terms}$ consists of a single pair and only a multiplicative error in the weight of the forest $F_{H}$ is allowed.
More formally, the following follows directly from
the spanner construction of Klein~\cite[Theorem 7.1]{klein:stoc06}.

\begin{theorem}\label{thm:wei:paths}
Let $\eps > 0$ be a fixed accuracy parameter,
and let $B$ be an edge-weighted brick.
Then one can find in $\Oh(\eps^{-1} |B| \log |B|)$ time a graph
$H \subseteq B$ such that
\begin{enumerate}
\item $\prm B \subseteq H$,
  \item $\wei{H} = \Oh(\eps^{-4} \wei{\prm B})$, and
  \item for any pair of vertices $s,t \in V(\prm B)$
there exists a path connecting $s$ and $t$ in $H$
of weight at most $(1+\eps)\dist_B(s,t)$.
\end{enumerate}
\end{theorem}

In Section~\ref{ss:wei:portals}, we present the $\portalbound$-variant of Theorem~\ref{thm:weighted}, where $\mc{\terms}$ contains at most $\portalbound$ terminal pairs and $\wei{H}$ depends polynomially on $\epsilon^{-1}$ and $\portalbound$. Finally, in Section~\ref{ss:wei:final}, we derive Theorem~\ref{thm:weighted}.

\subsection{Bounded number of terminal pairs}\label{ss:wei:portals}
We now prove a $\portalbound$-variant of Theorem~\ref{thm:weighted}. To be precise, we show:

\newcommand{\Bfin}{\mathcal{A}}
\newcommand{\Bfindown}{\Bfin^\downarrow}

\begin{theorem}\label{thm:wei:portals}
Let $\eps > 0$ be a fixed accuracy parameter, let $\portalbound$ be a positive integer,
and let $B$ be an edge-weighted brick.
Then one can find in $\poly(\eps^{-1}, \portalbound) |B| \log |B|$ time a graph
$H \subseteq B$ such that
\begin{enumerate}
\renewcommand{\theenumi}{\roman{enumi}}
\renewcommand{\labelenumi}{(\theenumi)}
\item $\prm B \subseteq H$,
\item $\wei{H} \leq \poly(\eps^{-1}, \portalbound) \wei{\prm B}$, and
\item for every set $\mathcal{\terms} \subseteq V(\prm B) \times V(\prm B)$ of size at most $\portalbound$,
there exists a Steiner forest $F_H$ that connects $\mathcal{\terms}$ in $H$
such that $\wei{F_H} \leq \wei{F_B} + \eps \wei{\prm B}$ for any Steiner forest $F_B$ that connects $\mathcal{\terms}$ in $B$.
\end{enumerate}
\end{theorem}

From a high-level perspective, we proceed similarly as in Section~\ref{sec:finish}.
The algorithm has two phases: in the first phase, we recursively use the decomposition tools developed in 
the previous sections to compute a brick covering $\Bfin$ of $B$, where each $B' \in \Bfin$ has the following property: either $\wei{\prm B'}$ is small, or
for every set $\mathcal{\terms} \subseteq V(\prm B) \times V(\prm B)$
of size at most $\portalbound$, there exist an optimal Steiner forest connecting $\mathcal{\terms}$ that does not contain any vertex of degree larger than $2$ that is strictly enclosed by $\prm B'$.

\subsubsection{Phase one: decomposing $B$}
We first initialize a family $\Bfin = \emptyset$. During the course of the algorithm all elements of this family will be subbricks of $B$.
Then we call a procedure $\mathtt{partition}$ on the input brick $B$.
The description of the procedure $\mathtt{partition}$, when called on a subbrick $B'$ of $B$, is as follows.

Call $B'$ \empty{tiny} if $\wei{\prm B'} \leq \frac{\eps}{\portalbound} \wei{\prm B}$, and \emph{large} otherwise. If $B'$ is tiny, then put $B'$ into $\Bfin$. If $B'$ is large, then invoke the algorithm of Theorem~\ref{thm:nice-testing-wei} for the brick $B'$
and parameter $\nice = \frac{1}{36}$. If the algorithm finds a $(3+2\nice)$-short $(\nice/2)$-nice brick covering $\Bb(B')$ of $B'$, then recursively invoke $\mathtt{partition}$ on all bricks of $\Bb(B')$.

If the algorithm of Theorem~\ref{thm:nice-testing-wei} finds that no short $\nice$-nice tree exists in $B$, then invoke the algorithm of Theorem~\ref{thm:core} for $\nice = \frac{1}{36}$ and $\delta=2\nice$ to find the core face $\coreface$, and then invoke the algorithm of Theorem~\ref{thm:ananas} for $\nice = \frac{1}{36}$ and the 
brick $B'$. Let $C$ be the cycle found by Theorem~\ref{thm:ananas}.
We find a sequence $p_1,p_2,\ldots,p_s$ of pegs on $C$ such that
for any $1 \leq i \leq s$ either $p_i,p_{i+1}$ are two consecutive vertices of $C$
or $\wei{C[p_i,p_{i+1}]} \leq 2\nice \wei{\prm B'}$ (here we assume $p_{s+1} = p_1$).
In a greedy manner (as in Section~\ref{sec:sliding}), we can find in linear time a sequence of such pegs with
\begin{equation}\label{eq:wei:portals:s}
s \leq \frac{1}{\nice} \cdot \frac{\wei{C}}{\wei{\prm B'}} \leq \frac{16}{\nice^3} = \Oh(1).
\end{equation}
Then we find, for each peg $p_i$, a shortest path $P_i$ between $p_i$ and $V(\prm B')$ that does not contain any edge strictly enclosed by $C$. Let $x_i$ be the second endpoint of $P_i$.
Observe that we may assume that the paths $P_i$ obtained in this manner are non-crossing in the following sense: whenever $P_i$ and $P_j$ meet at some vertex, they continue together towards a common endpoint $x_i=x_j$ on $V(\prm B)$.
Indeed, we can find the vertices $x_{i}$ by removing all edges and vertices that are strictly enclosed by $C$, adding a super-terminal $s_{0}$ in the outer face, and connecting $s_{0}$ to the vertices of $\prm B$ using edges of weight zero. The graph we just constructed is planar, and by constructing a shortest-path tree $T$ for $s_{0}$ in this graph (which takes linear time~\cite{planar-sp}), we can find the vertices $x_{i}$ in linear time. Then the paths $P_{i}$ are simply the $p_ix_i$-paths in $T$. By construction, these paths have the required property.

Now consider any $i$ such that $1 \leq i \leq s$ and $C[p_i,p_{i+1}] \neq \prm B'[x_i,x_{i+1}]$.
Let $W_{i}$ denote the closed walk $P_i \cup C[p_i,p_{i+1}] \cup P_{i+1} \cup \prm B'[x_i,x_{i+1}]$ in $B'$. Let $H_i$ be the graph consisting of all edges of $W_i$ that neighbour the outer face
of $W_i$ treated as a planar graph. 
%Note that $W_i$ and $H_i$ can be computed in linear time for fixed $i$.
By definition, each doubly-connected component of $H_i$ is a cycle or a bridge.
For each doubly-connected component that is a cycle, we create a brick consisting of all the edges
of $B$ that are enclosed by this cycle. Let $\Bb_i$ be the family of obtained bricks.
Observe that $\Bb_i$ can be computed in linear time for fixed $i$ and a face of $B'$ is enclosed by some brick
of $\Bb_i$ if and only if it is enclosed by $W_i$. For each $1 \leq i \leq s$, we recursively call
$\mathtt{partition}$ on all bricks of $\Bb_i$.

Finally, we put a brick $B^C$ consisting of all edges of $B$ enclosed by $C$ into $\Bfin$.

This concludes the description of the procedure $\mathtt{partition}$, and hence the description of the first phase of the algorithm.
We now analyse the family $\Bfin$ and the running time of the algorithm.

First, we establish some more notation that will be useful in the analysis.
For a fixed call $\mathtt{partition}(B')$, by $\Bfin(B')$ we denote all bricks that are inserted into $\Bfin$ during this call,
and by  $\Bfindown(B')$ we denote all bricks that are inserted into $\Bfin$ in any call in the subtree of the recursion tree
rooted at the call $\mathtt{partition}(B')$, including $\Bfin(B')$.

In the case when Theorem~\ref{thm:ananas} has been invoked, we denote $\Bb^r(B') = \bigcup_{i=1}^s \Bb_i$ and $\Bb(B') = \{B^C\} \cup \Bb^r(B')$.
In the case when Theorem~\ref{thm:nice-testing-wei} returned a brick covering $\Bb(B')$, we denote also $\Bb^r(B') = \Bb(B')$.
Observe that, regardless of whether Theorem~\ref{thm:ananas} has been invoked or not,
\begin{itemize}
\item $\Bb^r(B')$ is the family of subbricks of $B'$ for which a recursive call has been made;
\item $\Bb(B') = \Bfin(B') \cup \Bb^r(B')$;
\item $\Bb(B')$ is a brick covering of $B'$ with the additional property that $\bigcup_{B^\ast \in \Bb(B')} \prm B^\ast$ is connected.
\end{itemize}

\noindent
Using these properties, we analyse the family $\Bfin$.
\begin{lemma}\label{lem:wei:portals:cover}
$\Bfin$ is a brick covering of $B$ and, moreover, $\bigcup_{B_1 \in \Bfin} \prm B_1$ is connected.
\end{lemma}
\begin{proof}
By induction on the recursion tree of procedure $\mathtt{partition}$,
we prove that for any call $\mathtt{partition}(B')$,
the family $\Bfindown(B')$ is a brick covering of $B'$ and, moreover, $\bigcup_{B_1 \in \Bfindown(B')} \prm B_1$ is connected.
This is clearly true in the leaves of the recursion tree
when $\Bfindown(B') = \{B'\}$.
In an induction step, observe that the fact that $\Bfindown(B')$ is a brick covering of $B'$ follows from the fact that $\Bb(B')$ is a brick covering of $B'$ and the induction hypothesis for all elements of $\Bb^r(B')$. The fact that $\bigcup_{B_1 \in \Bfindown(B')} \prm B_1$ is connected follows from the fact that
$\bigcup_{B^\ast \in \Bb(B')} \prm B^\ast$ is connected, $\Bb(B') = \Bfin(B') \cup \Bb^r(B')$
and the induction hypothesis for all elements of $\Bb^r(B')$.
\end{proof}

\begin{lemma}\label{lem:wei:portals:one-correct}
For every set $\mathcal{\terms} \subseteq V(\prm B) \times V(\prm B)$
there exists a Steiner forest $F$ connecting $\mathcal{\terms}$ in $B$
of minimum possible length with the following additional property:
for every vertex $v$ of degree at least three in $F$,
there exists some $B_1 \in \Bfin$ such that either
\begin{enumerate}
\item $v \in V(\prm B_1)$, or
\item $v$ is strictly enclosed by $\prm B_1$ and $B_1$ is tiny.
\end{enumerate}
\end{lemma}
\begin{proof}
For any call $\mathtt{partition}(B')$ in the recursion tree,
and for any forest $F$ in $B'$, we say that a vertex $v$ is \emph{lame}
if (a) the degree of $v$ in $F$ is at least three, and (b) for any $B_1 \in \Bfindown(B')$
we have $v \notin V(\prm B_1)$, and (c) if $v$ is strictly enclosed by $\prm B_1$, then $B_1$ is large.
By induction on the recursion tree of the procedure $\mathtt{partition}$,
we prove that for any call $\mathtt{partition}(B')$
and any $\mathcal{\terms} \subseteq V(\prm B') \times V(\prm B')$
there exists a Steiner forest $F$ connecting $\mathcal{\terms}$ of minimum possible length
that does not contain lame vertices.
In the leaves of the recursion tree, the statement is clearly true as $\Bfindown(B')=\{B'\}$ and $B'$ is tiny.

Consider now a call $\mathtt{partition}(B')$, and let 
$\mathcal{\terms} \subseteq V(\prm B') \times V(\prm B')$.
By Theorem~\ref{thm:ananas}, there exists a Steiner forest $F$ connecting $\mathcal{\terms}$ in $B'$
of minimum possible length that additionally satisfies the following: if Theorem~\ref{thm:ananas} has been invoked
to obtain $\Bb(B')$, then no vertex of degree at least three in $F$ is strictly enclosed by $\prm B^C$.
Pick such $F$ that minimizes the number of lame vertices.
We claim that there are in fact no lame vertices; note that such a claim proves the induction step and finishes the proof of the lemma.
Assume the contrary, and let $v$ be any lame vertex for $F$.

As $v$ is not strictly enclosed by $\prm B^C$ in the case when Theorem~\ref{thm:ananas} has been invoked, we infer that
there exists $B^\ast \in \Bb^r(B')$ such that $\prm B^\ast$ encloses $v$. 
As $v \notin V(\prm B_1)$ for any $B_1 \in \Bfindown(B')$, $\prm B^\ast$ strictly encloses $v$.
Consider $F_1 := F \cap B^\ast$ and let $\mathcal{\terms}_1$ be the set of  pairs $(x,y)$
such that $x,y \in V(F_1) \cap V(\prm B^\ast)$, $x \neq y$, and $x,y$ belong to the same connected component of $F_1$.
By the induction hypothesis, there exists a forest $F_2$ connecting $\mathcal{\terms}_1$ in $B^\ast$ of length at most $\wei{F_1}$
that does not contain any lame vertices in $B^\ast$.
Hence, $F' := (F \setminus F_1) \cup F_2$ is a Steiner forest connecting $\mathcal{\terms}$ in $B'$ of length at most $\wei{F}$
that contains a strict subset of the set of lame vertices of $F$, a contradiction to the choice of $F$.
This finishes the induction step, and concludes the proof of the lemma.
\end{proof}

We now move to the analysis of the efficiency of the algorithm.
Our goal is to prove upper bounds on the size of $\Bfin$, on the total length of the perimeters of the bricks in $\Bfin$,
and on the running time of phase one.

\begin{lemma} \label{lem:wei:portals:bb}
Let $i\in\{1,\ldots,s\}$ be such that $C[p_i,p_{i+1}] \neq \prm B'[x_i,x_{i+1}]$. Then $\sum_{B_{1} \in \Bb_i} \wei{\prm B_{1}} \leq (1-2\nice)\wei{\prm B'}$.
\end{lemma}
\begin{proof}
Consider the walk $Q_i := P_i \cup C[p_i,p_{i+1}] \cup P_{i+1}$ that connects
$x_i$ and $x_{i+1}$. We claim that:
\begin{equation}\label{eq:wei:portals:Pi}
\wei{Q_i} \leq \textstyle \left(\frac{1}{2} - 2\nice\right) \wei{\prm B'}.
\end{equation}
Indeed, if $\wei{C[p_i,p_{i+1}]} \leq 2\nice \wei{\prm B'}$, then as each vertex of $C$ is at distance at most $(\frac{1}{4}-2\nice)\cdot\wei{\prm B'}$ from $V(\prm B')$ by the construction of $C$ and Theorem~\ref{thm:ananas}, the paths $P_{i}$ and $P_{i+1}$ have length at most $(\frac{1}{4}-2\nice)\wei{\prm B'}$, and the claim follows. Otherwise, by the construction of the pegs, $p_ip_{i+1}$ is an edge of $C$. Now, the claim follows from the fact that each point of $C$ (and in particular every point of the edge $p_ip_{i+1}$) is within distance at most $(\frac{1}{4}-2\nice)\wei{\prm B'}$ from $V(\prm B')$ and the assumption that $C[p_i,p_{i+1}] \neq \prm B'[x_i,x_{i+1}]$.

Observe that $W_i = Q_i \cup \prm B'[x_i,x_{i+1}]$. We claim that:
\begin{equation}\label{eq:wei:portals:Wi}
\wei{W_i} \leq \textstyle \left(1 - 2\nice\right) \wei{\prm B'}.
\end{equation}
If the paths $P_{i}$ and $P_{i+1}$ intersect, then $x_i=x_{i+1}$ by the construction of $P_{i}$ and $P_{i+1}$,
and~\eqref{eq:wei:portals:Wi} is immediate from~\eqref{eq:wei:portals:Pi} and the choice of $\nice$.
So assume that the paths $P_{i}$ and $P_{i+1}$ do not intersect. In particular, $x_i \not=x_{i+1}$.
Let $z_i$ be the vertex of $V(P_i) \cap V(C[p_i,p_{i+1}])$ that lies closest to $x_i$ on $P_i$; define $z_{i+1}$ similarly with respect to $P_{i+1}$.
Observe that $z_i$ lies closer to $p_i$ on $C[p_i,p_{i+1}]$ than $z_{i+1}$, as otherwise $P_i[p_i,z_i]$ and $P_{i+1}[p_{i+1},z_{i+1}]$
would intersect (recall that none of these paths contain an edge strictly enclosed by $C$).
Hence, $C[z_i,z_{i+1}]$ is a subpath of $C[p_i,p_{i+1}]$.
Let $Q_i' = P_i[x_i,z_i] \cup C[z_i,z_{i+1}] \cup P_{i+1}[z_{i+1},x_{i+1}]$.
Observe that $Q_i'$ is a simple path of length at most $\wei{Q_i} \leq (\frac{1}{2}-2\nice)\wei{\prm B'}$ by~\eqref{eq:wei:portals:Pi}.
Moreover, the closed walk $W_i' := Q_i' \cup \prm B'[x_i,x_{i+1}]$ does not enclose any point strictly enclosed by $C$.
%its subpath $C[z_i,z_{i+1}]$ neighbours the interior of $C$ to the left, traversing $C$ counter-clockwise, and $W_i'$ does not contain any edge strictly enclosed by $C$.
Hence, $Q_i' \cup \prm B'[x_{i+1},x_{i}]$ encloses the whole of $C$, and thus in particular the core face $\coreface$.
Thus, $(Q_i', \prm B'[x_{i+1},x_i])$ is not a $(2\nice)$-carve,
despite that $\wei{Q_i'} \leq \textstyle \left(\frac{1}{2} - 2\nice\right) \wei{\prm B'}$. Therefore, it must be that $\wei{\prm B'[x_{i+1},x_{i}]} > \frac{1}{2}\wei{\prm B'}$, and thus $\wei{\prm B'[x_i,x_{i+1}]} \leq \frac{1}{2}\wei{\prm B'}$.
Then~\eqref{eq:wei:portals:Wi} follows from~\eqref{eq:wei:portals:Pi}. 

It remains to observe that
$\sum_{B_{1} \in \Bb_i} \wei{\prm B_{1}} \leq \wei{W_i} \leq (1-2\nice)\wei{\prm B'}$.
\end{proof}

\begin{lemma}\label{lem:wei:portals:shrink}
If $\mathtt{partition}(B')$ recursively calls
$\mathtt{partition}(B^\ast)$, then $\wei{\prm B^\ast} \leq (1-\nice/2)\cdot\wei{\prm B'}$.
\end{lemma}
\begin{proof}
If a $(3+2\nice)$-short $\nice/2$-nice brick covering $\Bb$ has been found in $B'$,
   then the claim follows from the niceness of $\Bb$.
In the second case, when Theorem~\ref{thm:ananas} is invoked,
the claim follows by Lemma~\ref{lem:wei:portals:bb}.
\end{proof}

\begin{lemma}\label{lem:wei:portals:perimeter}
There exists a universal constant $C$ such that the following holds:
for any call  $\mathtt{partition}(B')$, we have $\sum_{B^\ast \in \Bb^r(B')} \wei{\prm B^\ast} \leq C\wei{\prm B'}$.
\end{lemma}
\begin{proof}
The claim is immediate for any $C \geq 3+2\nice$ in the case when
a $(3+2\nice)$-short $\nice/2$-nice brick partition $\Bb$ has been found in $B'$.
In the second case, when Theorem~\ref{thm:ananas} is invoked,
note that the claim follows for sufficiently large $C$ by Lemma~\ref{lem:wei:portals:bb}
and the bound of~\eqref{eq:wei:portals:s} that $s = \Oh(1)$.
\end{proof}

\begin{lemma}\label{lem:wei:portals:calls}
There exists a universal constant $c$ such that the following holds:
for any call $\mathtt{partition}(B')$, in the subtree of the recursion tree
rooted at this call there are at most
$$c \left(\frac{\portalbound}{\eps} \cdot \frac{\wei{\prm B'}}{\wei{\prm B}}\right)^c$$
calls to $\mathtt{partition}(B^\ast)$
where $B^\ast$ is large
(i.e., the call $\mathtt{partition}(B^\ast)$ does not finish after the first step).
\end{lemma}
\begin{proof}
We prove the claim by induction, proceeding from the leaves to the root of the recursion tree.
The claim is clearly true for any positive $c$ 
if $B'$ is tiny, as no recursive call
is made.

Consider now a call $\mathtt{partition}(B')$ where 
$B'$ is large.
We use Lemmata~\ref{lem:wei:portals:shrink} and~\ref{lem:wei:portals:perimeter};
let $C$ be the constant given by the latter.
By the induction hypothesis, for sufficiently large $c$
that depends on $\nice=\frac{1}{36}$ and $C$, the number of calls in question
is bounded by 
\begin{align*}
&1+\sum_{B^\ast \in \Bb^r(B')} c \left(\frac{\portalbound}{\eps} \cdot \frac{\wei{\prm B^\ast}}{\wei{\prm B}}\right)^c \\
&\qquad \leq 1 + c \frac{\portalbound^c}{\eps^c} \sum_{B^\ast \in \Bb^r(B')} \frac{\wei{\prm B^\ast}}{\wei{\prm B}} (1-\nice)^{c-1} \left(\frac{\wei{\prm B'}}{\wei{\prm B}} \right)^{c-1} \\
&\qquad \leq 1 + c \left( \frac{\portalbound}{\eps} \cdot \frac{\wei{\prm B'}}{\wei{\prm B}} \right)^c
(1-\nice)^{c-1} \cdot C \\
&\qquad \leq c \left(\frac{\portalbound}{\eps} \cdot \frac{\wei{\prm B'}}{\wei{\prm B}}\right)^c.
\end{align*}
The last inequality follows for sufficiently large $c$ as
$$\left( \frac{\portalbound}{\eps} \cdot \frac{\wei{\prm B'}}{\wei{\prm B}} \right) > 1.$$
\end{proof}

By applying Lemma~\ref{lem:wei:portals:calls} to the root call $\mathtt{partition}(B)$ we obtain the
following:
\begin{corollary}\label{cor:wei:portals:calls}
In the entire run of the algorithm there are at most
$\poly(\eps^{-1},\portalbound)$ calls to $\mathtt{partition}(B')$
where $B'$ is large
\end{corollary}
As a single call to $\mathtt{partition}(B')$ takes $\Oh(|V(B')| \log |V(B')|)$ time, we have also that:
\begin{corollary}\label{cor:wei:portals:time}
Phase one takes $\poly(\eps^{-1},\portalbound) |B| \log |B|$ time.
\end{corollary}

We now bound the size and the length of the bricks in $\Bfin$.

\begin{lemma}\label{lem:wei:portals:length}
The sum of the lengths of the perimeters of all bricks in $\Bfin$
is bounded by $\poly(\eps^{-1},\portalbound)\wei{\prm B}$.
\end{lemma}
\begin{proof}
By Lemma~\ref{lem:wei:portals:shrink}, in each call $\mathtt{partition}(B')$
we have $\wei{\prm B'} \leq \wei{\prm B}$.
Consider a call $\mathtt{partition}(B')$
where $B'$ is large.
By Lemma~\ref{lem:wei:portals:perimeter}, the sum of
lengths of all perimeters of bricks $B^\ast \in \Bb^r(B')$
that are tiny (and hence will be inserted into $\Bfin$)
is bounded by $C\wei{\prm B'}$.
Moreover, if Theorem~\ref{thm:ananas} has been invoked,
  we have $\wei{\prm B^C} \leq \frac{16}{\nice^2} \wei{\prm B'}$.
Finally, by Corollary~\ref{cor:wei:portals:calls}, there are at most $\poly(\eps^{-1},\portalbound)$ calls $\mathtt{partition}(B')$ where $B'$ is large. The lemma follows.
\end{proof}

\begin{lemma}\label{lem:wei:portals:size}
The total number of edges and vertices in all bricks of $\Bfin$ is bounded
by $\poly(\eps^{-1},\portalbound)|B|$.
\end{lemma}
\begin{proof}
Consider a call to $\mathtt{partition}(B')$
where $B'$ is large.
First, observe that in this call at most one brick is put into $\Bfin$.
Moreover, observe that the total number of edges and vertices in all recursive calls
$\mathtt{partition}(B^\ast)$ for $B^\ast \in \Bb^r(B')$ is $\Oh(|B'|)$.
Here we rely on the fact that in the algorithm of Theorem~\ref{thm:nice-testing-wei},
each face of $B'$ is contained in at most $7$ bricks of $\Bb(B')$,
and, if the algorithm of Theorem~\ref{thm:ananas} has been invoked,
then $\Bb(B')$ is a brick partition of $B'$.
Finally, recall that if 
$B'$ is tiny, then
we simply put $B'$ into $\Bfin$. The bound of the lemma follows from Corollary~\ref{cor:wei:portals:calls}.
\end{proof}

\subsubsection{Phase two: constructing $H$ from the decomposition}

In the second phase we derive the output graph $H$ from the brick covering $\Bfin$.

Consider first a graph $H_0 := \bigcup_{B_1 \in \Bfin} \prm B_1$.
By Lemma~\ref{lem:wei:portals:cover}, $H_0$ is connected and contains $\prm B$.
Pick any finite face $f$ of $H_0$. As $H_0$ is connected, the interior of $f$ is homeomorphic to an open disc.
Moreover, since $H_0$ is a union of simple cycles, there is no bridge in $H_0$ and, hence, each edge of $H_0$ appears on the boundary of $f$ at most once
(but $H_0$ may have articulation points, and one vertex may appear multiple times on the boundary of $f$).

Let $C^f$ be the walk in $B$ around the boundary of $f$
and let $G^f$ be the subgraph of $B$ consisting of all edges of $B$ that lie in $f$ or on the boundary of $f$ (i.e., all edges of $B$ that are enclosed by $C^f$).
Moreover, construct a brick $B^f$ from $G^f$ by `straightening' the boundary $C^f$, that is, for each appearance of a vertex $v$ on $C^f$, make a separate copy of $v$ adjacent to all edges
that were adjacent to this appearance. Observe that there is a natural homomorphism $\pi^f$ from $B^f$ to $G^f$ that is bijective on the edge set of $B^f$ and surjective on the vertex set.

For each brick $B^f$, apply Theorem~\ref{thm:wei:paths} to obtain a graph $H^f$. Output $H := \bigcup_f \pi^f(H^f)$, where the union
ranges over all finite faces of $H_0$. It remains to show that $H$ has the properties desired by Theorem~\ref{thm:wei:portals} and can be computed in the desired time.

As $\prm B^f \subseteq H^f$ for each face $f$, we have that $C^f \subseteq H$ for each $f$ and, consequently, $\prm B \subseteq H$.
By Theorem~\ref{thm:wei:paths} and Lemma~\ref{lem:wei:portals:length}, there is a universal constant $\gamma$ such that:
\begin{eqnarray*}
\wei{H} &\leq & \sum_f \wei{H^f} \\
& \leq & \sum_f \gamma\eps^{-4} \wei{\prm B^f} \\
    &=& \sum_f \gamma\eps^{-4} \wei{C^f} \\
    & \leq& \gamma\eps^{-4} 2\wei{H_0} \\
    &\leq& 2\gamma\eps^{-4} \sum_{B_1 \in \Bfin} \wei{\prm B_1}\\
    & \leq& 2\gamma\eps^{-4} \cdot \poly(\eps^{-1},\portalbound) \wei{\prm B}\\
    &\leq& \poly(\eps^{-1},\portalbound) \wei{\prm B}.
\end{eqnarray*}
Therefore, $\wei{H}$ satisfies the desired bound.

The following lemma shows that $H$ preserves approximate Steiner forests for any choice of terminal pairs on the perimeter of $B$.

\begin{lemma}\label{lem:wei:portals:H-ok}
For every set $\mathcal{\terms} \subseteq V(\prm B) \times V(\prm B)$ of size at most $\portalbound$, there exists a Steiner forest $F_H$ that connects $\mathcal{\terms}$ in $H$ such that $\wei{F_H} \leq \wei{F_B} + 2 \eps \wei{\prm B}$ for any Steiner forest $F_B$ that connects $\mathcal{\terms}$ in $B$.
\end{lemma}
\begin{proof}
Let $F_B$ be a Steiner forest connecting $\mathcal{\terms}$ in $B$ of minimum possible length
that additionally satisfies the properties promised by Lemma~\ref{lem:wei:portals:one-correct}.
We construct a subgraph $F_H \subseteq H$ connecting $\mathcal{\terms}$ of length
at most $(1+\eps)\wei{F_B} + \eps \wei{\prm B}$. Since $\wei{F_B} \leq \wei{\prm B}$ (as $\prm B$ connects $\mathcal{\terms}$),
this would conclude the proof of the lemma.

First, construct a subgraph $F$ as follows. Start with $F = F_B$.
As long as there exists a vertex $v$ that is of degree at least three in $F$ and does not belong to $V(\prm B_1)$ for any $B_1 \in \Bfin$,
find any tiny $B_2 \in \Bfin$ such that $\prm B_2$ strictly encloses $v$, delete from $F$ all edges strictly enclosed by $\prm B_2$, add $\prm B_2$ instead, and take any spanning forest of the obtained graph. In this procedure we never introduce a vertex of degree at least three into $F$ that does not belong
to $V(H_0) = \bigcup_{B_1 \in \Bfin} V(\prm B_1)$, and hence such a tiny $B_2$ always exists by the properties of $F_B$ promised by Lemma~\ref{lem:wei:portals:one-correct}.
Moreover, as $|\mathcal{\terms}| \leq \portalbound$, $F_B$ contains at most $\portalbound$ vertices of degree at least three,
and in the construction of $F$ we made at most $\portalbound$ replacements.
Consequently,
$$\wei{F} \leq \wei{F_B} + \portalbound \cdot \frac{\eps}{\portalbound} \wei{\prm B} = \wei{F_B}+ \eps \wei{\prm B}.$$

Consider the graph $F \setminus H_0$. Recall that $F_B$ is a forest,
$F \setminus F_B \subseteq H_0$ (in the process of constructing $F$ we have only added edges of $H_0$ to $F$),
and each vertex of degree at least three in $F$ belongs to $V(H_0)$.
Consider the following relation on the edge set of $F \setminus H_0$: two edges $e_1,e_2$ are in relation
if and only if there exists a path in $F \setminus H_0$ that contains $e_1$ and $e_2$ and 
no internal vertex of this path belongs to $V(H_0)$.
Observe that this is an equivalence relation.
Moreover, as each vertex of degree at least three in $F$ belongs to $V(H_0)$, each equivalence class in this relation
is a path $P$ that connects two vertices of $V(H_0)$, but all internal vertices of $P$ do not belong to $V(H_0)$.

Let $\mathcal{P}$ be the family of equivalence classes of the aforementioned relation in $F \setminus H_0$.
For each path $P \in \mathcal{P}$, proceed as follows. As no edge and no internal vertex of $P$ belongs to $H_0$,
there exists a finite face $f$ of $H_0$ that contains $P$. Moreover, $(\pi^f)^{-1}(P)$ is a path in $B^f$, connecting
two vertices of $\prm B^f$. By the properties of $H^f$ (and in particular by Theorem~\ref{thm:wei:paths}), there exists a path $Q$ in $H^f$ connecting the same endpoints
and of length at most $(1+\eps)\wei{(\pi^f)^{-1}(P)} = (1+\eps)\wei{P}$. Hence, $\pi^f(Q)$ is a walk in $G^f$ connecting
the endpoints of $P$ of length at most $(1+\eps)\wei{P}$. To obtain a graph $F_H$, replace each $P$ with $\pi^f(Q)$ in the graph $F$.

By construction, $F_H \subseteq H$ and $F_H$ connects $\mathcal{\terms}$. Moreover, as each path $P \in \mathcal{P}$
has been replaced by a path of length at most $(1+\eps)\wei{P}$, we have that
$\wei{F_H} \leq (1+\eps)\wei{F_B} + \eps\wei{\prm B}$. This concludes the proof of the lemma.
\end{proof}

Observe that the lemma obtains an additive error $2\eps\wei{\prm B}$ instead of $\eps\wei{\prm B}$. The error of Theorem~\ref{thm:wei:portals} can be obtained by appropriately rescaling $\eps$
at the beginning of the algorithm.

Finally, observe that Lemma~\ref{lem:wei:portals:size} ensures that $H_0$ can be computed
in $\poly(\eps^{-1},\portalbound)|B|$ time, and, consequently, the graph $H$ can be computed in
in $\poly(\eps^{-1},\portalbound)|B| \log |B|$ time. This completes the proof of Theorem~\ref{thm:wei:portals}.

\subsection{Wrap up}\label{ss:wei:final}

We now pipeline the mortar graph construction of Borradaile et al~\cite{klein:planar-st-eptas}
with Theorem~\ref{thm:wei:portals} to conclude the proof of
Theorem~\ref{thm:weighted}.
In the language of brick coverings, the mortar graph construction of~\cite{klein:planar-st-eptas}
can be summarized as follows.
\begin{theorem}[\cite{klein:planar-st-eptas}, in particular Theorem 10.7]\label{thm:mortar}
Given a brick $B$ and an accuracy parameter $\eps > 0$, one can
in $\poly(\eps^{-1}) |B| \log |B|$ time compute a brick partition $\Bb$ of $B$
of total perimeter $(1+18\eps^{-1}) \wei{\prm B}$ such that the perimeter $\prm B'$ of each brick $B' \in \Bb$ can be partitioned into four paths $\north_{B'} \cup \west_{B'} \cup \south_{B'} \cup \east_{B'}$ (the so-called north, west, south, and east boundaries, appearing in this counter-clockwise order), such that:\begin{enumerate}
\item the total length of all parts $\west_{B'}$ and $\east_{B'}$ in all bricks of $\Bb$ is bounded
by $\eps \wei{\prm B}$; and
\item for any subgraph $F \subseteq B'$ of a brick $B' \in \Bb$, there exists 
a subgraph $F' \subseteq B'$ with the following properties:\label{p:mortar:replace}
\begin{enumerate}
\item $\wei{F'} \leq (1+c_1\eps)\wei{F}$ for some universal constant $c_1$;
\item there are at most $\alpha(\eps^{-1}) = o(\eps^{-5.5})$ vertices of 
$V(\north_{B'}) \cup V(\south_{B'})$ that are incident to an edge of $F'$ that does not belong to
$\north_{B'} \cup \south_{B'}$;
\item if two vertices of $V(\north_{B'}) \cup V(\south_{B'})$ are connected by $F$, then they
are also connected by $F'$.
\end{enumerate}
\end{enumerate}
\end{theorem}
The algorithm of Theorem~\ref{thm:weighted} for a given brick $B'$ and accuracy parameter $\eps > 0$ can now be described as follows.
First, we compute the brick partition $\Bb$ of Theorem~\ref{thm:mortar}
for the parameter $\eps$ and brick $B$.
Second, for each $B' \in \Bb$, we invoke Theorem~\ref{thm:wei:portals}
for the brick $B'$, accuracy parameter $\eps' := \eps/(1+18\eps^{-1})$
and bound $\portalbound = (\alpha(\eps^{-1})+4)^2$.
Let $H(B')$ be the obtained subgraph for the brick $B'$.
We output $H = \bigcup_{B' \in \Bb} H(B')$. 

It remains to prove that $H$ has the properties desired by Theorem~\ref{thm:weighted} and can be computed in the desired time.
Clearly, $\prm B \subseteq H$.
By the bounds of Theorem~\ref{thm:wei:portals} and the fact that
$\alpha(\eps^{-1}) = o(\eps^{-5.5})$ we have that $\wei{H} \leq \poly(\eps^{-1}) \wei{B}$.
Moreover, as $\Bb$ is a brick partition, all calls to the algorithm
of Theorem~\ref{thm:wei:portals} run in total in $\poly(\eps^{-1}) |B| \log |B|$ time, and the
time bound of Theorem~\ref{thm:weighted} follows.
It remains to argue that $H$ preserves approximate Steiner forests
for terminals on the perimeter of $B$.

To this end, consider any $\mathcal{\terms} \subseteq V(\prm B) \times V(\prm B)$ and let
$F$ be a Steiner forest connecting $\mathcal{\terms}$ in $B$ of minimum possible length.
First, define $F_1 := F \cup \bigcup_{B' \in \Bb} \west_{B'} \cup \east_{B'}$ and
observe that $\wei{F_1} \leq \wei{F} + \eps \wei{\prm B}$ by point 1 of Theorem~\ref{thm:mortar}.
Then, for each $B' \in \Bb$ proceed as follows. Let $F_1(B')$ be the subgraph of $F_1$
consisting of all edges strictly enclosed by $\prm B'$.
Let $F_2(B')$ be the subgraph promised by 
point~\ref{p:mortar:replace} of Theorem~\ref{thm:mortar} for the subgraph
$F_1(B') \cup \west_{B'} \cup \east_{B'}$ of $B'$.
Define 
$$F_2 = \left(F_1 \setminus \bigcup_{B' \in \Bb} F_1(B')\right) \cup \bigcup_{B' \in \Bb} F_2(B').$$
By Theorem~\ref{thm:mortar}, we have
$$\wei{F_2(B')} \leq (1+c_1\eps)(\wei{F_1(B')} + \wei{\west_{B'}} + \wei{\east_{B'}}).$$
Hence, for some universal constant $c_2$,
$$\wei{F_2} \leq (1+c_1\eps)\wei{F_1} + (1+c_1\eps)\eps\wei{\prm B} \leq \wei{F_1} + c_2\eps \wei{\prm B}.$$

Observe that $\west_{B'},\east_{B'} \subseteq F_2$ for any $B' \in \Bb$.
For each $B' \in \Bb$, we now proceed as follows.
Define $F_2'(B')$ to be the subgraph of $F_2$ consisting of all edges strictly
enclosed by $\prm B'$; observe that $F_2'(B') \subseteq F_2(B')$.
Define $\mathcal{\terms}(B')$ to be the set of pairs $(x,y)$
for which $x,y \in V(\north_B) \cup V(\south_B)$, $x \neq y$,
and $x,y$ are in the same connected component of $F_2'(B') \cup \west_{B'} \cup \east_{B'}$.
Observe that if $(x,y) \in \mathcal{\terms}(B')$, then $x$ (and similarly $y$)
is an endpoint of $\north_{B'}$, an endpoint of $\south_{B'}$ or an endpoint
of an edge of $F_2'(B') \subseteq F_2(B')$ that is strictly enclosed by $\prm B'$.
By Theorem~\ref{thm:mortar} and our choice of $\portalbound$,
$|\mathcal{\terms}(B')| \leq \portalbound$.
Hence, by Theorem~\ref{thm:wei:portals},
  there exists a subgraph $F_3(B')$ that connects $\mathcal{\terms}(B')$ in $B'$,
  is contained in $H(B')$, and is of length
$$\wei{F_3(B')} \leq \wei{F_2'(B')} + \wei{\west_{B'}} + \wei{\east_{B'}} + \frac{\eps}{1+18\eps^{-1}} \wei{\prm B'}.$$
Define
$$F_3 = \left(F_2 \setminus \bigcup_{B' \in \Bb} F_2'(B')\right) \cup \bigcup_{B' \in \Bb} F_3(B').$$
As $\sum_{B' \in \Bb} \wei{\prm B'} \leq (1+18\eps^{-1}) \wei{\prm B}$
and $\sum_{B' \in \Bb} \wei{\west_{B'}} + \wei{\east_{B'}} \leq \eps\wei{\prm B}$,
we have that $\wei{F_3} \leq \wei{F} + c_3\eps \wei{\prm B}$ for some universal constant $c_3$.
Moreover, by construction $F_3 \subseteq H$.

We now argue that $F_3$ connects $\mathcal{\terms}$. 
As $F$ connects $\mathcal{\terms}$, so does $F_1$.
To analyse $F_2$ and $F_3$, we introduce the following notion:
for any $B' \in \Bb$ and $x \in V(\prm B')$, we set $\widehat{x}$
to be the common endpoint of $\north_{B'}$ and $\west_{B'}$ if $x \in V(\west_{B'})$,
the common endpoint of $\north_{B'}$ and $\east_{B'}$ if $x \in V(\east_{B'})$,
and $\widehat{x} = x$ otherwise.
Observe that 
if $x,y \in V(\prm B')$ are connected by $F_1(B')$,
then $\widehat{x}$ and $\widehat{y}$ are connected by
$F_1(B') \cup \west_{B'} \cup \east_{B'}$ and, consequently,
$\widehat{x}$ and $\widehat{y}$ are also connected by $F_2(B')$.
Moreover, an identical claim is true for $F_1(B')$ replaced by $F_2'(B')$
and $F_2(B')$ replaced by $F_3(B')$. As all west and east boundaries
of all bricks of $\Bb$ belong to $F_1$, $F_2$ and $F_3$, we infer that
$F_3$ indeed connects $\mathcal{\terms}$.
By taking $\eps/c_3$ instead of $\eps$ at the beginning of the algorithm,
   Theorem~\ref{thm:weighted} follows.

%!TEX root = pst-kernel.tex

\section{Applications: \pST{}, \pSF{} and \pemwcname{}}\label{sec:applications}

In this section we apply Theorem \ref{thm:main} to obtain polynomial kernels
for \pST{}, \pSF{} (parameterized by the number of edges in the tree or forest)
and \pemwcname{} (parameterized by the size of the cutset).
The applications to \pST{} and \pSF{} are rather straightforward, and rely
on the trick from the EPTAS~\cite{klein:planar-st-eptas} to cut the graph open
along an approximate solution.
For \pemwcname{} we need some more involved arguments to bound the diameter
of the dual of the input graph, before we apply Theorem \ref{thm:main}.

In all aforementioned problems, we consider the --- maybe more practical or natural ---
optimization variants of the problem, instead of the decision ones. That is,
we assume that the algorithm does not get the bound on the required tree, forest
or cut, but instead is required to kernelize the instance with respect to the
(unknown) optimum value. However, note that in all three considered problems
an easy approximation algorithm is known, and the output of such an algorithm
will be sufficient for our needs.

We also note that we do not care much about optimality of the exponents in the sizes
of the kernels, as any application Theorem \ref{thm:main} immediately
raises the exponents to the magnitude of hundreds. The main result
of our work is the existence of polynomial kernels, not the actual sizes.

\subsection{\pST{} and \pSF{}}\label{sec:pst-psf}
For both problems, we can apply the known trick of cutting open the graph along an approximate solution~\cite{klein:planar-st-eptas}, which when combined with Theorem~\ref{thm:main} gives the kernel.

\begin{theorem}[Theorem~\ref{thm:pst-intro} repeated] \label{thm:pst}
Given a \pST{} instance $(G,\terms)$,
one can in
$\Oh(k_{OPT}^{142} |G|)$ time find a set $F \subseteq E(G)$
of $\Oh(k_{OPT}^{142})$ edges that contains an optimal Steiner tree
connecting $\terms$ in $G$, where $k_{OPT}$ is the size of an optimal
Steiner tree.
\end{theorem}

\begin{proof}
We first manipulate the graph such that all terminals lie on the outer face. To do this, we find a $2$-approximate Steiner tree $\Tapx$ for $\terms$ in $G$ in the following way.
We run a breadth-first search in $G$ from each terminal in $\terms$ to determine
a shortest path between each pair of the terminals. This takes $\Oh(|\terms||G|) = \Oh(k_{OPT}|G|)$ time. Define an auxiliary complete graph $G'$ over $\terms$, where the length of an edge between two terminals is the length of the shortest path between these two terminals that we computed earlier. We then compute a minimum spanning tree in $G'$. This tree induces a Steiner tree in $G$, which is $2$-approximate. Note that $k_{OPT} \leq |\Tapx| \leq 2k_{OPT}$.

We now cut the plane open along tree $\Tapx$, cf.~\cite{klein:planar-st-eptas} (see Figure~\ref{fig:cutopen}). That is, we create an Euler tour of $\Tapx$ that traverses each edge twice in different directions, and respects the plane embedding of $\Tapx$. Then we duplicate every edge of $\Tapx$, replace each vertex $v$ of $\Tapx$ with $d-1$ copies of $v$, where $d$ is the degree of $v$ in $\Tapx$, and distribute the copies in the plane embedding so that we obtain a new face $F$ whose boundary corresponding to the aforementioned Euler tour.
Then fix an embedding of the resulting graph $\nG$ that has $F$ as its outer face.
Observe that there exists a natural mapping $\pi$ from $E(\nG)$ to $E(G)$, i.e., edges in $\nG$ are
mapped to edges from which they where obtained. Moreover, note that the terminals
$\terms$ lie only on the outer face of $\nG$, and that $|\prm \nG| \le 4k_{OPT}$.

\begin{figure}
\centering
\includegraphics[width=.6\linewidth]{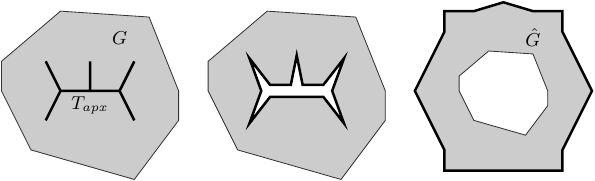}
\caption{(Figure~\ref{fig-intro:cutopen} repeated) The process of cutting open the graph $G$ along the tree $\Tapx$.}
\label{fig:cutopen}
\end{figure}

Finally, we obtain the kernel. Apply Theorem~\ref{thm:main} to $\nG$ to obtain a subgraph $\hat{H}$, which has size
$\Oh(|\prm \nG|^{142}) = \Oh(k_{OPT}^{142})$. Let $F = \pi(\hat{H})$. We show that $F$ is a kernel for $(G,\terms)$.
Clearly, $|\pi(\hat{H})|\leq |\hat{H}|\leq \Oh(k_{OPT}^{142})$. Let $T$ be an optimal Steiner tree in $G$ for $\terms$ and
consider $\pi^{-1}(T)$. If $\pi^{-1}(T)$ contains edges $e'$ and $e''$ for which there
exists an edge $e\in G$ such that $\pi(e')=\pi(e'')=e$, then arbitrarily remove either $e'$ or $e''$.
Let $\hat{T}$ denote the resulting graph. By construction, $|T| = |\hat{T}|$.
Observe that any connected component $C$ of $\hat{T}$ is a connector for $V(C) \cap V(\prm \nG)$.
Hence, there exists an optimal Steiner tree $T_C$ in $\hat{H}$ that connects $V(C) \cap V(\prm \nG)$.
Let $\hat{T}_H$ be the graph that is obtained from $\hat{T}$ by replacing $C$ with $T_C$ for each connected component of
$\hat{T}$. Observe that during each such replacement, $\pi(\hat{T}_H)$ remains connected, because $T$ was connected.
Again, by construction, $|\hat{T}_H| \le |\hat{T}|$. Now
observe that $\pi(\hat{T}_H)$ is a subgraph of $\pi(\hat{H})$ connecting $\terms$ in $G$, of not higher cost than $T$.
\end{proof}

For \pSF{}, we need to slightly preprocess the input instance, removing some obviously
unnecessary parts, to bound the diameter of each connected component.

\begin{theorem}[Theorem~\ref{thm:psf-intro} repeated]\label{thm:psf}
Given a \pSF{} instance $(G,\termpairs)$,
one can in
$\Oh(k_{OPT}^{710} |G|)$ time find a set $F \subseteq E(G)$
of $\Oh(k_{OPT}^{710})$ edges that contains an optimal Steiner forest
connecting $\termpairs$ in $G$, where $k_{OPT}$ is the size of an optimal
Steiner forest.
\end{theorem}
\begin{proof}
Let $(G,\termpairs)$ be a \pSF{} instance.
A forest with $k_{OPT}$ edges has at most $2k_{OPT}$ vertices, and thus 
%$k_{OPT} = \Omega(|\termpairs|^{1/2})$, if we assume $\termpairs$ does not contain
%duplicates. Hence,
$|\termpairs| = \Oh(k_{OPT}^2)$.
We construct an approximate solution $T_1$, by taking a union of shortest $s_1s_2$-paths
for all $(s_1,s_2) \in \termpairs$. Clearly,
$k_{OPT} \leq |T_1| \leq |\termpairs|k_{OPT} = \Oh(k_{OPT}^3)$.
Let $k_1 = |T_1|$.

We remove from $G$ all vertices (and incident edges) that are at distance
more than $k_1$ from all terminals of $\termpairs$.
Clearly, no such vertices or edges are used
in a minimal solution for $(G,\termpairs)$ with
at most $k_1$ edges.

Consider each connected component of $G$ separately.
Let $G_0$ be a component of $G$ and let $\terms_0$ be the family of terminals of $\termpairs$
in $G_0$. In $\Oh(|\terms_0| \cdot |G_0|)$ time, we construct a $2$-approximate
Steiner {\em{tree}} $T_0$ connecting $\terms_0$ in $G_0$.
Note that, as each vertex of $G_0$ is within a distance at most $k_1$ from $\terms_0$,
we have $|T_0| = \Oh(|\terms_0| k_1)$.
As in the proof of Theorem \ref{thm:pst}, cut the graph $G_0$ open along
$T_0$, obtaining a brick $\nG_0$ of perimeter $|\prm \nG_0| = \Oh(|\terms_0| k_1)$.
Then apply the algorithm of Theorem \ref{thm:main} to $\nG_0$, obtaining
a subgraph $\hat{H}$. Finally, put the edges of $G_0$ that correspond to $\hat{H}$ into
the constructed subgraph $F$. By similar arguments as in the proof
of Theorem \ref{thm:pst}, $F$ contains a minimum Steiner forest for $(G,\termpairs)$.
The time bound and the bound on $|F|$ follows from the bound
$k_1 = \Oh(k_{OPT}^3)$ and
the fact that the union of all sets $\terms_0$ has size $2|\termpairs| = \Oh(k_{OPT}^2)$.
\end{proof}
We observe that the size of the kernel can be improved to $\Oh(k_{OPT}^{426})$ by running a constant-factor approximation algorithm for \pSF{} to construct the forest $T_{1}$. However, when using the EPTAS for \pSF{}~\cite{klein-efficient-psf}, this makes the algorithm run in $\Oh(k_{OPT}^{426} |G| + |G| \log^{3} |G|)$ time, which is no longer linear in $|G|$.

Another observation is that the size of the kernel can be improved if we consider a `classic' kernel. That is, a kernel for the decision variant of the problem: does the planar graph $G$ have a Steiner forest of size at most $k$? Then we can use $k$ instead of $k_{1}$ in the above proof and return a kernel of size $\Oh(k^{426})$ in $\Oh(k^{426} |G|)$ time.

%!TEX root = pst-kernel.tex

\newcommand{\reach}[2]{\mathtt{reach}(#1,#2)}
\newcommand{\ut}{\hat{t}}
\newcommand{\rel}{\mathcal{R}}
\newcommand{\dupaterms}{\widehat{\terms}}

\subsection{\pemwcname{}}\label{sec:emwc}

We are left with the case of \pemwcname{}.

\begin{theorem}[Theorem~\ref{thm:emwc-intro} repeated]\label{thm:emwc}
Given a \pemwcname{} instance $(G,\terms)$,
one can in polynomial time find a set $F \subseteq E(G)$
of $\Oh(k_{OPT}^{568})$ edges that contains an optimal solution
to $(G,\terms)$, where $k_{OPT}$ is the size of this optimal solution.
\end{theorem}

The idea of the kernel is that the \pemwc{} problem
is some sort of \textsc{Steiner Forest}-like problem
in $G^\ast$, the dual of $G$.
However, to apply Theorem \ref{thm:main},
we need to cut $G^\ast$ open so that Theorem \ref{thm:main}
can be applied to the brick created by this cutting.
To bound the perimeter of this brick, it suffices
to bound the diameter of $G^\ast$.
This is done in Section \ref{sec:emwc:diam},
via a separate reduction rule.
Earlier, in Section \ref{sec:emwc:prelims},
we perform a few (well-known) regularization reductions on the input graph.
Finally, in Section \ref{sec:emwc:finish}, we show formally
how to cut open $G^\ast$
and apply Theorem \ref{thm:main} to obtain the promised
kernel.

Note that, contrary to the case of \pST{} and \pSF{},
the preprocessing for \pemwc{} takes superlinear time, in terms of $|G|$.

In the rest of this section we assume that
$(G,\terms)$ is an input to \pemwc{} that we aim to kernelize.
Note that, contrary to the previous sections, $G$ may contain
multiple edges.
We fix some planar embedding of $G$, where multiple edges
are drawn in parallel in the plane, without any other
element of $G$ between them.

In the course of the kernelization algorithm, we may perform
two types of operations on $G$.
First, if we deduce for some $e \in E(G)$ that
there exists a minimum solution $X$ not containing $e$,
then we may contract $e$ in $G$.
During this contraction, any self-loops are removed, but multiple
edges are kept.
This operation is safe, because if $F$ is a subgraph of $G/e$ that has the properties promised by Theorem~\ref{thm:emwc}, then the projection of $F$ into $G$ satisfies those same properties.
Second, if we deduce for some edge $e$ that some minimum solution
$X$ to \pemwc{} on $(G,\terms)$ contains $e$, we may delete $e$ from $G$,
analyze $G \setminus \{e\}$ obtaining a set $F$, and return $F \cup \{e\}$.
As the size of the minimum solution to \pemwc{} decreases
in $G \setminus e$, the size of $F$ satisfies the bound promised in Theorem \ref{thm:main}. Note that both edge contractions and edge deletions preserve planarity of $G$.

In the course of the arguments, we provide a number of reduction rules.
At each step, the lowest-numbered applicable rule is used.

\subsubsection{Preliminary reductions}\label{sec:emwc:prelims}

In this section, we provide a few reduction rules to clean up the instance.

\begin{reduction}\label{red:emwc:2terms}\label{red:emwc:first}
If there is an edge $e$ that connects two terminals, then delete $e$
and include it into the constructed set $F$.
\end{reduction}
\begin{reduction}\label{red:emwc:empty-cc}
If $|\terms| \leq 1$, then return $F = \emptyset$.
\end{reduction}

Now, we take care of the situation when the input instance $(G,\terms)$
is in fact a union of a few \pemwc{} instances.

\begin{reduction}\label{red:emwc:cc}
If $G \setminus \terms$ is not a connected graph, then consider each of its connected component separately.
That is, if $C_1,C_2,\ldots,C_s$ are connected components of $G \setminus \terms$, separately
run the algorithm on instances $I_i = (G[C_i \cup N_G(C_i)], N_G(C_i))$ for $i=1,2,\ldots,s$, obtaining
sets $F_1,F_2,\ldots,F_s$. Return $F = \bigcup_{i=1}^s F_i$.
\end{reduction}
To see that Rule \ref{red:emwc:cc} is safe, first note that
since $G[\terms]$ is edgeless (as Rule~\ref{red:emwc:first} has been performed exhaustively), the instances $(I_i)_{i=1}^s$ partition
the edge set of $G$.
Consequently, any path connecting two terminals in $G$, without
any internal vertex being a terminal, is completely contained in one instance $I_i$.
Hence, a minimum solution to $(G,\terms)$ is the union of minimum solutions
to the instances $(I_i)_{i=1}^s$, and thus
if $k_{OPT}$ is the size of an minimum solution to $(G,\terms)$ and $k_{i,OPT}$
is the size of an minimum solution to $I_i$, then $k_{OPT} = \sum_{i=1}^s k_{i,OPT}$.
Moreover, if $|F_i| \leq c k_{i,OPT}^{568}$ for some constant $c > 0$,
then $|F|  \leq ck_{OPT}^{568}$, as the function $x \mapsto cx^{568}$ is convex.

Therefore, in the rest of this section we may assume that $G \setminus \terms$ is connected.

%We assume that $G$ is connected: if this is not the case,
%we indepentently kernelize each connected component of $G$.
%As a minimum solution to \pemwc{} on $G$ consists
%of minimum solutions in each connected component of $G$,
%and the bound promised in Theorem \ref{thm:emwc} is convex with respect
%to $k_{OPT}$, the union of the sets $F$ obtained for each connected
%component satisfy the properties of Theorem \ref{thm:emwc} for the entire graph $G$.

We now introduce some notation with regards to cuts in a graph.
For two disjoint subsets $A,B \subseteq V(G)$
we say that $X \subseteq E(G)$ is a {\em{$(A,B)$-cut}}
if no connected component of $G \setminus X$ contains
both a vertex of $A$ and a vertex of $B$.
For $A = \{a\}$ or $B = \{b\}$ we shorten this notion to
$(a,b)$-cut, $(a,B)$-cut and $(A,b)$-cut.
An $(A,B)$-cut $X$ is {\em{minimal}} if no proper subset of $X$
is an $(A,B)$-cut, and {\em{minimum}} if $|X|$ is minimum possible.
For $X \subseteq E(G)$ and $A \subseteq V(G)$
we define $\reach{A}{X}$ as the set of those vertices $v \in V(G)$
that are contained in a connected component of $G \setminus X$
with at least one vertex of $A$. Note that $X$ is a $(A,B)$-cut
if and only if $\reach{A}{X} \cap \reach{B}{X} = \emptyset$,
and $X$ is a minimal $(A,B)$-cut if additionally each edge
of $X$ has one endpoint in $\reach{A}{X}$, and second endpoint
in $\reach{B}{X}$. For a vertex $t$, we write $\reach{t}{X}$ instead of $\reach{\{t\}}{X}$.
For any $Q \subseteq V(G)$, we define $\delta(Q)$
as the set of edges of $G$ with exactly one endpoint in $Q$.
Note that if $A \subseteq Q$ and $B \cap Q = \emptyset$,
then $\delta(Q)$ is a $(A,B)$-cut. Moreover,
if $X$ is a $(A,B)$-cut then $\delta(\reach{A}{X}) \subseteq X$
and if $X$ is a minimal $(A,B)$-cut
then $\delta(\reach{A}{X}) = X$.

This section relies on the submodularity of
the cut function $\delta()$:
\begin{lemma}[submodularity of cuts~\cite{gomory-hu}]
For any $P,Q \subseteq V(G)$ it holds that:
$$|\delta(P)| + |\delta(Q)| \geq |\delta(P \cup Q)| + |\delta(P \cap Q)|.$$
\end{lemma}

From the submodularity of cuts we infer that if
$X$ and $Y$ are minimum $(A,B)$-cuts, then $\delta(\reach{A}{X} \cup \reach{A}{Y})$ and $\delta(\reach{A}{X} \cap \reach{A}{Y})$ are minimum
$(A,B)$-cuts as well. Therefore, there exists a unique minimum
$(A,B)$-cut $K$ with inclusion-wise maximal $\reach{A}{K}$.
We call this cut {\em{the minimum $(A,B)$-cut furthest from $A$}}.
Moreover, this cut can be computed in polynomial time (see for example~\cite{Marx06}).

The submodularity of cuts also yields the following
known reduction rule (cf.~\cite{mwc-alg-4k}).
\begin{reduction}\label{red:emwc:behind-cut}
For all $t \in \terms$, let $K_t$ be the minimum $(t,\terms\setminus \{t\})$-cut
furthest from $t$. If
$K_t \neq \delta(t)$ for some $t \in \terms$, then contract all edges with both
endpoints in $\reach{t}{K_t}$ (i.e., contract $\reach{t}{K_t}$ onto $t$).
\end{reduction}
Clearly, Reduction \ref{red:emwc:behind-cut}
can be applied in polynomial time. Note that 
if this rule is not applicable, then $\delta(\ut)$ is the unique minimum $(\ut,\terms \setminus \{\ut\})$-cut.
For completeness, we provide the proof of its safeness.

\begin{lemma}\label{lem:emwc:behind-cut:correct}
Let $K_t$ be the minimum $(t,\terms\setminus \{t\})$-cut furthest from $t$.
Then there exists a minimum solution to $(G,\terms)$ that does not contain
any edge with both endpoints in $\reach{t}{K_t}$.
\end{lemma}
\begin{proof}
Let $X$ be a minimum solution of $(G,\terms)$.
Let $P = \reach{t}{X}$ and $Q = \reach{t}{K_t}$. Note that $P \cap \terms = \{t\}$
and, consequently, $\delta(P)$ is a $(t,\terms \setminus \{t\})$-cut.
By submodularity of the cuts, $|\delta(P \cup Q)| + |\delta(P \cap Q)| \leq |\delta(P)| + |\delta(Q)|$.
As $K_t$ is a minimum $(t,\terms \setminus \{t\})$-cut, $|\delta(P \cap Q)| \geq |\delta(Q)|$ and, consequently,
   $|\delta(P \cup Q)| \leq |\delta(P)|$.
We infer that, if we define
$$Y := (X \setminus (E(G[Q]) \cup \delta(P))) \cup \delta(P \cup Q),$$
we have $|Y| \leq |X|$, as $\delta(P) \subseteq X$.

We claim that $Y$ is a solution to $(G,\terms)$; as $|Y| \leq |X|$ and $Q \subseteq \reach{t}{Y}$, this would finish the proof of the lemma.
Assume otherwise, and let $R$ be a path between two terminals in $G \setminus Y$. As $X$ is a solution to $(G,\terms)$, $R$
contains an edge of $\delta(P)$ or a vertex of $Q$,  and, consequently, contains a vertex of $P \cup Q$. Note that at least one endpoint of $R$ is different than $t$;
hence, $R$ contains an edge of $\delta(P \cup Q)$, a contradiction, as $\delta(P \cup Q) \subseteq Y$.
\end{proof}

We now recall that the set of all minimum
$t-(\terms \setminus \{t\})$ cuts is a $2$-approximation for \pemwc{} (cf.~\cite{dahlhaus}).

\begin{lemma}\label{lem:emwc-2apx}
If Rule \ref{red:emwc:behind-cut} is not applicable to $(G,\terms)$, then
$\bigcup_{t \in \terms} \delta(t)$ is a solution to $(G,\terms)$
of size at most $2k_{OPT}$.
\end{lemma}
\begin{proof}
Observe that $\bigcup_{t \in \terms} \delta(t)$ is indeed a solution. It remains to prove the bound.
Let $X$ be a solution to $(G,\terms)$.
Note that
for each $t \in \terms$,
the set $\delta(\reach{t}{X})$ is a $(t,\terms \setminus \{t\})$-cut in $G$.
Consequently, $|\delta(\reach{t}{X})| \geq |\delta(t)|$.
On the other hand, each edge $e \in X$ belongs to $\delta(\reach{t}{X})$
for at most two terminals $t \in \terms$. Hence,
$$2k_{OPT} = 2|X| \geq \sum_{t \in \terms} |\delta(\reach{t}{X})| \geq \sum_{t \in \terms} |\delta(t)| \geq \left|\bigcup_{t \in \terms} \delta(t)\right|,$$
and the lemma follows.
\end{proof}
We infer that, once Rule \ref{red:emwc:behind-cut} is exhaustively applied,
$k := |\bigcup_{t \in \terms} \delta(t)|$ satisfies $k_{OPT} \leq k \leq 2k_{OPT}$.

We now state the last clean-up rule.

\begin{reduction}\label{red:emwc:multiple-edge}\label{red:emwc:last-prelim}
If there is an edge $e$ of multiplicity larger than $k$, then contract $e$.
\end{reduction}

\subsubsection{Bounding the diameter of the dual}\label{sec:emwc:diam}

We are now ready to present a reduction rule that bounds the diameter
of the dual of $G$. Recall that we assume that $G$ is connected.

Arbitrarily, pick one terminal $\ut \in \terms$. We construct a sequence
of $(\ut,\terms \setminus \{\ut\})$-cuts $K_1,K_2,\ldots,K_r$ as follows.
We start with $K_1 = \delta(\ut)$; recall that, once Rule \ref{red:emwc:behind-cut}
is not applicable, $\delta(\ut)$ is the unique minimum $(\ut,\terms \setminus \{\ut\})$-cut.
Having constructed $K_i$, we proceed as follows. If there exists an edge in $K_i$
that is not incident to a terminal in $\terms \setminus \{\ut\}$, we pick
one such edge $uv$ arbitrarily and take
$K_{i+1}$ to be the minimum $(\reach{\ut}{K_i} \cup \{u,v\},\terms \setminus \{\ut\})$-cut
furthest from $\reach{\ut}{K_i} \cup \{u,v\}$. Otherwise, we terminate the process.
Note that the sequence $K_1,K_2,\ldots,K_r$ can be computed in polynomial time.

We note the following properties of the sequence $K_1,K_2,\ldots,K_r$.
\begin{lemma}\label{lem:emwc:Ki-props}
If Rules \ref{red:emwc:first}--\ref{red:emwc:last-prelim} are not applicable,
   then the following holds:
\begin{enumerate}
\item $K_r = \bigcup_{t \in \terms \setminus \{\ut\}} \delta(t)$;
\item $1 \leq |\delta(\ut)| = |K_1| < |K_2| < \ldots < |K_r| < 2k_{OPT}$;
\item $r < 2k_{OPT}$;
\item for each $1 \leq i < r$, $\reach{\ut}{K_i} \subsetneq \reach{\ut}{K_{i+1}}$.
\end{enumerate}
\end{lemma}
\begin{proof}
We first show that when $K_i \neq \bigcup_{t \in \terms \setminus \{\ut\}} \delta(t)$, for some $i$, then
$K_{i+1}\neq K_i$. 
As $\bigcup_{t \in \terms \setminus \{\ut\}}  \delta(t)$ is a $(\ut,\terms \setminus \{\ut\})$-cut,
and $K_i$ is a minimal $(\ut,\terms \setminus \{\ut\})$-cut, we infer that there exists an edge $vt \notin K_i$
incident to a terminal $t \neq \ut$. As Rule \ref{red:emwc:cc} is not applicable, $G \setminus \terms$ is connected
and thus there exists a $\ut v$-path $Q$, such that only the first edge of $Q$ is incident to a terminal.
We infer that $Q$ intersects $K_i$, and $K_i$ contains an edge not incident to $\terms \setminus \{\ut\}$.
Consequently, $K_{i+1}$ can be constructed. This concludes the proof of the first claim.

For the second claim, note that $K_i$ is the unique minimum $(\reach{\ut}{K_i},\terms \setminus \{\ut\})$-cut, thus $|K_{i+1}| > |K_i|$ for all $1 \leq i < r$.
By Lemma \ref{lem:emwc-2apx}, $|\bigcup_{t \in \terms} \delta(t)| \leq 2k_{OPT}$.
As Rule \ref{red:emwc:2terms} is not applicable, the sets $\delta(t)$
are pairwise disjoint. As Rules \ref{red:emwc:empty-cc} and \ref{red:emwc:cc} are not applicable,
$\delta(\ut) \neq \emptyset$.
We infer that
$$|K_r| =
\left|\bigcup_{t \in \terms \setminus \{\ut\}} \delta(t)\right| <
\left|\bigcup_{t \in \terms} \delta(t)\right| \leq 2k_{OPT}.$$

The third claim follows directly from the second one, and the last claim is straightforward
from the construction.
\end{proof}

The main claim of this section is the following.
\begin{lemma}\label{lem:emwc:Ki-dist}
Assume Rules \ref{red:emwc:first}--\ref{red:emwc:last-prelim} are not applicable
to the \pemwc{} instance $(G,\terms)$.
Moreover, assume there exists an edge $e \in G$
such that the distance, in the dual of $G$,
between $e$ and $\bigcup_{i=1}^r K_i$ is greater than $k$.
Then there exists a minimum solution to $(G,\terms)$ that does not contain $e$.
\end{lemma}
\begin{proof}
Let $X$ be a minimum solution to $(G,\terms)$. If $e \notin X$, there is nothing
to prove, so assume otherwise.
As $e$ is distant from $\bigcup_{i=1}^r K_i$, in particular
$e \notin \bigcup_{i=1}^r K_i$.
Recall that, since we assume $G$ is connected,
$$\{\ut\} = \reach{\ut}{K_1} \subsetneq \reach{\ut}{K_2} \subsetneq \ldots \subsetneq \reach{\ut}{K_r} = V(G) \setminus (\terms \setminus \{\ut\}).$$
Hence, there exists a unique index $\iota$, $1 \leq \iota < r$, such that both endpoints of $e$
belong to $\reach{\ut}{K_{\iota+1}} \setminus \reach{\ut}{K_\iota}$.

Consider now $X$ as an edge subset of the dual of $G$, and let $Y$ be the connected component of $X$ that contains $e$. Let
$\terms_Y = \terms \setminus \reach{\ut}{Y}$, i.e., $\terms_Y$ is the set of terminals
separated in $G$ from $\ut$ by $Y$.
Finally, we define $\overline{Y}$ to be the set of edges of $G$ that are incident to
a face of $G$ that is incident to at least one edge of $Y$, i.e., the set of edges
that are incident to the endpoints of $Y$ in the dual of $G$.

We first claim the following.
\begin{claim}\label{cl:emwc:Yiota}
$\overline{Y}$ is a connected subgraph of $G$,
disjoint from $\bigcup_{i=1}^r K_i$ and
the endpoints of $\overline{Y}$ in $G$ lie in $\reach{\ut}{K_{\iota+1}} \setminus \reach{\ut}{K_{\iota}}$.
\end{claim}
\begin{proof}
Since $G$ is connected, the edges incident to a face of $G$ form a closed walk,
and, consequently, $\overline{Y}$ is a connected subgraph of $G$.
As $Y \subseteq X$, $|Y| \leq |X| = k_{OPT} \leq k$. Hence, any face incident
to an edge of $Y$ is, in the dual of $G$,
within distance less than $k$ from a face incident to $e$. Consequently,
by the definition of $\overline{Y}$ and the choice of $e$,
$\overline{Y}$ cannot contain any edge of $\bigcup_{i=1}^r K_i$.
By the connectivity of $\overline{Y}$, for any $1 \leq i \leq r$, $\overline{Y}$
is either fully contained in $G[\reach{\ut}{K_i}]$ or fully contained in
$G \setminus \reach{\ut}{K_i}$. Hence,
the last claim follows from the definition of $\iota$.
\cqed\end{proof}

Intuitively, Claim \ref{cl:emwc:Yiota} asserts that $Y$ is a connected part of the solution
that lives entirely between $K_\iota$ and $K_{\iota + 1}$. The role of $Y$
in the solution $X$ is to separate $\terms_Y$ from $\ut$ (and/or
other terminals of $\terms \setminus \terms_Y$), and, possibly,
separate some subsets of $\terms_Y$ from each other.
Define $Z$ to be the set of those edges of $K_{\iota+1}$ whose endpoints
are separated from $\ut$ by $Y$, i.e., both do not belong to $\reach{\ut}{Y}$.
Note that, as $\overline{Y} \cap K_{\iota+1} = \emptyset$, for any $e' \in K_{\iota+1}$,
either both endpoints of $e'$ belong or both endpoints do not belong to $\reach{\ut}{Y}$.

\begin{claim}\label{cl:emwc:Y-in-iota}
$K := (K_{\iota+1} \setminus Z) \cup Y$ is a $(\ut,\terms \setminus \{\ut\})$-cut.
Moreover, $\reach{\ut}{K_\iota} \cup V(K_\iota \setminus K_{\iota+1})
\subseteq \reach{\ut}{K}$.
\end{claim}
\begin{proof}
The second claim of the lemma is straightforward, as, by Claim \ref{cl:emwc:Yiota},
no edge of $Y$
belongs to $K_\iota \setminus K_{\iota+1}$ nor does it have both endpoints in $\reach{\ut}{K_{\iota}}$.
For the first claim, assume the contrary, and let $P$
be a $\ut t$-path in $G \setminus K$ for some $t \in \terms\setminus\{\ut\}$.
As $K_{\iota+1}$ is a $(\ut,t)$-cut, $P$ contains an edge of $Z$.
However, by the definition of $Z$, $P$ contains an edge of $Y$, and
$P$ intersects $K$, a contradiction.
\cqed\end{proof}
Recall now that $K_{\iota+1}$ is a minimum $(\reach{\ut}{K_\iota} \cup \{u,v\},\terms \setminus \ut)$-cut for some $uv \in K_{\iota}$.
By Claim \ref{cl:emwc:Y-in-iota},
$K$ is also a $(\reach{\ut}{K_\iota} \cup \{u,v\},\terms \setminus \ut)$-cut.
Hence, $|K| \geq |K_{\iota+1}|$ and, consequently, $|Y| \geq |Z|$.

We are now ready to make the crucial observation.
\begin{claim}\label{cl:emwc:Zfix}
The set $X' := (X \setminus Y) \cup Z$ is a solution to \pemwc{} on $(G,\terms)$.
\end{claim}
\begin{proof}
Assume the contrary, and let $P$ be a path connecting two terminals in $G \setminus X'$.
We consider two cases, depending on whether there exists an endpoint of $P$ that
belongs to $\terms_Y$. If there exists such an endpoint, $Z$ should
substitute $Y$ as a separator and should intersect $P$. Otherwise,
$Y$ does not play any substantial role in intersecting $P$ as a part of the solution $X$,
and $X \setminus Y$ should already intersect $P$. We now proceed with formal
argumentation.

In the first case, assume that $t \in \terms_Y$ is an endpoint of $P$.
As $X$ is a solution to $(G,\terms)$, $P$ contains an edge of $Y$.
Let $uv$ be such an edge on $P$ that is closest to $t$, where $u$ lies before $v$ on $P$.
%on $P$ earlier than $v$.
Note that $P[t,u]$ does not contain any edge of $K_{\iota+1}$:
as $t \notin \reach{\ut}{Y}$ and $P[t,u] \cap Y = \emptyset$,
$P[t,u]$ is contained in $G \setminus \reach{\ut}{Y}$
but $Z = K_{\iota+1} \cap (G \setminus \reach{\ut}{Y})$ and $P$ avoids $Z$.
Recall that all endpoints of the edges of $Y$
lie in $\reach{\ut}{K_{\iota+1}}$; hence, there exists a path
$Q$ connecting $\ut$ with $u$ that avoids $K_{\iota+1}$. Hence,
$Q \cup P[u,t]$ is a $\ut t$-path avoiding $K_{\iota+1}$, a contradiction
to the definition of $K_{\iota+1}$.

In the second case, both endpoints of $P$ belong to $\reach{\ut}{Y}$.
Denote them $t_1$ and $t_2$. As $X$ is a solution to $(G,\terms)$,
$P$ contains at least one edge of $Y$. Let $e_1$ be the first such
edge, and let $e_2$ be the last one. Moreover, let $v_1$ be the endpoint
of $e_1$ closer to $t_1$ on $P$ and $v_2$ be the endpoint of $e_2$ closer
to $t_2$ on $P$. Note that $v_1,v_2 \in \reach{\ut}{Y}$, as both
$P[t_1,v_1]$ and $P[t_2,v_2]$ do not contain any edge of $Y$.
We also note that it may happen that $e_1=e_2=v_1v_2$,
but $v_1 \neq v_2$ and $v_1$ is closer to $t_1$ on $P$ than $t_2$. Observe that since $P[t_1,v_1]$ and $P[t_2,v_2]$ avoid both $X'$ and $Y$, they also avoid $X$.

As $Y$ is connected in the dual of $G$, there exists a unique face $f_Y$
of $(G \setminus Y)[\reach{\ut}{Y}]$,
that contains $Y$. As $\reach{\ut}{Y}$ is connected
by definition and the interior of each face of a connected graph is isomorphic to an open disc (since we are working on the euclidean plane),
the closed walk around $f_Y$ in $\reach{\ut}{Y}$ connects
all vertices incident to $Y$ that belong to $\reach{\ut}{Y}$ and,
by the definition of $\overline{Y}$, all edges of this closed walk belong
  to $\overline{Y} \setminus Y$. We infer that $v_1$ and $v_2$ lie in the
  same connected component of $\overline{Y} \setminus Y$.\footnote{%
Note that the argument of this paragraph fails if we assume only that $G$ is embedded on, say, a torus, instead of a plane. We do not know how to fix it for graphs of higher genera.}
% We definitely need to keep this sentence to use the plural form of 'genus'!!!

By the definition of $Y$ and $\overline{Y}$, we have $X \cap \overline{Y} = Y$.
Hence, $v_1$ and $v_2$ lie in the same connected component of $G \setminus X$
and the same holds for $t_1$ and $t_2$ (via paths $P[t_1,v_1]$ and $P[t_2,v_2]$),
a contradiction to the fact that $X$ is a solution to $(G,\terms)$.
This finishes the proof of Claim \ref{cl:emwc:Zfix}.
\cqed\end{proof}

Clearly, as $|Y| \geq |Z|$ and $Y \subseteq X$, we have $|X'| \leq |X|$.
As $Z \subseteq K_{\iota+1}$, we have $e \notin X'$.
Thus, by Claim \ref{cl:emwc:Zfix}, $X'$ is a minimum solution to \pemwc{} on
$(G,\terms)$ that does not contain $e$. This concludes the proof of the lemma.
\end{proof}

Lemma \ref{lem:emwc:Ki-dist}
allows us to state the following reduction rule.
\begin{reduction}\label{red:emwc:Ki-dist}
Compute a choice of cuts $K_1,K_2,\ldots,K_r$ for some arbitrarily chosen $\ut \in \terms$.
If there exists an edge $e$ in $G$
whose distance from $\bigcup_{i=1}^r K_i$ in the dual of $G$
is greater than $k$, contract $e$.
\end{reduction}
Note that Rule \ref{red:emwc:Ki-dist} may be applied in polynomial time.
Moreover, it bounds the diameter of the dual of $G$. To prove this claim, we need the following easy fact.
\begin{lemma}\label{lem:small-diam}
Let $H$ be a connected graph, and let $D\subseteq V(H)$ be a subset of vertices such that every vertex of $H$ is in distance at most $r$ from some element of $X$. Then the diameter of $H$ is bounded by $(2r+1)|D|-1$.
\end{lemma}
\begin{proof}
For a vertex $w\in V(H)$, let $\pi(w)$ be a vertex of $D$ closest to $v$, breaking ties arbitrarily. For sake of contradiction assume that there exist two vertices $u,v\in V(H)$ such that the shortest path $P$ in $H$ between $u$ and $v$ is of length at least $(2r+1)|D|$. Then $|V(P)|\geq (2r+1)|D|+1$, and by the pigeon-hole principle there must exist a vertex $x\in D$ such that $x=\pi(w)$ for at least $2r+2$ vertices of $V(P)$. Let $w_1$ be the first of these vertices and $w_2$ be the last; note that the distance between $w_1$ and $w_2$ on $P$ is at least $2r+1$, since there are at least $2r$ vertices on $P$ between them. Now obtain a walk $P'$ by removing $P[w_1,w_2]$ from $P$, and inserting first a shortest path from $w_1$ to $x$ and then a shortest path from $x$ to $w_2$. By assumption, both these paths are of length at most $r$, so $P'$ is shorter than $P$. This contradicts the minimality of $P$.
\end{proof}

We are ready to give a bound on the diameter of the dual of $G$.

\begin{lemma}\label{lem:emwc:dual-bound}
If Rules \ref{red:emwc:first}--\ref{red:emwc:Ki-dist} are not applicable,
then the diameter of the dual of $G$
is $\Oh(k_{OPT}^3)$.
\end{lemma}
\begin{proof}
By Lemma~\ref{lem:small-diam}, since the dual of $G$ is connected, it suffices to identify a set $D$ of $\Oh(k_{OPT}^2)$ vertices of $G$ such that every vertex of $G$ is in distance at most $k+1$ from $D$. We claim that $D=V(\bigcup_{i=1}^r K_i)$ is such a set. By Lemma~\ref{lem:emwc:Ki-props} we have that $|D|\leq \Oh(k_{OPT}^2)$. Take now any vertex $v\in V(G)$ and, since Rules~\ref{red:emwc:empty-cc} and~\ref{red:emwc:cc} are not applicable, let $e$ be an arbitrary edge incident to $v$. Since Rule \ref{red:emwc:Ki-dist} is not applicable, $e$ is in distance at most $k$ from $D$, so also $v$ is in distance at most $k+1$ from $D$.
\end{proof}

%For any $e \in G$, let $f(e)$ be an element of $\bigcup_{i=1}^r K_i$
%that minimizes the distance between $f(e)$ and $e$ in the dual of $G$.
%As Rule \ref{red:emwc:Ki-dist} is not applicable, this distance is at most $k$.
%Hence, if $f(e) = f(e')$, then the distance in the dual between $e$ and $e'$
%is at most $2k+1$ (the extra one comes from the fact that we measure distance
%between edges).
%Consequently, in any shortest path in the dual of $G$, two edges $e$ and $e'$
%with $f(e) = f(e')$ lie within distance at most $2k+1$. We infer that
%the diameter of the dual of $G$ is bounded
%by $\Oh(k\sum_{i=1}^r |K_i|) = \Oh(k^2r) = \Oh(k_{OPT}^3)$.

\subsubsection{Cutting the dual open and applying Theorem \ref{thm:main}}\label{sec:emwc:finish}

We now proceed to the application of Theorem \ref{thm:main}.
We start with the following observation.

\begin{lemma}\label{lem:emwc:cacti}
% Yeah! Another cool plural!
If Rules \ref{red:emwc:cc} and \ref{red:emwc:2terms} are not applicable,
then each $2$-connected component of $\bigcup_{t \in \terms} \delta(t)$ in the dual of $G$
is a cycle. That is, $\bigcup_{t \in \terms} \delta(t)$ is a set of cacti in the dual of $G$.
\end{lemma}
\begin{proof}
Let $H_0=\bigcup_{t \in \terms} \delta(t)$ be a subgraph of the dual of $G$. First, note that if Rule \ref{red:emwc:2terms} is not applicable, then
$\delta(t)$, for $t \in \terms$, are edge-disjoint cycles in $G^\ast$. We claim that these cycles are precisely $2$-connected components of $H_0$. For the sake of contradiction, assume that there exists a simple cycle $C$ in $H_0$ that contains edges from cycles $\delta(t_1),\delta(t_2),\ldots,\delta(t_p)$, where $p\geq 2$. Since $C$ is simple, we can assume that for each $i$, there exists an edge of $\delta(t_i)$ not contained in $C$. Let $\gamma$ be the curve on the plane corresponding to cycle $C$. Observe that edges of $G$ crossing $\gamma$ are precisely the primal edges of $C$. Take $t_1$ and observe that in $G$ there is an edge incident to $t_1$ crossing $\gamma$, and there is an edge incident to $t_1$ not crossing $\gamma$. Since Rule \ref{red:emwc:2terms} is not applicable, we conclude that there exist nonterminal vertices on both sides of the curve $\gamma$. As each edge of $H_0$ is incident to a terminal, removing $\terms$ from $G$ disconnects nonterminal vertices on different sides of $\gamma$, and Rule~\ref{red:emwc:cc} would be applicable. This is a contradiction.
\end{proof}

%Moreover, each edge of $\bigcup_{t \in \terms} \delta(t)$ connects $\terms$ with $G \setminus \terms$.
%As $G \setminus \terms$ is connected (Rule \ref{red:emwc:cc} is not applicable),
%in the dual of $G$, each edge of $\bigcup_{t \in \terms} \delta(t)$ is incident to the face
%that contains the drawing of the (primal) graph $G \setminus T$.
%Hence, $\delta(t)$ for $t \in \terms$ are exactly the doubly connected components of $\bigcup_{t \in \terms} \delta(t)$ in the dual of $G$.

We now construct two subgraphs $H_0$ and $H_s$ of the dual of $G$.
Let $H_0 = \bigcup_{t \in \terms} \delta(t)$. We note that, by Lemma \ref{lem:emwc:cacti},
for each connected component $C$ of $H_0$, the closed walk around the outer face of $C$ is an Eulerian tour of $C$ -- as shown on Figure~\ref{fig:cacti}~(a).

\begin{figure}[h]
\centering
\svg{0.8\textwidth}{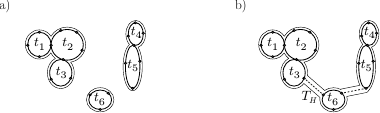}
\caption{Panel a) shows the set of cacti together with their Eulerian tours, i.e., the graph $H_0$. Panel b) shows the construction of graph $H_s$, 
where the three Eulerian tours are jointed together using two copies of paths $P$ and $P'$.}
\label{fig:cacti}
\end{figure}

We now construct a connected subgraph $H_s$ of the dual that contains as subgraph $H_0$. We first contract all connected components of $H_0$ to 
vertices, and find a minimum spanning tree $T_H$ over these vertices (i.e., a $2$-approximate Steiner tree). We set $H_s = H_0 \cup T_H$. Observe that 
$|H_0| = \sum_{t \in \terms} |\delta(t)| = k \leq 2k_{OPT}$. Moreover, by Lemma \ref{lem:emwc:dual-bound} the distance
between any two terminals in the dual is bounded by $\Oh(k_{OPT}^3)$, so the cost of the MST $T_H$ is bounded by $\Oh(k_{OPT}^4)$. 
We infer that $|H_s| = \Oh(k_{OPT}^4)$.

Now consider a multigraph $H_{2s}$ obtained by taking a union of $H_0$ and two copies of $T_H$. 
We observe that $H_{2s}$ is Eulerian, and let $W$ be its Eulerian tour. 
Note that $W$ is a closed walk around the outer face of $H_s$ and each edge
of $H_0$ appears exactly once on $W$ and each edge of $H_s \setminus H_0=T_H$ appears exactly twice on $W$.
Hence, $|W| = \Oh(k_{OPT}^4)$.
We cut the dual of $G$ open along $W$. That is, we start with $G^\ast$, the dual of $G$,
we duplicate each edge of $H_s \setminus H_0$ and, for each
vertex $v \in V(H_s)$, we create a of copies of $v$ equal to the number of appearances of $v$ on $W$.
Let $\nG^\ast$ be the graph obtained in this way. In $\nG^\ast$ the walk $W$ becomes a simple cycle,
enclosing a face $f_W$. We fix an embedding of $\nG^\ast$ where $f_W$ is the outer face.
In this way $\nG^\ast$ is a brick with perimeter of length $\Oh(k_{OPT}^4)$.
Let $\pi$ be a mapping that assigns to each edge of $\nG^\ast$ its corresponding edge of $G$ and $G^\ast$.

We apply Theorem \ref{thm:main} to the brick $\nG^\ast$, obtaining a set $F'$
of size $\Oh(k_{OPT}^{568})$. The set $F'$ naturally projects to a set $F \subseteq E(G)$ via the mapping $\pi$.
We claim that we may return the set $F$ in our algorithm. That is, to finish the proof of Theorem \ref{thm:emwc}
we prove the following lemma.
\begin{lemma}\label{lem:emwc:finish}
There exists a minimum solution $X$ to \pemwc{} on $(G,\terms)$ that is contained in $F$.
\end{lemma}
\begin{proof}
Let $X$ be a solution to \pemwc{} on $(G,\terms)$ that minimizes $|X \setminus F|$.
By contradiction, assume $X \setminus F \neq \emptyset$.

We define the following binary relation $\rel$ on $X$: $\rel(e,e')$
if and only if there exists a walk in $G^\ast$ containing $e$ and $e'$, with all edges in $X$ and all internal vertices
not in $V(H_s)$. Clearly, $\rel$ is symmetric and reflexive. We show that it is also transitive.
Assume $\rel(e,e')$ and $\rel(e',e'')$, with witnessing paths $P$ and $P'$.
If $e=e'$ or $e'=e''$, the claim is obvious, so assume otherwise.
We may assume that $P$ starts with $e$ and ends with $e'$ and $P'$ starts with $e'$ and ends with $e''$.
If $P$ and $P'$ traverse $e'$ in the same direction then $P \cup P'$ is a witness to $\rel(e,e'')$, as $P$ and $P'$
are of length at least two.
In the other case, $(P \setminus e') \cup (P' \setminus e')$ is a witness to $\rel(e,e'')$.
Thus, $\rel$ is an equivalence relation.

Note that any edge of $X \cap H_s$ is in a singleton equivalence class of $\rel$.
Let $Y$ be the equivalence class of $\rel$ that contains an element of $X \setminus F$. As $H_s \subseteq F$, we infer that $Y \cap H_s = \emptyset$
and, consequently, $Y$ is also a subgraph (subset of edges) of $\nG^\ast$.
Let $\dupaterms = V(\prm \nG^\ast) \cap V(Y)$ in $\nG^\ast$. We note that $Y$ is a connected subgraph of $\nG^\ast$
that connects $\dupaterms$ -- see Figure~\ref{fig:connector-pemwc}.

\begin{figure}[h]
\centering
\svg{0.8\textwidth}{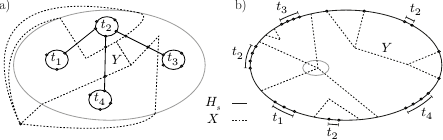}
\caption{The figure shows the solution $X$ to \pemwc{} and set $Y$ in a) the dual graph $G^\ast$ and b) the cut open dual graph $\nG^\ast$.}
\label{fig:connector-pemwc}
\end{figure}

By the properties of $F'$, there exists a set $Z' \subseteq F'$ that connects $\dupaterms$ in $\nG^\ast$
and $|Z'| \leq |Y|$. Let $Z = \pi(Z') \subseteq F$. We claim that $X' := (X \setminus Y) \cup Z$ is a solution
to $(G,\terms)$ as well. This would contradict the choice of $X$, as $|X'| \leq |X|$ and $|X' \setminus F| < |X \setminus F|$.

So assume the contrary, and let $P$ be a path connecting two terminals $t^1$ and $t^2$ in $G \setminus X'$. We may assume
that $P$ does not contain any terminal as an internal vertex.
Note that $P$ starts and ends with an edge of $H_0 \subseteq H_s$. As Rule \ref{red:emwc:2terms} is not applicable,
$P$ is of length at least two. Let $e_0,e_1,e_2,\ldots,e_d$ be the edges of $P \cap H_s$, in the order of their appearance
on $P$, let $e_i = u_iv_i$, where $u_i$ lies closer on $P$ to $t^1$ than $v_i$ does.
Note that $e_0,e_d \in H_0$ but $e_i \in H_s \setminus H_0$ for $1 \leq i < d$.
Let $\nG^{\ast\ast}$ be the dual of $\nG^\ast$. For each $i = 1,2,\ldots,d$, we define a cycle $Q_i$ in $\nG^{\ast\ast}$ as follows. Consider first path $P[u_{i-1},v_i]$, and observe that every edge of this path apart from the first and the last is present in $\nG^{\ast\ast}$. Therefore, in $P[u_{i-1},v_i]$ replace the edge $e_{i-1}$ (belonging to $H_s$) with the copy of $e_{i-1}$ in $\nG^{\ast\ast}$
that leads from the outer face of $\nG^\ast$ to the face $v_{i-1}$, and replace the edge $e_i$ with a copy
of $e_i$ in $\nG^{\ast\ast}$ that leads from $u_i$ to the outer face of $\nG^\ast$.
Although $Q_i$ is a cycle in $\nG^{\ast\ast}$, we call the aforementioned copy of $e_{i-1}$ {\em{the first arc}} of $Q_i$, and the copy of $e_i$ {\em{the last arc}}.

The set $\dupaterms$ splits $\prm \nG^\ast$ into a number of arcs $A_1,A_2,\ldots,A_{\max(1,|\dupaterms|)}$. If, for some $1 \leq i \leq d$, the first and the last
edge of the cycle $Q_i$ lies in different arcs $A_\alpha$ and $A_\beta$, then $Q_i$ intersects $Z'$, and, consequently,
$P$ intersects $Z$, a contradiction to the choice of $P$ -- see Figure~\ref{fig:path-pemwc}~(a). Hence, for all $1 \leq i \leq d$, the first and the last arc of $Q_i$
lies in the same arc $A_{\alpha(i)}$. We now reach a contradiction by showing that $t^1$ and $t^2$ lie in the same connected component
of $G \setminus X$.

\begin{figure}[h]
\centering
\svg{0.8\textwidth}{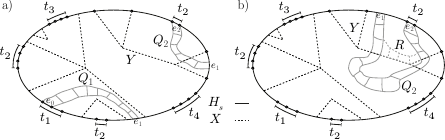}
\caption{The path $P$ in $G$ can be seen as a sequence of faces in the dual. On panel (a) the 
last and the first edge of $Q_2$ lie in different arcs, whereas on (b) these edge belong to the same arc.}
\label{fig:path-pemwc}
\end{figure}

As $P$ avoids $X'$ and $Y \cap H_s = \emptyset$, $P$ avoids $X \cap H_s$ and $e_0,e_1,\ldots,e_d \notin X$. Let $i$ be the smallest integer such that $e_i$ does not lie in the same
connected component of $G \setminus X$ as $t^1$. If such $i$ does not exist, the claim is proven as $t^2$ is an endpoint of $e_d$.
Consider $P[v_{i-1},u_i]$; note that this is also a subpath of $Q_i$, as it does not contain any edge of $H_s$.
Recall that $P$ avoids $X \setminus Y$. Hence, $P[v_{i-1},u_i]$ intersects $Y$. Moreover, 
the first and the last edge of $P[v_{i-1},u_i]$ lies on the same arc $A_{\alpha(i)}$, so $P[v_{i-1},u_i]$ intersects $Y$ at least twice. 
We treat now $P[v_{i-1},u_i]$ as a subpath of $Q_i$, i.e., a path in $\nG^{\ast\ast}$.
Let $f_1$ be the first face of $\nG^\ast$ on $P[v_{i-1},u_i]$ that is incident to an edge of $Y$, and let $f_2$ be the last such face. Observe that the prefix of $P[v_{i-1},u_i]$ up to $f_1$ and the suffix of $P[v_{i-1},u_i]$ from $f_2$ avoid both $X'$ and $Y$, so they also avoid $X$.

Let us now show that there exists a path $R$ in $\nG^{\ast\ast}$ connecting $f_1$ and $f_2$
that uses only edges that in $\nG^\ast$ are incident to the endpoints of $Y$, but do not belong to $Y$ nor $\prm \nG^\ast$. Existence of path $R$ can be inferred as follows. Take the set of faces $F_{\alpha(i)}$ of $\nG^\ast$ that are reachable in $\nG^{\ast\ast}$ from edges of the arc $A_{\alpha(i)}$ without passing through the infinite face of $\nG^\ast$, or traversing edges of $Y$. Consider also $F_{\alpha(i)}$ as a subset of plane obtained by gluing these faces together along all the edges between them that are not contained in $Y$. By the definition of $Y$ as an equivalence class of $\rel$, the boundary of $F_{\alpha(i)}$ is a closed walk that consists of arc $A_{\alpha(i)}$ and edges of $Y$ that are incident to faces of $F_{\alpha(i)}$. By the definition, both of $f_1$ and $f_2$ are incident to the part of boundary of $F_{\alpha(i)}$ that is contained in $Y$. Path $R$ can be then obtained by traversing faces of $F_{\alpha(i)}$ along its boundary, choosing the direction of the traversal so that part of the boundary of $F_{\alpha(i)}$ that is the arc $A_{\alpha(i)}$ is not traversed -- see Figure~\ref{fig:path-pemwc}~(b).

Since $Y$ is an equivalence class of $\rel$, edges of $R$ do not belong to $X$ (as otherwise they would be in relation with the edges of $Y$). Let $R'$ be $P[v_{i-1},u_i]$ with subpath between faces $f_1$ and $f_2$ replaced with $R$.
If we now project $R'$ to $G^\ast$ and $G$ using $\pi$, we infer that $v_{i-1}$ and $u_i$ lie in the same connected
component of $G \setminus X$, a contradiction to the choice of $i$. This finishes the proof of the lemma, and concludes the proof of Theorem \ref{thm:emwc}.
\end{proof}

%!TEX root = pst-kernel.tex

\section{Extending to bounded-genus graphs}\label{sec:genus}

In this section we extend the results from planar graphs to bounded genus graphs, using the framework of Borradaile et al.~\cite{cora:genus}.
The idea is to reduce the bounded genus case to the planar case by cutting the graph embedded on a surface of bounded genus into a planar graph, using only a cutset of small size.

As in~\cite{cora:genus}, we assume that we are given a combinatorial embedding of genus $g$ of an input graph $G$, where the interior of each face is homeomorphic to an open disc.
We proceed as in Sections~4.1 and~4.2 of~\cite{cora:genus}: given a brick embedded on a surface of genus $g$ (i.e., a graph with a designated face),
we may cut along a number of ``short'' cutpaths to make the brick planar. More precisely, the following theorem summarizes the results of~\cite{cora:genus} in our terminology, in particular the proved guarantees about the behaviour of procedures \verb-Preprocess- and \verb-Planarize- in~\cite{cora:genus}.

\begin{theorem}[\cite{cora:genus}, with adjusted terminology and parameter $\mu$ set to $1$]\label{thm:cora-summary}
Let $G$ be a connected graph embedded into a surface of genus $g$, and let $\terms\subseteq V(G)$ be a set of terminals in $G$. Let $OPT$ be the weight of an optimum Steiner tree connecting $\terms$ in $G$. Then one can in $\Oh(|G|)$ time find subgraphs $CG$ and $G'$ of $G$ such that the following holds:
\begin{itemize}
\item $CG\subseteq G'\subseteq G$, $CG$ and $G'$ are connected, and $CG$ contains all the terminals of $\terms$;
\item all the vertices and edges of $G'$ are at distance at most $4OPT$ from $\terms$ in $G'$, and $G'$ contains all the vertices and edges of $G$ that are at distance at most $2OPT$ from $\terms$ in $G$;
\item cutting $G'$ along $CG$ results in a planar graph $G_p$ with the infinite face (corresponding to cut-open $CG$) being a simple cycle of length at most $8(2g+2)OPT$.
\end{itemize}
\end{theorem}

Let us remark that a combinatorial embedding of $G'$ can be easily derived from a combinatorial embedding of $G$ by removing all the vertices and edges not present in $G'$, and replacing each new face whose interior ceased to be homeomorphic to an open disc with a number of disc faces.

By combining Theorem~\ref{thm:cora-summary} with Theorem~\ref{thm:main} we obtain the following.

\begin{theorem}[Main Theorem for graphs of bounded genus]\label{thm:main-genus}
Let $B$ be a connected graph, with a combinatorial embedding into a surface of genus $g$. Let $f$ be a simple face of $B$. Then one can find in $\Oh(|\prm f|^{142}\cdot (g+1)^{142}\cdot |B|)$ time a subgraph $H\subseteq B$ such that
\begin{itemize}
\item[(i)] $\prm f \subseteq H$,
\item[(ii)] $|E(H)| = \Oh(|\prm f|^{142}\cdot (g+1)^{142})$, and
\item[(iii)] for every set $\terms \subseteq V(\prm f)$, $H$ contains some optimal Steiner tree in $B$ connecting $\terms$.
\end{itemize}
\end{theorem}
\begin{proof}
Let $\terms_0=V(\prm f)$. Observe that if $OPT$ is the optimum weight of a Steiner tree connecting $\terms_0$ in $B$, then $OPT\leq |\prm f|$. We apply the algorithm of Theorem~\ref{thm:cora-summary} to $B$, obtaining graphs $B'$ and $CB$ with the promised guarantees. Note that if $B_p$ is the planar brick obtained from $B'$ by cutting open along $CB$, then $|\prm B_p|\leq 8(2g+2)|\prm f|$. The theorem now follows from an application of Theorem~\ref{thm:main} to the brick $B_p$, and projecting the obtained subgraph $H_p\subseteq B_p$ back to $B'$. Note here that no edge of $B$ that is not present in $B'$ can participate in any optimum Steiner tree connecting any subset of $\terms_0$.
\end{proof}

Using Theorem~\ref{thm:main-genus} instead of Theorem~\ref{thm:main}, 
we immediately obtain bounded-genus variants 
of Theorems~\ref{thm:pst} and \ref{thm:psf}.

\begin{theorem}\label{thm:pst-genus}
Given a \textsc{Steiner Tree} instance $(G,\terms)$
together with an embedding of $G$ into a surface of genus $g$
where the interior of each face is homeomorphic to an open disc,
one can in 
$\Oh(k_{OPT}^{142} (g+1)^{142} |G|)$ time find a set $F \subseteq E(G)$
of $\Oh(k_{OPT}^{142} (g+1)^{142})$ edges that contains an optimal Steiner tree
connecting $\terms$ in $G$, where $k_{OPT}$ is the size of an optimal
Steiner tree.
\end{theorem}

\begin{theorem}\label{thm:psf-genus}
Given a \textsc{Steiner Forest} instance $(G,\termpairs)$
together with an embedding of $G$ into a surface of genus $g$
where the interior of each face is homeomorphic to an open disc,
one can in 
$\Oh(k_{OPT}^{710} (g+1)^{710} |G|)$ time find a set $F \subseteq E(G)$
of $\Oh(k_{OPT}^{710} (g+1)^{710})$ edges that contains an optimal Steiner forest
connecting $\termpairs$ in $G$, where $k_{OPT}$ is the size of an optimal
Steiner forest.
\end{theorem}

We note that the arguments of Section~\ref{sec:emwc} for \pemwcname{}
heavily rely on the planarity of the input graph,
and the question of a polynomial kernel for \textsc{Multiway Cut}
on graphs of bounded genus remains open.

We can plug the kernel given by Theorem~\ref{thm:pst-genus} directly into the algorithm of Tazari~\cite{tazari:mfcs10} for \problemST{} on graphs of bounded genus to obtain the following result:

\begin{corollary} \label{cor:pst-genus-subexp}
Given a graph $G$ with an embedding into a surface of genus $g$
where the interior of each face is homeomorphic to an open disc, a terminal set $\terms \subseteq V(G)$, and an integer $k$,
one can in $2^{\Oh_{g}(\sqrt{k \log k})} + \Oh(k_{OPT}^{142} (g+1)^{142} |G|)$ time decide whether the \pST{} instance $(G,\terms)$ has a solution with at most $k$ edges.
\end{corollary}

In this corollary, the hidden constant in $\Oh_{g}(\cdot)$ is some computable function of $g$.

%!TEX root = pst-kernel.tex

\section{\pemwcname{}: Subexponential-Time Algorithm}\label{sec:emwc-subexp}

In this section we show that the approach of Tazari for \problemST{}~\cite{tazari:mfcs10}
can be extended to \textsc{Edge Multiway Cut}.
\begin{theorem}\label{thm:mwc-subexp}
Given a planar graph $G$, a terminal set $\terms \subseteq V(G)$, and an integer $k$, one can in $|G|^{\Oh(\sqrt{k})}$ time decide whether the \pemwcname{} instance $(G,\terms)$ has a solution with at most $k$ edges.
\end{theorem}
\begin{proof}
First, assume that $(G,\terms,k)$ is a YES-instance and let $X$ be an arbitrary minimum solution.
We follow Baker's approach in $G^\ast$, the dual of $G$.
Let $f$ be an arbitrary vertex of $G^\ast$.
Perform breadth-first search in $G^\ast$, starting from $f$, and let
$E_j$, $j=0,1,2,\ldots$
be the set of edges of $G^\ast$ that connect the vertices of
distance $j$ from $f$ with vertices of distance $(j+1)$.
Note that the sets $E_j$ are pairwise disjoint, but $\bigcup_j E_j$ may be a proper
subset of $E(G^\ast)$.
Denote $\ell = \lceil \sqrt{k} \rceil$.
For $0 \leq i < \ell$, let $L_i = \bigcup_{j \geq 0} E_{i+j\ell}$.
Branch into $\ell$ subcases, guessing an index
$0 \leq i < \ell$ where $|X \cap L_i| \leq \sqrt{k}$.
Furthermore, branch into $(\lfloor \sqrt{k} \rfloor + 1)$ subcases
guessing $|X \cap L_i|$ and branch into at most $|V(G)|^{\lfloor \sqrt{k} \rfloor}$
subcases guessing the set $X \cap L_i$ itself. 
Label each branch with a pair $(i,Y)$: the index of the layer $L_i$
and the set $Y \subseteq L_i$ guessed (that is supposed to be $X \cap L_i$).
Contract the edges of $L_i \setminus Y$ in the graph $G$ (keeping multiple edges).
Let $H$ be the obtained graph.

We claim that after this operation the treewidth of $H$ is bounded by $\Oh(\sqrt{k})$.
By~\cite{tw-dual}, it suffices to bound the treewidth of $H^\ast$, the dual of $H$.
Recall that a contraction of an edge in a planar graph corresponds to a deletion of this edge
in the dual. Hence, $H^\ast$ is isomorphic to $G^\ast \setminus (L_i \setminus Y)$.
However, each connected component of $G^\ast \setminus L_i$ is $\ell$-outerplanar,
and $|Y| \leq \sqrt{k}$. This finishes the proof of the treewidth bound of $H^\ast$
and, consequently, of $H$.

To finish the proof of the theorem it suffices to note that a given \textsc{Multiway Cut}
instance $(G,\terms)$, equipped with a tree decomposition of $G$ of width $t$, one can decide whether this instance has a solution of size at most $k$ in $(|S|t)^{\Oh(t)} \mathrm{poly}(|G|)$ time by a straightforward dynamic-programing routine\footnote{We observe that this straightforward algorithm can be easily improved to a $t^{\Oh(t)}\mathrm{poly}(|G|)$-time algorithm, since for a connected component intersecting the bag we do not need to remember precisely which terminal is contained in it, but only whether such a terminal exists or not. This running time can be further refined to $2^{\Oh(t)}\mathrm{poly}(|G|)$ using the framework of sphere-cut decompositions and Catalan structures~\cite{DornFT12}.}.
Indeed, suppose we consider a bag $B$ in the tree decomposition
and we define $A \subseteq V(G)$ to be union of bags in the subtree rooted at $B$ (including $B$ itself). Then
in a state of the dynamic-programing algorithm we need to
remember the following information ($F$ is a solution that conforms to the state): for each vertex $z \in B$, which terminal lies in the same connected component
of $G[A]\setminus F$ as the vertex $z$, and how the vertices of $B$ are partitioned by the connected components
of $G[A]\setminus F$. 
Since $|\terms| \leq |G|$ and $t \leq |G|$, this implies a $|G|^{\Oh(t)}$ algorithm. In our case $t$ is $\Oh(\sqrt{k})$, which implies the theorem.
\end{proof}
By pipelining the kernelization algorithm of Theorem~\ref{thm:emwc-intro}
with Theorem~\ref{thm:mwc-subexp} we obtain the second claim
of Corollary~\ref{cor:subexp}.

%!TEX root = pst-kernel.tex

\section{\pSF{}: No Subexponential-Time Algorithm}\label{sec:sf-lb}
In this section, we prove Theorem~\ref{thm:psf-eth}, which states that no algorithm can decide in $2^{o(k)} \poly(|G|)$ time whether \pSF{} instances $(G,\mc{\terms})$ have a solution with at most $k$ edges, unless the Exponential Time Hypothesis fails. The Exponential Time Hypothesis was proposed by Impagliazzo, Paturi, and Zane~\cite{eth}. Using the formulation by Fomin and Kratsch~\cite{fomin:book}, it hypothesizes that no algorithm can decide instances of \pSAT{} in $2^{o(n)}$ time, where $n$ is the number of variables in the formula of the instance. Using the Sparsification Lemma~\cite{eth}, this is equivalent (see~\cite{fomin:book}) to the hypothesis that no algorithm can decide instances of \pSAT{} in $2^{o(m)}$ time, where $m$ is the number of clauses in the formula of the instance. It is this formulation of the Exponential Time Hypothesis that we rely on here.

To prove Theorem~\ref{thm:psf-eth}, we need a reduction from \pSAT{} to \pSF{}. We use the following intermediate problem, which was also considered by Bateni~\etal\cite{marx-bateni} in their NP-hardness reduction of \pSF{} on planar graphs of treewidth~$3$. Let the boolean relation $R(f,g,h)$ be equal to $(f = h) \vee (g = h)$. Then an \emph{$R$-formula} is a conjunction of relations $R(f,g,h)$, where each of $f,g,h$ can be a boolean variable, true ($1$), or false ($0$). For example, $R(x_{1},x_{2},x_{3}) \wedge R(x_{1},0,x_{2}) \wedge R(0,1,x_{3})$ is a valid $R$-formula. We explicitly mention here that it is critical that in $R(f,g,h)$ none of $f,g,h$ can be the negation of a boolean variable. Then one can define the following problem:\\

\defproblemnoparam{\rSAT}{An $R$-formula $\phi$.}{Decide whether $\phi$ is satisfiable.}\\

Bateni~\etal\cite{marx-bateni} essentially show the following result as part of their Theorem~8.2:

\begin{lemma}[\cite{marx-bateni}] \label{lem:psf-reduction}
Let $\phi$ be an $R$-formula on $n$ variables and $m$ clauses. Then in polynomial time one can construct an instance $(G_{\phi},\mc{\terms}_{\phi})$ of \pSF{} such that $G_{\phi}$ is a planar graph of treewidth $3$, and $(G_{\phi},\mc{\terms}_{\phi})$ has a solution with at most $n+3m$ edges if and only if $\phi$ is satisfiable.
\end{lemma}
We can use this lemma to prove the following result, which is stronger than Theorem~\ref{thm:psf-eth}, and thus implies it.

\begin{theorem}
If there is an algorithm that can decide in time $2^{o(k)} \poly(|G|)$ whether \pSF{} instances $(G,\mc{\terms})$, where $G$ has treewidth $3$, have a solution with at most $k$ edges, then the Exponential Time Hypothesis fails.
\end{theorem}
\begin{proof}
Consider an instance of \pSAT{} and let $\psi$ be the CNF-formula of this instance. Let $n$ denote the number of variables that appear in $\psi$ and let $m$ denote the number of clauses of $\psi$. Since each clause contains at most three variables, $m \geq n/3$ and thus $n \leq 3m$.

We first construct an $R$-formula $\phi$ that is equivalent to $\psi$.
For each variable $x_{i}$ ($i\in\{1,\ldots,n\}$) that appears in $\psi$, add the \emph{variable relations} $R(x_{i}^{+},x_{i}^{-},1)$ and $R(x_{i}^{+},x_{i}^{-},0)$ to $\phi$. Here $x_{i}^{+}$ and $x_{i}^{-}$ are new variables, which indicate whether $x_{i}$ will be true or false respectively. Note that the relations ensure that $T'(x_{i}^{+})\not= T'(x_{i}^{-})$ for any truth assignment $T'$ that satisfies both relations.
Now consider a clause $C_{j} = (a \vee b \vee c)$ of $\psi$ ($j\in\{1,\ldots,m\}$) --- if $C_{j}$ actually contains at most two literals, then we pretend that $c=0$; if $C_{j}$ contains one literal, then we also pretend that $b=0$. Define $a'$ as follows. If $a$ is a variable $x_{i}$, then let $a' = x_{i}^{+}$. If $a$ is the negation of a variable $x_{i}$, then let $a' = x_{i}^{-}$. Otherwise, \ie if $a = 0$ or $a=1$, then let $a'=a$. Define $b'$ and $c'$ similarly. Then, add to $\phi$ two new variables $y_j^+$ and $y_j^-$, and the following \emph{clause relations}: $R(a',b',y_{j}^{+})$, $R(0, c', y_{j}^{-})$, and $R(y_{j}^{+},y_{j}^{-},1)$. We claim that $\psi$ is satisfiable if and only if $\phi$ is satisfiable. 

Suppose that $\psi$ is satisfiable, and let $T$ be a satisfying truth assignment for $\psi$. We extend $T$ to also cover negations of variables, \ie $T(\lnot x_{i}) = \lnot T(x_{i})$. We construct a satisfying truth assignment $T'$ for $\phi$ as follows. If $T(x_{i}) = 1$, then let $T'(x_{i}^{+}) = 1$ and $T'(x_{i}^{-}) = 0$; otherwise, let $T'(x_{i}^{+}) = 0$ and $T'(x_{i}^{-}) = 1$. This satisfies all variable relations. Consider any clause $C_{j} = (a \vee b \vee c)$ of $\psi$. If $T(a) = 1$ or if $T(b)=1$, then set $T'(y_{j}^{+}) = 1$ and $T'(y_{j}^{-}) = 0$. Otherwise, \ie if $T(a)=0$ and $T(b) =0$, then $T(c) = 1$, and set $T'(y_{j}^{+}) = 0$ and $T'(y_{j}^{-}) = 1$. This satisfies all clause relations of $\phi$. Hence, $T'$ is a satisfying truth assignment for $\phi$.

Suppose that $\phi$ is satisfiable, and let $T'$ be a satisfying truth assignment for $\phi$. We construct a satisfying truth assignment $T$ for $\psi$ as follows: set $T(x_{i}) = T'(x_{i}^{+})$ for each variable in $\psi$. Again, we extend $T$ to also cover negations of variables, \ie $T(\lnot x_{i}) = \lnot T(x_{i})$. Consider any clause $C_{j} = (a \vee b \vee c)$ of $\psi$. If $T'(y_{j}^{-}) = 1$, then it follows from the clause relations that $T(c) = 1$. Otherwise, \ie if $T'(y_{j}^{-}) = 0$, then it follows from the clause relations that $T'(y_{j}^{+}) = 1$ and thus $T(a)=1$ or $T(b)=1$. Therefore, the clause is satisfied. Hence, $T$ is a satisfying truth assignment for $\psi$. This proves the claim.

Observe that $\phi$ has $2n+2m$ variables and $2n+3m$ relations. Moreover, $\phi$ can be constructed in polynomial time. Now apply the construction of Lemma~\ref{lem:psf-reduction} to $\phi$ in polynomial time. This yields an instance $(G_{\phi},\mc{\terms}_{\phi})$ of \pSF{} such that $G_{\phi}$ is a planar graph of treewidth $3$, and $(G_{\phi},\mc{\terms}_{\phi})$ has a solution with at most $8n+11m$ edges if and only if $\phi$ is satisfiable.
Using the above claim, $(G_{\phi},\mc{\terms}_{\phi})$ has a solution with at most $8n+11m$ edges if and only if $\psi$ is satisfiable. Note that $8n+11m \leq 35m$. Therefore, the existence of an algorithm as in the theorem statement would imply an algorithm that decides instances of \pSAT{} in $2^{o(m)}$ time. This proves the theorem.
\end{proof}

%moved to intro
%\input{conclusions}

\subsection*{Acknowledgements}

We thank Daniel Lokshtanov and Saket Saurabh for showing us
the application of Baker's approach to \textsc{Planar Multiway Cut} (Theorem~\ref{thm:mwc-subexp})
and for allowing us to include the proof in this paper.
Moreover, we acknowledge the discussions with Daniel Lokshtanov that lead to the discovery
that the NP-hardness proof for \textsc{Steiner Forest} on planar graphs of treewidth~$3$ of Bateni~\etal\cite{marx-bateni} can be strengthened to also refute a subexponential-time algorithm.

We would like also to acknowledge the support and extremely productive atmosphere
at Dagstuhl Seminars 13121, 13421 and 14071.
At the first one, major technical ideas of the proof of Theorem~\ref{thm:main} were developed.
At the second one, the fundaments of the weighted variant (Theorem~\ref{thm:weighted})
were laid, whereas many important details were discussed and straightened during the third one.

\bibliographystyle{plain}
\bibliography{pst-kernel}

\end{document}